\documentclass[11pt]{article}
\usepackage{amsmath}
\usepackage{tikz}

\usepackage{float}
\usepackage{amssymb}
\usepackage{amsfonts}
\usepackage{graphicx}
\usepackage{epstopdf}
\usepackage{color}
\usepackage{multirow}
\usepackage{algorithmic}
\usepackage[title]{appendix}
\usepackage{arydshln} 

\graphicspath{{figures/}}
\usepackage[labelformat=simple]{subcaption}

\ifpdf
  \DeclareGraphicsExtensions{.eps,.pdf,.png,.jpg}
\else
  \DeclareGraphicsExtensions{.eps}
\fi

\usepackage{xcolor, colortbl}
\definecolor{Gray}{gray}{0.90}
\definecolor{LightCyan}{rgb}{0.65,1,1}
\definecolor{Cyan}{rgb}{0.95,1,1}
\definecolor{GRAY}{gray}{0.75}

\usepackage[T1]{fontenc}
\usepackage[utf8]{inputenc}
\usepackage[font=small,labelfont=bf]{caption}

\usepackage[numbers,sort&compress]{natbib}

\newcommand{\bsub}{\begin{subequations}}
\newcommand{\esub}{\end{subequations}$\!$}

\renewcommand\theequation{\thesection.\arabic{equation}}

\renewcommand{\theequation}{\arabic{section}.\arabic{equation}}

\usepackage[numbers,sort&compress]{natbib}

\usepackage{tabularx}
\usepackage{parskip}
\usepackage{dirtytalk}
\usepackage{amsthm}
\usepackage{wrapfig}
\usepackage{mathtools}

\usepackage[format=plain,
            labelfont={it},
            textfont=it]{caption}

\newtheorem{remark}{Remark}
\newtheorem{prop}{Proposition}
\newtheorem{lemma}{Lemma}[section]

\newcommand{\R}{\mathbb{R}}

\renewcommand{\v}[1]{\mathbf{#1}}

\newcommand{\x}{\v{x}}

\newcommand{\y}{\v{y}}

\newcommand{\eps}{\varepsilon}

\newcommand{\lb}{{\pmb l}}

\title{Synchronized Memory-Dependent Intracellular Oscillations for a
  Cell-Bulk ODE-PDE Model in $\R^2$} \author{Merlin
  Pelz\thanks{Dept. of Mathematics, Univ. of British Columbia,
    Vancouver, B.C., Canada.}  \and Michael
  J. Ward\footnotemark[1]\,\, \thanks{corresponding author,
    \texttt{ward@math.ubc.ca}}}

\begin{document}

\maketitle

\begin{abstract}
  {For a cell-bulk ODE-PDE model in $\R^2$, a hybrid
    asymptotic-numerical theory is developed to provide a new
    theoretical and computationally efficient approach for studying
    how oscillatory dynamics associated with spatially segregated
    dynamically active “units” or “cells” are regulated by a PDE bulk
    diffusion field that is both produced and absorbed by the entire
    cell population. The study of oscillator synchronization in a PDE
    diffusion field was one of the initial aims of Yoshiki Kuramoto's
    foundational work.  For this cell-bulk model, strong localized
    perturbation theory, as extended to a time-dependent setting, is
    used to derive a new integro-differential ODE system that
    characterizes intracellular dynamics in a memory-dependent
    bulk-diffusion field. For this nonlocal reduced system, a novel
    fast time-marching scheme, relying in part on the {\em
      sum-of-exponentials method} to numerically treat convolution
    integrals, is developed to rapidly and accurately compute
    numerical solutions to the integro-differential system over long
    time intervals.  For the special case of Sel'kov reaction
    kinetics, a wide variety of large-scale oscillatory dynamical
    behavior including phase synchronization, mixed-mode oscillations,
    and quorum-sensing are illustrated for various ranges of the
    influx and efflux permeability parameters, the bulk degradation
    rate and bulk diffusivity, and the specific spatial configuration
    of cells. Results from our fast algorithm, obtained in under one
    minute of CPU time on a laptop, are benchmarked against PDE
    simulations of the cell-bulk model, which are performed with a
    commercial PDE solver, that have run-times that are orders of
    magnitude larger.}
\end{abstract}

\section{Introduction}\label{sec:intro}

{Over the past four decades, the Kuramoto model \cite{kuram_orig}
\begin{equation}\label{intro:model}
  \dot{\theta}_i = \omega_i + \frac{K}{N}\sum_{j=1}^N \sin(\theta_j -\theta_i)
  \,,
\end{equation}
and its variants have provided the primary theoretical modeling
framework for analyzing the synchronization properties of a collection
of $N$ nonlinear oscillators. In (\ref{intro:model}), $\theta_i$ is
the phase of the $i^\text{th}$ oscillator, $\omega_i$ is its natural
frequency, and $K$ is the coupling strength. Key quantities for
analyzing the synchronization behavior for (\ref{intro:model}) are
the Kuramoto order parameters $r(t)$ and $\psi(t)$ defined by
$r e^{i\psi} = N^{-1} \sum_{j=1}^N e^{i\theta_j}$. In the mean field
limit $N\to \infty$, phase transitions to complete phase
synchronization $r=1$ as $K$ is increased above a threshold have been
characterized. For a survey of results see \cite{strogatz2000},
\cite{acebron2005kuramoto} and \cite{da2018kuramoto}.

One key aspect of Kuramoto-type phase oscillator systems is that they
can be formally derived from nonlinear dynamical systems of coupled
oscillators by using asymptotic phase reduction methods that are valid
in the limit of weak coupling (cf.~\cite{nakao_phasered},
\cite{nakao_philtrans} and \cite{kuram_cont}).  An important current
research theme is the study of oscillator synchronization on networks,
as surveyed in \cite{sync1}, owing to their many diverse applications
such as power grids \cite{grzybowski2016} and cortical brain activity
\cite{breakspear2010generative}, among others.

With the aim of incorporating diffusive effects, in
\cite{kuramoto1995} Kuramoto emphasized the need for characterizing
synchronization behavior for a collection of coupled oscillators or
``cells'', each fixed in space, in which an extra substance secreted
from each cell can diffuse over the entire space and effectively
mediate the inter-cell interaction.  His phenomenological model in
\cite{kuramoto1995} of this interaction has the form
\begin{equation}\label{intro:kuram_pde}
  \v{u}_j^{\prime} = \v{F}(\v{u}_j) + \v{g}(A(\x_j,t)) \,, \quad
  \varepsilon A_t = D\Delta A - \eta A + \sum_{k=1}^{N} h(\v{u}_k)
  \delta(\x-\x_k) \,,
\end{equation}
where $\v{g}$ and $h$ are prescribed coupling functions.  For $\x\in \R^1$
and in the limit $\varepsilon\to 0$, (\ref{intro:kuram_pde}) reduces
to
\begin{equation}\label{intro:kuram_adia}
  \v{u}_j^{\prime} = \v{F}(\v{u}_j) + \v{g}\left(\sum_{k=1}^{N}
    \sigma(|\x_j-\x_k|) h(\v{u}_k)\right) \,.
\end{equation}
where $\sigma(r)=Ce^{-r\sqrt{\eta/D}}$ is the 1-D quasi-static Green's
function for some $C>0$.  Further approximating
(\ref{intro:kuram_adia}) near a Hopf bifurcation point of the
uncoupled dynamics $\v{u}_j^{\prime}=\v{F}(\v{u}_{j})$ by using a
weakly nonlinear analysis, the normal form of the underlying dynamics
leads to a nonlocally coupled Complex Ginzburg-Landau (CGL)
system. These nonlocal CGLs yield a wide variety of highly intricate
spatio-temporal dynamics, commonly known as {\em chemical turbulence}
(\cite{kuramoto_book}, \cite{kuramoto1995}, \cite{nakao_heat}). More
recently, in a series of articles (\cite{viana_heat}, \cite{viana} and
\cite{viana_2023}) there has been a renewed focus on analyzing
synchronization properties for phase oscillator systems with coupling
as in (\ref{intro:kuram_pde}), with the main emphasis being for the
$\varepsilon\to 0$ limit (\ref{intro:kuram_adia}).}

{Two primary limitations of the phenomenological model
  (\ref{intro:kuram_pde}) is that $A(\v{x},t)$ has a singularity at
  $\x=\x_j$ in $\R^2$ and $\R^3$, and that it is not clear how to
  choose the coupling functions $g$ and $h$ for a specific
  application. For many problems in biology and chemical physics, it
  is essential to {\em explicitly} model the diffusive coupling of
  dynamically active units, without incorporating extraneous coupling
  functions. In particular, for many microbial systems, cell-cell
  communication occurs through the diffusion of extracellular
  signaling molecules, referred to as autoinducers, that are both
  produced and absorbed by the entire collection of cells
  (cf.~\cite{dunny1997cell}, \cite{taga2003chemical},
  \cite{dano2}). For a colony of \textit{Dictyostelium discoideum}, it
  is known that the autoinducer cyclic adenosine monophosphate (cAMP)
  triggers intracellular oscillations that initiate the process of the
  spatial aggregation of the colony in low nutrient environments
  (cf.~\cite{gregor2010}, \cite{nandy1998},
  \cite{goldbeter1997biochemical}). In addition, the autoinducer
  acetaldehyde (Ace) leads to glycolytic oscillations in a colony of
  yeast cells (cf.~\cite{dano}, \cite{dano2}, \cite{de2007dynamical},
  \cite{hauser}), while acylated homoserine lactones (AHLs) are
  implicated in triggering bioluminescence behavior associated with
  the marine bacterium {\em Aliivibrio fischeri} that resides in the light
  organ of certain species of tropical squid
  (cf.~\cite{taga2003chemical}). For these microbial systems, one key
  aspect is to characterize quorum-sensing behavior, whereby
  collective dynamics can only occur if the cell population exceeds a
  threshold (cf.~\cite{schwab}, see \cite{perez2016} for a survey).
  In a chemical physics context, collective oscillatory dynamics of
  catalyst loaded pellets can occur owing to the diffusive coupling of
  Belousov-Zhabotinsky (BZ) reagents in a chemical mixture
  (cf. \cite{show_epstein}, \cite{taylor1}, \cite{taylor2},
  \cite{tinsley1}, \cite{tinsley2}). In addition, certain
  microemulsions of dynamically active surface-stabilized BZ droplets
  that are immersed in oil lead to synchronized oscillatory dynamics
  (cf.~\cite{tompkins}). More recently, chaotic oscillations have been
  observed for a compartmentalized surface reaction nanosystem
  \cite{raab}.}

{An ideal modeling framework to investigate collective dynamics
  arising from the coupling of a bulk diffusion field are the
  cell-bulk models originating from \cite{muller2006} and
  \cite{muller2013} in a 3-D setting and from \cite{gou2d} and
  \cite{smjw_diff} in 2-D domains. From a mathematical viewpoint,
  cell-bulk models are compartmental-reaction diffusion (RD) systems
  with rich dynamics that can be used for analyzing how intracellular
  oscillations associated with spatially segregated dynamically active
  ``units'' or ``cells'' are mediated by one or more diffusing
  substances that are both produced and absorbed by the entire cell
  population. This modeling framework is also well-aligned with
  Kuramoto's original aims of investigating oscillator synchronization
  through the effect of diffusive coupling as discussed in
  \cite{kuramoto1995}. The survey article \cite{show_epstein} has also
  emphasized the need for theoretical modeling frameworks that couple
  discretely interacting dynamical units.
  
  Previous studies of synchronization or pattern-forming properties of
  cell-bulk systems, i.e., systems that involve a cell membrane
  connected to a diffusion field (the bulk), include the 1-D analysis
  in \cite{gomezmarin} with oscillatory FitzHugh-Nagumo kinetics on
  the diffusively coupled boundaries and the bulk-membrane analysis of
  \cite{levine2005} in disk-shaped domains. In a 1-D context, and with
  one bulk diffusing species, this compartmental-RD system modeling
  paradigm has been shown to lead to triggered oscillatory
  instabilities for various reaction kinetics involving conditional
  oscillators (cf.~\cite{gou2015}, \cite{gou2016},
  \cite{gou2017}). Amplitude equations characterizing the local
  branching behavior for these triggered oscillations have been
  derived in \cite{paquin_1d} using a weakly nonlinear
  analysis. Extensions of this framework incorporating time-delay
  effects have been used to model intracellular polarization and
  oscillations in fission yeast (cf.~\cite{xu}). In a
  2-D domain, cell-bulk models with one diffusing bulk species have
  been formulated and used to model quorum-sensing behavior
  (cf.~\cite{gou2d}, \cite{smjw_diff}, \cite{ridgway},
  \cite{q_survey}). In a 2-D bounded domain with no-flux boundary
  conditions, and in the limit of large bulk diffusivity and small
  cell radii, the study of intracellular dynamics for cell-bulk models
  can be asymptotically reduced to the study of an ODE system with
  global coupling (cf.~\cite{smjw_diff}, \cite{smjw_quorum}), where
  the global mode arises from the approximately spatially uniform bulk
  diffusion field.  With two-bulk diffusing species, and as inspired
  by the trans-membrane signal transduction study in \cite{rauch2004},
  in 1-D \cite{turing1d} and in 2-D \cite{turing2d} domains it has
  been shown that a symmetry-breaking bifurcation leading to a
  linearly stable asymmetric pattern can occur for equal bulk
  diffusivities when the ratio of reaction rates across the cell
  boundaries for the two bulk species is sufficiently large.}

\begin{figure}[htbp]
  \centering
    \begin{subfigure}[b]{0.48\textwidth}  
      \includegraphics[width=\textwidth]{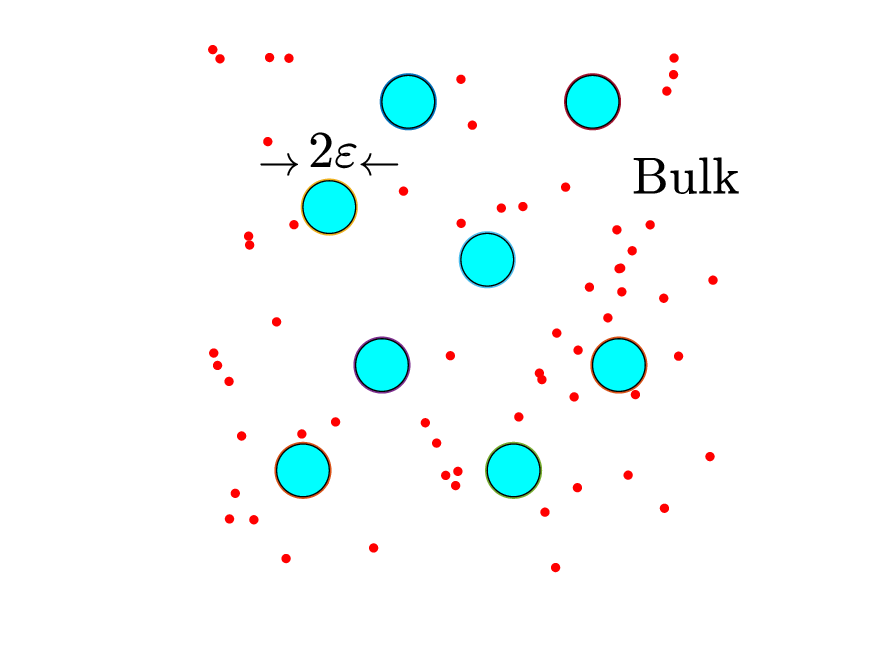}
        \caption{Cells coupled by bulk diffusion}
        \label{fig:schematic}
    \end{subfigure}
    \begin{subfigure}[b]{0.48\textwidth}
      \includegraphics[width=\textwidth]{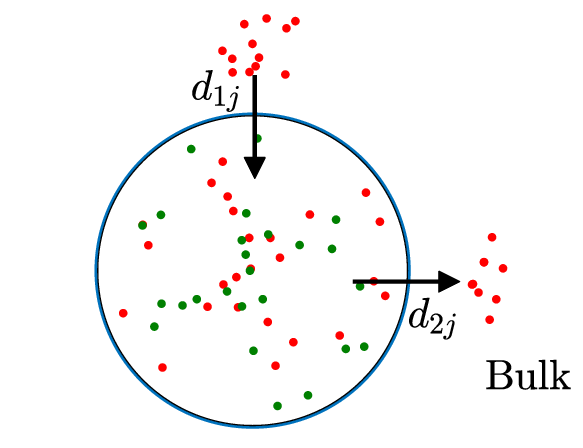}
        \caption{Exchange between a cell and the bulk} 
        \label{fig:schematic_zoom}
    \end{subfigure}
    \caption{Left: Dynamically active circular signaling compartments
      (in cyan) coupled through a bulk diffusion field (red dots) in
      $\R^2$. Right: Zoom of the chemical exchange between the bulk
      and any given cell when there are two intracellular species
      (green and red dots). Only the red species is secreted into the
      extracellular bulk region, while the bulk species can bind to
      the cell membrane. A Robin boundary condition models this
      chemical exchange.}
\label{fig:onebulk_scheme}
\end{figure}

{We now formulate the coupled dimensional cell-bulk ODE-PDE model of
\cite{gou2d} and \cite{smjw_diff}, as inspired by \cite{muller2006}
and \cite{muller2013}, with one bulk species in $\R^2$. We assume that
there are $N$ dynamically active cells, modeled by a collection of
disks of a common radius $R_0$, denoted by $\Omega_j$ and centered at
$\x_j\in \R^2$ for $j\in \lbrace{1,\ldots,N\rbrace}$. In the bulk, or
extracellular, region $\R^2\setminus \cup_{j=1}^{N}\,\Omega_j$, the
concentration $\mathcal{U}(\v{X},T)$ of the autoinducer or bulk
signal, with diffusivity $D_B$ and bulk degradation $k_B$, is assumed
to satisfy
\begin{subequations}\label{Dim_bulk} 
\begin{align}
   \partial_{T}\mathcal{U} = & D_B \, \Delta \,
     \mathcal{U} - k_B \, \mathcal{U}\,, \quad T>0\,,\quad 
    \v{X} \in \R^2\,\setminus \cup_{j=1}^{N}\,\Omega_j\,, \label{Dim_bulka}\\
    D_B \, \partial_{n_{\v{X}}} \, \mathcal{U}  &=
  \beta_{1j} \, \mathcal{U} - \beta_{2j} \, \mu_{1j}\,, \quad
    \x \in \partial \Omega_j\,,\qquad j \in \lbrace{1, \ldots, N\rbrace}\,.
                                                  \label{Dim_bulkb}
\end{align}
The permeabilities $\beta_{1j}>0$ and $\beta_{2j}>0$ control the
influx and efflux into and out of the $j^{\mbox{th}}$ cell,
respectively, while $\partial_{n_{\v{X}}}$ is the outer normal
derivative pointing into the bulk region.

Inside each cell we assume that there are $m$ species
$\v{\mu}_j \equiv (\mu_{1j},\ldots,\mu_{mj})^{T}$ that undergo interactions
with the local reaction kinetics $\v{F}_j$, and that one species,
labeled by $\mu_{1j}$, can permeate across the cell boundary. The
cells are assumed to be sufficiently small so that we can
neglect any spatial gradients in the intracellular species within each
cell. The intracellular dynamics in the $j^{\mbox{th}}$ cell are
coupled to the bulk diffusion field (\ref{Dim_bulk}) via an
integration of the diffusive flux across the cell membrane as
\begin{equation}\label{Dim_Intra}
\begin{split}
  \frac{d\v{\mu}_j}{dT} & = k_R\, \mu_c \, \v{F}_j
  \left( \v{\mu}_j/\mu_c  \right)  + \v{e}_1 \int_{\partial \Omega_j}\,
  \left( \beta_{1j} \, \mathcal{U} - \beta_{2j} \, \mu_{1j} \right) \,
  {\rm d}S_{\v{X}} \,, \quad j\in \lbrace{1,\ldots,N\rbrace}\,.
\end{split}
\end{equation}
\end{subequations}}
{In (\ref{Dim_Intra}),
$\v{e}_1 \equiv (1,0,\ldots,0)^{T}$, $k_R > 0$ is the dimensional
intracellular reaction rate and $\mu_c>0$ is a typical magnitude of
the vector $\v{\mu}_j$.  In this formulation, $\mu_{1j}$ can permeate
the cell membrane with an efflux rate $\beta_{2j}$ per unit
length. The influx rate $\beta_{1j}$ controls the feedback into the
$j^{\mbox{th}}$ cell from the global bulk diffusion field that is
produced by the entire collection of spatially segregated cells.

For our asymptotic limit, we will assume that the common radius $R_0$
of the signaling compartments is small relative to the minimum
inter-cell separation distance $L$, and so we introduce a small
parameter $\eps \equiv R_0/L \ll 1$. From the non-dimensionalization of
the cell-bulk model \eqref{Dim_bulk} given in Appendix \ref{app:nondim}
we obtain that the dimensionless bulk field $U(\v{x},t)$ satisfies
(see Fig.~\ref{fig:onebulk_scheme} for a schematic)
\begin{subequations}\label{DimLess_bulk}
\begin{align}
  \partial_{t} U =&\,  D \, \Delta U -  \sigma \, U\,, \quad
   t > 0\,,\quad \x \in \R^2 \setminus \cup_{j=1}^{N}\,
   \Omega_{\varepsilon_j}\,; \quad U(\v{x},0)=0 \,, \label{DimLess_bulka} \\
   \varepsilon  D\, \partial_n U  &= d_{1j} \, U - d_{2j} \, u_{1j}\,,
   \quad \x \in \partial \Omega_{\varepsilon_j}\,,\qquad j \in \lbrace{1,
                                    \ldots, N\rbrace} \,,
    \label{DimLess_bulkb}
\end{align}
which is coupled to the dimensionless intracellular dynamics within
the $j^{\mbox{th}}$ cell by
\begin{align}\label{DimLess_Intra}
  \frac{\text{d} \v{u}_j}{\text{d}t} & =  \,
 \v{F}_j \left( \v{u}_j  \right)  + \frac{\v{e}_1}{\eps}
  \int_{\partial \Omega_{\varepsilon_j}}\, ( d_{1j} \, U - d_{2j} \, u_{1j}) \,\,
              {\rm d}S_{\v{x}} \,,\qquad j \in \lbrace{1, \ldots, N\rbrace}\,.
\end{align}
\end{subequations}
Here $\v{u}_j=(u_{1j},\ldots,u_{mj})^{T}$ is the dimensionless vector
of intracellular species in the $j^{\mbox{th}}$ cell, labeled by
$\Omega_{\varepsilon_j}\equiv \lbrace{\v{x} \, \vert \, \,
  |\v{x}-\v{x}_j| \leq \varepsilon\rbrace}$. We assume that the
centers of the cells are well-separated in the sense that
$\mbox{dist}(\v{x}_j,\v{x}_k)={\mathcal O}(1)$ for $j\neq k$. In 
\eqref{DimLess_bulk}, the key dimensionless parameters are
\begin{equation}\label{dim:param}
  D \equiv \frac{D_B}{k_R L^2} \,, \quad
  d_{1j} \equiv \varepsilon \frac{\beta_{1j}}{k_R L} = {\mathcal O}(1) \,, \quad
  d_{2j} \equiv \varepsilon \frac{\beta_{2j} L}{k_R} = {\mathcal O}(1) \,, \quad
   \sigma \equiv \frac{k_B}{k_R} \,.
\end{equation}
Here $D$ and $\sigma$ are the effective bulk diffusivity and
bulk degradation rate. In \eqref{dim:param}, the ratios
${\beta_{1j}/(k_RL)}$ and ${\beta_{2j} L/k_R}$ are chosen as
$\mathcal{O}(\varepsilon^{-1})$ so that there is an
${\mathcal O}(1)$ transport across the cell membrane. 

Our goal for the conceptual cell-bulk model (\ref{DimLess_bulk}) is to
develop a hybrid asymptotic-numerical approach to study how intracellular
oscillations are both initiated and synchronized by the bulk
diffusion field that is created by the entire collection of cells. In
contrast to the study in \cite{smjw_diff} that focused only on the steady-state
problem and the spectral properties of the linearization of the steady-state,
the main goal in this paper is to use a hybrid approach to
study the {\em large-scale} intracellular dynamics for (\ref{DimLess_bulk})
that occur away from stable steady-states.

In \S \ref{2:ode_memory} we will extend the strong localized
perturbation theory, as surveyed in \cite{ward2018spots}, to the
time-dependent setting of (\ref{DimLess_bulk}) in order to derive, in
the limit $\varepsilon\to 0$ of small cell radii, a new
integro-differential ODE system that characterizes how intracellular
dynamics are coupled to the bulk diffusion field that is produced by
the entire cell population. This reduced integro-differential system,
as described in Proposition \ref{prop:2ode}, is asymptotically
accurate to all orders in $\nu={-1/\log\varepsilon}$. As discussed in
Remark \ref{rem:kuram} below, the ``outer problem'' in our analysis
shares some common features with Kuramoto's original model
(\ref{intro:kuram_pde}).

Assuming that $U(\x,0)=0$, in \S \ref{sec:transient} we analyze the
short-time behavior of the cell-bulk system on the time-scale
$\tau = \varepsilon^2 t$ within an ${\mathcal O}(\varepsilon)$
neighborhood of each cell. This transient solution, involving
Ramanujun's integral, is essential for providing the initial behavior
for the integro-differential ODE system in Proposition \ref{prop:2ode}
that is valid on the ${\mathcal O}(1)$ time-scale. In \S \ref{sec:ss},
we show that the steady-states of the cell-bulk model
(\ref{DimLess_bulk}), as obtained from strong localized perturbation
theory, coincide with the steady-state limiting values of the
integro-differential ODE system. In addition, the linear stability
properties of the steady-states are shown to be characterized by the
roots of a nonlinear matrix eigenvalue problem, referred to as the
globally coupled eigenvalue problem (GCEP).

We emphasize that a direct numerical study of the integro-differential
system in Proposition \ref{prop:2ode} is highly challenging owing to
the fact that ${N(N+1)/2}$ memory-dependent convolution integrals,
each with integrable singularities, would have to be computed
numerically up to time $t$ in order to advance the solution one time
step to time $t+\Delta t$. To overcome this difficulty, in \S
\ref{sec:numerics} we develop a fast time-marching numerical scheme,
which relies significantly on the {\em sum-of-exponentials} method of
\cite{greengard_2015} and \cite{beylkin_2005} (see also
\cite{beylkin_2010}) together with Duhamel's principle, that allows us
to both rapidly and accurately compute solutions to the
integro-differential system for arbitrary reaction kinetics over long
time intervals.

In \S \ref{sec:selkov} we illustrate our hybrid asymptotic-numerical
theory for the choice of Sel'kov reaction dynamics, which have been
used previously in simple models of glycolysis oscillations
(cf.~\cite{nandy1998}, \cite{Selkov}). With Sel'kov kinetics the
steady-state solution branches of (\ref{DimLess_bulk}) have no
transcritical or fold bifurcations and, as a result, steady-states are
destabilized only via Hopf bifurcations. For various specific spatial
configurations of cells, phase diagrams in the ${1/\sigma}$ versus $D$
parameter plane, as computed from the winding-number of the
determinant of the nonlinear matrix eigenvalue problem, are provided
to identify the number of destabilizing modes of the linearization of the
steady-state. The corresponding eigenvector of the GCEP is shown to
encode both the relative magnitude of the signaling gradient near each
cell, as well as the phase shift in small amplitude oscillations
between cells that occur near an unstable steady-state. Quorum-sensing
and diffusion-sensing behavior, as well as the key effects on
intracellular oscillations of varying the influx and efflux
permeabilities, are examined for various spatial configurations of
cells. In particular, we show that a single oscillating cell can
trigger coherent intracellular oscillations in an entire collection of
cells.  Moreover, by computing the Kuramoto order parameter, we show
an apparent phase transition to complete phase coherence as the bulk
diffusivity is increased for a one-shell hexagonal arrangement of
identical cells. These cells would be in a quiescent state without any
cell-bulk coupling. Finally, in \S \ref{sec:discussion} we 
discuss several related cell-bulk problems that are now tractable to
study with our hybrid approach.}

\section{Transient analysis near the cells}\label{sec:transient}

We first analyze the transient solution near the $j^{\mbox{th}}$ cell,
which is valid for $t={\mathcal O}(\eps^2)$ when $U(\v{x},0)=0$ and
for arbitrary initial values $\v{u}_{j}(0)$ for
$j\in \lbrace{1,\ldots,N\rbrace}$. This short-time analysis will provide
the initial conditions for the long-time dynamics studied in \S
\ref{2:ode_memory}.

We first assume that $u_{1j}(0)\neq 0$. In the $j^{\mbox{th}}$ inner
region we let $\y=\eps^{-1}(\v{x}-\v{x}_j)$ with $\rho=|\v{y}|$ and we
introduce the short time-scale $\tau$, defined by $t=\eps^2
\tau$. From (\ref{DimLess_bulka}) and (\ref{DimLess_bulkb}), we obtain
that $V(\rho,\tau)=U(\v{x}_j+\eps\v{y}, \eps^2\tau)$ satisfies, to
leading-order, the local problem
\begin{equation}\label{3:local_trans}
\begin{split}
  \partial_{\tau} V &=D \left(\partial_{\rho\rho} V + \frac{1}{\rho}
    \partial_{\rho} V \right) \,,
  \quad \rho>1\,, \,\,\, \tau\geq 0\,;
  \quad V(\rho,0) = 0  \,, \\
  D \partial_{\rho} V &= d_{1j} V - d_{2j} u_{1j}(0) \,, \quad \mbox{on} \,\,\,
  \rho=1 \,.
\end{split}
\end{equation}
By taking the Laplace transform of (\ref{3:local_trans}), we
calculate $\hat{V}(\rho,s)={\mathcal L}\left[V(\rho,\tau)\right]\equiv
\int_{0}^{\infty} V(\rho, \tau) e^{-s\tau}d\tau$ as
\begin{equation}\label{3:lap_trans}
  \hat{V}(\rho,s) = u_{1j}(0) \frac{d_{2j}}{s}
  \left( \frac{K_{0}\left(\rho \sqrt{s/D}\right)}{
      d_{1j} K_{0}\left(\sqrt{s/D}\right) +
      \sqrt{sD} K_{1}\left(\sqrt{s/D}\right)}\right) \,,
\end{equation}
where the branch cut is taken along $\mbox{Re}(s)\leq 0$ with
$\mbox{Im}(s)=0$.  Here $K_{0}(z)$ and $K_{1}(z)$ are modified Bessel
functions of the second kind.

Since the singularity of (\ref{3:lap_trans}) with the largest real
part is at $s=0$, to determine the large $\tau$ behavior of $V$
for fixed $\rho>1$ we must first determine the behavior of
(\ref{3:lap_trans}) as $s\to 0$.  Upon using
$K_{0}(z)\sim -\log{z} - \gamma_e + \log{2}+o(1)$ and
$K_{1}(z)\sim z^{-1}+o(1)$ as $z\to 0$, where $\gamma_e$ is Euler's
constant, we get
\begin{subequations}
\begin{equation}\label{3:small_s}
  \hat{V}(\rho,s) \sim \frac{d_{2j} u_{1j}(0)}{d_{1j} s} +
  \frac{2 d_{2j} u_{1j}(0)}{d_{1j}} \left(\log\rho + \frac{D}{d_{1j}}
  \right) \frac{1}{s\log(\kappa_{0j} s)} \,, \quad \mbox{as} \quad s\to 0 \,,
\end{equation}
where $\kappa_{0j}$ is defined by
\begin{equation}\label{3:small_F}
  \kappa_{0j} \equiv \frac{1}{4D} e^{2\left(\gamma_e - {D/d_{1j}}\right)}\,.
\end{equation}
\end{subequations}

To determine $V(\rho,\tau)$ for $\tau\gg 1$ and
$\rho={\mathcal O}(1)$, we must determine a function whose
Laplace transform is analytic in $\mbox{Re}(s)>0$ and where its
singularity with the largest real part occurs at $s=0$ with the
leading-order local behavior $\left[s\log(\kappa_{0j} s)\right]^{-1}$ as
$s\to 0$. We emphasize that $\hat{V}(\rho,s)$ in (\ref{3:lap_trans})
is analytic in $\mbox{Re}(s)>0$, and that the pole at $s={1/\kappa_{0j}}$
in the local behavior as $s\to 0$ in (\ref{3:small_s}) is spurious. As
shown in \cite{LS}, and summarized in Appendix \ref{app:trans}, for
any $\kappa_{0j}>0$ we have the exact relation
\begin{equation}\label{3:N1}
  {\mathcal L}\left[ N\left({\tau/\kappa_{0j}}\right)\right]
  = \frac{1}{s-{1/\kappa_{0j}}} - \frac{1}{s\log(s\kappa_{0j})} \,,
\end{equation}
where $N(x)$ for $x>0$ is Ramanujan's integral defined by
\begin{equation}\label{3:rama}
  N(x) \equiv \int_{0}^{\infty} \frac{e^{-x\xi}}{\xi\left[\pi^2+
      (\log \xi)^2\right]} \, d\xi \,.
\end{equation}
It is readily shown that $s={1/\kappa_{0j}}$ is a removable singularity
for (\ref{3:N1}) and that
${\mathcal L}\left[N\left({\tau/\kappa_{0j}}\right)\right]$ is analytic
in $\mbox{Re}(s)>0$ and the singularity with the largest real part
satisfies
${\mathcal L}\left[ N\left({\tau/\kappa_{0j}}\right)\right]\sim
-\left[s\log(s\kappa_{0j})^{-1} \right]+{\mathcal O}(1)$ as $s\to 0$. By
using the well-known integral asymptotics (cf.~\cite{Wyman},
\cite{Bouk}, \cite{LS})
\begin{equation}\label{3:R_asy_1}
  N(x) \sim \frac{1}{\log{x}} - \frac{\gamma_e}{(\log{x})^2} +
  {\mathcal O}((\log{x})^{-3}) \,, \qquad \mbox{as} \quad x\to +\infty \,,
\end{equation}
{where $\gamma_e$ is Euler's constant, we absorb the second
  term into the leading-order approximation to obtain 
\begin{equation}\label{3:R_asy}
  N(x) \sim \frac{1}{\log\left(xe^{\gamma_e}\right)} -
  {\mathcal O}(\left[\log\left(xe^{\gamma_e}\right)\right]^{-3}) \,,
  \qquad \mbox{as} \quad x\to +\infty \,.
\end{equation}}
By using this result and (\ref{3:small_F}), we readily obtain 
for $\tau\gg 1$ and $\rho>1$ fixed that
\begin{subequations}\label{3:v_asy}
\begin{equation}\label{3:v_asy_1}
  V(\rho,\tau) \sim u_{1j}(0) \frac{d_{2j}}{d_{1j}} +
  \frac{B_j}{2\pi D} \left(\log\rho + \frac{D}{d_{1j}}\right)
   \,,
\end{equation}
where $B_j$ satisfies
\begin{equation}\label{3:v_asy_2}
  {B_j\sim -u_{1j}(0)\frac{4\pi D d_{2j}}{d_{1j}
    \log\left({\tau/\left[\kappa_{0j}e^{-\gamma_e}\right]}
      \right)} + {\mathcal O}\left(
      \left[\log\left({\tau/(\kappa_{0j}e^{-\gamma_e})}\right)
        \right]^{-2}\right) \,, \quad
   \mbox{for} \quad \tau\gg 1\,.}
\end{equation}
\end{subequations}

A similar analysis can be done for the case where $u_{1j}(0)=0$. By
evaluating (\ref{DimLess_Intra}) at $t=0$, and using $U(\v{x},0)=0$ and
$u_{1j}(0)=0$, we conclude for $\tau={\mathcal O}(1)$ that
\begin{equation}\label{3:uprime}
  u_{1j}\sim \eps^2 \tau u_{1j}^{\prime}(0) \,, \quad
  \mbox{where} \quad u_{1j}^{\prime}(0)=\v{e}_1^{T} \v{F}_j(\v{u}_j(0))\,.
\end{equation}
As a result, we replace $u_{1j}(0)$ in the boundary condition on $\rho=1$
in (\ref{3:local_trans}) with $\eps^2\tau u_{1j}^{\prime}(0)$ to derive,
in place of (\ref{3:lap_trans}), that
\begin{equation}\label{3:lap_trans_1}
  \hat{V}(\rho,s) = d_{2j} \eps^2 u_{1j}^{\prime}(0)
  \left( \frac{K_{0}\left(\rho \sqrt{s/D}\right)}{s^2\left[
      d_{1j} K_{0}\left(\sqrt{s/D}\right) +
      \sqrt{sD} K_{1}\left(\sqrt{s/D}\right)\right]}\right) \,.
\end{equation}
To determine the large $\tau$ behavior for $\rho>1$ fixed, we let
$s\to 0$ in (\ref{3:lap_trans_1}) to conclude that
\begin{equation}\label{3:small_s_new}
  V(\rho,\tau) \sim \eps^2 \tau u_{1j}^{\prime}(0) \frac{d_{2j}}{d_{1j}} +
  + 2\eps^2 u_{1j}^{\prime}(0) \frac{d_{2j}}{d_{1j}} \left( \log\rho +
      \frac{D}{d_{1j}}\right) {\mathcal L}^{-1} \left[
      \frac{1}{s^2\log(\kappa_{0j} s)} \right]\,,
\end{equation}
where $\kappa_{0j}$ is given in (\ref{3:small_F}).  By convolving
${\mathcal L}^{-1}\left[ (s\log(\kappa_{0j} s)^{-1}\right]=-
N\left({\tau/\kappa_{0j}}\right)$ and ${\mathcal L}^{-1}\left[s^{-1}\right]=1$, we
obtain from (\ref{3:small_s_new}) for $\tau\gg 1$ and
$\rho={\mathcal O}(1)$ that
\begin{subequations}\label{3:v_asy_new}
\begin{equation}\label{3:v_asy_1new}
  V(\rho,\tau) \sim \eps^2\tau u_{1j}^{\prime}(0) \frac{d_{2j}}{d_{1j}} +
  \frac{B_j}{2\pi D} \left(\log\rho + \frac{D}{d_{1j}}\right)
   \,,
\end{equation}
where, in terms of Ramanujan's integral $N(x)$ in (\ref{3:rama}) and
with $u_{1j}^{\prime}(0)=\v{e}_1^{T} \v{F}_j(\v{u}_j(0))$, $B_j$
now satisfies
\begin{equation}\label{3:v_asy_2new}
  B_j\sim -\eps^2 u_{1j}^{\prime}(0)\frac{4\pi D d_{2j}}{d_{1j}}
  \int_{0}^{\tau} N\left({\xi/\kappa_{0j}}\right)\, d\xi \,, \quad \mbox{for}
  \quad \tau \gg 1\,.
\end{equation}
\end{subequations}

Finally, to estimate $B_j$ for the range ${\mathcal O}(1)\ll \tau\ll
{\mathcal O}(\eps^{-2})$, we observe that since $\int_{0}^{\infty}
N\left({\xi/\kappa_{0j}}\right)\, d\xi$ is not integrable, the dominant
contribution to the integral in (\ref{3:v_asy_2new}) arises from
the upper endpoint $\xi=\tau\gg 1$ where we can use the asymptotics
(\ref{3:R_asy}). Upon integrating by parts, we calculate that
\begin{equation}\label{asy:int}
  \int_{0}^{\tau} \frac{1}{\log\left({\xi/(\kappa_{0j}e^{-\gamma_e})}\right)}\,
  d\xi =
  -\kappa_{0j} E_1\left[-\log\left({\tau/(\kappa_{0j}e^{-\gamma_e})}\right)
  \right] \sim
  \frac{\tau}{\log\left({\tau/(\kappa_{0j}e^{-\gamma_e})}\right)} \,,  \quad
  \mbox{as} \quad \tau\to \infty\,,
\end{equation}
where we used $E_{1}(z)\sim {e^{-z}/z}$ for $z\to +\infty$.  In this
way, we obtain from (\ref{3:R_asy}), (\ref{3:v_asy_2new}) and
(\ref{asy:int}) that
\begin{equation}\label{3:v_asy_2neww}
  {B_{j} \sim -\eps^2 \int_{0}^{\tau}
  \frac{u_{1j}^{\prime}(0) \gamma_j}
  {\log\left({\xi/(\kappa_{0j}e^{-\gamma_e})}\right)} \, d\xi =
  -\eps^2 
  \frac{u_{1j}^{\prime}(0) \gamma_j \tau}
  {\log\left({\tau/(\kappa_{0j}e^{-\gamma_e})}\right)} \,, \quad
  {\mathcal O}(1)\ll \tau \ll {\mathcal O}(\eps^{-2}) \,.}
\end{equation}

The explicit far-field behavior of the transient solution given in
(\ref{3:v_asy_2}) and (\ref{3:v_asy_2neww}), valid for
${\mathcal O}(1)\ll \tau\ll {\mathcal O}(\eps^{-2})$ and
$\rho={\mathcal O}(1)$, is essential for providing the initial
behavior for the integro-differential system, derived below in \S
\ref{2:ode_memory}, that is valid on long time scales.

\section{Derivation of ODE system for intracellular kinetics}\label{2:ode_memory}

We now derive a reduced ODE system with memory for the intracellular
kinetics valid that is valid for $t\gg {\mathcal O}(\eps^2)$.

For $t={\mathcal O}(1)$, in the $j^{\mbox{th}}$ inner region we let
$\y=\eps^{-1}(\x-\x_j)$ with $\rho=|\y|$, to obtain from
(\ref{DimLess_bulka}) and (\ref{DimLess_bulkb}) that
$V(\v{y},t)=U(\x_j+\eps\v{y},t)$ satisfies the leading-order
quasi-steady problem
\begin{equation}\label{2:inn_qss}
  \Delta_{\y} V = 0 \,, \quad \mbox{for} \quad \rho>1 \,; \qquad
  D \partial_{\rho} V= d_{1j}V - d_{2j} u_{1j} \,, \quad \mbox{on} \quad \rho=1\,.
\end{equation}
In terms of some $B_j(t)$ to be found, the radially symmetric solution to
(\ref{2:inn_qss}) is
\begin{equation}\label{2:inn_sol}
  V(y,t) =\frac{B_{j}(t)}{2\pi D} \log|\y| +
  \frac{B_{j}(t)}{2\pi d_{1j}} + \frac{d_{2j}}{d_{1j}} u_{1j}(t) \,.
\end{equation}
By calculating the flux as
$\eps^{-1}\int_{\Omega_{j}} \left(d_{1j}U - d_{2j} u_{1j}\right)\,
{\rm d}S_{\v{x}} = B_j(t)$, we obtain from (\ref{DimLess_Intra}) that
\begin{equation}\label{2:intra_1}
  \frac{d\v{u}_j}{dt} = \v{F}_j(\v{u}_j) +  B_j(t) \v{e}_1\,,
  \quad j\in \lbrace{1,\ldots, N\rbrace}\,.
\end{equation}
Next, by asymptotically matching (\ref{2:inn_sol}) to the outer solution
$U$, we conclude that $U$ must satisfy
\begin{equation}\label{2:bulk}
  \begin{split}
\partial_{t} U = D \Delta U &- \sigma U \,, \quad \v{x}\in \R^2\backslash
    \lbrace{\v{x}_1,\ldots,\v{x}_N\rbrace} \,; \qquad U(\v{x},0)=0 \,, \\
    U \sim \frac{B_j}{2\pi D}\log|\v{x}-\v{x}_j| &+ \frac{B_j}{2\pi
      D\nu} + \frac{B_j}{2\pi d_{1j}} + \frac{d_{2j}}{d_{1j}} u_{1j}
    \,, \quad \mbox{as} \quad \v{x}\to \v{x}_j \,, \quad
    j\in \lbrace{1,\ldots,N\rbrace} \,,
  \end{split}
\end{equation}
with $U(\v{x},t)\to 0$ as $|\v{x}|\to\infty$ for each $t>0$. Here
$\nu\equiv {-1/\log\eps}$, $B_{j}=B_{j}(t)$, and $u_{1j}=u_{1j}(t)$.

\begin{remark}\label{rem:kuram} {The ``outer'' problem
    (\ref{2:intra_1}) coupled to the bulk diffusion field via
    (\ref{2:bulk}) will self-consistently lead to an
    integro-differential ODE system for $B_j(t)$ and $\v{u}_j(t)$ for
    $j\in\lbrace{1,\ldots,N\rbrace}$. For the cell-bulk model,
    (\ref{2:intra_1}) together with (\ref{2:bulk}) replaces the
    phenomenological diffusive coupling model of \cite{kuramoto1995}
    as given in (\ref{intro:kuram_pde}).}
\end{remark}

To solve (\ref{2:bulk}) we first consider an auxiliary problem for
$v_{k}(\v{x},t)$ given by
\begin{equation} \label{2:vk}
\begin{split}
 \partial_{t} v_{k} = D \Delta v_{k}- &\sigma v_{k} -
  B_k(t)\delta(\v{x}-\v{x}_k)\,;
  \quad v_k\to 0  \quad \mbox{as} \quad |\v{x}|\to \infty \,;
  \quad v_k(\v{x},0)=0 \,,\\
  v_{k}(\v{x},t) &\sim \frac{B_k(t)}{2\pi D} \log|\v{x}-\v{x}_k| +
  R_{k}(t) + o(1)\,, \qquad \mbox{as} \quad \x \to \x_k \,,
\end{split}
\end{equation}
where $R_{k}(t)$ is to be determined in terms of $B_k(t)$. As shown in
our analysis of the transient solution in \S \ref{sec:transient}, we must
impose that $B_k(t)\to 0$ as $t\to 0^{+}$.

By taking the Laplace transform of (\ref{2:vk}) we obtain that
$\hat{v}_{k}(\v{x},s)= {\mathcal L}\left[v_{k}(\v{x},t)\right]$ satisfies
\begin{equation}\label{2:lt_vk}
  \begin{split}
   \Delta \hat{v}_{k} &- \frac{(\sigma+s)}{D} \hat{v}_k = \frac{\hat{B}_k(s)}{D}
    \delta(\v{x}-\v{x}_k)\,; \quad \hat{v}_k\to 0  \quad \mbox{as} \quad
    |\v{x}|\to \infty \,,\\
    \hat{v}_{k} &\sim \frac{\hat{B_{k}(s)}}{2\pi D} \log|\v{x}-\v{x}_k| +
        \hat{R}_{k}(s)\,, \qquad \mbox{as} \quad \x \to \x_k \,,
  \end{split}
\end{equation}
where $\hat{B}_k(s)={\mathcal L}\left[B_{k}(t)\right]$ and
$\hat{R}_k(s)={\mathcal L}\left[R_{k}(t)\right]$. By decomposing
$\hat{v}_k=-\hat{B}_k \hat{G}_k$, we find that $\hat{G}_k$ satisfies
\begin{equation}\label{2:lt_gk}
   \Delta \hat{G}_{k}- \frac{ (\sigma+s) }{D} \hat{G}_k = -\frac{1}{D}
    \delta(\v{x}-\v{x}_k)\,; \quad \hat{G}_k\to 0  \quad \mbox{as} \quad
    |\v{x}|\to \infty \,.
\end{equation}
In terms of the modified Bessel function $K_0(z)$ of the second kind of
order zero, the solution to (\ref{2:lt_gk}) is
\begin{equation}\label{2:l_gk}
  \hat{G}_k(\v{x},s) = \frac{1}{2\pi D} K_0\left(\sqrt{\frac{s+\sigma}{D}}
    |\v{x}-\v{x}_k|\right) \,.
\end{equation}
In (\ref{2:l_gk}) we specify the principal branch of the square root,
so that $\hat{G}_k$ is analytic in the complex $s$ plane except along
the branch cut where $\mbox{Re}(s)\leq -\sigma$ and $\mbox{Im}(s)=0$.

Upon using $K_{0}(z)\sim -\log{z} - \gamma_e + \log{2}$ as
$z\to 0$, where $\gamma_e$ is Euler's constant, we can calculate
$\hat{v}_k(\v{x},s)$ as $\v{x}\to\v{x}_k$ and enforce that this
limiting behavior agrees with that required in (\ref{2:lt_vk}). By
specifying the principal branch of $\log(s+\sigma)$, this determines that
\begin{equation}\label{2:L_RK}
  \hat{R}_k(s) = \frac{1}{4\pi D}\hat{B}_k(s) \log(s+\sigma)  +
   \frac{\hat{B}_k(s)}{2\pi D} \left(\gamma_e - \log(2\sqrt{D})\right)\,.
\end{equation}
Then, by using the well-known inverse Laplace transform (cf.~\cite{AS})
\begin{equation*}
  {\mathcal L}^{-1}\left[K_{0}(a\sqrt{s})\right] = \frac{1}{2t} e^{-a^2/{(4t)}}
     \quad \mbox{for} \quad a>0 \,,
\end{equation*}
together with the shift property of the Laplace transform, we conclude from
(\ref{2:lt_gk}) that
\begin{equation}\label{2:gk_inv}
  G_{k}(\v{x},t) = {\mathcal L}^{-1}\left[\hat{G}_{k}(\v{x},s)\right]
  =\frac{1}{4\pi D t}
    e^{-\sigma t} e^{-{|\v{x}-\v{x}_k|^2/(4Dt)}} \,,
\end{equation}
which is simply the fundamental solution of the heat equation with bulk
degradation.  By using the convolution property we calculate
$v_{k}(\v{x},t) = -{\mathcal L}^{-1}\left[\hat{B}_k\hat{G}_{k}\right]$
as
\begin{equation}\label{2:vk_solve}
  v_{k}(\v{x},t) = -\int_{0}^{t} B_{k}(\tau) G_{k}(\v{x},t-\tau) \, d\tau = 
  - \int_{0}^{t} \frac{B_{k}(\tau) e^{-\sigma(t-\tau)}}{4\pi D
    (t-\tau)} e^{-|\v{x}-\v{x}_k|^2/(4D(t-\tau))} \, d\tau\,.
\end{equation}

Next, we invert (\ref{2:L_RK}) under the assumption that $B_k(0)=0$ to
derive that
\begin{equation}\label{2:L_RK_a}
  R_k(t) = \frac{1}{4\pi D} {\mathcal L}^{-1}\left[ s \hat{B}_k(s)
    \left( \frac{\log(s+\sigma)}{s}\right)\right] +
  \frac{B_{k}(t)}{2\pi D} \left(\gamma_e - \log\left(2\sqrt{D}\right)
    \right)\,.
\end{equation}
From the inverse transforms
$B_k^{\prime}(t)={\mathcal L}^{-1}\left[s\hat{B}_k(s)\right]$ and
${\mathcal L}\left[s^{-1}\log(s+\sigma)\right]=E_{1}(\sigma t) +
\log\sigma$ for $\sigma>0$ (cf.~\cite{AS}), we obtain from
the convolution property and $B_k(0)=0$ that
\begin{equation}\label{2:L_RK_b}
  {\mathcal L}^{-1} \left[\hat{B}_k \log(s+\sigma)\right] 
  = \int_{0}^{t} B_{k}^{\prime}(\tau) E_1(\sigma(t-\tau))
  \, d\tau + B_k(t) \log\sigma \,.
\end{equation}
In this way, in terms of the exponential integral $E_1(z)$, we 
conclude from (\ref{2:L_RK_a}) that
\begin{equation}\label{2:L_RK_fin}
  R_k(t) = \frac{B_{k}(t)}{2\pi D}\left[\gamma_{e} -
    \log\left(2 \sqrt{\frac{D}{\sigma}
      }\right)\right] + \frac{1}{4\pi D}
  \int_{0}^{t} B_{k}^{\prime}(\tau) E_1(\sigma(t-\tau))\, d\tau\,.
\end{equation}

We represent the solution to (\ref{2:bulk}) as the superposition
$U(\x,t)=\sum_{k=1}^{N} v_k(\x,t)$. By letting $\x\to\x_j$, and
enforcing that the limiting behavior of $U(\x,t)$ agrees with that
required in (\ref{2:bulk}), we conclude that
\begin{equation}\label{2:match}
  R_j(t)+ \sum_{\stackrel{k=1}{k\neq j}}^{N} v_k(\x_j,t) =
  \frac{B_j(t)}{2\pi D} \left(
    \frac{1}{\nu} + \frac{D}{d_{1j}} \right) + \frac{d_{2j}}{d_{1j}} u_{1j}(t)\,,
  \quad j\in \lbrace{1,\ldots,N\rbrace} \,.
\end{equation}
Upon using (\ref{2:L_RK_fin}) for $R_{j}(t)$, we observe that
(\ref{2:match}) determines $B_j(t)$ in terms of $u_{1j}(t)$.

Finally, upon combining (\ref{2:match}) with the intracellular dynamics
(\ref{2:intra_1}) we obtain an integro-differential system for the
intracellular species $\v{u}_j(t)$ that is coupled to the
time-history of the source strengths $B_j(t)$. We summarize our result
in the following formal proposition.

\begin{prop}\label{prop:2ode} For $\eps\to 0$, and with the initial condition
  $U(\x,0)=0$, the solution $U(\x,t)$ and $\v{u}_j(t)$, for
  $j\in\lbrace{1,\ldots,N\rbrace}$, to the cell-bulk model
  (\ref{DimLess_bulk}) is approximated for $t\gg {\mathcal O}(\eps^2)$
  by 
  \begin{subequations}\label{2:reduced}
  \begin{align}
    \frac{d\v{u}_j}{dt} &= \v{F}_j(\v{u}_j) + \v{e}_1 B_j(t) \,, 
    \label{2:reduced_1} \\
    \int_{0}^{t} B_{j}^{\prime}(\tau) E_1(\sigma(t-\tau))\, d\tau &=
                   \eta_j B_{j}(t) + \gamma_j u_{1j}(t) \nonumber \\
          & \qquad + \sum_{\stackrel{k=1}{k\neq j}}^{N}
  \int_{0}^{t} \frac{B_{k}(\tau) e^{-\sigma(t-\tau)}}{
    t-\tau} e^{-|\v{x}_j-\v{x}_k|^2/(4D(t-\tau))} \, d\tau\,,\label{2:reduced_2}
  \end{align}
  for $j\in\lbrace{1,\ldots,N\rbrace}$. In this integro-differential
  system $\eta_j$ and $\gamma_j$ are defined by
  \begin{equation}\label{2:reduced_3}
    \eta_j \equiv 2 \left( \frac{1}{\nu} + \frac{D}{d_{1j}} +
      \log\left(2 \sqrt{\frac{D}{\sigma} }\right) - \gamma_e \right)=
    -\log(\eps^2\kappa_{0j}\sigma) \,, \qquad
    \gamma_j \equiv \frac{4\pi D d_{2j}}{d_{1j}} \,,
  \end{equation}
\end{subequations}
where $\kappa_{0j}$ is defined in (\ref{3:small_F}).  In terms of
$B_{j}(t)$, the approximate solution in the bulk region is 
\begin{equation}\label{2:reduced_bulk}
  U(\x,t) \sim -\frac{1}{4\pi D} \sum_{j=1}^{N}\int_{0}^{t}
  \frac{B_{j}(\tau) e^{-\sigma(t-\tau)}}{t-\tau}
  e^{-|\v{x}-\v{x}_j|^2/(4D(t-\tau))} \, d\tau \,,
\end{equation}
while in the vicinity of the $j^{\mbox{th}}$ cell we have for
$\rho=\eps^{-1}|\v{x}-\v{x}_j|={\mathcal O}(1)$ that
\begin{equation}\label{2:reduced_local}
  U \sim\frac{B_{j}(t)}{2\pi D} \log\rho +
  \frac{B_{j}(t)}{2\pi d_{1j}} + \frac{d_{2j}}{d_{1j}} u_{1j}(t) \,.
\end{equation}
\end{prop}

\subsection{Matching to the Transient Solution}\label{sec:match}

We now study the limiting behavior as $t\to 0^{+}$, but with
$t\gg {\mathcal O}(\eps^2)$, for $B_j(t)$ in (\ref{2:reduced}), which
satisfies $B_j\to 0$ as $t\to 0$. Our small time analysis is based on
taking the Laplace transform of (\ref{2:reduced_2}) and letting
$s\to \infty$. More specifically, we will derive that
\begin{equation}\label{mat:bj}
{  B_{j}(t)\sim  \left\{\begin{array}{cl}
   -{u_{1j}^{\prime}(0) \gamma_j t/\log\left({t/(\kappa_je^{-\gamma_e})}\right)}
                         \,, &  \mbox{as}
  \,\,\,  t \to 0^{+} \,\,\, \mbox{if}\,\,\, u_{1j}(0)=0 \,, \\[10pt]
  -{u_{1j}(0) \gamma_j /\log\left({t/(\kappa_je^{-\gamma_e})}\right)} \,,
                             &  \mbox{as}
     \,\,\,  t \to 0^{+} \,\,\, \mbox{if}\,\,\, u_{1j}(0)\neq 0\,,
                        \end{array} \right.}
\end{equation}
where $\kappa_j=\eps^2\kappa_{0j}$ with $\kappa_{0j}$ defined in
(\ref{3:small_F}). These limiting results agree with those derived in
\S \ref{sec:transient} from the far-field behavior of the transient
solution. Moreover, we have
\begin{equation}\label{mat:bjp}
{  B_{j}^{\prime}(t)\sim  \left\{\begin{array}{cl}
       -\frac{u_{1j}^{\prime}(0) \gamma_j}
           {\log\left({t/(\kappa_je^{-\gamma_e})}\right)}
    \left(1 - \frac{1}{\log\left({t/(\kappa_je^{-\gamma_e})}\right)}\right)\,,
                                  &  \mbox{as}
  \,\,\,  t \to 0^{+} \,\,\, \mbox{if}\,\,\, u_{1j}(0)=0 \,, \\[10pt]
     {u_{1j}(0) \gamma_j/\left(t \left[\log\left({t/(\kappa_j e^{-\gamma_e})}
                                  \right)\right]^2\right)} \,,
    &  \mbox{as} \,\,\,  t \to 0^{+} \,\,\, \mbox{if}\,\,\, u_{1j}(0)\neq 0\,.
                        \end{array} \right.}
\end{equation}
Therefore, $B_{j}^{\prime}(t)\to 0$ as $t\to 0^{+}$ if
$u_{1j}(0)=0$, while $|B_{j}^{\prime}(t)|\to \infty$ as $t\to 0^{+}$
if $u_{1j}(0)\neq 0$.  

To derive (\ref{mat:bj}), we will first assume that
$u_{1j}(0)=0$. Setting $B_j(0)=0$, we obtain from (\ref{2:reduced_1})
that $u_{1j}(t)\sim t u_{1j}^{\prime}(0)$ as $t\to 0^{+}$, where
$u_{1j}^{\prime}(0)= \v{e}_1^{T}\v{F}(\v{u}_j(0))$. By taking the Laplace
transform of (\ref{2:reduced_2}) and by using
${\mathcal L}\left[B_j^{\prime}(t)\right]=s\hat{B}_j(s)$ and
${\mathcal L}\left[E_1(\sigma
  t)\right]={\log\left(1+{s/\sigma}\right)/s}$ we readily calculate
that
\begin{equation}\label{ode:mat_1}
  \hat{B}_j \left[\log\left(1+{s/\sigma}\right) - \eta_j\right]
  - 2 \sum_{\stackrel{k=1}{k\neq j}}^{N} \hat{B}_k
  K_0 \left( \sqrt{ \frac{s+\sigma}{D}} |\v{x}_j-\v{x}_k|\right) =
  \hat{u}_{j1} \gamma_j\,.
\end{equation}
By using the exponential decay of $K_0$ and
$\hat{u}_{j1}\sim {u_{1j}^{\prime}(0)/s^2}$ for $s\to \infty$, (\ref{ode:mat_1})
becomes
\begin{equation}\label{ode:mat_2}
  \hat{B}_j \left[\log\left( \frac{se^{-\eta_j}}{\sigma}\right) +
    {\mathcal O}(s^{-1})\right] \sim \frac{\gamma_j}{s^2} u_{1j}^{\prime}(0)\,,
  \quad \mbox{as} \quad s\to \infty\,.
\end{equation}
By neglecting the ${\mathcal O}(s^{-1})$ term, we use 
(\ref{2:reduced_3}) for $\eta_j$ to calculate from (\ref{ode:mat_2})
that
\begin{equation}\label{ode:mat_4}
  \hat{B}_j \sim \frac{u_{1j}^{\prime}(0) \gamma_j}{s^2\log(\kappa_j  s)} \,,
  \quad \mbox{as} \quad s\to \infty \,, \quad \mbox{where} \quad
  \kappa_j \equiv \frac{e^{-\eta_j}}{\sigma} = \eps^2 \kappa_{0j} \,.
\end{equation}
As a result, for $t\to 0^{+}$, but with $t\gg {\mathcal O}(\eps^2)$, we
use the convolution property and $\kappa_j=\eps^2\kappa_{0j}$ to obtain
\begin{equation}\label{ode:mat_5}
  B_j(t) \sim -u_{1j}^{\prime}(0)\gamma_j \int_{0}^{t}
  N\left(\frac{\xi}{\eps^2\kappa_{0j}}\right) \, d\xi\,, \quad
  \mbox{for} \quad {\mathcal O}(\eps^2)\ll t \ll {\mathcal O}(1)\,,
\end{equation}
where $N(x)$ is Ramanujan's integral (\ref{3:rama}), $\kappa_{0j}$ is
defined in (\ref{3:small_F}) and
$u_{1j}^{\prime}(0)=\v{e}_1^{T} \v{F}_j(\v{u}_j(0))$. This limiting
result matches identically with the result derived in
(\ref{3:v_asy_2new}) of \S \ref{sec:transient}. Moreover, by setting
$t=\eps^2\tau$ in (\ref{3:v_asy_2neww}) we obtain the first result in
(\ref{mat:bj}).

A very similar short-time analysis can be done when $u_{1j}(0)\neq
0$. In place of (\ref{ode:mat_4}), we obtain that
\begin{equation}\label{node:mat_3}
  \hat{B}_j \sim \frac{u_{1j}(0)\gamma_j}{s\log(\kappa_j  s)} \,,
  \quad \mbox{as} \quad s\to \infty \,.
\end{equation}
By inverting (\ref{node:mat_3}), and recalling that
$\gamma_j={4\pi D d_{2j}/d_{1j}}$, we obtain the second result in
(\ref{mat:bj}). This result matches identically with the far-field
behavior (\ref{3:v_asy_2}) as obtained from our transient analysis
in \S \ref{sec:transient}.
  
\section{Steady-state and linear stability analysis}\label{sec:ss}

In this section, we first classify the long-time dynamics of
(\ref{2:reduced}) for solutions that tend to limiting values as
$t\to\infty$. We show that this limiting dynamics coincides with the
steady-state solution, which is constructed using strong localized
perturbation theory from the steady-state of the cell-bulk PDE system
(\ref{DimLess_bulk}). We also formulate the linear stability problem
for steady-state solutions. We begin with the following lemma:

\begin{lemma} Suppose that $B_j(0)=0$, $B_j(t)$ is bounded for $t>0$
  and that $B_j(t)\to B_{j\infty}$ as $t\to\infty$ for
  $j\in\lbrace{1,\ldots,N\rbrace}$, where $B_{j\infty}$ are constants
  for $j\in\lbrace{1,\ldots,N\rbrace}$. Then,
  \begin{subequations}
    \begin{align}
&      \lim_{t\to\infty} C_{jk}(t)  =
    2B_{j\infty} K_0\left(\sqrt{\frac{\sigma}{D}}|\v{x}_j-\v{x}_k|\right) \,,
                    \qquad \lim_{t\to\infty} D_{j}(t) =0 \,, \label{4:limit_1}\\
      \noindent \mbox{where} \qquad  & C_{jk}(t)\equiv
    \int_{0}^{t} \frac{B_j(\tau)e^{-\sigma(t-\tau)}} {t-\tau}
        e^{-|\v{x}_j-\v{x}_k|^2/(4D(t-\tau))} \, d\tau
        \,, \label{4:limit_2} \\
      \noindent \mbox{and} \qquad
      & D_{j}(t) \equiv \int_{0}^{t} B_{j}^{\prime}(\tau)\, E_1\left(\sigma
      (t-\tau)\right) \, d\tau\,. \label{4:limit_3}
    \end{align}
   \end{subequations}
 \end{lemma}

 \begin{proof} To establish the first limit in (\ref{4:limit_1}), we
   take the Laplace transform of (\ref{4:limit_2}) to obtain
   \begin{equation}\label{4:plim_1}
     \hat{C}_{jk}(s) =  {\mathcal L}\left[\frac{e^{-\sigma t}}{t}
       e^{-|\v{x}_j-\v{x}_k|^2/(4Dt)}\right] \hat{B}_j(s) = 2
     K_0\left(\sqrt{\frac{s+\sigma}{D}}|\v{x}_j-\v{x}_k|\right)
     \hat{B}_j(s) \,.
   \end{equation}
   To determine the long-time behavior of $C_{jk}(t)$, we will use the
   Tauberian theorem. Since $B_j(t)$ is bounded for $t>0$ and
   satisfies $B_j(t)\to B_{j\infty}$ as $t\to \infty$, it follows that
   $\hat{B}_j(s)$ is analytic for $\mbox{Re}(s)>0$ and that the singularity
   with the largest real part is a simple pole at $s=0$ for which
   $\hat{B}_j(s)\sim {B_{j\infty}/s}$ as $s\to 0$. As a result, the
   singularity of $\hat{C}_{jk}(s)$ with the largest real part is a
   simple pole at $s=0$ and we have
   \begin{equation}
     \lim_{t\to\infty}C_{jk}(t) = \lim_{s\to 0} s \hat{C}_{jk}(s) =
     2 K_0\left(\sqrt{\frac{\sigma}{D}}|\v{x}_j-\v{x}_k|\right)
     \lim_{s\to 0} s \hat{B}_j(s)\,,
   \end{equation}
   where $\lim_{s\to 0} s \hat{B}_j(s)=B_{j\infty}$. This proves the first
   result in (\ref{4:limit_1}).

   The second result in (\ref{4:limit_1}) is established in a similar
   way.  Upon taking the Laplace transform of $D_j(t)$ in
   (\ref{4:limit_3}), we use $B_{j}(0)=0$ to obtain that
   \begin{equation}\label{4:plim_2}
     \hat{D}_j(s) = {\mathcal L}\left[B_{j}^{\prime}(t)\right]
     {\mathcal L}\left[ E_1(\sigma t)\right] = s \hat{B}_j(s)
     \frac{\log\left(1+{s/\sigma}\right)}{s} = \hat{B}_j(s)
     \log\left(1+{s/\sigma}\right) \,,
   \end{equation}
   where the branch cut for $\log(1+{s/\sigma})$ occurs for
   $\mbox{Re}(s)\leq -\sigma<0$ and $\mbox{Im}(s)=0$.  To establish
   that $D_j(t)\to 0$ as $t\to\infty$ it suffices to show that
   $\hat{D}_j(s)$ is analytic in $\mbox{Re}(s)\geq 0$. This follows
   since from (\ref{4:plim_2}) we observe that, although
   $\hat{B}_j\sim {B_{j\infty}/s}$ as $s\to 0$, the point $s=0$ is a
   removable singularity for $\hat{D}_j(s)$.
 \end{proof}

 With this lemma, we can readily characterize those solutions to
 (\ref{2:reduced}) that tend to limiting values as $t\to \infty$. We
 summarize this result as follows:

 \begin{prop}\label{prop:limit} Suppose that
   $B_j(0)=0$ and that $B_j(t)$ and $\v{u}_{j}(t)$ are bounded for
   $t\geq 0$ with limiting values $B_{j}(t)\to B_{j\infty}$ and
   $\v{u}_{j}\to \v{u}_{j\infty}$ as $t\to \infty$ for each
   $j\in\lbrace{1,\ldots,N\rbrace}$. Then, with
  $\eta_j$ and $\gamma_j$ as defined in (\ref{2:reduced_3}), $B_{j\infty}$ and
   $\v{u}_{j\infty}$ satisfy the $N(m+1)$ dimensional nonlinear
   algebraic system  (NAS)
   \begin{subequations}\label{4:limit_sys}
     \begin{align}
       \v{F}_j(\v{u}_{j\infty}) + \v{e}_1 B_{j\infty} &= 0 \,,
          \quad j\in\lbrace{1,\ldots,N\rbrace}\,,\\
       \eta_j B_{j\infty} +  2  \sum_{\stackrel{k=1}{k\neq j}}^{N}
     B_{k\infty} K_0\left(\sqrt{\frac{\sigma}{D}}|\v{x}_j-\v{x}_k|\right)
                &=-\gamma_j \v{e}_1^{T}\v{u}_{j\infty} \,,
         \quad j\in\lbrace{1,\ldots,N\rbrace} \,.
      \end{align}
    \end{subequations}
 \end{prop}

 We now show that the limiting NAS (\ref{4:limit_sys}) is
 precisely the same system as can be derived from a strong localized
 perturbation analysis for steady-state solutions of
 (\ref{DimLess_bulk}) by following the methodology of
 \cite{smjw_diff}. In the limit $\eps\to 0$, the steady-state
 problem for the outer bulk solution $U_{s}(\v{x})$ is
\begin{equation}\label{4:ssbulk}
  \begin{split}
   &  D \Delta U_{s} - \sigma U_{s} =0\,, \quad \v{x}\in \R^2\backslash
    \lbrace{\v{x}_1, \ldots,\v{x}_N\rbrace} \,,\\
    U_s \sim \frac{B_{js}}{2\pi D}\log|\v{x}-\v{x}_j| &+ \frac{B_{js}}{2\pi
      D\nu} + \frac{B_{js}}{2\pi d_{1j}} + \frac{d_{2j}}{d_{1j}} u_{1js}
    \,, \,\,\, \mbox{as} \,\,\, \v{x}\to \v{x}_j \,, \quad
    j\in\lbrace{1,\ldots,N\rbrace} \,,
  \end{split}
\end{equation}
where $\nu\equiv {-1/\log\eps}$. This system is coupled to the steady-state
of the intracellular kinetics, given by
\begin{equation}\label{4:ssintra}
  \v{F}_{j}\left(\v{u}_{js}\right) + \v{e}_1 B_{js} =0 \,, \quad
  j\in\lbrace{1,\ldots,N\rbrace}\,.
\end{equation}
In terms of the modified Bessel function $K_0(z)$, the
solution to (\ref{4:ssbulk}) is
\begin{equation}\label{4:ssbulk_sol}
  U_{s}(\v{x}) = -\sum_{k=1}^{N} \frac{B_{ks}}{2\pi D}
  K_0\left( \sqrt{\frac{\sigma}{D}} |\v{x}-\v{x}_k| \right) \,.
\end{equation}
To determine $B_{js}$, for $j\in\lbrace{1,\ldots,N\rbrace}$, we
enforce that the local behavior of $U_{s}(\v{x})$ as
$\v{x}\to \v{x}_j$ agrees with that required in (\ref{4:ssbulk}). Upon
using $K_0(z)\sim -\log\left({z/2}\right)- \gamma_e+o(1)$ as $z\to 0$,
this condition yields
\begin{equation}\label{4:ss_nas}
    \eta_j B_{js} + \gamma_j u_{1js} +  2
       \sum_{\stackrel{k=1}{k\neq j}}^{N} B_{ks}
      K_0\left(\sqrt{\frac{\sigma}{D}}|\v{x}_j-\v{x}_k|\right)=0 \,,
                \quad j\in\lbrace{1,\ldots,N\rbrace} \,.
\end{equation}
The coupled system (\ref{4:ss_nas}) and (\ref{4:ssintra})
characterizing the steady-state solution is identical to that in
(\ref{4:limit_sys}).

\subsection{Linear Stability Analysis}\label{sec:stab}

We now derive a globally coupled eigenvalue problem (GCEP) for the
linearization of the cell-bulk model (\ref{DimLess_bulk}) around the
steady-state solution. To derive the GCEP we first perturb around the
steady-state solution by introducing the eigen-perturbation
\begin{equation}\label{5:pert}
  U=U_{s}(\v{x}) + e^{\lambda t} \Phi(\v{x}) \,, \quad
  \v{u}_j = \v{u}_{js} + e^{\lambda t} \v{\zeta}_j \,, \quad
  j\in\lbrace{1,\ldots,N\rbrace}\,,
\end{equation}
into (\ref{DimLess_bulk}). Upon linearizing, we obtain the eigenvalue
problem
\begin{subequations}\label{5:eig_prob}
\begin{align}
 &  \Delta \Phi - \frac{(\lambda+\sigma)}{D} \Phi =0\,, \quad
      \x \in \R^2 \setminus \cup_{j=1}^{N}\,\Omega_{\varepsilon_j}\,,
      \label{5:eig_prob_1} \\
 &  \varepsilon  D\, \partial_n \Phi  = d_{1j} \, \Phi - d_{2j} \, \zeta_{j1}\,,
   \quad \x \in \partial \Omega_{\varepsilon_j}\,,\qquad
   j \in \lbrace{1, \ldots, N\rbrace} \,, \label{5:eig_prob_2}
\end{align}
which is coupled to the linearized intracellular dynamics within the
$j^{\mbox{th}}$ cell  by
\begin{align}\label{5:eig_prob_3}
  \lambda \v{\zeta}_j & =  J_{j} \v{\zeta}_j +
   \frac{\v{e}_1}{\eps}\int_{\partial \Omega_{\varepsilon_j}}\, ( d_{1j} \,
                        \Phi - d_{2j} \, \zeta_{j1}) \, {\rm d}S_{\v{x}}\,,
                        \qquad j \in\lbrace{1, \ldots, N \rbrace} \,.
\end{align}
\end{subequations}
Here $J_j\equiv D_{\v{u}_j} \v{F}_j$ is the Jacobian of the
intracellular kinetics evaluated at $\v{u}_j=\v{u}_{js}$.

For $\eps\to 0$, we now analyze (\ref{5:eig_prob}) using strong
localized perturbation theory. In the inner region near the
$j^{\mbox{th}}$ cell, we have to leading-order from
(\ref{5:eig_prob_1}) and (\ref{5:eig_prob_2}) that
\begin{equation}\label{5:inn_eig}
  \Delta_{\y} \Phi = 0 \,, \quad \mbox{for} \quad \rho>1 \,; \qquad
  D \Phi_{\rho}= d_{1j}\Phi - d_{2j} \zeta_{j1} \,, \quad \mbox{on} \quad \rho=1\,.
\end{equation}
In terms of some $c_j$ to be found, (\ref{5:inn_eig}) has the radially
symmetric solution
\begin{equation}\label{5:inn_sol}
   \Phi =\frac{c_j}{2\pi D} \log|\y| +
  \frac{c_j}{2\pi d_{1j}} + \frac{d_{2j}}{d_{1j}} \zeta_{j1} \,.
\end{equation}
By calculating the surface integral 
$\eps^{-1}\int_{\Omega_{j}} \left(d_{1j}\Phi - d_{2j}
  \zeta_{j1}\right)\, {\rm d}S_{\v{x}}= c_j$, we obtain that
(\ref{5:eig_prob_3}) becomes
\begin{equation}\label{5:intra_1}
  \left(\lambda I - J_j\right) \v{\zeta}_j = \v{e}_1 c_j \,,
  \quad j\in\lbrace{1,\ldots, N\rbrace}\,,
\end{equation}
where $I$ is the $m\times m$ identity matrix. Upon assuming that
$\lambda I -J_j$ is invertible, we calculate $\zeta_{j1}$ as
\begin{equation}\label{5:zeta_j1}
  \zeta_{j1}=  {\mathit K}_{j} c_j \quad \mbox{where} \quad
  {\mathit K}_j \equiv \v{e}_1^{T} \left(\lambda I - J_j\right)^{-1}\v{e}_1 \,.
\end{equation}

By writing (\ref{5:inn_sol}) in terms of the outer variable, the
asymptotic matching condition shows that in the outer region $\Phi$
satisfies
\begin{subequations}\label{5:outer}
  \begin{align}
 & \Delta \Phi - \frac{(\lambda+\sigma)}{D} \Phi = 0\,, \quad
      \x \in \R^2 \setminus \cup_{j=1}^{N}\,\Omega_{\varepsilon_j}\,,
      \label{5:outer_1} \\
  &   \Phi \sim \frac{c_j}{2\pi D}\log|\v{x}-\v{x}_j| + \frac{c_j}{2\pi
      D\nu} + \frac{c_j}{2\pi d_{1j}} + \frac{d_{2j}}{d_{1j}} \zeta_{j1}
    \,, \quad \mbox{as} \quad \v{x}\to \v{x}_j \,, \label{5:outer_2}
  \end{align}
\end{subequations}
for $j\in\lbrace{1,\ldots,N\rbrace}$, where $\nu\equiv
{-1/\log\eps}$. The solution to (\ref{5:outer}) is
\begin{equation}\label{5:eigbulk_sol}
  \Phi(\v{x}) = -\sum_{k=1}^{N} \frac{c_j}{2\pi D}
  K_0\left( \sqrt{\frac{\sigma+\lambda}{D}} |\v{x}-\v{x}_k| \right) \,.
\end{equation}
Upon enforcing that the limiting behavior of $\Phi(\v{x})$ as
$\v{x}\to\v{x}_j$ agree with that required in (\ref{5:outer_2}), and
where $\zeta_{j1}$ is given in terns of $c_j$ by (\ref{5:zeta_j1}), we
conclude for $j\in \lbrace{1,\ldots,N\rbrace}$ that
\begin{equation}\label{5:lin_sys}
 2 \left(\frac{1}{\nu} + \frac{D}{d_{1j}} - \gamma_e +
   \log\left(2 \sqrt{\frac{D}{\sigma+\lambda}}\right)\right) c_j +
\frac{4\pi D d_{2j}}{d_{1j}} \zeta_{j1} +
  2 \sum_{\stackrel{k=1}{k\neq j}}^{N} c_{k}
  K_0\left( \sqrt{\frac{\sigma+\lambda}{D}} |\v{x}_j-\v{x}_k| \right) =0
  \,.
\end{equation}
Finally, upon writing (\ref{5:lin_sys}) and (\ref{5:zeta_j1}) in matrix
form, we obtain that the discrete eigenvalues
$\lambda$ of the linearization around a steady-state solution of
(\ref{DimLess_bulk}) are obtained from a nonlinear matrix eigenvalue
problem, which we refer to as the GCEP. Our result is summarized as
follows:
  
\begin{prop}\label{prop:stab} For $\eps\to 0$, the discrete eigenvalues
  $\lambda$ associated
  with the linearization around a steady-state solution to
  (\ref{DimLess_bulk}), for which $\det(\lambda I-J_j)\neq 0$ for
  any $j\in\lbrace{1,\ldots,N\rbrace}$, are the set of values
  \begin{subequations}\label{5:gcep}
    \begin{equation}\label{TransDent}
   \Lambda({\mathcal M}) \equiv \lbrace{ \lambda \, \vert \,
    \det \mathcal{M}(\lambda)=0 \rbrace} \,,
\end{equation}
where the $N\times N$ dimensional matrix ${\mathcal M}(\lambda)$ is
defined by
\begin{equation}\label{5:gcep_2}
   \mathcal{M}(\lambda) \equiv I + 2\pi \nu \mathcal{G}_{\lambda}  + \nu\,D\,P_1 +
 2\pi \nu D P_2 \mathcal{K}(\lambda) \,,
\end{equation}
\end{subequations}
with $\nu={-1/\log\eps}$. In (\ref{5:gcep_2}), $P_1$, $P_2$ and
${\mathcal K}$ are the diagonal matrices
\begin{subequations}
\begin{align}
  P_1 &\equiv \mbox{diag}\Big{(} \frac{1}{d_{11}}, \ldots,
        \frac{1}{d_{1N}} \Big{)}
  \,, \qquad P_2 \equiv \mbox{diag}\Big{(}\frac{d_{21}}{d_{11}}, \ldots,
  \frac{d_{2N}}{d_{1N}}\Big{)} \,, \label{5:p1p2} \\
   \mathcal{K}(\lambda) &\equiv \mbox{diag}\left( \mathit{K}_1,\ldots,
  \mathit{K}_N\right)\,, \quad \mbox{where} \quad {\mathit K}_j \equiv
  \v{e}_1^{T} \left(\lambda I - J_j
  \right)^{-1}\v{e}_1 \,.\label{5:kjac}
\end{align}
\end{subequations}
In addition, ${\mathcal G}_{\lambda}$ is the eigenvalue-dependent
Green's matrix with matrix entries
\begin{align}\label{GreenMat}
  ({\mathcal G}_{\lambda})_{ij} &= ({\mathcal G}_{\lambda})_{ji} \equiv
  \frac{1}{2\pi}K_0\left( \sqrt{\frac{\sigma+\lambda}{D}}
    |\v{x}_j-\v{x}_k| \right) \,, \quad i\neq j\,, \\
  ({\mathcal G}_{\lambda})_{jj} &= R_{\lambda j} \equiv \frac{1}{2\pi} \left(
 \log\left(2 \sqrt{\frac{D}{\sigma+\lambda}}\right) - \gamma_e\right)\,.
\end{align}
For any specific $\lambda_0\in \Lambda({\mathcal M})$, we have
$\det\mathcal{M}(\lambda_0)=0$, and so
${\mathcal M}(\lambda_0) \v{c}=\v{0}$ has a nontrivial solution
$\v{c}=(c_1,\ldots,c_N)^{T}$ that can be normalized as $|\v{c}|=1$. We
conclude that the steady-state is linearly stable if, whenever
$\lambda\in \Lambda({\mathcal M})$, we have $\mbox{Re}(\lambda)<0$.
\end{prop}

We remark that the normalized nullvector $\v{c}=(c_1,\ldots,c_N)^{T}$
encodes the relative magnitude of the perturbation of the spatial
gradient of the bulk signal near the cell boundaries. It also can be
used to predict the relative magnitude and phase shift of
intracellular oscillations for the permeable species $u_{1j}$ that can
arise from bifurcations of the steady-state. To see this, we use the
steady-state solution $U_{s}(\v{x})$ and the perturbation (\ref{5:pert})
to calculate for the $j^{\mbox{th}}$ cell that
\begin{subequations}\label{nstabform:c}
\begin{align}
  D \partial_{\rho} U \vert_{\rho=1} &\sim \frac{1}{2\pi} \left(
          B_{js} + \sum_{\lambda_0\in\Lambda({\mathcal M})} c_j
                                e^{\lambda_0 t}\right) \,,
  \quad j\in\lbrace{1,\ldots,N\rbrace} \,,  \label{nstabform:c1}  \\
 u_{1j} &\sim u_{1js} +
  \sum_{\lambda_0\in\Lambda({\mathcal M})} {\mathit K}_{j}(\lambda_0) c_j
          e^{\lambda_0 t} \,, \quad j\in\lbrace{1,\ldots,N\rbrace} \,,
          \label{nstabform:c2}
\end{align}
\end{subequations}
{where
  ${\mathit K}_j(\lambda_0)=\v{e}_1^{T} \left(\lambda_0 I -
    J_j\right)^{-1}\v{e}_1$.  From (\ref{nstabform:c2}) we observe
  that if $\lambda_0$ is complex-valued, the real and imaginary parts
  of the $j^{\mbox{th}}$ component of the complex-valued
  matrix-eigenvector product ${\mathcal K}\v{c}$ encode both the
  relative magnitude and phase shift of oscillations for the permeable
  species $u_{1j}$ in the cell population. Moreover, from
  (\ref{nstabform:c1}), the components of the eigenvector $\v{c}$
  determine the strength and phase shift of the eigen-perturbation of
  the signaling gradients near the cells.}

\section{Time-marching scheme for the integro-differential system}
\label{sec:numerics}

A direct numerical approach to solve (\ref{2:reduced}) would require
at each time step a numerical quadrature of ${\mathcal O}(N^2)$
memory-dependent convolution integrals. This naive approach would be
prohibitively expensive for large $N$ and would also require storing
the full time history of each $B_{j}(t)$ in order to advance one
time-step.

As such, we now develop a time-marching algorithm to compute solutions
to (\ref{2:reduced}). This approach is based on a highly accurate
approximation of the kernels in the nonlocal terms of
(\ref{2:reduced_2}) by a sum of exponentials, which leads naturally to
an exponential time differencing marching scheme. Rigorous results for
the \emph{sum-of-exponentials} approximation, together with the
development of time-marching methods for convolution integrals in
other contexts are given in \cite{beylkin_2005}, \cite{beylkin_2010},
\cite{lopez_1}, \cite{lopez_2}, \cite{lopez_3} and
\cite{greengard_2015}.

With an exponential kernel, our derivation of a time-marching scheme
relies on a Duhamel-type lemma:

\begin{lemma}\label{prop:duh} Let $f(t)$ be continuous and define the
  convolution
  ${\mathcal F}(t)\equiv\int_{0}^{t} e^{\omega (t-\tau)} f(\tau)\,
  d\tau$.  Then, we have
  ${\mathcal F}^{\prime}(t)=\omega {\mathcal F}(t) + f(t)$ with
  ${\mathcal F}(0)=0$. Moreover, we have the marching scheme
\begin{equation}\label{duh:update}
    {\mathcal F}(t+\Delta t)={\mathcal F}(t) e^{\omega \Delta t}
    +\, {\mathcal U}(t,\Delta t) \,, \quad \mbox{with} \quad
    {\mathcal U}(t,\Delta t)\equiv e^{\omega \Delta t} \int_{0}^{\Delta t}
    e^{-\omega z} f(t+z) \, dz \,.
\end{equation}
An {\em exponential time differencing} ETD2 scheme
(cf.~\cite{lopez_1}), ensuring that the update integral ${\mathcal U}$
exact for linear functions $f(t)$, yields, with an error
${\mathcal O}\left( (\Delta t)^3\right)$, the approximation
\begin{equation}\label{duh:etd2}
  {\mathcal F}(t+\Delta t) \approx {\mathcal F}(t) e^{\omega \Delta t}
  + f(t)\left( \frac{e^{\omega \Delta t}-1}{\omega}\right) +
  \left[ f(t+\Delta t)-f(t)\right] \left( \frac{e^{\omega \Delta t} - 1 -
      \omega \Delta t}{\omega^2\Delta t}\right) \,.
\end{equation}
\end{lemma}

\begin{proof} The proof of (\ref{duh:update}) is immediate. To derive
  (\ref{duh:etd2}) we substitute
  \begin{equation*}
    f(t+z) = f(t) +\frac{z}{\Delta t} \left[ f(t+\Delta t)-f(t)\right]
    +{\mathcal O}\left((\Delta t)^2\right) \,,
  \end{equation*}
  into the update integral ${\mathcal U}$ in (\ref{duh:update}) and
  integrate the resulting expression explicitly.
\end{proof}

We now develop a time-marching scheme for (\ref{2:reduced_2}), which we write
compactly as
\begin{equation}\label{duh:r2}
  D_{j}(t) = \eta_j B_{j}(t) + \gamma_j u_{1j}(t) +
  \sum_{\stackrel{k=1}{k\neq j}}^{N} C_{jk}(t) \,,
\end{equation}
where $C_{jk}(t)$ and $D_{j}(t)$ are defined in (\ref{4:limit_2}) and
(\ref{4:limit_3}), respectively. We observe that both memory integrals
$C_{jk}(t)$ and $D_j(t)$ are improper, as their kernels each have an
integrable singularity at $t=\tau$.

\subsection{Sum-of-Exponentials Approximation for $D_{j}(t)$
  and $C_{jk}(t)$}\label{sec:sum_of_exponetials}

We first develop a \emph{sum-of-exponentials} approximation for
$D_{j}(t)$. To do so, we use the sector analyticity of the Laplace
transform of the exponential integral to deform the initial vertical
Bromwich line $\Gamma_B$ to the curve $\Gamma$, defined by
(\ref{dh:curve}), with endpoints at infinity in the left-half plane
$\mbox{Re}(s)<0$ (see Fig.~\ref{fig:laplace_discret} below). This
yields that
\begin{equation}\label{dj:lap}
  E_1(\sigma t) =- \frac{1}{2\pi i} \int_{\Gamma} {\mathcal E}(s) e^{st} \, ds
  \quad
  \mbox{where} \quad {\mathcal E}(s)=\frac{\log\left(1 + {s/\sigma}\right)}
  {s} \,.
\end{equation}
Observe that ${\mathcal E}(s)$ is analytic except along the branch cut
$\mbox{Re}(s)\leq -\sigma$ with $\mbox{Im}(s)=0$.

The \emph{sum-of-exponentials} method following \cite{beylkin_2005} (see
also \cite{greengard_2015}) establishes rigorous results for the
quadrature of (\ref{dj:lap}) along a family of hyperbolic shaped
curves
\begin{equation}\label{dh:curve}
  \Gamma \equiv \lbrace{\, s=\chi P(x) \,, \,\, x\in \R\rbrace} \,, \quad
  \mbox{where} \quad P(x)\equiv 1-\sin(\alpha + i\, x) \,,
\end{equation}
with $0<\alpha<{\pi/2}$ and $\chi>0$, where $s=\chi(1-\sin\alpha)>0$
at $x=0$. The curve $\Gamma$  has the limiting behavior
\begin{equation*}
  \mbox{Im}(s) \to \mp\infty\,,\quad \mbox{Re}(s)\to -\infty\,,\quad
  \frac{\mbox{Im}(s)}{\mbox{Re}(s)}\to \chi\cot\alpha \,, \quad
  \mbox{as} \,\,\, x \to \pm \infty\,.
\end{equation*}

To evaluate (\ref{dj:lap}) on $\Gamma$ we use $ds=\chi P^{\prime}(x)\, dx$
and (\ref{dh:curve}) for $P(x)$ to calculate that
\begin{equation}\label{dj:lap_x}
  E_1(\sigma t) = \frac{\chi}{2\pi} \int_{-\infty}^{\infty}
  e^{\chi P(x) t} {\mathcal E}\left[\chi P(x)\right]\, \cos(\alpha + i\, x)
    \, dx \,.
\end{equation}
By discretizing (\ref{dj:lap_x}) uniformly in $x$, with
$x_{\ell}=\ell h$ for $|\ell|\leq n$, and by labeling $s_{\ell}=\chi P(x_{\ell})$,
we get
\begin{subequations}\label{dj:lap_dsic}
\begin{equation}\label{dj:lap_disca}
  E_1(\sigma t) \approx E_n(t) \equiv \sum_{\ell=-n}^{n} e_{\ell} \,
  e^{s_{\ell} t}\,,
\end{equation}
where the coefficients are given explicitly by
\begin{equation}\label{dj:lap_discb}
  e_{\ell} = \frac{\chi h}{2\pi} \cos\left(\alpha + i \, \ell h\right)
\frac{\log\left(1 + {s_{\ell}/\sigma}\right)}
  {s_{\ell}} \,, \quad
  s_{\ell} = \chi\left[1-\sin(\alpha)\cosh(\ell h)\right]
  -i \chi \cos(\alpha)\sinh(\ell h) \,.
\end{equation}
\end{subequations}

Since $|{\mathcal E}(s)|\leq {1/\left(2|s|^{1/2}\right)}$ in the cut
plane $\mathbb{C}\backslash \, (-\infty,0)$ as $|s|\to \infty$, Lemma
1 and Corollary 1 of \cite{greengard_2015}, as adapted from
\cite{lopez_3}, provides the estimate for the difference
$|E_{1}(\sigma t)-E_n(t)|$.  The result is as follows.

\begin{lemma}\label{prop:approx}(cf.~\cite{greengard_2015})
  Consider the time interval
    $0<\delta\leq t\leq T_f$, with $T_f\geq 1000\delta$ and let
    $\varepsilon_f$ with $0<\varepsilon_f<0.1$ be a prescribed
    error-tolerance. Then, for the choice of parameters $h$ and $\chi$
    defined by
    \begin{equation}\label{approx:cond}
      h=\frac{a(\theta)}{n} \,, \quad
      \chi=\frac{2\pi \beta n(1-\theta)}{T_f a(\theta)} \,, \quad
      \mbox{where} \quad
      a(\theta)\equiv \mbox{cosh}^{-1}\left(\frac{2T_f}{\delta(1-\theta)
          \sin\alpha}\right) \,, \quad
    \end{equation}
    with $0<\alpha-\beta<\alpha+\beta<{\pi/2}$ and $0<\theta<1$, we have
    the uniform estimate
    \begin{equation}\label{approx:thm}
      \| E_{1}(\sigma t) - E_n(t)\|\leq \frac{\varepsilon_f}{\sqrt{t}} 
      \quad \mbox{on} \quad \delta\leq t\leq T_f \,,
    \end{equation}
    when $n$ is sufficiently large of the order
    \begin{equation}\label{approx:n_cond}
   n = {\mathcal O} \left( \left(-\log\varepsilon_f + \log\log\left({T_f/\delta}
          \right) \right) \log\left({T_f/\delta}\right)\right) \,.
    \end{equation}
\end{lemma}

In (\ref{approx:cond}), the parameters $\alpha$, $\beta$ and $\theta$,
satisfying the constraint in Lemma \ref{prop:approx}, can be
optimized so as to minimize the number of terms needed to achieve a
prescribed accuracy. As in \cite{greengard_2015}, we chose $\alpha=0.8$
and $\beta=0.7$ in all the results below with $\theta$ in the range
$[0.90,0.95]$ (see below). The key conclusion is that the number of terms in
the sum grows very slowly as either the tolerance $\varepsilon_f$
decreases or as ${T_f/\delta}$ increases.

In a similar way, we can develop a sum-of-exponentials approximation
for $C_{jk}(t)$ defined in (\ref{4:limit_2}). From the
Laplace-transform pair
\begin{equation}\label{cjk:laplace}
{\mathcal L}\left[G(a_{jk},t)\right]  = 2
K_0\left(2 a_{jk} \sqrt{s+\sigma}\right)  \,, \quad
\mbox{where}  \quad G(a_{jk},t) \equiv \frac{e^{-\sigma t}}{t}
e^{-a_{jk}^2/t} \,, \quad a_{jk}\equiv
\frac{|\x_j-\x_k|}{\sqrt{4D}}\,,
\end{equation}
we observe that $K_0\left(2 a_{jk} \sqrt{s+\sigma}\right)$ is analytic
in the complex $s$-plane except across the branch cut
$\mbox{Re}(s)\leq -\sigma$ with $\mbox{Im}(s)=0$.  Therefore, we can
approximate $G(a_{jk},t)$ by a sum of exponentials, similar to that
done in (\ref{dj:lap_dsic}), to obtain
\begin{subequations}\label{cj:lap_dsic}
\begin{equation}\label{cj:lap_disca}
  G(a_{jk},t) \approx G_{n}(a_{jk},t) \equiv \sum_{\ell=-n}^{n} \omega_{jk\ell} \,
  e^{s_{{\ell}} t}\,,
\end{equation}
where the coefficients are
\begin{equation}\label{cj:lap_discb}
\omega_{jk\ell} = \frac{\chi h}{2\pi} \cos\left(\alpha + i \, \ell h\right)
 2 K_0\left(2 a_{jk} \sqrt{s_{\ell}+\sigma}\right) \,, \qquad
  s_{\ell} = \chi P(x_{\ell})\,.
\end{equation}
\end{subequations}
Owing to the same branch cut structure as for approximating
$E_1(\sigma t)$, for simplicity in our approximation
(\ref{cj:lap_disca}) we choose the same values of the parameters
$\alpha=0.8$ and $\beta=0.7$ used in (\ref{dj:lap_dsic}), which
provides a comparable accuracy. We emphasize that the
coefficients $\omega_{jk\ell}$ will depend on $j$ and $k$ owing to the
inter-cell distances $|\x_j-\x_k|$. For $N$ cells, there are
$N(N-1)/2$ sets
$\lbrace{\omega_{jk\ell} \, \vert \, -n\leq \ell \leq n\rbrace}$ of
coefficients that need to be calculated once and stored for the
algorithm developed below in \S \ref{sec:all_march}.

We now provide some numerical results for three
\emph{sum-of-exponentials} approximations. Under the conditions of
Lemma \ref{prop:approx} the difference between a function $f(x, t)$,
with Laplace transform $\hat{f}(x,s)$, and its
\emph{sum-of-exponentials} approximation
$f_{a}(x,t)=\sum_{-n}^n w_k e^{s_k t}$ with an error magnitude
$\varepsilon_f$ is
\begin{equation} \label{eq:erroreps}
  \| f(x, t) - f_a(x, t)\| \leq \frac{\varepsilon_f}{\sqrt{t}}\,, \quad
  \mbox{on} \quad  t \in [\delta,T]\,.
\end{equation}
In the tables below $f(x,t)$ is taken either as the 1-D heat-kernel
$G_{1D}(x,t)$ considered in \cite{greengard_2015}, the 2-D heat-kernel
with bulk degradation
$G_{2D}(x, t) \equiv (4 \pi t)^{-1} e^{-x^2/(4 t) - \sigma
  t}=(4\pi)^{-1}G(x,t)$, where $G(x,t)$ is given in
(\ref{cjk:laplace}), or the exponential integral $E_{1}(\sigma t)$. Recall
that the discretization points are
\begin{equation*}
  w_k = \frac{\chi h}{2
\pi} \cos(\alpha + ikh) \hat{f}(x,s_k)\,, \quad \mbox{where} \quad
    s_k = \chi (1 - \sin(\alpha + ikh))\,.
\end{equation*}
In order to estimate $\sqrt{t} \, \| f(x, t) - f_a(x, t)\|$ depending
on the number $2n+1$ of terms in the \emph{sum-of-exponentials}
approximation we proceed as in \cite{greengard_2015}. We take a
$50 \times 1000$ grid $(x_j, t_k)$ where $x_0=0$, $x_j=2^{-16+j}$ for
$j\in \{1, ..., 49\}$ (not applicable for $E_1(\sigma t)$), and
$t_k = \delta e^{k\blacktriangle}$ with
$\blacktriangle\equiv{\log\left({T/\delta}\right)/999}$ for
$k\in\lbrace{0,\ldots, 999\rbrace}$ (equi-spaced on a logarithmic scale). We
then compute
$\max_{j, k} \, \sqrt{t} \, | f(x_j, t_k) - f_a(x_j, t_k)|$, which
provides the error results given in the tables below on the three
different intervals $I_1 \equiv [\delta, T] = [10^{-3}, 1]$,
$I_2 \equiv [\delta, T] = [10^{-3}, 10^3]$ and
$I_3 \equiv [\delta, T] = [10^{-5}, 10^4]$.

In Table \ref{tab:G1eps} we first reproduce the numerical error table
of \cite{greengard_2015} for the \emph{sum-of-exponentials}
approximation of the 1-D heat kernel
$G_{1D}(x, t) \equiv \left(4 \pi t\right)^{-1/2} e^{-x^2/(4t)}$.

\begin{table}[h!]
\centering
\begin{tabular}{ c c c c }
        \hline
 $\varepsilon_f$   &$n(I_1)$       &$n(I_2)$ & $n(I_3)$ \\
        \hline
$10^{-3}$   &$15  \,\, (6.598\cdot 10^{-4})$    &$23 \,\, (9.768\cdot 10^{-4})$
&$32  \,\, (9.095\cdot 10^{-4})$ \\
$10^{-6}$   &$31  \,\, (7.795\cdot 10^{-7})$    &$50 \,\, (8.709\cdot 10^{-7})$
&$68  \,\, (9.643\cdot 10^{-7})$ \\
$10^{-9}$   &$47  \,\, (9.199\cdot 10^{-10})$   &$77 \,\, (9.107\cdot 10^{-10})$
&$105 \,\, (8.817\cdot 10^{-10})$ 
\end{tabular}
\caption{The number $2n+1$ of terms needed to approximate the 1-D heat
  kernel $G_{1D}(x, t) \equiv (4\pi t)^{-1/2} e^{-x^2/(4t)}$ with
  Laplace transform
  ${\mathcal L}\left[G_1(x,t)\right]= \frac{1}{2 \sqrt{s}}
  e^{-\sqrt{s}|x|}$ to a precision $\varepsilon_f$ corresponding to
  \eqref{eq:erroreps} for the three time intervals
  $I_1 = [10^{-3}, 1]$, $I_2 = [10^{-3}, 10^3]$ and
  $I_3 = [10^{-5}, 10^4]$.  This table essentially reproduces Table 1
  in \cite{greengard_2015} with a few additional explicit error
  values.}
\label{tab:G1eps}
\end{table}

In Table \ref{tab:G2eps} and Table \ref{tab:E1eps} we provide similar
numerical error tables for the 2-D heat kernel with degradation
$G_{2D}(x, t) = (4\pi t)^{-1} e^{-x^2/(4 t) - \sigma t}$ with
$\sigma=1$ and for the exponential integral
$E_1(\sigma t) = \int_{\sigma t}^\infty \eta^{-1}e^{-\eta}\, d\eta$
with $\sigma=1$, respectively.  For the 2-D heat kernel, we observe
that we require a slightly higher number of terms in the
\emph{sum-of-exponentials} approximation than for the 1-D heat kernel
or the exponential integral.  However, this causes essentially no time
constraint for our overall numerical scheme in \S \ref{sec:all_march}
since the exponential sum representation has to be created only once
at the beginning of the time marching stepping for a given spatial
arrangement of cells. Creating these exponential sum
approximations for both the 2-D heat kernel and the exponential
integral takes less than $2s$ on a laptop for $n(I_2) = 77$.

\begin{table}[h!]
\centering
\begin{tabular}{ c c c c }
        \hline
 $\varepsilon_f$  &$n(I_1)$   &$n(I_2)$  &$n(I_3)$ \\
        \hline
$10^{-3}$   &$31 \,\, (8.085\cdot 10^{-4})$    &$49 \,\, (8.663\cdot 10^{-4})$
&$91  \,\, (8.089\cdot 10^{-5})$ \\
$10^{-6}$   &$45  \,\, (6.690\cdot 10^{-7})$    &$75 \,\, (9.272\cdot 10^{-7})$
&$114  \,\, (9.799\cdot 10^{-7})$ \\
$10^{-9}$   &$64  \,\, (7.028\cdot 10^{-10})$   &$110 \,\, (8.417\cdot 10^{-11})$
&$150 \,\, (9.343\cdot 10^{-10})$ 
\end{tabular}
\caption{Number $2n+1$ of terms needed to approximate the 2-D heat
  kernel with degradation
  $G_{2D}(x, t) = (4\pi t)^{-1} e^{-x^2/(4 t) - \sigma t}$ with
  Laplace transform
  ${\mathcal L}\left[G_{2D}(x,t)\right] = (2\pi)^{-1} K_0(x
  \sqrt{s+\sigma})$ to a precision $\varepsilon_f$ for $\sigma=1$
  corresponding to \eqref{eq:erroreps} for the three time intervals
  $I_1 = [10^{-3}, 1]$, $I_2 = [10^{-3}, 10^3]$ and
  $I_3 = [10^{-5}, 10^4]$.}
\label{tab:G2eps}
\end{table}

\begin{table}[h!]
\centering
\begin{tabular}{ c c c c }
        \hline
$\varepsilon_f$   &$n(I_1)$               &$n(I_2)$   &$n(I_3)$ \\
        \hline
$10^{-3}$   &$31  \,\, (8.824\cdot 10^{-5})$    &$49 \,\, (9.457\cdot 10^{-5})$
&$91  \,\, (8.424\cdot 10^{-5})$ \\
$10^{-6}$   &$45  \,\, (8.362\cdot 10^{-7})$    &$75 \,\, (7.902\cdot 10^{-7})$
&$114 \,\, (9.532\cdot 10^{-7})$ \\
$10^{-9}$   &$64  \,\, (9.620\cdot 10^{-10})$   &$110 \,\, (9.229\cdot 10^{-10})$
&$150 \,\, (9.525\cdot 10^{-10})$ 
\end{tabular}
\caption{Number $2n+1$ of terms needed to approximate the exponential
  integral $E_1(\sigma t)$ for $\sigma =1$ with Laplace transform
  ${\mathcal L}\left[E_1(\sigma t)\right] = {\log\left(1
      +{s/\sigma}\right)/s}$ to a precision $\varepsilon_f$
  corresponding to \eqref{eq:erroreps} for the three time intervals
  $I_1 = [10^{-3}, 1]$, $I_2 = [10^{-3}, 10^3]$ and
  $I_3 = [10^{-5}, 10^4]$.}
\label{tab:E1eps}
\end{table}

We remark that $\theta=0.95$ was chosen for Table \ref{tab:G2eps}
while $\theta=0.90$ was chosen for Table \ref{tab:E1eps}. For Table
\ref{tab:G1eps}, $\theta=0.9$ was chosen only for $I_1$, with
$\theta = 0.95$ otherwise. In our time-stepping algorithm developed
below in \S \ref{sec:all_march} and implemented in \S \ref{sec:selkov}
for Sel'kov reaction kinetics, we primarily used $\theta=0.95$ and
$n=75$ for the \emph{sum-of-exponentials} approximation for
$E_1(\sigma t)$ and $G_{2D}(x,t)$, so that the discretization points
$s_{\ell}$ were common to both approximations. Overall, this achieved
an accuracy of roughly $10^{-9}$ for $E_1(\sigma t)$ on the time
interval $[10^{-3},10^{3}]$. In \S \ref{cell:hexagonal} below, where
we will study the phase coherence of intracellular oscillations by
computing the Kuramoto order parameter over long time intervals, we
will use $n=114$.

\begin{figure}[htbp]
  \centering
 \includegraphics[width=1.0\textwidth,height=4.5cm]{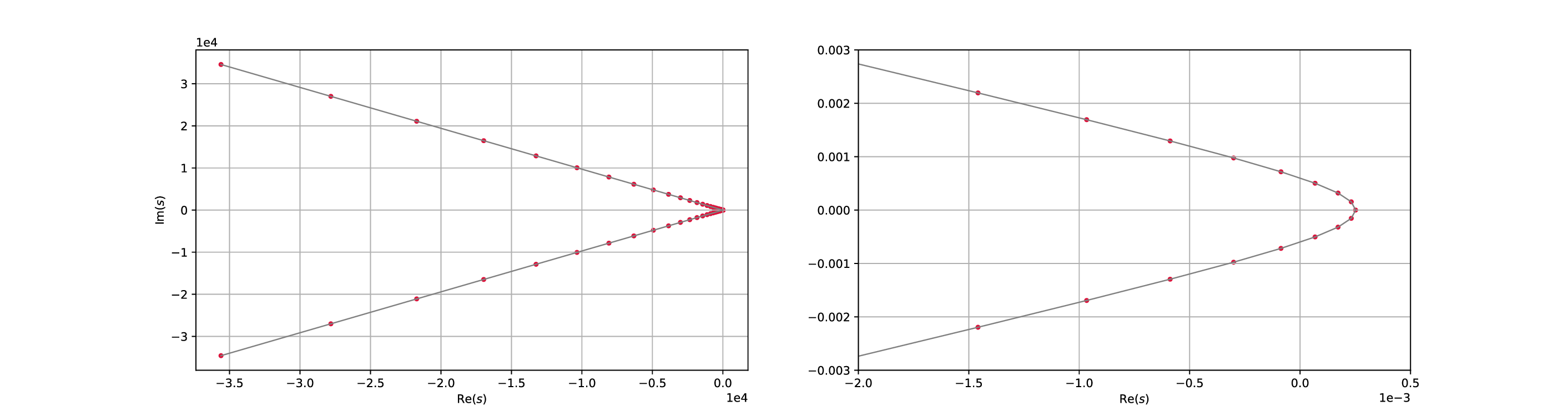}
      \caption{Left: Plot of the hyperbolic-shaped contour $\Gamma$ in
        the Laplace transform plane defined by (\ref{dh:curve}) with
        $\theta=0.95$, $\alpha=0.8$ and $\beta=0.7$ in
        (\ref{approx:cond}). The discretization points
        $\lbrace{s_{\ell}\rbrace}$ for the \emph{sum-of-exponentials}
        approximation of $E_{1}(\sigma t)$ with $\sigma=1$ and $n=75$
        are the red dots.  Right: Zoomed plot showing the
        discretization near the origin.}
    \label{fig:laplace_discret}
\end{figure}

\subsection{Time-marching scheme for $D_{j}(t)$}\label{sec:dj_march}

We first decompose $D_{j}(t)$ into the sum of a local term $D_{Lj}(t)$ near
the singularity of the kernel and a history term $D_{Hj}(t)$ as
\begin{subequations}\label{dj:decomp}
\begin{equation}\label{dj:decomp_1}
  D_{j}(t)=D_{Hj}(t)+D_{Lj}(t) \,,
\end{equation}
where for some $\Delta t$, with $0<\Delta t\ll 1$, we define
\begin{align}
  D_{Hj}(t) &\equiv \int_{0}^{t-\Delta t} B_{j}^{\prime}(\tau) E_1\left[\sigma
              (t-\tau)\right] \, d\tau \,, \label{dj:decomp_H} \\
  D_{Lj}(t) &\equiv \int_{t-\Delta t}^{t} B_{j}^{\prime}(\tau) E_1\left[\sigma
              (t-\tau)\right] \, d\tau = \int_{0}^{\Delta t} B_{j}^{\prime}(t-z)
              E_1(\sigma z) \, dz \,. \label{dj:decomp_L}
\end{align}
\end{subequations}

{The approximation given in the next lemma for the local
  contribution $D_{Lj}(t)$ ensures that the quadrature is exact for
  linear functions on ranges where $B_{j}^{\prime}(\tau)$ is
  continuous.  However, the estimate for $D_{Lj}(\Delta t)$ is more delicate
  since, owing to the singular behavior for $B_j^{\prime}(\tau)$ given
  in (\ref{mat:bj}) when $u_{1j}(0)\neq 0$, we cannot in general
  assume that $B_j^{\prime}(\tau)$ is approximately linear on
  $0<\tau<\Delta t$.}

\begin{lemma}\label{prop:dlj} For $0<\Delta t\ll 1$ and for $t\geq 2\Delta t$,
  we have
  \begin{subequations}\label{dlj:approx}
  \begin{equation}\label{dlj:approx_1}
      D_{Lj}(t) = B_j^{\prime}(t) b_1 + B_{j}^{\prime}(t-\Delta t) b_2 +
      {\mathcal O}\left((\Delta t)^3\right)\,,
  \end{equation}
  where $b_2$ and $b_1$ are defined by
  \begin{equation}\label{dlj:b1b2}
    b_2 = \frac{\Delta t}{2} E_1(\sigma \Delta t) -
    \frac{e^{-\sigma \Delta t}}{2\sigma}
    + \frac{\left(1-e^{-\sigma \Delta t}\right)}{2\sigma^2 \Delta t} \,, \qquad
    b_1 = \Delta t E_1(\sigma \Delta t) +
    \frac{\left(1-e^{-\sigma \Delta t}\right)}{
      \sigma} - b_2 \,.
  \end{equation}
\end{subequations}
{Moreover, for $t=\Delta t$, and with $B_j$ satisfying the limiting
behavior in (\ref{mat:bj}), we have the exact relation}
\begin{equation}\label{dlj:first}
{  D_{Lj}(\Delta t) = E_1(\sigma \Delta t) B_{j}(\Delta t) +
  \int_{0}^{\Delta t} \left( B_{j}(\Delta t) - B_{j}(\Delta t-z)\right)
  \frac{ e^{-\sigma z}}{z} \, dz \,.}
\end{equation}
\end{lemma}

\begin{proof} On $0\leq z\leq \Delta t$, and assuming $t\ge 2\Delta t$, we
  can approximate
  \begin{equation*}
    B_{j}^{\prime}(t-z) = B_{j}^{\prime}(t) + \frac{z}{\Delta t}
    \left[ B_{j}^{\prime}(t-\Delta t) - B_{j}^{\prime}(t) \right] +
    {\mathcal O}\left((\Delta t)^2\right) \,,
  \end{equation*}
  in $D_{Lj}(t)$ of (\ref{dj:decomp_L}). By integrating the resulting
  expression explicitly we get
  \begin{equation*}
    D_{Lj}(t) = \frac{B_{j}^{\prime}(t)}{\Delta t} \int_{0}^{\sigma \Delta t}
    E_1(\eta) \, d\eta + \frac{\left[B_{j}^{\prime}(t-\Delta t)-B_j^{\prime}(t)
      \right]}{\sigma^2\Delta t} \int_{0}^{\sigma\Delta t} \eta E_1(\eta)\,
    d\eta + {\mathcal O}\left((\Delta t)^3\right)\,.
  \end{equation*}
  Using $\int_{0}^{x} E_1(\eta) \, d\eta =x E_1(x) +1 - e^{-x}$ and
  $\int_{0}^{x} \eta E_{1}(\eta) \, d\eta = \frac{1}{2}\left[x^2 E_1(x) - (x+1)
      e^{-x} + 1 \right]$, we get (\ref{dlj:approx}).

  {To derive (\ref{dlj:first}), we integrate
    $D_{Lj}(\Delta t)$ by parts using $B_j(0)=0$ to obtain
    \begin{equation}\label{dlj:first_1}
      \begin{split}
  D_{Lj}(\Delta t) &=  -\lim_{z\to 0} E_1(\sigma z)\left[ B_j(\Delta t) -
        B_j(\Delta t-z)\right]  +
      E_1(\sigma \Delta t) B_{j}(\Delta t) \\
       & \qquad +
        \int_{0}^{\Delta t} \left( B_{j}(\Delta t) - B_{j}(\Delta t-z)\right)
        \frac{ e^{-\sigma z}}{z} \, dz \,.
      \end{split}
    \end{equation}
    Upon using the local behavior (\ref{mat:bj}) for $B_j$, together
    with $E_1(\sigma z)\sim -\log(\sigma z)+{\mathcal O}(1)$ as
    $z\to 0$, we readily obtain that the first term on the right-hand
    side of (\ref{dlj:first_1}) vanishes when either $u_{1j}(0)=0$ or
    $u_{1j}(0)\neq 0$. This yields (\ref{dlj:first}). In particular,
    for $u_{1j}(0)\neq 0$ we have the estimate
    \begin{equation*}
      \lim_{z\to 0} E_1(\sigma z)\left[ B_j(\Delta t) -
        B_j(\Delta t-z)\right]= -\frac{u_{1j}(0)\gamma_j}{\Delta t
        \log\left({\Delta t/(\kappa_j e^{-\gamma_e})}\right)}
      \lim_{z\to 0}\left[-z\log(\sigma z)\right]=0 \,.
     \end{equation*}}
\end{proof}

For the history term $D_{Hj}(t)$, we substitute (\ref{dj:lap_disca}) into
(\ref{dj:decomp_H}) to obtain
\begin{equation}\label{djh:main}
  D_{Hj}(t) \approx \sum_{\ell=-n}^{n} e_{\ell} H_{Dj\ell}(t) 
  \quad \mbox{where} \quad H_{Dj\ell}(t) \equiv \int_{0}^{t-\Delta t}
  B_{j}^{\prime}(\tau) e^{s_{\ell}(t-\tau)} \, d\tau \,.
\end{equation}
For each \emph{history-mode} $H_{Dj\ell}(t)$ we adapt Lemma
\ref{prop:duh} to readily obtain for $t\geq \Delta t$ that
\begin{subequations}\label{dhj:history_all}
\begin{equation}\label{djh:history}
  \begin{split}
  H_{Dj\ell}(t+\Delta t) &= H_{Dj\ell}(t) e^{s_{\ell}\Delta t} +
  U_{Dj\ell}(t,\Delta t) \,, \qquad \mbox{with} \quad H_{Dj\ell}(\Delta t)=0\,, \\
  \noindent\mbox{where} \qquad
  U_{Dj\ell}(t,\Delta t)  &\equiv e^{2s_{\ell}\Delta t} \int_{0}^{\Delta t}
  e^{-s_{\ell} z} B_{j}^{\prime}(t-\Delta t + z) \, dz \,.
 \end{split}
\end{equation}
{To approximate (\ref{djh:history}) for $t\geq 2\Delta t$, we
  use the ETD2 scheme that ensures that the update integral
  $U_{D\ell j}$ is exact for linear functions. This yields that
\begin{equation}\label{djh:updateU_1}
  U_{Dj\ell}(t,\Delta t) \approx
    B_{j}^{\prime}(t) b_{3\ell} + B_{j}^{\prime}(t-\Delta t) b_{4\ell} \,, 
    \qquad \mbox{for} \quad t\geq 2\Delta t \,,
\end{equation}
where $b_{3\ell}$ and $b_{4\ell}$ are defined by
\begin{equation}\label{djh:b3b4}
  b_{3\ell} = \frac{e^{s_{\ell}\Delta t}}{s_{\ell}^2 \Delta t}
  \left( e^{s_{\ell} \Delta t} - 1 - s_{\ell}\Delta t \right)\,, \qquad
  b_{4\ell} = e^{s_{\ell} \Delta t} \frac{ \left(e^{s_{\ell}\Delta t}-1\right)}{s_{\ell}}
  - b_{3\ell} \,.
\end{equation}
For $t=\Delta t$, we integrate 
$U_{Dj\ell}(\Delta t,\Delta t)$ from (\ref{djh:history}) by parts and use
$B_j(0)=0$ to get
\begin{equation}\label{dhj:first_1}
  U_{Dj\ell}(\Delta t, \Delta t) = e^{s_{\ell} \Delta t} B_{j}(\Delta t)
  + s_{\ell} e^{2 s_{\ell} \Delta t} \int_{0}^{\Delta t} e^{-s_{\ell} z} B_{j}(z)\, dz\,.
\end{equation}
By using an ETD1 scheme to estimate
$\int_{0}^{\Delta t} e^{-s_{\ell} z} B_j(z)\, dz\approx B_j(\Delta t)
\int_{0}^{\Delta t} e^{-s_{\ell} z}\, dz$, we readily obtain after
calculating the integral $\int_{0}^{\Delta t} e^{-s_{\ell} z}\, dz$
explicitly that}
\begin{equation}\label{djh:first}
{  U_{Dj\ell}(\Delta t,\Delta t) \approx
    B_{j}(\Delta t) e^{2s_{\ell} \Delta t} \,.    }
\end{equation}
\end{subequations}

{In this way, by combining (\ref{djh:main})--(\ref{djh:updateU_1}), we
obtain for $t\geq 2 \Delta t$ that
\begin{equation}\label{djh:march_sup}
  D_{Hj}(t+\Delta t) \approx \sum_{\ell=-n}^{n} e_{\ell} \, e^{s_{\ell} \Delta t}
  H_{Djl}(t) + \left( \sum_{\ell=-n}^{n} e_{\ell} \, b_{3\ell}\right)
   B_{j}^{\prime}(t) +  \left( \sum_{\ell=-n}^{n} e_{\ell} \, b_{4\ell}\right)
   B_{j}^{\prime}(t-\Delta t)\,.
 \end{equation}
Moreover, we have
\begin{equation}\label{djh:march_sup_first}
  D_{Hj}(2\Delta t) \approx \left(\sum_{\ell=-n}^{n} e_{\ell} \, e^{2s_{\ell} \Delta t}
 \right) B_j(\Delta t) \,, \qquad D_{Hj}(\Delta t)=0 \,.
 \end{equation}}

{Finally, by using (\ref{dlj:approx}), (\ref{djh:march_sup}) and
(\ref{djh:march_sup_first}) we readily obtain the
marching scheme
\begin{subequations}\label{dhj:march}
\begin{equation}\label{dhj:march_1}
  \begin{split}
    D_j(t+\Delta t) &\approx B_{j}^{\prime}(t+\Delta t) b_1 +
    B_{j}^{\prime}(t) \left(b_2 +  \sum_{\ell=-n}^{n} e_{\ell} b_{3\ell}\right) +
   \left( \sum_{\ell=-n}^{n} e_{\ell} \, b_{4\ell}\right)
   B_{j}^{\prime}(t-\Delta t) \\
   & \qquad + \sum_{\ell=-n}^{n} e_{\ell} \, e^{s_{\ell} \Delta t}H_{Djl}(t) \,,
   \qquad t\geq 2\Delta t \,,\\
   D_j(2\Delta t) &\approx  B_j^{\prime}(2\Delta t) b_1 +
   B_{j}^{\prime}(\Delta t)b_2 + \left(\sum_{\ell=-n}^{n} e_{\ell} e^{2s_{\ell}\Delta t}
   \right) B_j(\Delta t) \,,\\
   D_{j}(\Delta t) &= E_1(\sigma \Delta t) B_{j}(\Delta t) +
  \int_{0}^{\Delta t} \left( B_{j}(\Delta t) - B_{j}(\Delta t-z)\right)
  \frac{ e^{-\sigma z}}{z} \, dz \,.
\end{split}
\end{equation}

From (\ref{dhj:history_all}) we have the following update rule for the
next time-step:
\begin{equation}\label{dhj:hd_update}
  H_{Dj\ell}(t+\Delta t) = H_{Dj\ell}(t)e^{s_{\ell} \Delta t} +
  \begin{cases}
    B_{j}^{\prime}(t) b_{3\ell} + B_{j}^{\prime}(t-\Delta t) b_{4\ell} \,, & \quad
    t\geq 2\Delta t \,,\\
    B_{j}(\Delta t) e^{2s_{\ell} \Delta t} \,, & \quad t=\Delta t \,.
  \end{cases}
\end{equation}
\end{subequations}
In (\ref{dhj:march}), $b_1$, $b_2$, $b_{3\ell}$ and $b_{4\ell}$ are
defined in (\ref{dlj:b1b2}) and (\ref{djh:b3b4}).}

\subsection{Time-marching scheme for $C_{jk}(t)$}\label{sec:cjk_march}
A similar approach can be used to derive a time-marching scheme for
$C_{jk}(t)$ defined in (\ref{4:limit_2}). As in (\ref{dj:decomp}), we
decompose $C_{jk}(t)$ as
\begin{subequations}\label{cj:decomp}
\begin{equation}\label{cj:decomp_1}
  C_{jk}(t)=C_{Hjk}(t) + C_{Ljk}(t) \,,
\end{equation}
where for some $\Delta t$, with $0<\Delta t\ll 1$, we define the local and
history-dependent terms as
\begin{align}
  C_{Hjk}(t) &\equiv \int_{0}^{t-\Delta t} B_{k}(\tau) G(a_{jk},t-\tau) \,
                d\tau \,, \label{cj:decomp_H} \\
  C_{Ljk}(t) &\equiv \int_{t-\Delta t}^{t} B_{k}(\tau) G(a_{jk},t-\tau) \,
               d\tau = \int_{0}^{\Delta t} B_{k}(t-z) G(a_{jk},z) \,
               dz \,. \label{cj:decomp_L}
\end{align}
\end{subequations}

The marching scheme for the history term is derived in the same way as in
\S \ref{sec:dj_march}. We have 
\begin{equation}\label{cjh:main}
  C_{Hjk}(t) \approx \sum_{\ell=-n}^{n} \omega_{jk\ell} H_{Ck\ell}(t) 
  \quad \mbox{where} \quad H_{Ck\ell}(t) \equiv \int_{0}^{t-\Delta t}
  B_{k}(\tau) e^{s_{\ell}(t-\tau)} \, d\tau \,.
\end{equation}
For each \emph{history-mode} $H_{Ck\ell}(t)$ we use Lemma \ref{prop:duh} to
obtain the update scheme
\begin{equation}\label{cjh:history}
  \begin{split}
  H_{Ck\ell}(t+\Delta t) &= H_{Ck\ell}(t) e^{s_{\ell}\Delta t} +
  U_{Ck\ell}(t,\Delta t) \,, \qquad \mbox{with} \quad H_{Ck\ell}(\Delta t)=0\,, \\
  \noindent\mbox{where} \qquad
  U_{Ck\ell}(t,\Delta t)  &\equiv e^{2s_{\ell}\Delta t} \int_{0}^{\Delta t}
  e^{-s_{\ell} z} B_{k}(t-\Delta t + z) \, dz \,.
 \end{split}
\end{equation}
To approximate (\ref{cjh:history}) for $t\ge 2\Delta t$, we use the
ETD2 scheme to obtain
\begin{equation}\label{cjh:updateU_1}
{  U_{Ck\ell}(t,\Delta t) \approx B_{k}(t) b_{3\ell} +
  B_{k}(t-\Delta t) b_{4\ell} \,, \quad \mbox{for} \quad t\ge 2\Delta t\,,}
\end{equation}
where $b_{3\ell}$ and $b_{4\ell}$ are defined in (\ref{djh:b3b4}).
{For $t=\Delta t$, we use an ETD1 scheme to get
\begin{equation}\label{cjh:updateU_0}
 U_{Ck\ell}(\Delta t,\Delta t) \approx 
  B_{k}(\Delta t) b_{40\ell} \,, \qquad \mbox{where} \qquad
 b_{40\ell}\equiv  e^{s_{\ell} \Delta t}
  \left( \frac{e^{s_{\ell} \Delta t}-1}{s_{\ell}} \right)\,.
\end{equation}}

The approximation of the local contribution $C_{Ljk}(t)$ is simpler
than for $D_{Ljk}$ owing to the exponential decay of $G(a_{jk},z)$ for
$z>0$. We use an ETD1 scheme and estimate
\begin{equation}\label{cjh:local_approx}
  C_{Ljk}(t) \approx B_{k}(t)\int_{0}^{\Delta t} G(a_{jk},z)\, dz
  \approx B_{k}(t) E_{1}\left(\frac{a_{jk}^2}{\Delta t}\right) \,.
\end{equation}
By using $E_1(z)\sim {e^{-z}/z}$ as $z\to \infty$, we conclude that
$C_{Ljk}(t)$ is exponentially small when ${a_{jk}^2/\Delta t}\gg 1$.
However, if the diffusivity $D$ is large so that ${a_{jk}^2/\Delta t}$
is not so large, then we may need to use the higher order approximation
for $C_{Ljk}(t)$ given by
\begin{equation}\label{cjh:local_approx_better}
  C_{Ljk}(t) \approx B_{k}(t) E_{1}\left(\frac{a_{jk}^2}{\Delta t}\right)
   - \left(\Delta t e^{-a_{jk}^2/\Delta t} - a_{jk}^2
    E_{1}\left(\frac{a_{jk}^2}{\Delta t}\right) \right) \left(\sigma B_{k}(t) +
    \frac{B_{k}(t)- B_{k}(t-\Delta t)}{\Delta t} \right)\,.
\end{equation}

In this way, by combining (\ref{cjh:main})--(\ref{cjh:updateU_0}) for
the history modes and (\ref{cjh:local_approx}) for the local term, and
using only the leading-order result for $C_{Ljk}(t)$, we obtain the
following marching scheme for $C_{jk}(t)$ for $t\geq \Delta t$:
\begin{subequations}\label{chj:march}
\begin{equation}\label{chj:march_1}
  \begin{split}
{    C_{jk}(t+\Delta t)} &{\approx B_{k}(t+\Delta t)
    E_1\left(\frac{a_{jk}^2}{\Delta t}\right) +
       \left(\sum_{\ell=-n}^{n} \omega_{jk\ell} b_{3\ell}\right) B_{k}(t)} \\
      & \qquad +
      {\left( \sum_{\ell=-n}^{n} \omega_{jk\ell} \, b_{4\ell}\right)
        B_{k}(t-\Delta t)
   + \sum_{\ell=-n}^{n} \omega_{jk\ell} \, e^{s_{\ell} \Delta t}H_{Ck\ell}(t) \,,
   \,\,\, \mbox{for} \,\,\, t\ge 2\Delta t\,,}\\
{   C_{jk}(2\Delta t)} &\approx {B_{k}(2\Delta t)
   E_1\left(\frac{a_{jk}^2}{\Delta t}\right) +
   \left(\sum_{\ell=-n}^{n} \omega_{jk\ell} b_{40\ell} \right)B_{k}(\Delta t) \,,}
   \\
{ C_{jk}(\Delta t)} &\approx {B_{k}(\Delta t)
   E_1\left(\frac{a_{jk}^2}{\Delta t}\right)\,.}
 \end{split}
\end{equation}
In addition, from (\ref{cjh:updateU_1}) and (\ref{cjh:updateU_0}),
  we have the update rule
\begin{equation}\label{chj:march_update}
{  H_{Ck\ell}(t+\Delta t) = \begin{cases}
    H_{Ck\ell}(t) e^{s_{\ell}\Delta t} + B_{k}(t) b_{3\ell} + B_{k}(t-\Delta t)
    b_{4\ell}\,, & \quad t\ge 2\Delta t\,, \\
      B_{k}(\Delta t)b_{40\ell} \,, & \quad t= \Delta t\,.
  \end{cases}}
\end{equation}
\end{subequations}
{Here $b_{3\ell}$, $b_{4\ell}$ and $b_{40\ell}$  are defined in
(\ref{dlj:b1b2}), (\ref{djh:b3b4}) and (\ref{cjh:updateU_0}),
respectively.}

{
\subsection{Marching scheme for integro-differential system}
\label{sec:all_march}

By substituting (\ref{chj:march}) and (\ref{dhj:march}) into
(\ref{duh:r2}), and recalling (\ref{2:reduced_1}) for the ODE
intracellular dynamics, we now develop a time-marching algorithm for
approximating solutions to (\ref{2:reduced}).  Let
$\lbrace{t_0, ..., t_q\rbrace}$ be a discretization of the time domain
$[0, T]$ where $T$ is some final time of choice. Then, $t_0=0$,
$t_q=T$, $t_{i+1} - t_i = \Delta t$ for $i\in\{0, ..., q-1\}$, where
$\Delta t = {T/q}$ and $t_i = i \Delta t$. In the formulation of our
algorithm we will use
$B_j^{\prime}(t_i)\approx {\left(
    B_{j}(t_i)-B_{j}(t_{i-1})\right)/\Delta t}$, for $i\geq 2$ in
(\ref{dhj:march}).
Below, we denote the numerical approximation to $u_{1j}(t_i)$
by $u_{1j}^{(i)}$ and the vector of all $u_{1j}^{(i)}$ for
$j\in\lbrace{1,\ldots,N\rbrace}$ by $\v{u}_1^{(i)}$. We use a similar
notation for the reaction kinetic functions.

{\bf Step 1; $0 \curvearrowright \Delta t$:} For this first time-step
we discretize the reaction kinetics in each cell by the explicit
Runge-Kutta RK4 method \cite{numerics} and use
\begin{equation}\label{alg:step0}
  \begin{split}
  \v{u}_{1}^{(1)} &= \v{u}_{1}^{(0)} + \frac{\Delta t}{6} (k_1^{(0)} +
                  2 k_2^{(0)} + 2 k_3^{(0)} + k_4^{(0)}) + \v{B}^{(0)} \Delta t\,, \\
  \v{u}_{2}^{(1)} &= \v{u}_{2}^{(0)} + \frac{\Delta t}{6} (\tilde{k}_1^{(0)}
  + 2 \tilde{k}_2^{(0)} + 2 \tilde{k}_3^{(0)} + \tilde{k}_4^{(0)}) \,,
\end{split}
\end{equation}
where the RK4 weights are
\begin{eqnarray*}
	k_1^{(0)} &=& \v{F}_1(\v{u}_{1}^{(0)}, \v{u}_{2}^{(0)})\,, \quad
	\tilde{k}_1^{(0)}= \v{F}_2(\v{u}_{1}^{(0)}, \v{u}_{2}^{(0)})\,, \\
  k_2^{(0)} &=&  \v{F}_1(\v{u}_{1}^{(0)} + \frac{\Delta t}{2} k_1^{(0)}\,,
                \v{u}_{2}^{(0)} + \frac{\Delta t}{2} \tilde{k}_1^{(0)})\,, \quad
  \tilde{k}_2^{(0)} =\v{F}_2(\v{u}_{1}^{(0)} + \frac{\Delta t}{2} k_1^{(0)},
                     \v{u}_{2}^{(0)} + \frac{\Delta t}{2} \tilde{k}_1^{(0)})\,, \\
  k_3^{(0)} &=&  \v{F}_1(u_{1}^{(0)} + \frac{\Delta t}{2} k_2^{(0)},
                \v{u}_{2}^{(0)} + \frac{\Delta t}{2} \tilde{k}_2^{(0)})\,, \quad
  \tilde{k}_3^{(0)} = \v{F}_2(\v{u}_{1}^{(0)} + \frac{\Delta t}{2} k_2^{(0)}\,,
                   \v{u}_{2}^{(0)} + \frac{\Delta t}{2} \tilde{k}_2^{(0)})\,, \\
  k_4^{(0)} &=&  \v{F}_1(\v{u}_{1}^{(0)} + \Delta t k_3^{(0)}, \v{u}_{2}^{(0)} +
                \Delta t \tilde{k}_3^{(0)})\,, \quad
  \tilde{k}_4^{(0)}  =\v{F}_2(\v{u}_{1}^{(0)} + \Delta t k_3^{(0)},
                        \v{u}_{2}^{(0)} + \Delta t \tilde{k}_3^{(0)})\,.
\end{eqnarray*}
In (\ref{alg:step0}) we imposed the explicit short-time behavior for
$\v{B}^{(0)}$ given in (\ref{mat:bj}), which is valid since
${\mathcal O}(\eps^2)\ll \Delta t \ll {\mathcal O}(1)$. The truncation
error for this approximation is $\mathcal{O}(\Delta t)$, whereas RK4
gives  a truncation error of $\mathcal{O}(\Delta t^5)$ for
the reaction kinetics.

In terms of the computed $\v{u}_{1}^{(1)}\approx \v{u}_1(t_1)$, in
Appendix \ref{app:bdeltat} we derive an improved approximation for
$\v{B}^{(1)}=(B_{1}^{(1)},\ldots,B_{N}^{(1)})^{T}$, in which the components
are
\begin{equation}\label{alg:b1}
  B_{j}^{(1)} = -\frac{u_{1j}^{(1)} \gamma_j}
  {\log\left({\Delta t/(\kappa_je^{-\gamma_e})}\right)} \left(
    1 - \frac{\pi^2}{6\,\left[\log\left({\Delta t/(\kappa_je^{-\gamma_e})}
          \right)\right]^2}  \right)\,.
  \end{equation}
  
{\bf Step 2; $\Delta t \curvearrowright 2\Delta t$:} We use an RK4 scheme
for the reaction kinetics with RK4 weights as given in the first time step,
and use a lagged $\v{B}^{(1)}$ so that
\begin{equation}\label{alg:step2}
  \begin{split}
  \v{u}_{1}^{(2)} &= \v{u}_{1}^{(1)} + \frac{\Delta t}{6}
 (k_1^{(1)} + 2 k_2^{(1)} + 2 k_3^{(1)} + k_4^{(1)}) + \v{B}^{(1)} \Delta t, \\
  \v{u}_{2}^{(2)} &= \v{u}_{2}^{(1)} + \frac{\Delta t}{6} (\tilde{k}_1^{(1)}
  + 2 \tilde{k}_2^{(1)} + 2 \tilde{k}_3^{(1)} + \tilde{k}_4^{(1)})\,.
  \end{split}
\end{equation}
In terms of the computed $\v{u}_{1}^{(2)}$, we determine $\v{B}^{2}$
from the linear system
\begin{equation}\label{alg:b2}
  \v{B}^{(2)} = A^{-1} \left( M_1 B^{(1)} + \Delta t \Gamma u_{1}^{(2)} -
  b_2 \Delta t \v{B}^{\prime}(\Delta t)\right)\,,
\end{equation}
where the components of $\v{B}^{\prime}(\Delta t)$ are given in
(\ref{mat:bjp}). In (\ref{alg:b2}),
$\Gamma\equiv \mbox{diag}(\gamma_1,\ldots,\gamma_N)$, while the
matrices $A$ and $M_1$ have the following entries for
$k\neq j \in \{1, ..., N\}$:
\begin{equation}\label{alg:a_m1}
	\begin{array}{rclrcl}
          A_{jj} &=& b_1 - \eta_j \Delta t\,,
 \qquad &A_{jk} &=& - \Delta t E_1\left({a_{jk}^2/\Delta t}\right)\,, \\
          M_{1,jj} &=& b_1-\Delta t \left(\sum_{\ell=-n}^{n}
                       e_{\ell} e^{2s_{\ell} \Delta t}
                       \right)\,, \qquad
          &M_{1,jk} &=& \Delta t \sum_{\ell=-n}^n \omega_{jk\ell} b_{40\ell}\,.
	\end{array}
\end{equation}

{\bf Recursive step; $t_i \curvearrowright t_{i+1}$, for $i \geq 2$:} We
use the RK4 method with a lagged $\v{B}^{(i)}$, so that
\begin{equation}\label{alg:step_i}
  \begin{split}
  \v{u}_{1}^{(i+1)} &= \v{u}_{1}^{(i)} + \frac{\Delta t}{6} (k_1^{(i)} +
        2 k_2^{(i)} + 2 k_3^{(i)} + k_4^{(i)}) + \v{B}^{(i)} \Delta t\,, \\
  \v{u}_{2}^{(i+1)} &= \v{u}_{2}^{(i)} + \frac{\Delta t}{6}
  (\tilde{k}_1^{(i)} + 2 \tilde{k}_2^{(i)} + 2 \tilde{k}_3^{(i)} +
  \tilde{k}_4^{(i)})  \,.
\end{split}
\end{equation}
In terms of the computed $\v{u}_{1}^{(i+1)}$ we calculate $\v{B}^{(i+1)}$ as
\begin{equation}\label{step:bi}
  \v{B}^{(i+1)} = A^{-1} \left( M \v{B}^{(i)} + \mathcal{N} \v{B}^{(i-1)} +
  \left( \sum_{\ell=-n}^n e_\ell b_{4\ell} \right) \, \v{B}^{(i-2)} +
  \Delta t \Gamma \v{u}_{1}^{(i+1)} -
  \Delta t \v{H}^{(i)}\right)\,,
\end{equation}
where the matrices $M$ and $\mathcal{N}$ have the entries
\begin{equation}\label{alg:a_m_n}
	\begin{array}{rclrcl}
          M_{jj} &=& b_1 - b_2 - \sum_{\ell=-n}^n e_\ell b_{3\ell}\,, \qquad
          &M_{jk} &=& \Delta t \sum_{\ell=-n}^n \omega_{jk\ell} b_{3\ell}\,, \\
  \mathcal{N}_{jj} &=& b_2 + \sum_{\ell=-n}^n e_\ell (b_{3\ell} - b_{4\ell})\,, \qquad
     &\mathcal{N}_{jk} &=& \Delta t \sum_{\ell=-n}^n \omega_{jk\ell} b_{4\ell}, \\
	\end{array}
\end{equation}
for $k\neq j \in \{1, ..., N\}$. In (\ref{step:bi}) the history vector
$\v{H}^{(i)}\equiv (H_{1}^{(i)},\ldots,H_{N}^{(i)})^T$ has entries
\begin{subequations}\label{alg:hji_update}
\begin{equation}
  H_{j}^{(i)} = \sum_{\ell=-n}^n \left(e_\ell e^{s_\ell\Delta t} H_{Dj\ell}^{(i)}
    - \sum_{k=1, k\neq j}^N \omega_{jk\ell} e^{s_\ell \Delta t} H_{Ck\ell}^{(i)}\right)\,,
\end{equation}
and is updated with the scheme
\begin{equation}
  \begin{split}
  \v{H}_{D\ell}^{(i)} &= \v{H}_{D\ell}^{(i-1)} e^{s_{\ell} \Delta t} +
  b_{3\ell} \frac{\left(\v{B}^{(i-1)}-\v{B}^{(i-2)}\right)}{\Delta t} +
  \begin{cases}
   b_{4\ell} \frac{\left(\v{B}^{(i-2)}-\v{B}^{(i-3)}\right)}{\Delta t} \,,
    & \,\, \mbox{if} \,\, i\geq 4 \,,\\
    \v{B}^{\prime}(\Delta t) b_{4\ell}     & \,\, \mbox{if} \,\, i=3 \,,
  \end{cases} \\
  \v{H}_{D\ell}^{(2)} &= e^{2s_{\ell} \Delta t} \v{B}(\Delta t) \,,
 \end{split}
\end{equation}
together with
\begin{equation}
  \begin{split}
  \v{H}_{C\ell}^{(i)} &= \v{H}_{C\ell}^{(i-1)} e^{s_{\ell} \Delta t} +
  b_{3\ell} \v{B}^{(i-1)} + b_{4\ell} \v{B}^{(i-2)} \,, \quad
  \mbox{if} \quad i\geq 3\,,\\
  \v{H}_{C\ell}^{(2)} &= b_{40\ell} \v{B}(\Delta t) \,.
\end{split}
\end{equation}

\end{subequations}

Overall, our formulation is an an operator-splitting scheme of a
semi-implicit kind in the sense that the right-hand sides of the
reaction kinetic vector fields are treated explicitly using the RK4
method with a lagged $\v{B}^{(i-1)}$, while in (\ref{step:bi})
$\v{u}_1^{(i)}$ appears implicitly in the update to $\v{B}^{(i)}$.}

\section{Numerical results: Sel'kov intracellular
  dynamics}\label{sec:selkov}

To illustrate our analysis of \eqref{DimLess_bulk}, we will consider
the two-component Sel'kov kinetics.  For this choice, the
intracellular kinetics
$\v{F}_j(u_{1j},u_{2j})\equiv
\left(F_{1j}(u_{1j},u_{2j}),F_{j2}(u_{1j},u_{2j})\right)^{T}$ are
given by
\begin{equation}\label{isolated:selkov:ode}
  F_{1j} = \alpha_j u_{2j}+ u_{2j} u_{1j}^2- u_{1j}\,,
  \qquad F_{2j} =\zeta_j
  \left[\mu_j-\left(\alpha_j u_{2j}+u_{2j} u_{1j}^2\right)\right] \,.
\end{equation}
We refer to $\alpha_j>0$, $\zeta_j>0$ and $\mu_j>0$ as the
\emph{reaction-kinetic parameters} for the $j^{\mbox{th}}$ cell.

\subsection{An Isolated Cell}\label{sec:selkov_isolated}

For an isolated cell uncoupled from the bulk, the unique steady-state for
${d\v{u}_j/dt}=\v{F}_j$ is $u_{1j}=\mu_j$ and
$u_{2j} = {\mu_j/\left(\alpha_j + \mu_j^2\right)}$. At this steady-state
we calculate 
\begin{equation}\label{isolated:selkov:jac}
  \det(J_{j})= \zeta_j \left(\alpha_j + \mu_j^2\right)
  \,, \quad \mbox{tr}(J_{j}) = \frac{1}{\alpha_j+\mu_j^2}
  \left(2\mu_j \mu_j - (\alpha_j+\mu_j^2) - \zeta_j
    (\alpha_j + \mu_j^2)^2 \right) \,,
\end{equation}
for the determinant and trace of the Jacobian $J_j$ of the kinetics.
Since $\det(J_{j})>0$, the unique steady-state for an isolated cell is
linearly stable if and only if $\mbox{tr}(J_{j})<0$.  The Hopf
bifurcation boundary in the $\alpha_j$ versus $\mu_j$ parameter plane
occurs when $\mbox{tr}(J_{j})=0$. From the Poincar\'{e}-Bendixson
theorem, we conclude that whenever the steady-state is unstable the
isolated cell will have limit cycle oscillations. For $\zeta_j=0.15$,
in Fig.~\ref{fig:isolated:plane} we show the region in the $\alpha_j$
versus $\mu_j$ plane where limit cycle oscillations occur. In
Fig.~\ref{fig:isolated:limit} we plot the limit cycle in the
$u_{2j}$ versus $u_{1j}$ plane for $\alpha_j=0.4$ and $\mu_j=2.0$.

\begin{figure}[htbp]
  \centering
    \begin{subfigure}[b]{0.48\textwidth}  
      \includegraphics[width=\textwidth,height=4.5cm]{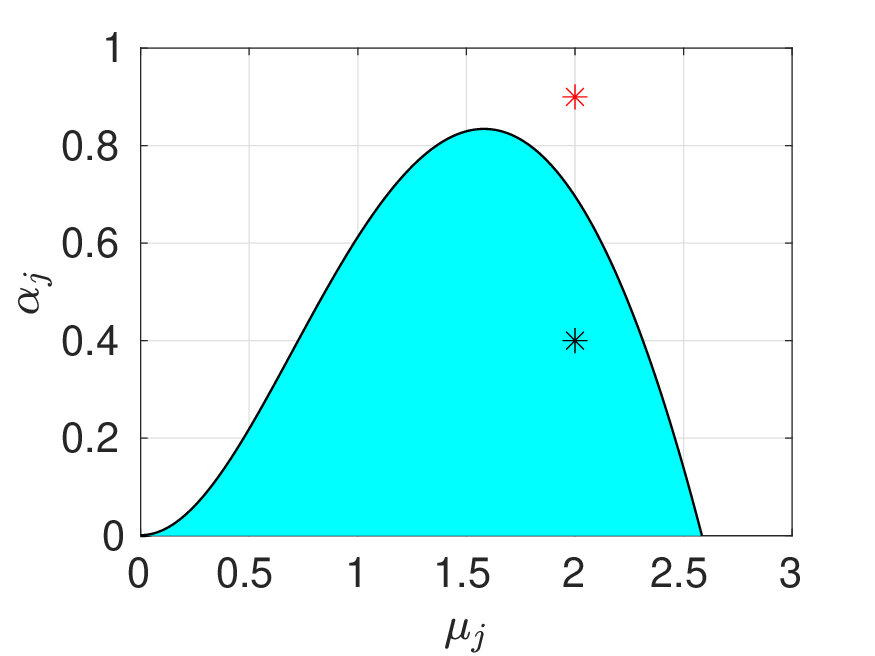}
        \caption{Instability region: isolated cell}
        \label{fig:isolated:plane}
    \end{subfigure}
    \begin{subfigure}[b]{0.48\textwidth}
      \includegraphics[width=\textwidth,height=4.4cm]{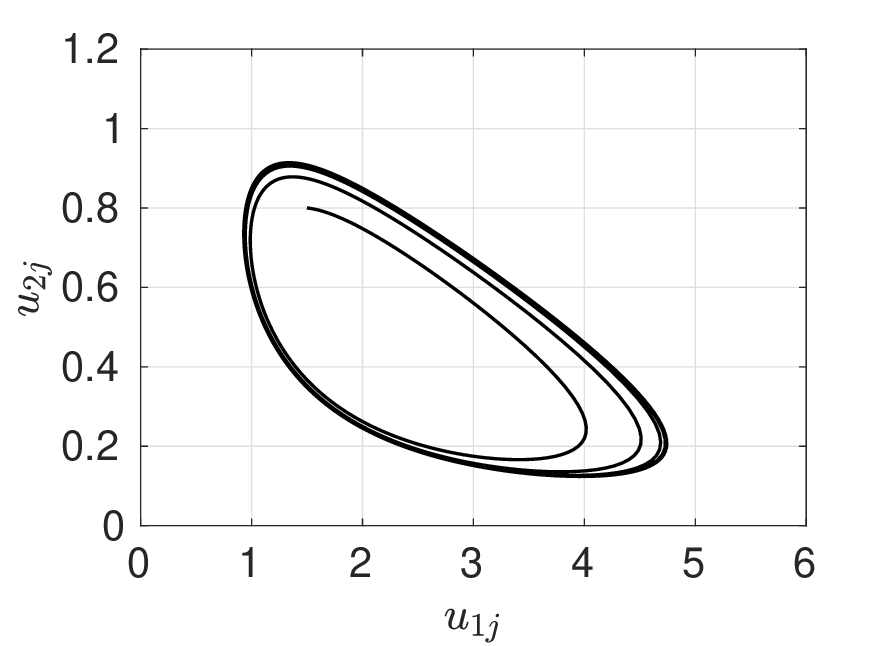}
        \caption{Limit cycle in phase plane} 
        \label{fig:isolated:limit}
    \end{subfigure}
    \caption{Left: Blue-shaded region of instability in the $\alpha_j$
      versus $\mu_j$ plane for the steady-state of an isolated cell
      when $\zeta_j=0.15$. In this region, since $\mbox{tr}(J_{j})>0$
     the unique steady-state is unstable, and a time-periodic
      (limit cycle) solution occurs for an isolated cell. In the unshaded
      region the steady-state is linearly stable for an isolated
      cell. Right: For $\alpha_j=0.4$ and $\mu_j=2.0$ (black star in
      left panel), there is a limit cycle in the $u_{2j}$ versus
      $u_{1j}$ plane.}
\label{fig:selkov}
\end{figure}

{Next, we determine how the instability region for an isolated cell
  changes when we include the effect of efflux across the cell
  membrane, but neglect any influx from the bulk medium (recall the
  schematic in Fig.~\ref{fig:schematic_zoom}). Setting $U=0$ in
  (\ref{DimLess_Intra}), for an isolated cell with boundary efflux the
  unique steady-state for
  ${d\v{u}_j/dt}=\v{F}_j-2\pi d_{2j} u_{1j}\v{e}_1$ is now
  $u_{1j}={\mu_j/(1+2\pi d_{2j})}$ and
  $u_{2j} = {\mu_j/\left(\alpha_j + u_{1j}^2\right)}$. The Hopf
  bifurcation boundary in the $\alpha_j$ versus $d_{2j}$ plane occurs
  when $\mbox{tr}(J_{j})=0$, which yields
\begin{equation}\label{isolated:efflux}
  \alpha_j = -\frac{\mu_j^{2}}{(1+2\pi d_{2j})^2} +
  \frac{1}{2\zeta_j} \left[-(1+2\pi d_{2j}) + \sqrt{
      (1+2\pi d_{2j})^2 + \frac{8\mu_j^2\zeta_j}{1+2\pi d_{2j}}}\right]\,.
\end{equation}
For $\zeta_j=0.15$ and $\mu_j=2.0$, in Fig.~\ref{fig:selkov_efflux} we
show that when $d_{2j}$ increases past a threshold, the range of
$\alpha_j$ where intracellular oscillations occur decreases significantly when
there is efflux from an isolated cell.}

\begin{figure}[htbp]
  \centering
      \includegraphics[width=0.5\textwidth,height=4.5cm]{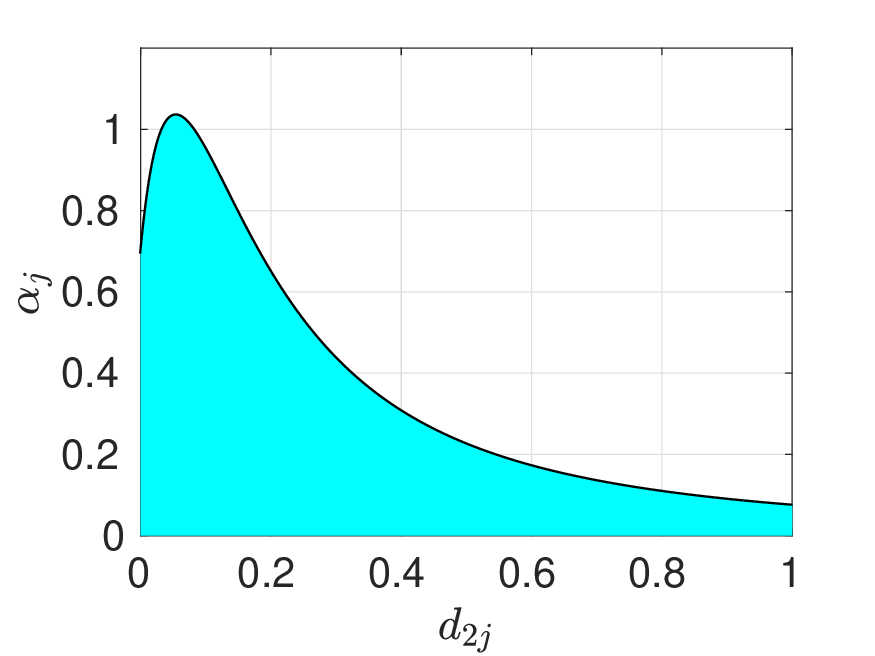}
      \caption{{Blue-shaded region of instability in the $\alpha_j$
          versus $d_{2j}$ plane for the steady-state of an isolated
          cell with boundary efflux when $\zeta_j=0.15$ and
          $\mu_j=2.0$. As $d_{2j}$ increases past a threshold, the
          range in $\alpha_j$ where limit cycle oscillations will
          occur for the isolated cell decreases.}}
\label{fig:selkov_efflux}
\end{figure}

\subsection{Steady-states and their stability with cell-bulk coupling}\label{sec:selkov_couple}

With cell-bulk coupling, the steady-state solutions of
(\ref{DimLess_bulk}) are obtained from the solution to the NAS
(\ref{4:ssintra}) and (\ref{4:ss_nas}) where $\v{F}_j$ is given in
(\ref{isolated:selkov:ode}). We readily obtain that
\begin{subequations}\label{selkov:equil}
  \begin{equation}\label{selkov:equil_1}
    u_{1js}= \mu_j + B_{js} \,, \qquad u_{2js}= \frac{\mu_j}{\alpha_j +
      u_{1js}^2} \,,
  \end{equation}
where, with $\eta_j$ and $\gamma_j$ as defined in (\ref{2:reduced_3}),
the steady-state source strengths $B_{js}$ satisfy 
\begin{equation}\label{selkov:equil_2}
  \left(\gamma_j + \eta_j\right) B_{js} + 2
  \sum_{\stackrel{k=1}{k\neq j}}^{N} B_{ks}
  K_0\left( \sqrt{\frac{\sigma}{D}} |\v{x}_j-\v{x}_k| \right) =
  -\gamma_j \mu_j \,, \quad j\in\lbrace{1,\ldots,N\rbrace} \,.
\end{equation}
\end{subequations}
We remark that although (\ref{4:ss_nas}) and (\ref{4:ssintra})
generally yields a nonlinear algebraic system characterizing
steady-state solutions, with Sel'kov kinetics one must only solve the
linear algebraic system (\ref{selkov:equil_2}).

To analyze the linear stability of this steady-state solution, we must
solve the nonlinear eigenvalue problem (\ref{TransDent}) of
Proposition \ref{prop:stab}. In calculating ${\mathcal M}(\lambda)$
from (\ref{5:gcep_2}), it is only the diagonal matrix
${\mathcal K}(\lambda)$ defined in (\ref{5:kjac}) that depends on the
choice of the intracellular kinetics. For the Sel'kov kinetics
(\ref{isolated:selkov:ode}), 
\begin{equation}\label{selkov:kmat}
 \mathcal{K}(\lambda) \equiv \mbox{diag}\left( \mathit{K}_1,\ldots,
  \mathit{K}_N\right)\,, \quad \mbox{where} \quad {\mathit K}_j \equiv
\frac{\lambda + \det{(J_{j})}}{\lambda^2 - \mbox{tr}(J_j)\lambda +
  \det{(J_j)}} \,, 
\end{equation}
where $\det{(J_j)}$ and $\mbox{tr}(J_j)$, evaluated at the steady-state of
the cell-bulk model, are now given by
\begin{equation}\label{coupled:selkov:jac}
  \det(J_{j}) = \zeta_j \left(\alpha_j + u_{1js}^2\right)>0
  \,,\quad
  \mbox{tr}(J_{j}) = \frac{1}{\alpha_j+u_{1js}^2}
  \left(2\mu_j u_{1js} - \left(\alpha_j+u_{1js}^2\right) - \zeta_j
   \left(\alpha_j + u_{1js}^2\right)^2 \right) \,.
\end{equation}
Here $u_{1js}=\mu_j+B_{js}$, where $B_{js}$ for
$j\in\lbrace{1,\ldots,N\rbrace}$ are obtained from
(\ref{selkov:equil_2}).  Since the linear system
(\ref{selkov:equil_2}) is always solvable, there are no transcritical
or fold bifurcation points along the steady-state solution branch as
parameters are varied for Sel'kov kinetics. As such, since
$\lambda=0$ is never a root of $\det{\mathcal M}(\lambda)=0$, the
steady-state solution is never destabilized by a zero-eigenvalue
crossing (see Proposition 1 of \cite{smjw_diff}).

In this way, for reaction-kinetic parameters for which each isolated
cell has a stable steady-state, any instability that arises from the
cell-bulk coupling must occur through a Hopf bifurcation. As a result,
we will identify stability boundaries in the ${1/\sigma}$ versus $D$
parameter plane for steady-state solutions with Sel'kov kinetics by
seeking Hopf bifurcation (HB) thresholds for which
$\lambda=i\lambda_{I}\in \Lambda({\mathcal M})$ with
$\lambda_I\in \R$. For arbitrary locations $\x_j$, with
$j\in\lbrace{1,\ldots,N\rbrace}$, of a collection of non-identical
cells, with possibly cell-dependent permeability and reaction-kinetic
parameters, we must compute all paths in the ${1/\sigma}$ versus $D$
parameter space where $\det{\mathcal M}(i\lambda_I)=0$ for some
$\lambda_I>0$. HB boundaries in parameter space can readily be 
computed numerically for a ring arrangement of identical cells, with
and without a defective center cell, as the matrix spectrum of
${\mathcal M}(\lambda)$ is known analytically (cf.~\cite{smjw_diff}).
When explicit analytical formulae are known for the eigenvalues of
${\mathcal M}(\lambda)$, simple scalar root-finding algorithms can be
used to determine the HB thresholds in parameter space using a
pseudo-arclength continuation scheme in $D$ (see \cite{smjw_diff} for
details). However, for a general spatial arrangement of non-identical
cells, the numerical solution of the nonlinear matrix eigenvalue
problem $\det{\mathcal M}(i\lambda_I)=0$ in the ${1/\sigma}$
versus $D$ parameter space is highly challenging. Nonlinear matrix
eigenvalue problems and effective solution strategies for various
classes of matrices are discussed in \cite{betcke}, \cite{betcke2} and
\cite{guttel}.

To determine regions of instability in open sets of the ${1/\sigma}$
versus $D$ parameter plane off of the HB boundaries {and to
  count the number, ${\mathcal Z}$, of destabilizing eigenvalues of the
  linearization of the cell-bulk model (\ref{DimLess_bulk}) around the
  steady-state, as defined by the number of
  $\lambda\in \Lambda({\mathcal M})$ with $\mbox{Re}(\lambda)>0$
  (counting multiplicity)}, we use the argument principle of complex
analysis applied to
$\mathcal{F}(\lambda) \equiv \det(\mathcal{M}(\lambda))$, where
${\mathcal M}(\lambda)$ is the complex symmetric GCEP matrix
(\ref{5:gcep_2}) in which ${\mathcal K}(\lambda)$ is given in
(\ref{selkov:kmat}). In the right-half plane, we take the contour
$\Gamma_{\mathcal R}$ as the union of the imaginary axis
$\Gamma_{I} = i \lambda_I$, for $|\lambda_I|\leq \mathcal{R}$, and the
semi-circle $C_{\mathcal{R}}$, defined by $|\lambda| = \mathcal{R}>0$
with $|\mbox{arg}(\lambda)| \leq \pi/2$. Provided that there are no
zeroes or poles on $\Gamma_{\mathcal R}$ for ${\mathcal F}(\lambda)$,
the argument principle yields that the number
${\mathcal Z}_{\mathcal R}$ of zeroes of $\mathcal{F}(\lambda) = 0$
inside $\Gamma_{\mathcal R}$ is
\begin{equation}\label{ArgMent_Principle}
 {\mathcal Z}_{\mathcal R} = \frac{1}{2\pi} \big[ \mbox{arg} \,
  \mathcal{F}(\lambda) \big]_{\Gamma_{\mathcal R}} + {\mathcal P}_{\mathcal R}\,.
\end{equation}
Here ${\mathcal P}_{\mathcal R}$ is the number of poles of
$\mathcal{F}(\lambda)$ inside $\Gamma_{\mathcal R}$, while
$\big[\mbox{arg}\,\mathcal{F}(\lambda)\big]_{\Gamma_{\mathcal R}}$
denotes the change in the argument of $\mathcal{F}(\lambda)$ over the
counter-clockwise oriented contour $\Gamma_{\mathcal R}$. We pass to
the limit ${\mathcal R}\to \infty$ and use (\ref{selkov:kmat}) to
obtain that ${\mathcal K}(\lambda)\to 0$ as ${\mathcal R}\to \infty$
on the semi-circle $C_{\mathcal R}$. Moreover, for
${\mathcal R}\gg 1$, (\ref{GreenMat}) and
(\ref{5:gcep_2}) yields that
${\mathcal M}(\lambda)\approx -({\nu/2}) \log\left({\mathcal R}\right)
I+{\mathcal O}(1)$ on ${\mathcal C}_{\mathcal R}$. As a result, we
conclude that
$\lim_{{\mathcal R}\to \infty} \big[ \mbox{arg} \,
\mathcal{F}(\lambda) \big]_{C_{\mathcal R}}=0$. Finally, to evaluate
the change of argument on the imaginary axis we use
${\mathcal F}(\overline{\lambda})=\overline{{\mathcal F}(\lambda)}$
with $\lambda=i\lambda_I$ to reduce the computation to the positive
imaginary axis. In this way, by letting
${\mathcal R}\to\infty$ we get that the number of zeroes
${\mathcal Z}$ of ${\mathcal F}(\lambda)$ in $\mbox{Re}(\lambda)>0$ is
\begin{equation}\label{wind:form}
  {\mathcal Z} = {\mathcal P} - \frac{1}{\pi}\big[ \mbox{arg} \,
  \mathcal{F}(i\lambda_I) \big]_{\Gamma_{I+}} \,,
\end{equation}
where ${\mathcal P}$ denotes the number of poles (counting
multiplicity) of ${\mathcal F}(\lambda)$ in $\mbox{Re}(\lambda)>0$,
and $\Gamma_{I+}$ denotes the entire positive imaginary axis now
directed upwards starting from $\lambda_I=0$.  To determine
${\mathcal P}$, we observe that since $\mathcal{G}_{\lambda}$ is
analytic in $\mbox{Re}(\lambda) > 0$ any singularity of
$\mathcal{F}(\lambda)$ must arise from the poles of the diagonal
matrix ${\mathcal K}(\lambda)$ with entries given in
(\ref{selkov:kmat}). Since $\det(J_j)>0$, we conclude that
${\mathcal P}=2p$ where $p$ is the total number of integers
$j\in\lbrace{1,\ldots,N \rbrace}$ for which $\mbox{tr}(J_j)>0$.  {For
  a given parameter set, we will use (\ref{wind:form}) to numerically
  calculate ${\mathcal Z}$ in a fine discretization of the
  ${1/\sigma}$ versus $D$ parameter plane. Below, we refer to this
  ``phase diagram'' as a {\em scatter plot}. For the cell
  configurations considered below, where the matrix spectrum of
  ${\mathcal M}(\lambda)$ is known, the determinant is readily
  evaluated at a fine discretization along the imaginary axis. The
  axis-crossing algorithm of \cite{mirandawinding} can then be used to
  compute the winding number.}

\begin{figure}[htbp]
  \centering
      \begin{subfigure}[b]{0.32\textwidth}
      \includegraphics[width=\textwidth]{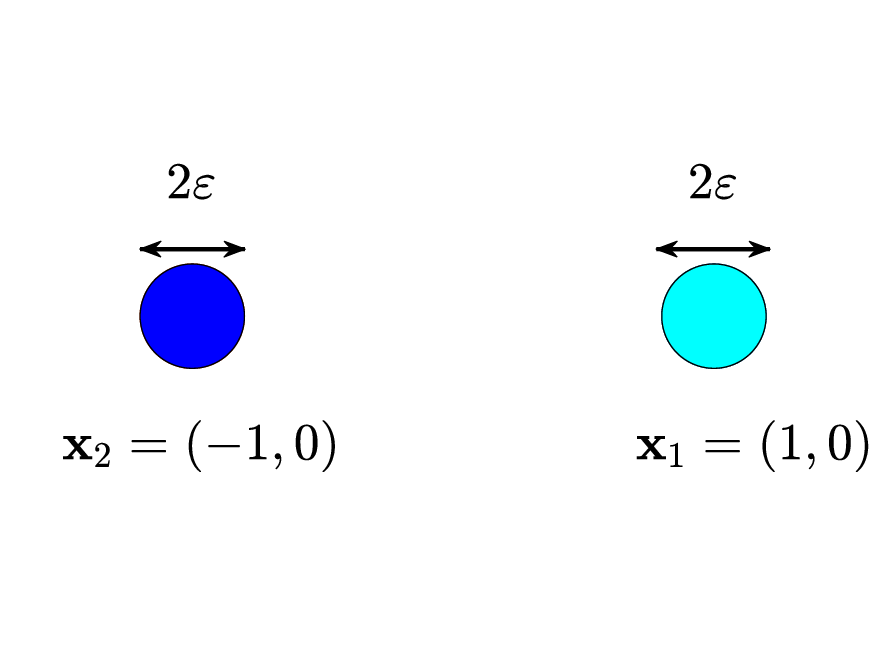}
        \caption{Two-cell configuration}
        \label{fig:schem:twocell}
      \end{subfigure}
    \begin{subfigure}[b]{0.32\textwidth}
      \includegraphics[width=\textwidth]{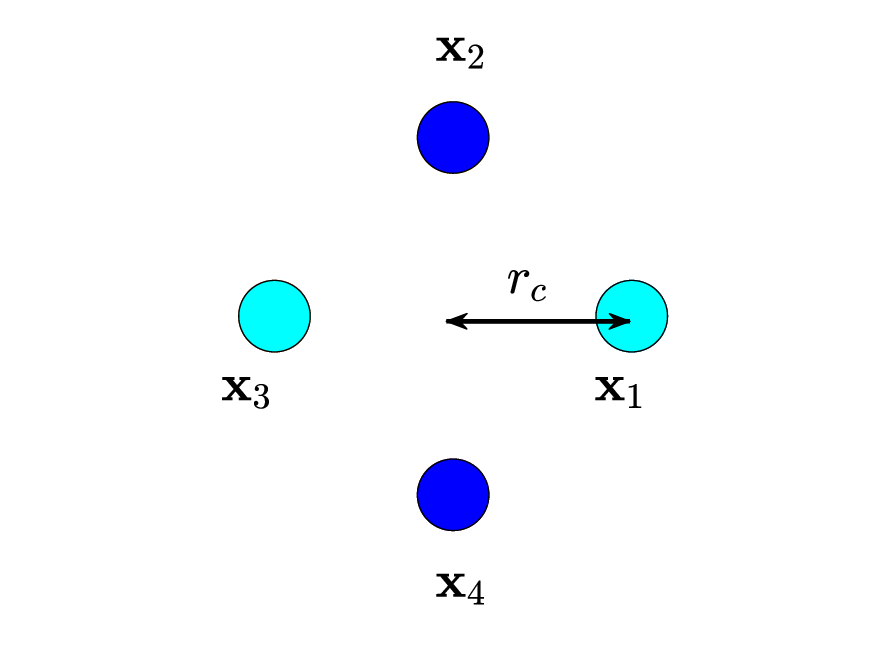}
        \caption{Four-cell configuration}
        \label{fig:fourcell}
      \end{subfigure}
          \begin{subfigure}[b]{0.32\textwidth}  
      \includegraphics[width=\textwidth]{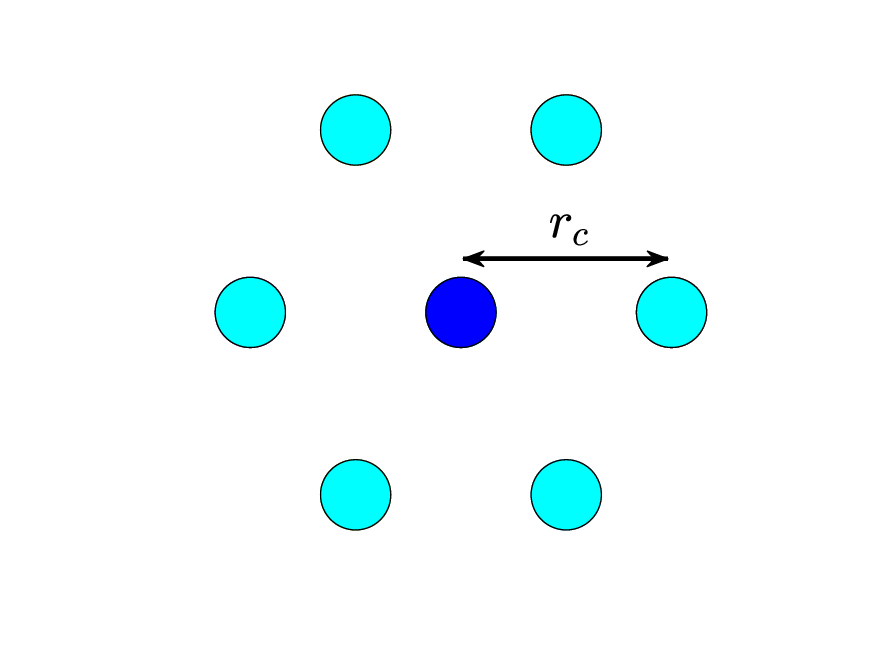}
        \caption{Ring with center cell}
        \label{fig:ringcenter}
    \end{subfigure}
    \caption{Left: A two-cell configuration. Middle: four
      equally-spaced cells on a ring of radius $r_c$.  Right:
      Hexagonal ring arrangement of identical cells where the center
      cell is a pacemaker or signaling cell.  Equally-colored
      pairs of cells are taken to have identical permeabilities and
      reaction-kinetic parameters.}
\label{fig:schema:two-four-eight}
\end{figure}

In the results given below in \S
\ref{cell:twocell}--\ref{cell:hexagonal} we have fixed the cell radius
as $\eps=0.03$ and the Sel'kov reaction-kinetic parameters as
$\mu_j=2$ and $\zeta_j=0.15$ for $j\in\lbrace{1,\ldots,N\rbrace}$. We
will explore the effect of changing the kinetic parameter $\alpha_j$
and the influx and efflux parameters $d_{1j}$ and $d_{2j}$,
respectively. More specifically, we will study the effect of choosing
pairs $(d_{2j},\alpha_j)$ either inside or outside the blue-shaded
region of instability, as shown in Fig.~\ref{fig:selkov_efflux}, for
an isolated cell with boundary efflux. Initial conditions will be
chosen near the steady-state values so that the scatter plots are
informative for predicting the onset of any instability.

\subsection{Two-Cell Configurations}\label{cell:twocell}

\begin{figure}[htbp]
  \centering
  \begin{subfigure}[b]{0.32\textwidth}  
     \includegraphics[width=\textwidth,height=4.4cm]{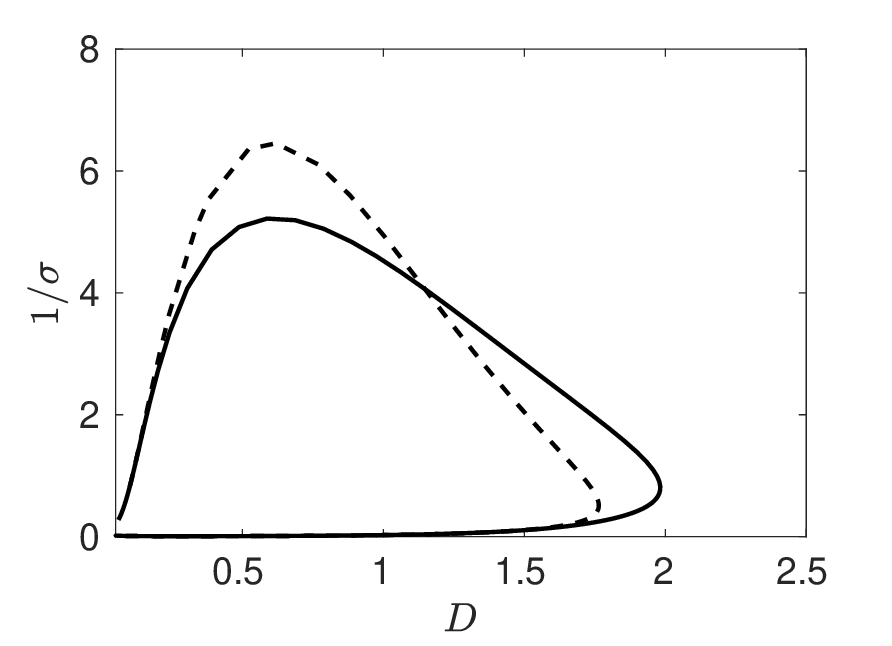}
        \caption{$d_{1j}=0.4$}
        \label{fig:twocell_hopf_b}
    \end{subfigure}
  \begin{subfigure}[b]{0.32\textwidth}  
     \includegraphics[width=\textwidth,height=4.4cm]{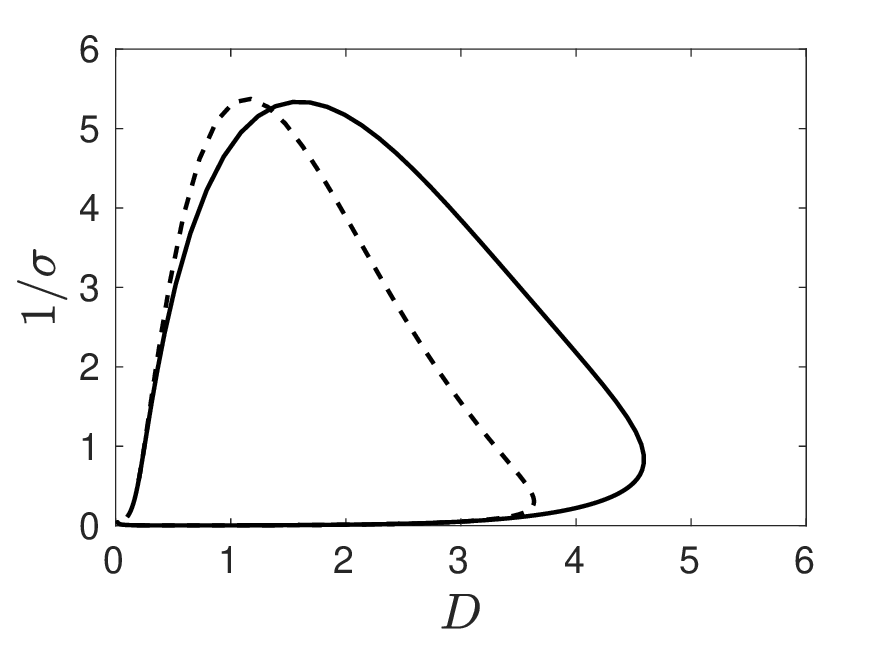}
        \caption{$d_{1j}=0.8$}
        \label{fig:twocell_hopf_a}
      \end{subfigure}
        \begin{subfigure}[b]{0.32\textwidth}  
            \includegraphics[width=\textwidth,height=4.4cm]{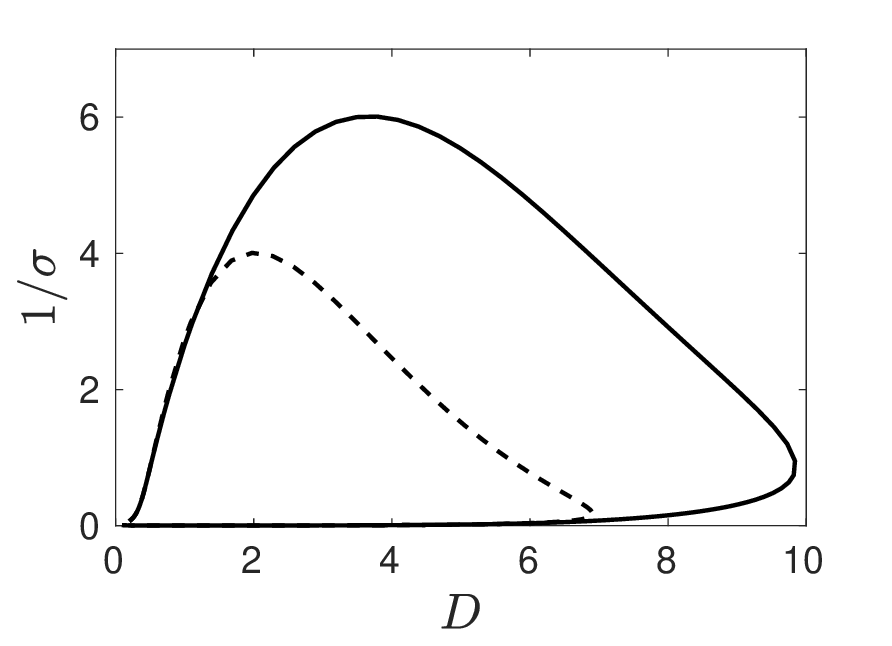}
        \caption{$d_{1j}=1.5$}
        \label{fig:twocell_hopf_d}
    \end{subfigure}
    \caption{Left: HB boundaries in the ${1/\sigma}$ versus $D$
      parameter plane for the linearization of the steady-state for
      the two-cell configuration of Fig.~\ref{fig:schem:twocell} for
      identical cells with permeabilities $d_{1j}=0.4$, $d_{2j}=0.2$
      and kinetic parameter $\alpha_j=0.9$. The HB boundaries for the
      in-phase $\v{c}=(1,1)^T$ and anti-phase $\v{c}=(1,-1)^T$ modes
      are the solid and dashed curves, respectively.  The steady-state
      is unstable to the specific mode inside each lobe.  Middle:
      $d_{1j}=0.8$. Right: $d_{1j}=1.5$. The range in $D$ of the lobes
      of instability increases with the influx parameter $d_{1j}$.}
\label{fig:twocell_hopf}
\end{figure}

For the linearization of the steady-state for the two-cell
configuration of Fig.~\ref{fig:schem:twocell}, and for three values of
the influx parameter $d_{1j}$, in Fig.~\ref{fig:twocell_hopf} we plot
the HB boundaries for the in-phase mode (solid curves) and the
anti-phase mode (dashed curves) in the ${1/\sigma}$ versus $D$
parameter plane when two identical cells are initially in a quiescent
state when uncoupled from the bulk.  The other parameters are given in
the figure caption. The steady-state is unstable to in-phase and the
anti-phase perturbations only inside the corresponding lobes, while
the steady-state is linearly stable outside the union of the two
lobes.  From Fig.~\ref{fig:twocell_hopf} we observe the possibility of
either a purely anti-phase instability or a purely in-phase
instability for some parameter pairs $\left({1/\sigma},D\right)$. Upon
comparing Figs.~\ref{fig:twocell_hopf_b}--\ref{fig:twocell_hopf_d},
the range in $D$ where the lobes of instability occur become larger as
$d_{1j}$ is increased.

\begin{figure}[htbp]
  \centering
  \begin{subfigure}[b]{0.32\textwidth}  
     \includegraphics[width=\textwidth,height=4.4cm]{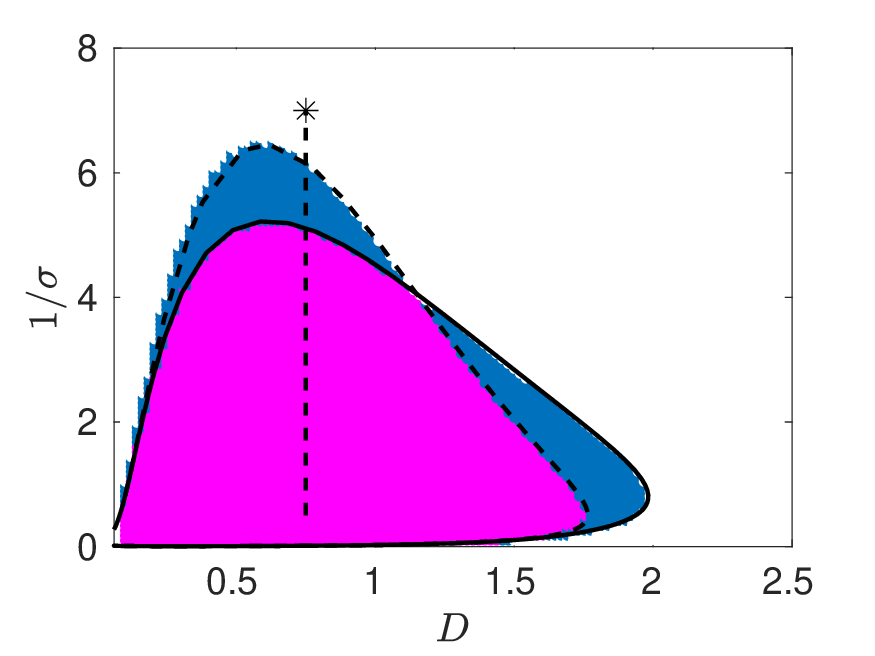}
        \caption{$d_{1j}=0.4$}
        \label{fig:twocell_hopf_scatt_b}
    \end{subfigure}
  \begin{subfigure}[b]{0.32\textwidth}  
     \includegraphics[width=\textwidth,height=4.4cm]{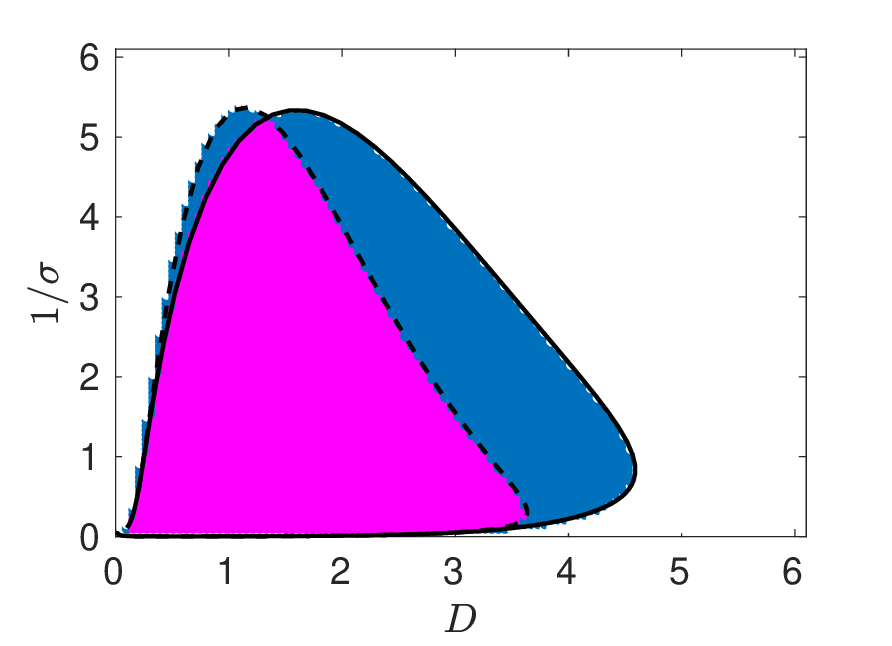}
        \caption{$d_{1j}=0.8$}
        \label{fig:twocell_hopf_scatt_a}
      \end{subfigure}
        \begin{subfigure}[b]{0.32\textwidth}  
  \includegraphics[width=\textwidth,height=4.4cm]{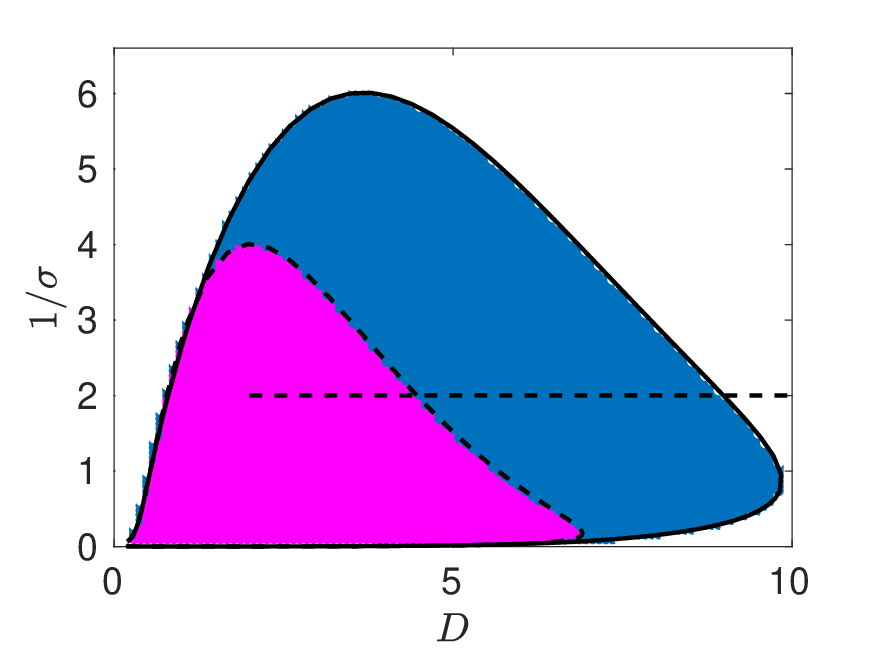}
        \caption{$d_{1j}=1.5$}
        \label{fig:twocell_hopf_scatt_d}
    \end{subfigure}
    \caption{Scatter plot of the number ${\mathcal Z}$ of destabilizing
      eigenvalues satisfying $\mbox{Re}(\lambda)>0$ in the
      ${1/\sigma}$ versus $D$ parameter plane corresponding to the
      linearization of the steady-state for the two-cell configuration
      of Fig.~\ref{fig:schem:twocell} for identical cells with
      parameter values $d_{1j}=0.4$ (left), $d_{1j}=0.8$ (middle),
      $d_{1j}=1.5$ (right). Here, ${\mathcal Z}=0$ is white,
      ${\mathcal Z}=2$ is blue, and ${\mathcal Z}=4$ is magenta. The
      HB boundaries are superimposed. Remaining parameters are as in
      Fig.~\ref{fig:twocell_hopf}.}
\label{fig:twocell_hopf_scatt}
\end{figure}

For Fig.~\ref{fig:twocell_hopf_scatt} we numerically implemented the
winding number criterion (\ref{wind:form}) to provide a scatter plot
in the ${1/\sigma}$ versus $D$ plane that indicates the number of
destabilizing eigenvalues in $\mbox{Re}(\lambda)>0$ associated with the
linearization of the steady-state. In Fig.~\ref{fig:twocell_slice_b_d}
we plot the path of both the real and imaginary parts of the two
dominant eigenvalues $\lambda\in {\mathbb C}$ along the parameter path
indicated by the vertical and horizontal dotted lines in
Fig.~\ref{fig:twocell_hopf_scatt_b} and
Fig.~\ref{fig:twocell_hopf_scatt_d}, respectively. Along the vertical
path in Fig.~\ref{fig:twocell_hopf_scatt_b}, we observe from
Fig.~\ref{fig:twocell_vertical_real_b} that the in-phase and
anti-phase modes of instability have comparable growth rates when
$D=0.75$ and $\sigma=0.5$ and that the in-phase mode becomes stable
before the anti-phase mode as ${1/\sigma}$ is increased. Since the
imaginary part $\mbox{Im}(\lambda)$ is roughly the same for both
modes, and shows little variation with $\sigma$, we conclude that the
temporal frequencies of small amplitude oscillations are roughly
similar for both modes. In contrast, along the horizontal parameter
path in Fig.~\ref{fig:twocell_hopf_scatt_d}, we observe from
Fig.~\ref{fig:twocell_horizontal_real_d} that the growth rate for the
in-phase mode is larger than that for the anti-phase mode as $D$ is
increased.

\begin{figure}[htbp]
  \centering
  \begin{subfigure}[b]{0.48\textwidth}  
     \includegraphics[width=\textwidth,height=4.2cm]{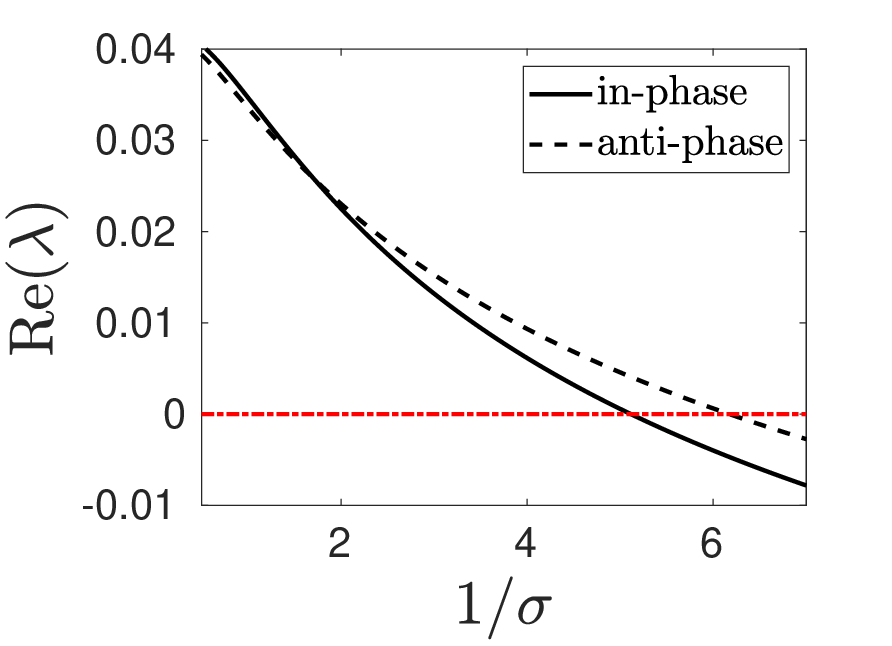}
  \caption{$\mbox{Re}(\lambda)$ vertical slice in
   Fig.~\ref{fig:twocell_hopf_scatt_b}}
        \label{fig:twocell_vertical_real_b}
      \end{subfigure}
  \begin{subfigure}[b]{0.48\textwidth}  
     \includegraphics[width=\textwidth,height=4.2cm]{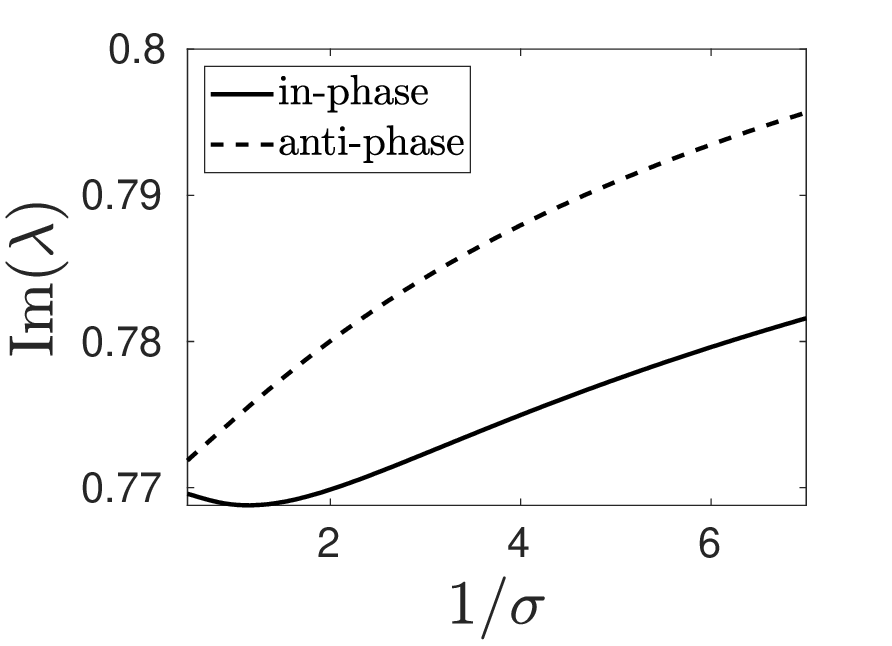}
  \caption{$\mbox{Im}(\lambda)$ vertical slice in
   Fig.~\ref{fig:twocell_hopf_scatt_b}}
        \label{fig:twocell_vertical_imag_b}
      \end{subfigure}
  \begin{subfigure}[b]{0.48\textwidth}  
     \includegraphics[width=\textwidth,height=4.2cm]{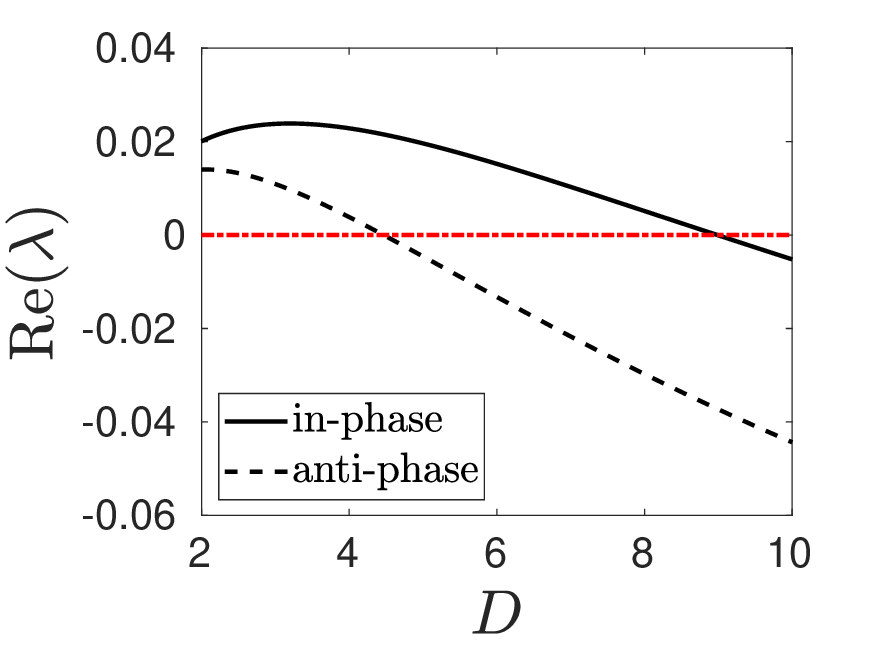}
  \caption{$\mbox{Re}(\lambda)$ horizontal slice in
   Fig.~\ref{fig:twocell_hopf_scatt_d}}
        \label{fig:twocell_horizontal_real_d}
      \end{subfigure}
  \begin{subfigure}[b]{0.48\textwidth}  
     \includegraphics[width=\textwidth,height=4.2cm]{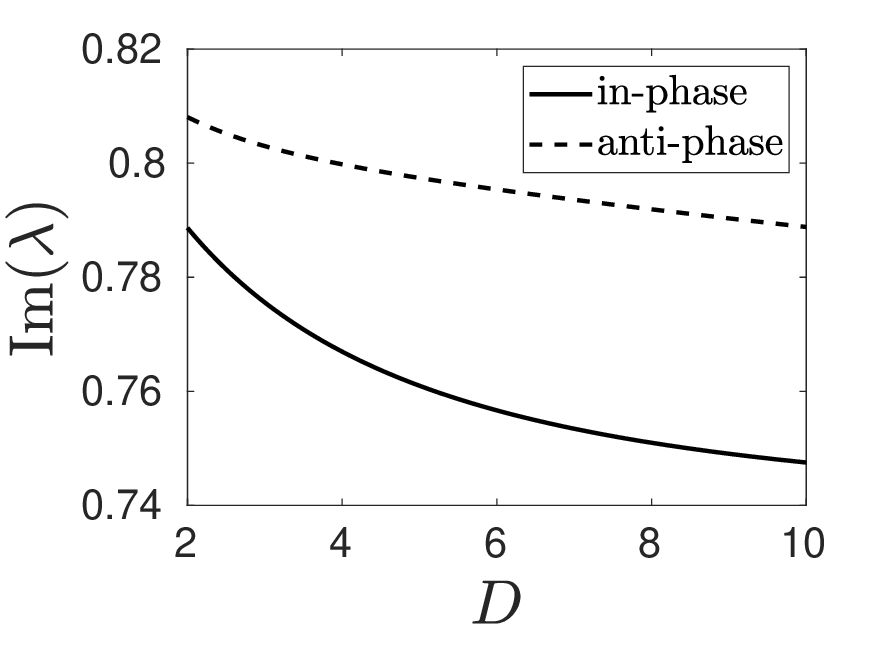}
  \caption{$\mbox{Im}(\lambda)$ horizontal slice in
   Fig.~\ref{fig:twocell_hopf_scatt_d}}
        \label{fig:twocell_horizontal_imag_a}
      \end{subfigure}
      \caption{Top row: two dominant eigenvalues, $\mbox{Re}(\lambda)$
        (left) and $\mbox{Im}(\lambda)$ (right), computed by solving
        $\det{\mathcal M} (\lambda)=0$ along the vertical dotted path
        with $D=0.75$ indicated in
        Fig.~\ref{fig:twocell_hopf_scatt_b}. Bottom row: two dominant
        eigenvalues, $\mbox{Re}(\lambda)$ (left) and
        $\mbox{Im}(\lambda)$ (right), computed by solving
        $\det{\mathcal M} (\lambda)=0$ along the horizontal dotted
        path with $\sigma=0.5$ as shown in
        Fig.~\ref{fig:twocell_hopf_scatt_d}. The horizontal red lines
        in the left panels is the marginal stability threshold
        $\mbox{Re}(\lambda)=0$.}
\label{fig:twocell_slice_b_d}
\end{figure}

{We now compare our numerical solution of the reduced
  integro-differential system (\ref{2:reduced}), as computed using our
  algorithm in \S \ref{sec:all_march}, with that of the full cell-bulk
  system (\ref{DimLess_bulk}), as computed with the commercial PDE
  solver FlexPDE Professional 6.50/L64 \cite{flexpde2015} using the
  domain-truncation approach given in Appendix \ref{app:artificial}.
  The comparison is made at the star-labeled point shown in
  Fig.~\ref{fig:twocell_hopf_scatt_b} where rather intricate long-time
  dynamics occur.  For the algorithm in \S \ref{sec:all_march}, we
  used a time step size $\Delta t = 0.002$ and chose $n=75$ for the
  discretization of the Laplace space contour for the
  sum-of-exponentials approximation (see
  Fig.~\ref{fig:laplace_discret}).  In Fig.~\ref{fig:D0p75Sig1over7}
  we show a very close qualitative agreement between our hybrid
  asymptotic-numerical results and the full numerical results. We
  observe that our hybrid approach is able to capture a transiently
  decaying oscillation, which transitions to a mixed-mode oscillation
  on some intermediate time scale, and that ultimately tends to a
  steady-state solution as obtained from the solution to
  (\ref{selkov:equil_2}). These final steady-state values correctly
  approximates the FlexPDE computed steady-state to several decimal
  places of accuracy. This intricate mixed-mode behavior over the
  rather long time scale seen in Fig.~\ref{fig:D0p75Sig1over7} stems
  from the proximity to the anti-phase Hopf bubble in parameter space
  as shown in Fig.~\ref{fig:twocell_hopf_scatt_b}, together with the
  two dominant, but rather small magnitude, closely-spaced eigenvalues
  shown in Fig.~\ref{fig:twocell_vertical_real_b} with
  $\sigma={1/7}$. As a more quantitative validation of our hybrid
  approach, in Fig.~\ref{fig:D0p75Sig1over7_detail} we show a very
  close comparison for the amplitude and period of the intracellular
  species $u_{11}$ and $u_{12}$ as extracted numerically from both our
  fast algorithm for (\ref{2:reduced}) and from the full FlexPDE
  solution to (\ref{DimLess_bulk}).

  Overall we conclude that both the integro-differential system, as
  derived in (\ref{2:reduced}) under the asymptotic assumption
  $\varepsilon\ll 1$, and the mixed-order numerical time-marching
  scheme formulated in \S \ref{sec:all_march}, are able to replicate
  with a high degree of accuracy detailed fine features in
  intracellular oscillations for the full cell-bulk model
  (\ref{DimLess_bulk}) over long time intervals. We emphasize that for
  the FlexPDE solution, the time-integration to $t\approx 669$ took
  many hours of CPU time, owing to the need for a fine spatial mesh
  in the boundary layers near the two cells at each time step. In
  contrast, our fast algorithm, implemented in {\em Fortran77} on a
  {\em Dell Precision} laptop with an {\em Intel Core I7} processor,
  completed in roughly one minute.

\begin{figure}[htbp]
  \centering
  \begin{subfigure}[b]{0.49\textwidth}
        \includegraphics[width=\textwidth]{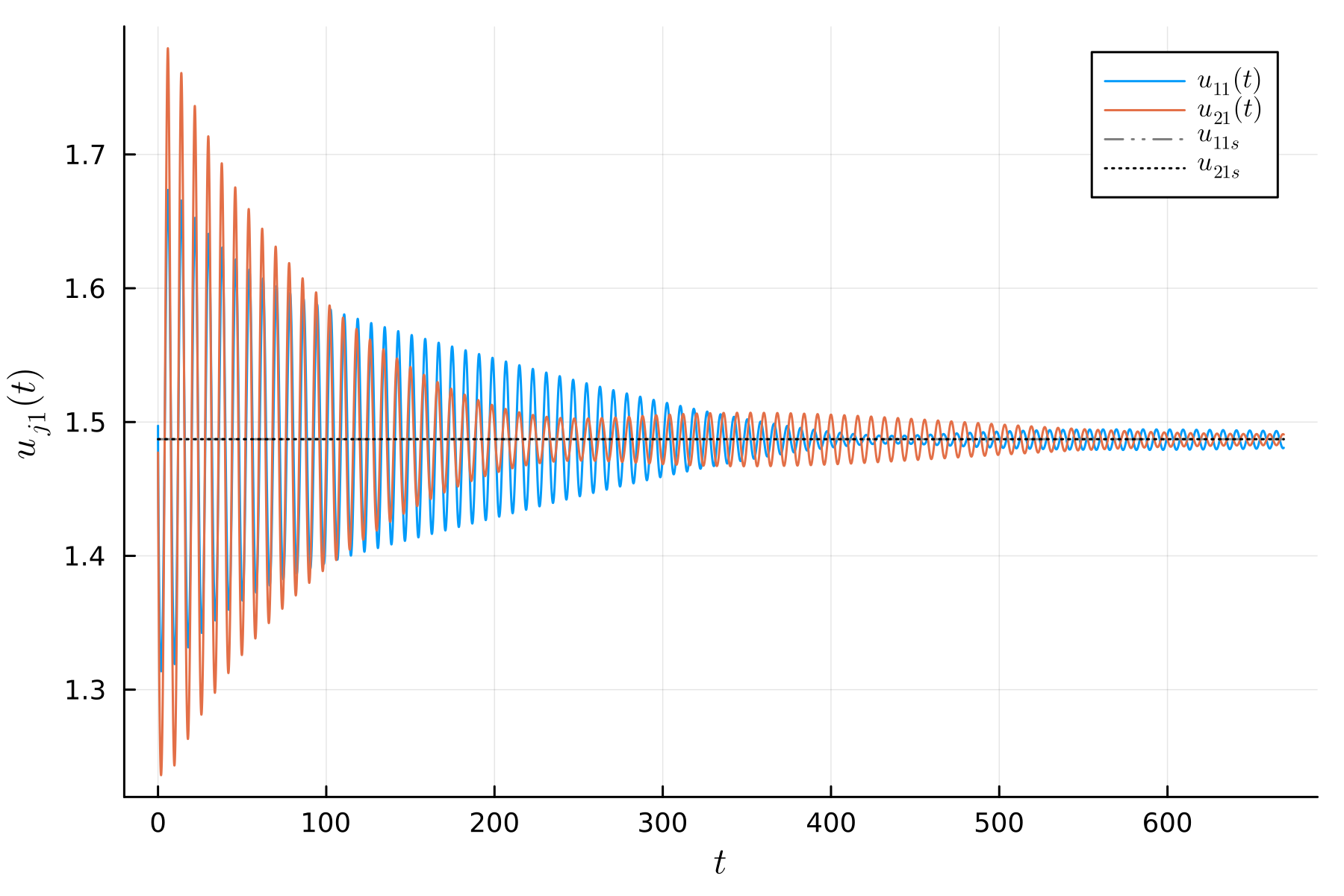}
        \caption{$u_{1j}$ hybrid scheme}
  \end{subfigure}
  \begin{subfigure}[b]{0.49\textwidth}
        \includegraphics[width=\textwidth]{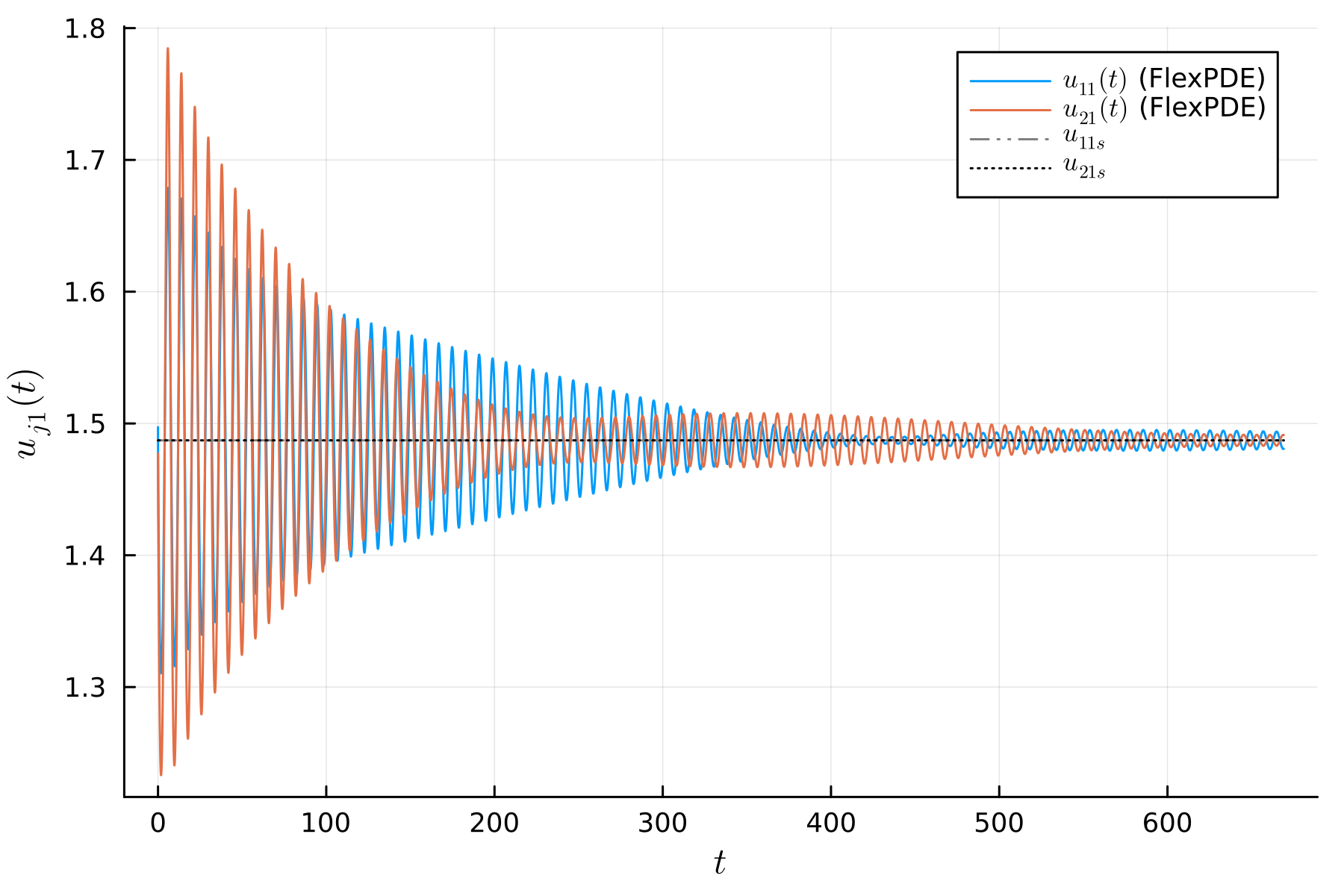}
   \caption{$u_{1j}$ FlexPDE}
  \end{subfigure}
  \begin{subfigure}[b]{0.49\textwidth}
        \includegraphics[width=\textwidth]{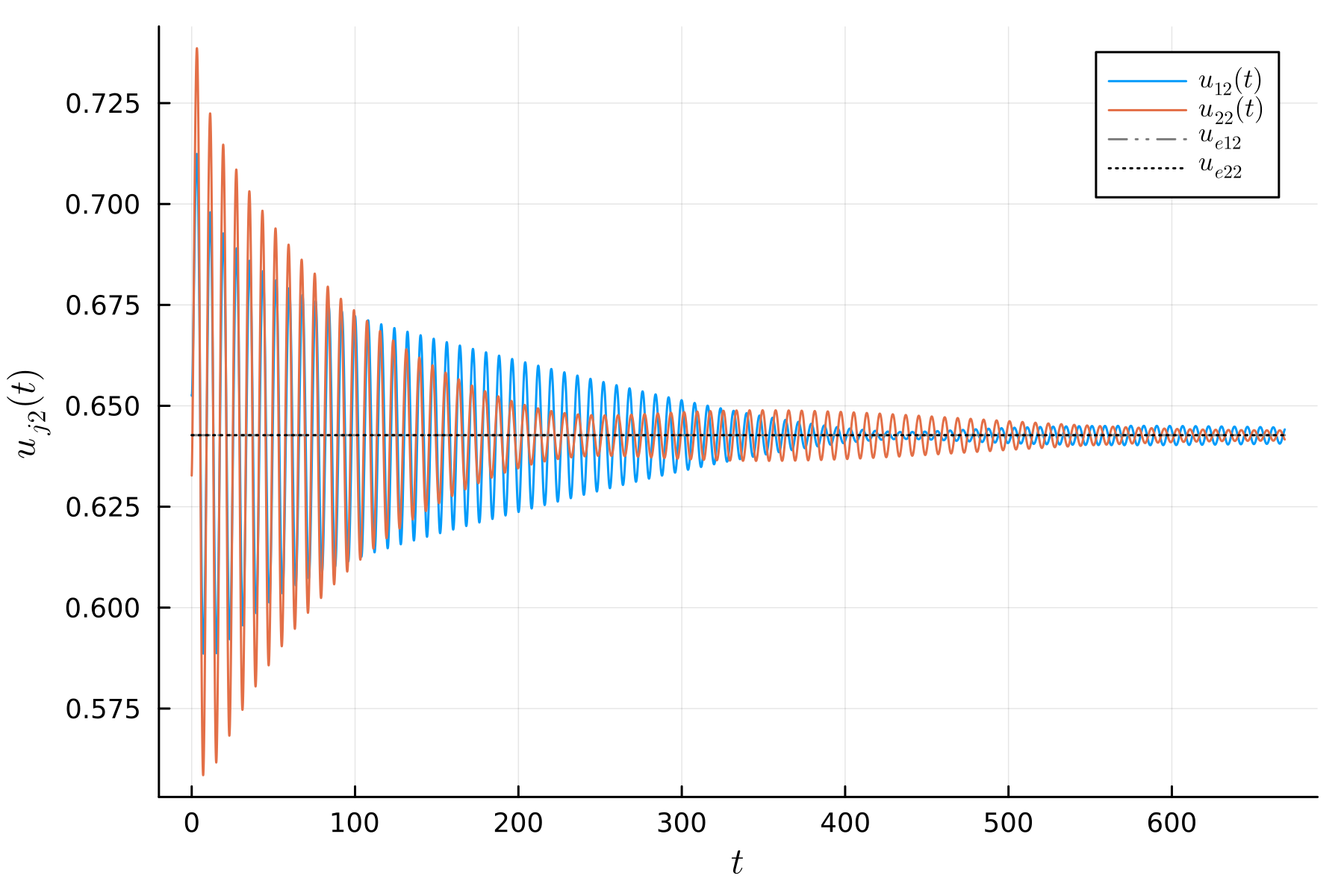}
        \caption{$u_{2j}$ hybrid scheme}
  \end{subfigure}
  \begin{subfigure}[b]{0.49\textwidth}
        \includegraphics[width=\textwidth]{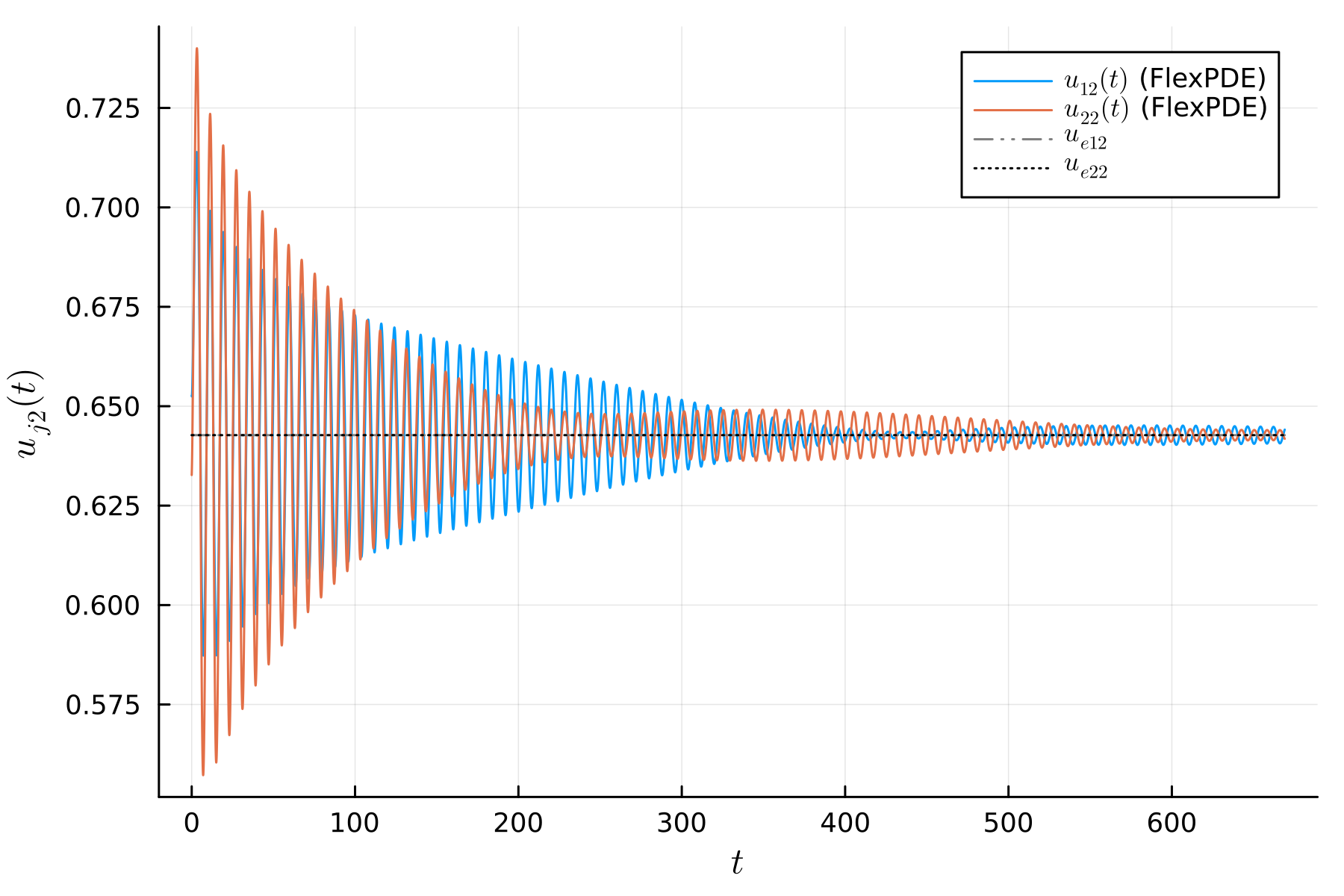}
   \caption{$u_{2j}$ FlexPDE}
  \end{subfigure}
  \caption{{Comparison of the numerical solutions computed
      from the integro-differential system (\ref{2:reduced}) using the
      algorithm in \S \ref{sec:all_march} with that computed from the
      cell-bulk model \eqref{DimLess_bulk} using FlexPDE
      \cite{flexpde2015} with the artificial boundary condition given
      in (\ref{app:art_eq}) of Appendix
      \ref{app:artificial}. Parameters: $D=0.75$, $\sigma={1/7}$ with
      $d_{1j}=0.4$, $d_{2j}=0.2$ and $\alpha_j=0.9$ corresponding to
      the star-labeled point in Fig.~\ref{fig:twocell_hopf_scatt_b}. The
      inital condition imposed was the steady-state with an anti-phase
      perturbation:
      $\v{u}_1^{(0)} = (u_{11s}, u_{21s})^T + 0.01 \cdot (1, -1)^T$, and
      similarly for $\v{u}_2^{(0)}$. The initial bulk solution for
      \eqref{DimLess_bulk} was $U(\v{x},0)=0$. The steady-state is the
    black horizontal line.}}
    \label{fig:D0p75Sig1over7}
\end{figure}}

\begin{figure}[htbp]
  \centering
  \begin{subfigure}[b]{0.49\textwidth}
        \includegraphics[width=\textwidth,height=4.2cm]{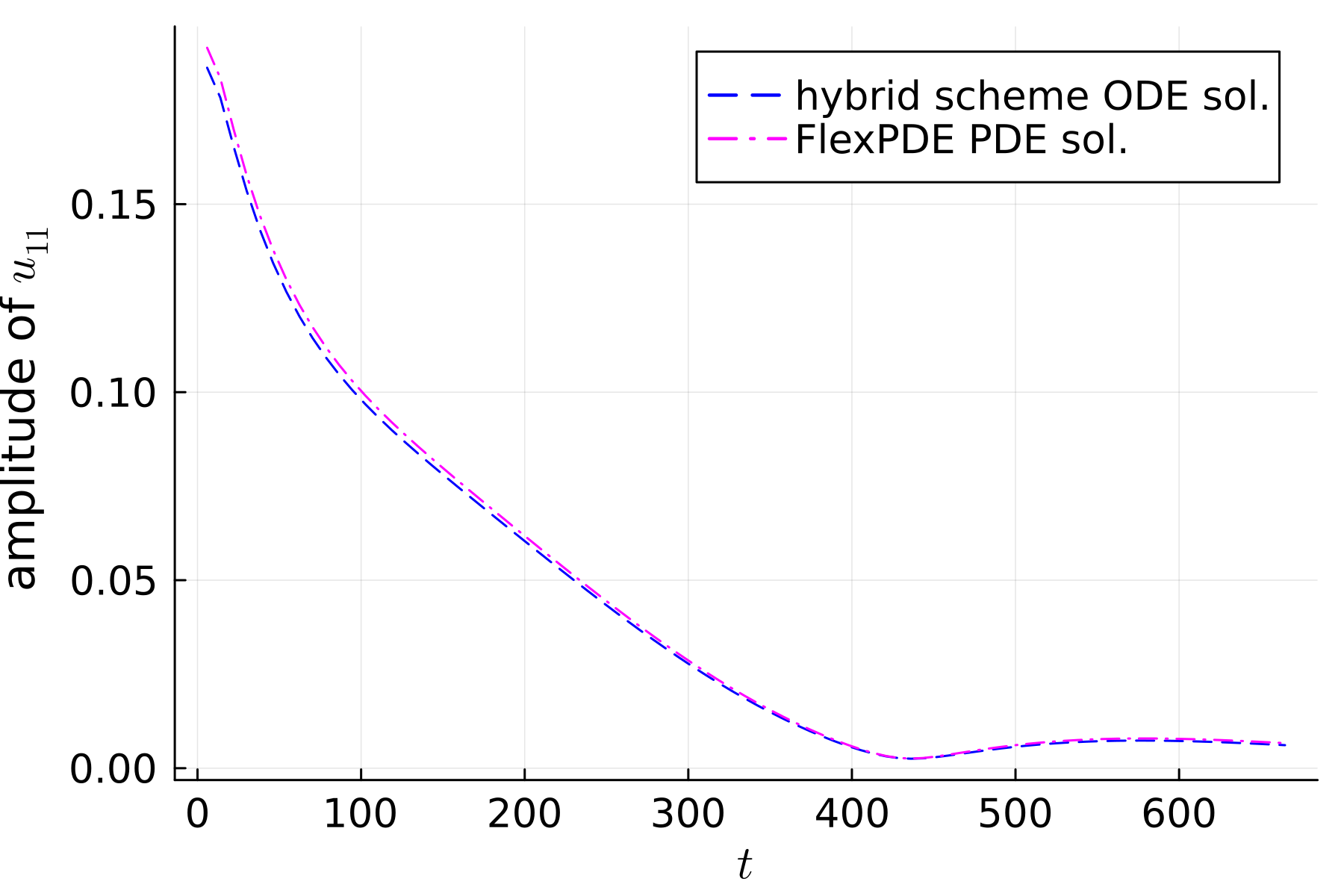}
        \caption{Comparison of $u_{11}$ amplitudes}
  \end{subfigure}
  \begin{subfigure}[b]{0.49\textwidth}
        \includegraphics[width=\textwidth,height=4.2cm]{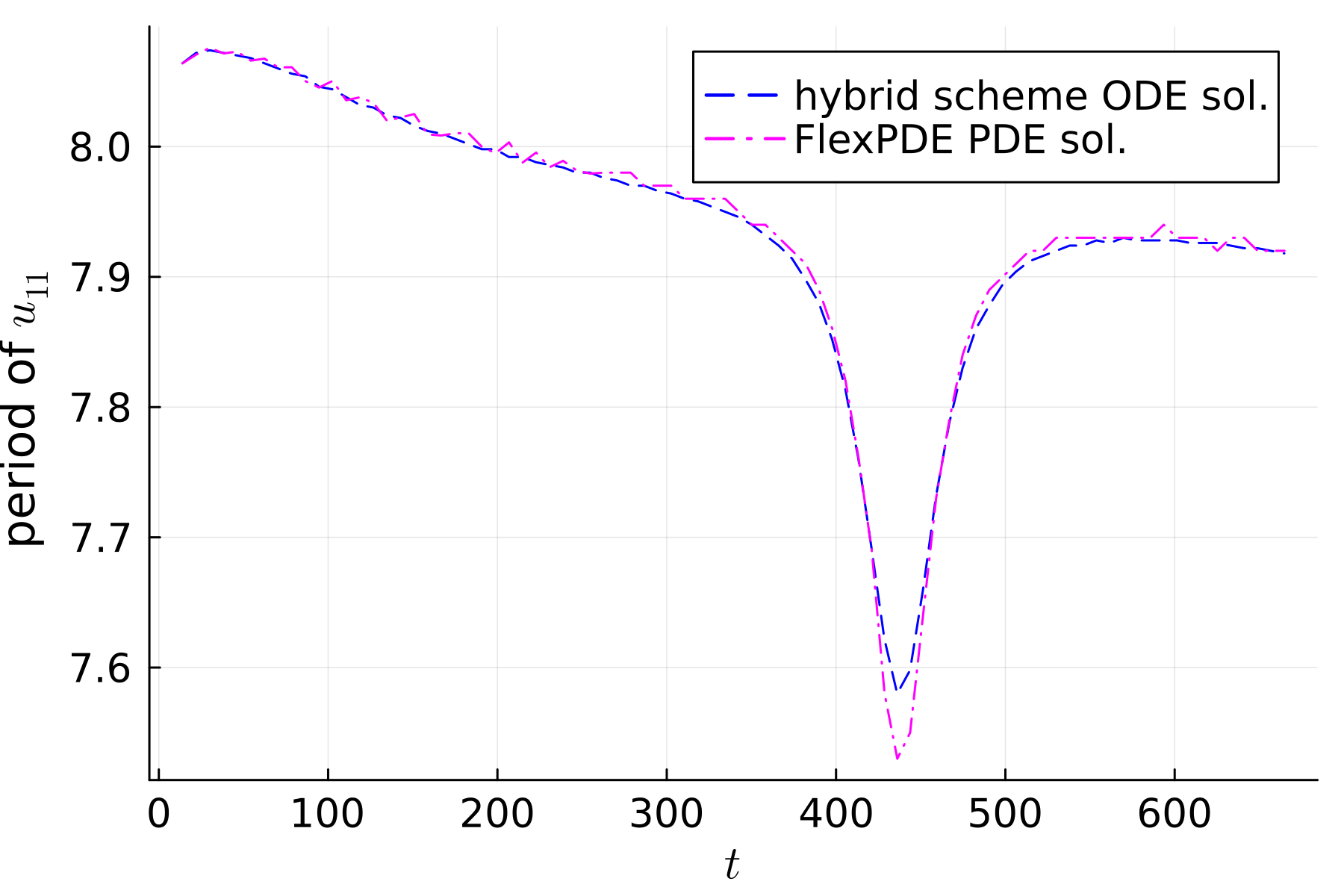}
   \caption{Comparison of $u_{11}$ periods}
  \end{subfigure}
  \begin{subfigure}[b]{0.49\textwidth}
        \includegraphics[width=\textwidth,height=4.2cm]{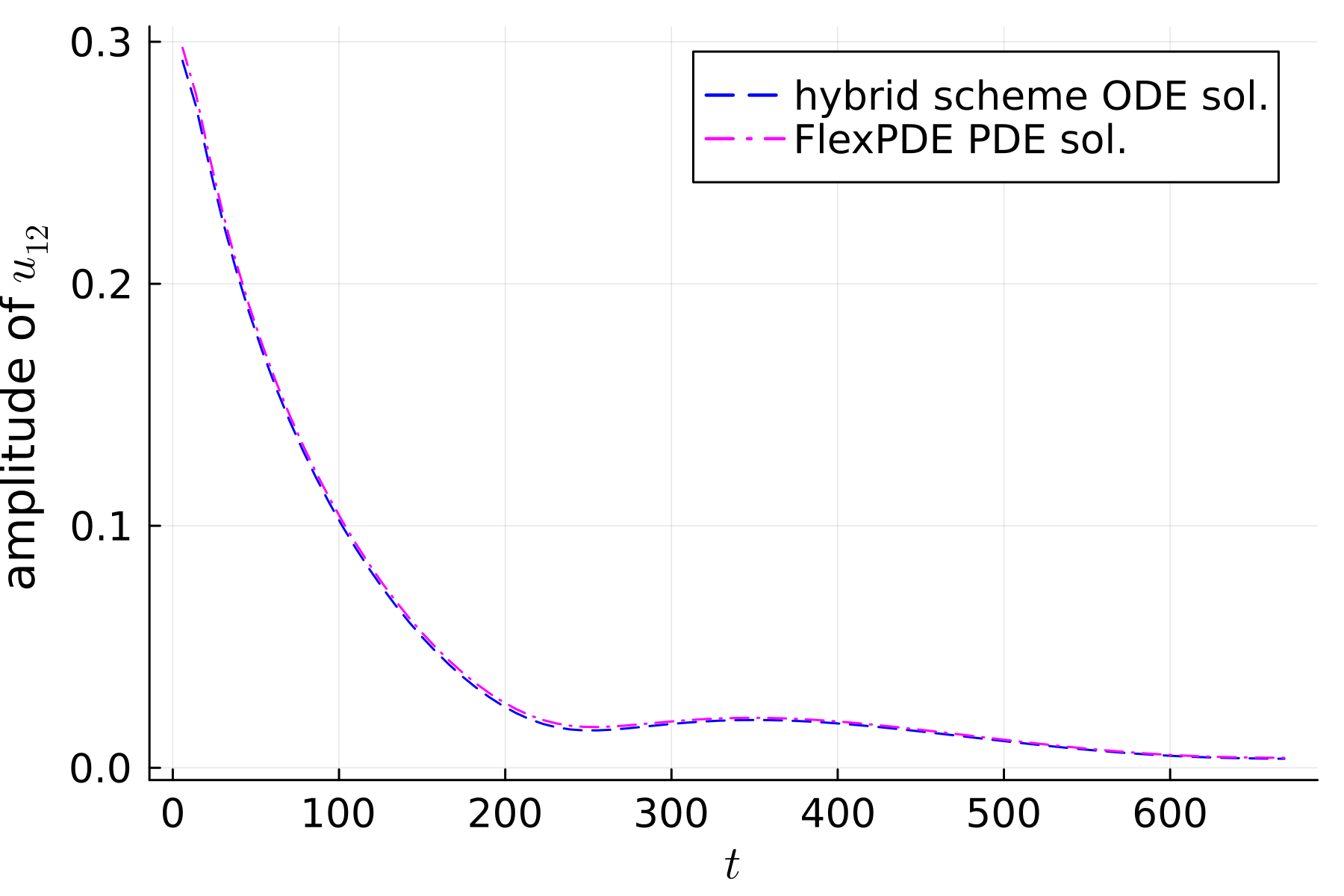}
        \caption{Comparison of $u_{12}$ amplitudes}
  \end{subfigure}
  \begin{subfigure}[b]{0.49\textwidth}
        \includegraphics[width=\textwidth,height=4.2cm]{periodmu1plot_D0p75_sig0p1429_antipert.png}
   \caption{Comparison of $u_{12}$ periods}
  \end{subfigure}
  \caption{Comparison of the amplitude and period of the numerical
    solutions to $u_{1j}$ for $j\in \lbrace{1,2\rbrace}$ as
    computed from the integro-differential system (\ref{2:reduced})
    using the fast algorithm in \S \ref{sec:all_march} with that computed
    from the cell-bulk model (\ref{DimLess_bulk}) using FlexPDE
    \cite{flexpde2015}. The parameters are as in
    in Fig.~\ref{fig:D0p75Sig1over7}.} \label{fig:D0p75Sig1over7_detail}
\end{figure}

Next, we consider the same parameters as in
Figs.~\ref{fig:twocell_hopf_d} and \ref{fig:twocell_hopf_scatt_d}
except that we now modify the kinetic parameter to $\alpha_1=0.4$ for
the first cell centered at $\x_1=(1,0)$. {For $\alpha_1=0.4$,
  an isolated cell with no boundary efflux would have limit cycle
  oscillations (see the blue star in
  Fig.~\ref{fig:isolated:plane}). Since $d_{21}=0.2$, we further
  observe from Fig.~\ref{fig:selkov_efflux} that the cell centered at
  $\x_1$ would have limit cycle oscillations even with boundary efflux
  when it is uncoupled from the bulk. As a result, we refer to this cell as
  the {\em signaling cell}.} In the scatter plot shown in
Fig.~\ref{fig:twocell_hopf_scatt_c} we observe that the steady-state
is always unstable in the ${1/\sigma}$ versus $D$ parameter plane. In
the blue-shaded region, the unique destabilizing mode is the one for which
the amplitude of intracellular oscillations for the signaling cell
centered at $\x_{1}$ is large, and where the second cell centered at
$\x_2$ has very small oscillations, and so is effectively silent. In
Fig.~\ref{fig:twocell_hopf_scatt_c} there is only one HB boundary and
it corresponds to a marginal mode where intracellular oscillations
emerge in the otherwise silent cell centered at $\v{x}_2=(-1,0)$ and
where the signaling cell centered at $\x_1$ is effectively
silent. However, by plotting the path of the two dominant eigenvalues
in Fig.~\ref{fig:twocell_horizontal_c} and
Fig.~\ref{fig:twocell_vertical_c} along the horizontal and vertical
slices shown in Fig.~\ref{fig:twocell_hopf_scatt_c}, respectively, we
observe that the mode with the largest growth rate is always the one
for which the cell centered at $\x_2$ has only very small amplitude
oscillations in comparison to that for the signaling cell.  By
computing the eigenvector $\v{c}$ for the two modes at some points
along the horizontal slice in Fig.~\ref{fig:twocell_hopf_scatt_c}, we
use the criterion in (\ref{nstabform:c2})
to determine the relative magnitude of the intracellular oscillations
in the two cells for the linearized problem. As shown in Table
\ref{table:eigvec_d}, the mode with the largest growth rate is the one
for which the second cell is effectively quiet.

By using our algorithm in \S \ref{sec:all_march} with
$\Delta t=0.002$, and with random initial values near the steady-state
of magnitude $0.01$, in Fig.~\ref{fig:twocell_unstable_c_dyn} we plot
$u_{2j}$ versus $t$ on the time window $600<t<800$ at the three
star-labeled points along the horizontal and vertical slices in
Fig.~\ref{fig:twocell_hopf_scatt_c}. We conjecture that when both the
in-phase and anti-phase modes are destabilizing for the linearization of
the steady-state solutions, such as for $D=2$, $\sigma={1/2}$ in
Fig.~\ref{fig:unstable_center} and for $D=4$, $\sigma={1/2}$ in
Fig.~\ref{fig:unstable_right}, the oscillations for the second cell
centered at $\x_2$ are of wave-packet type, as is characteristic of
mixed-mode oscillations.

\begin{figure}[htbp]
  \centering
  \begin{subfigure}[b]{0.32\textwidth}  
     \includegraphics[width=\textwidth,height=4.4cm]{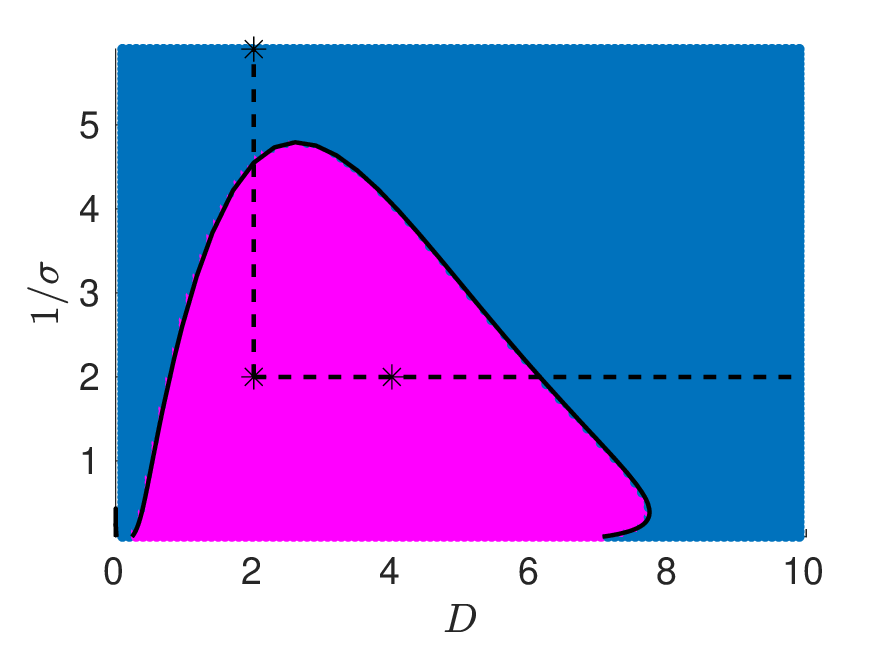}
        \caption{one non-quiescent cell}
        \label{fig:twocell_hopf_scatt_c}
      \end{subfigure}
        \begin{subfigure}[b]{0.32\textwidth}  
           \includegraphics[width=\textwidth,height=4.4cm]{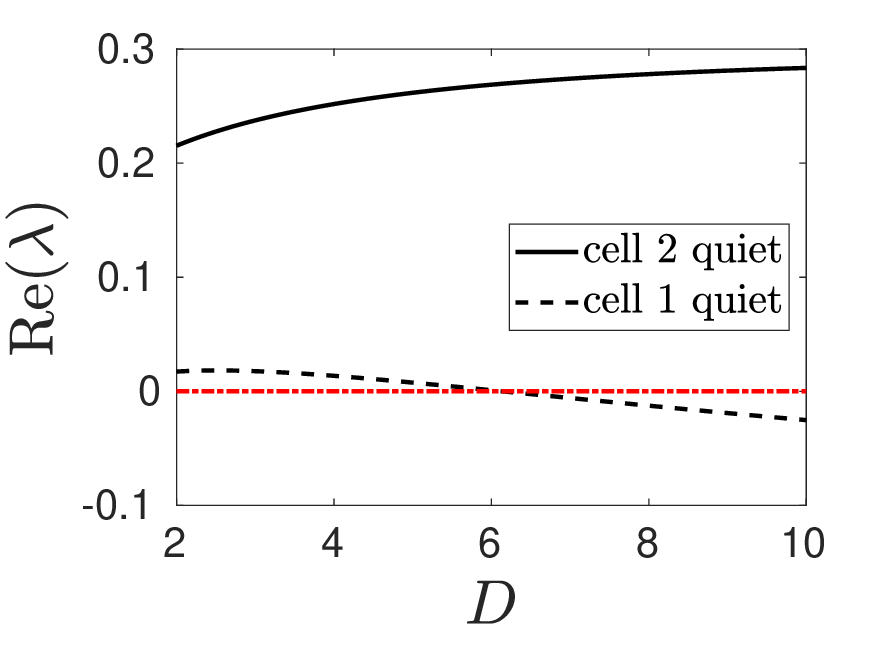}
        \caption{$\mbox{Re}(\lambda)$ horizontal slice}
        \label{fig:twocell_horizontal_c}
      \end{subfigure}
        \begin{subfigure}[b]{0.32\textwidth}  
 \includegraphics[width=\textwidth,height=4.4cm]{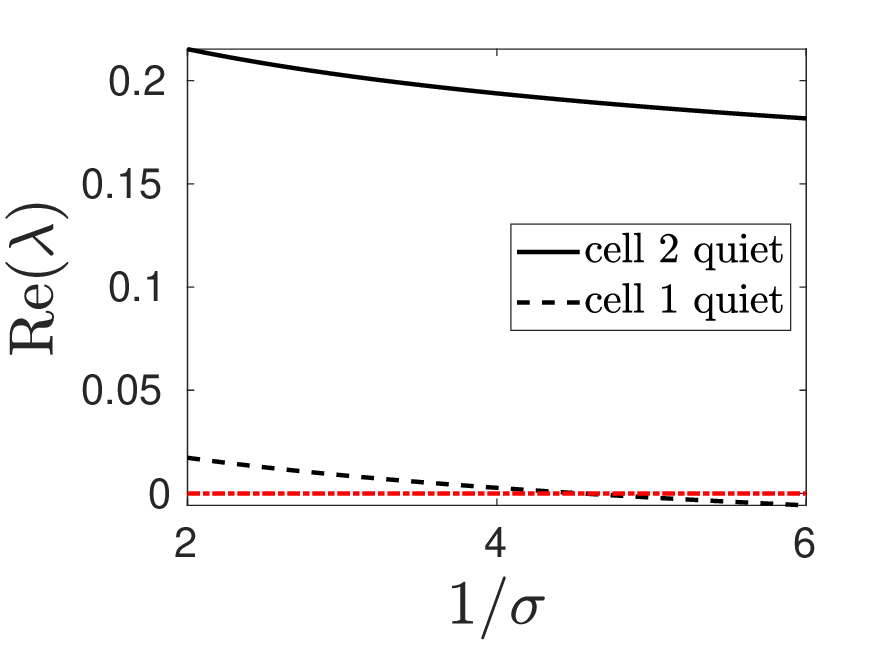}
        \caption{$\mbox{Re}(\lambda)$ vertical slice}
        \label{fig:twocell_vertical_c}
      \end{subfigure}
      \caption{Left: Scatter plot of the number ${\mathcal Z}$ of
        destabilizing eigenvalues, with ${\mathcal Z}=2$ (blue) and
        ${\mathcal Z}=4$ (magenta), for the linearization of the
        steady-state in the ${1/\sigma}$ versus $D$ plane for two
        non-identical cells with different kinetics parameter
        $\alpha_1=0.4$ and $\alpha_2=0.9$. The remaining parameters
        are $d_{1j}=1.5$ and $d_{2j}=0.2$ for
        $j\in\lbrace{1,2\rbrace}$. Only the cell with $\alpha_1=0.4$
        (signaling cell) would have limit cycle oscillations with
        boundary efflux when it is uncoupled to the bulk. The sole HB
        boundary is superimposed. The steady-state is now always
        unstable. Middle: $\mbox{Re}(\lambda)$ versus $D$ along the
        horizontal slice. Right: $\mbox{Re}(\lambda)$ versus
        ${1/\sigma}$ along the vertical slice. The dominant mode of
        instability is for intracellular oscillations to be
        concentrated to the first cell, while the second cell is
        essentially quiescent.}
\label{fig:twocell_hopf_unstable}
\end{figure}

\begin{figure}[htbp]
  \centering
  \begin{subfigure}[b]{0.32\textwidth}  
 \includegraphics[width=\textwidth,height=4.4cm]{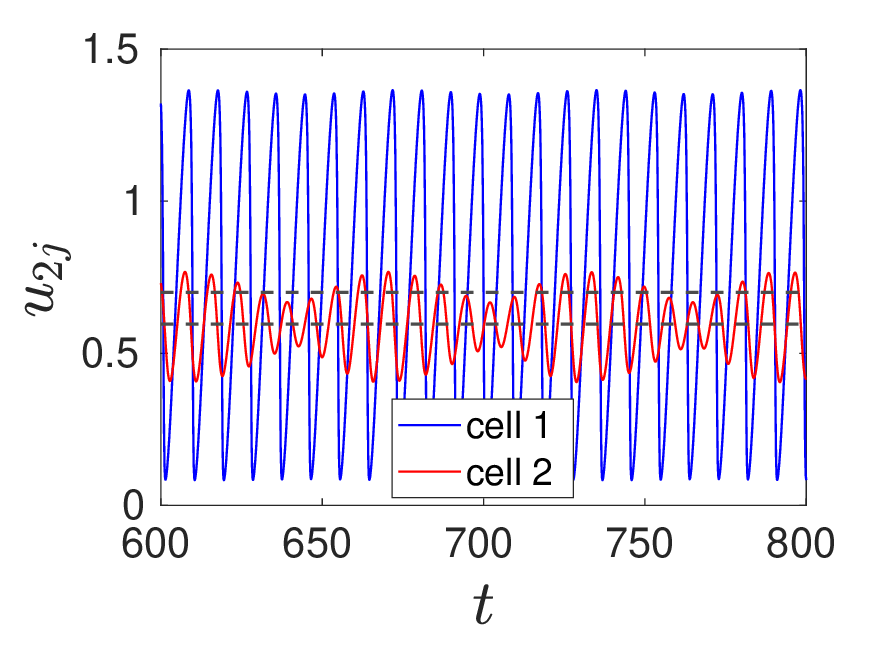}
  \caption{$D=2$, $\sigma={1/2}$.} 
        \label{fig:unstable_center}
      \end{subfigure}
  \begin{subfigure}[b]{0.32\textwidth}  
 \includegraphics[width=\textwidth,height=4.4cm]{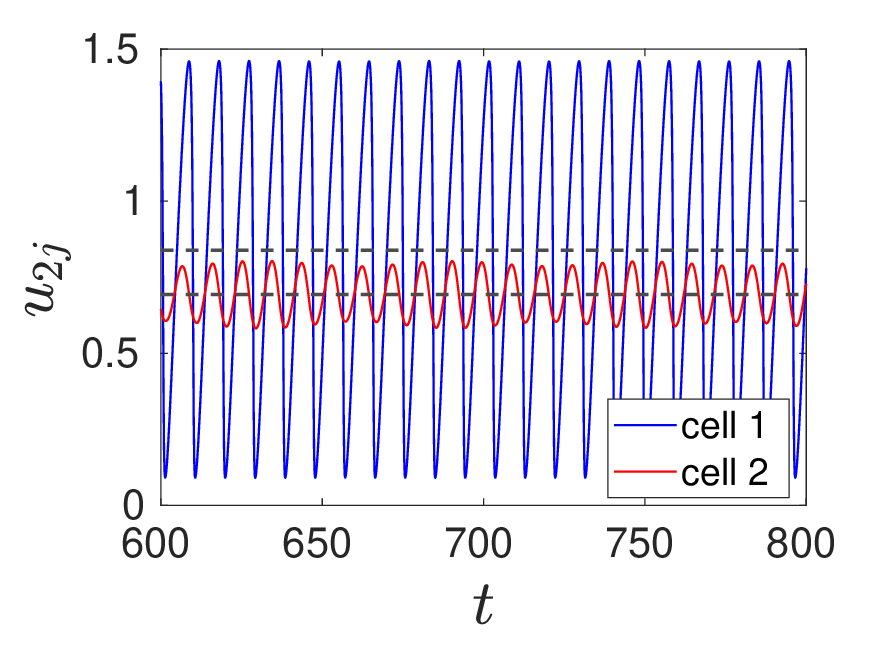}
  \caption{$D=4$, $\sigma={1/2}$.} 
        \label{fig:unstable_right}
      \end{subfigure}
  \begin{subfigure}[b]{0.32\textwidth}  
 \includegraphics[width=\textwidth,height=4.4cm]{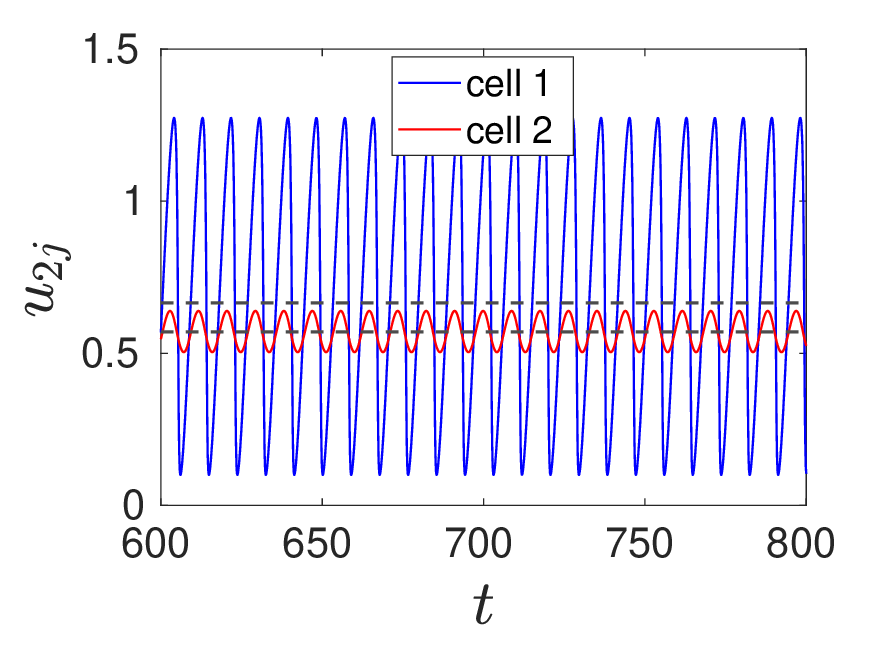}
  \caption{$D=2$, $\sigma={1/6}$.}
  \label{fig:unstable_up}
      \end{subfigure}
      \caption{{Intracellular dynamics $u_{2j}(t)$ between $600<t<800$
          at the three star-labeled points along the vertical and
          horizontal slices in
          Fig.~\ref{fig:twocell_hopf_scatt_c}. Left: $D=2$,
          ${\sigma}={1/2}$: wave-packet solution for second
          cell. Middle: $D=4$, ${\sigma}={1/2}$: wave-packet solution
          transitions to regular oscillations. Right: $D=2$,
          ${\sigma}={1/6}$: regular oscillations occur for the second
          cell but with smaller amplitude.  The dashed horizontal
          lines are the unstable steady-state values for $u_{2j}$.
          Parameters: $\alpha_1=0.4$, $\alpha_2=0.9$, with
          $d_{1j}=1.5$ and $d_{2j}=0.2$ for
          $j\in\lbrace{1,2\rbrace}$.}}
\label{fig:twocell_unstable_c_dyn}
\end{figure}

\begin{table}[!httbp]
\centering
  \begin{tabular}{ c | c | c | c | c | c | c| c| c| } \hline 
    $D$ &$|\v{c}_1|$ &  $|\v{c}_2|$ &  $\Delta \theta_c (rad)$ &
$|({\mathit K}\v{c})_1|$ & $|({\mathit K}\v{c})_2|$ & $\Delta \theta_k (rad)$ &
  $\mbox{Re}(\lambda)$ & Quiet \\
    \hline \hline   
    $2$ &  $1.0000$ & $0.0078$ & $2.71$ & $0.9989$ & $0.0467$ & $0.78$ & $0.215$ & cell 2 \\
    $2$ &  $0.1018$ & $0.9948$ & $2.31$ & $0.0468$ & $0.9989$ & $2.25$  & $0.017$ & cell 1 \\
    $4$ &  $0.9999$ & $0.0127$ & $1.26$ & $0.9973$ & $0.0736$ & $0.59$  & $0.252$ & cell 2 \\
    $4$ &  $0.1422$ & $0.9898$ & $3.77$ & $0.0731$ & $0.9973$ & $2.46$  & $0.0135$ & cell 1 \\
    $6$ &  $0.9998$ & $0.0182$ & $0.93$ & $0.9965$ & $0.0838$ & $0.49$  & $0.269$ & cell 2 \\
    $6$ &  $0.1556$ & $0.9878$ & $3.69$ & $0.0834$ & $0.9965$ & $2.56$  & $0.0010$ & cell 1 \\
    $8$ &  $0.9998$ & $0.0210$ & $0.80$ & $0.9962$ & $0.0872$ & $0.45$  & $0.278$ & cell 2  \\
    $8$ &  $0.1588$ & $0.9873$ & $3.64$ & $0.0804$ & $0.9963$ & $2.60$  & $-0.013$ & cell 1 \\
  \hline
\end{tabular}
\caption{Moduli of the components $c_j$ of the eigenvector
  $\v{c}=(c_1,c_2)^T$, normalized by $\v{c}^H\v{c}=1$, and of the
  normalized components of ${\mathit K}\v{c}$ at some specific points
  along the horizontal path through the scatter plot in
  Fig.~\ref{fig:twocell_hopf_scatt_c} for the two roots of
  $\det{\mathcal M}(\lambda)=0$ with the largest real part.
  $\Delta \theta_c$ and $\Delta \theta_k$ denote the phase shift for
  the eigenvector $\v{c}$ and for ${\mathit K}\v{c}$, respectively,
  between the cells. The dominant mode of instability always
  corresponds to the cell at $\x_2$ having very small amplitude
  oscillations as compared to that for the signaling cell. Parameters
  are as in Fig.~\ref{fig:twocell_hopf_unstable}.}
\label{table:eigvec_d}
\end{table}

For $D=3$ and $\sigma={1/2}$ in the scatter plot of
Fig.~\ref{fig:twocell_hopf_scatt_c}, in Fig.~\ref{fig:D3sig0.5_comp}
we show that our fast algorithm with $\Delta t=0.002$ is able to
reproduce with high accuracy the delicate wave-packet type
oscillations for cell 2 over long time scales that occurs in the
FlexPDE numerical solution of \eqref{DimLess_bulk}. The FlexPDE
results required hours of CPU time, whereas the hybrid algorithm
completed in less than a minute on a laptop.

\begin{figure}[htbp]
  \centering
  \begin{subfigure}[b]{0.49\textwidth}
 \includegraphics[width=\textwidth]{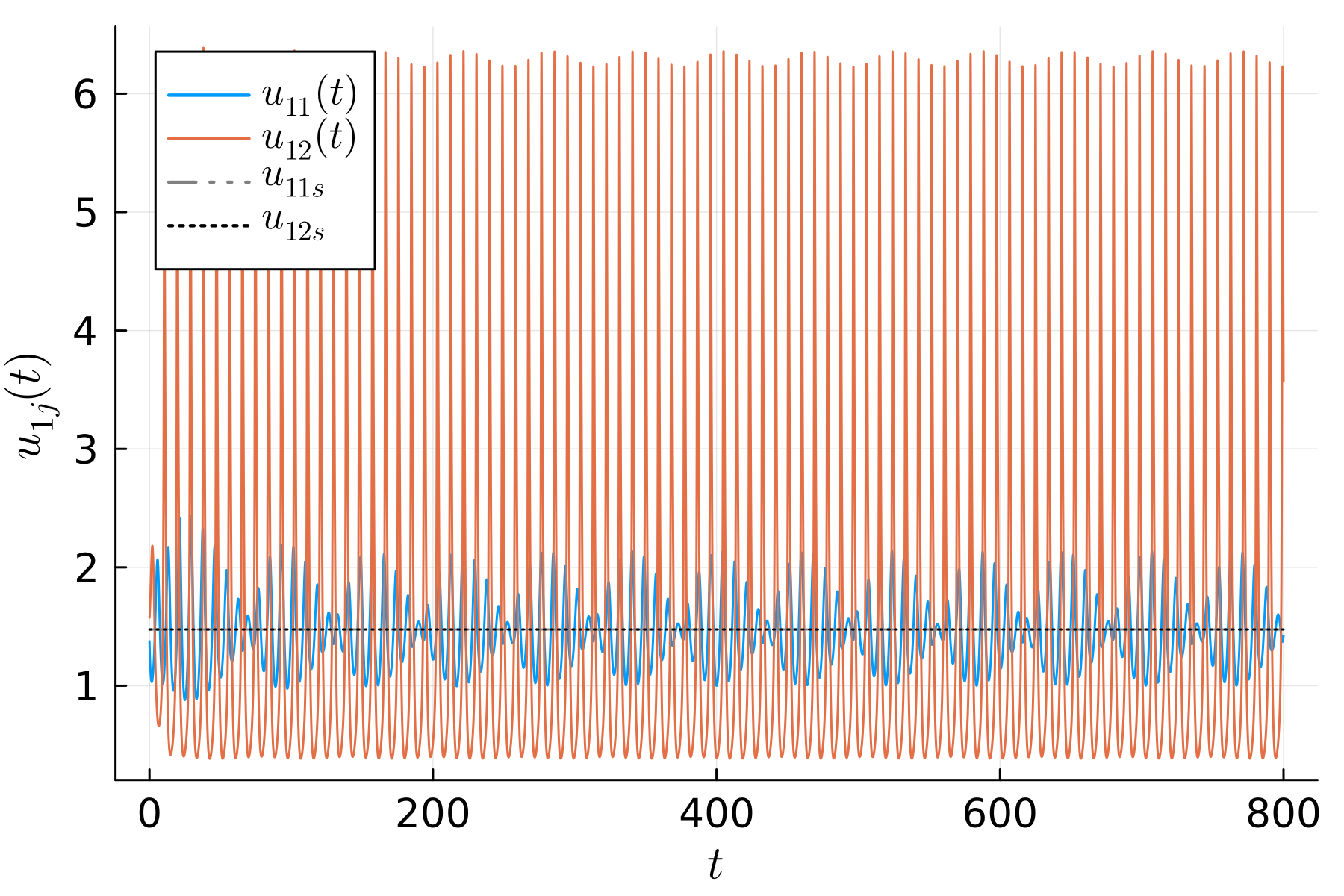}
        \caption{$u_{1j}$ hybrid scheme}
  \end{subfigure}
  \begin{subfigure}[b]{0.49\textwidth}
  \includegraphics[width=\textwidth]{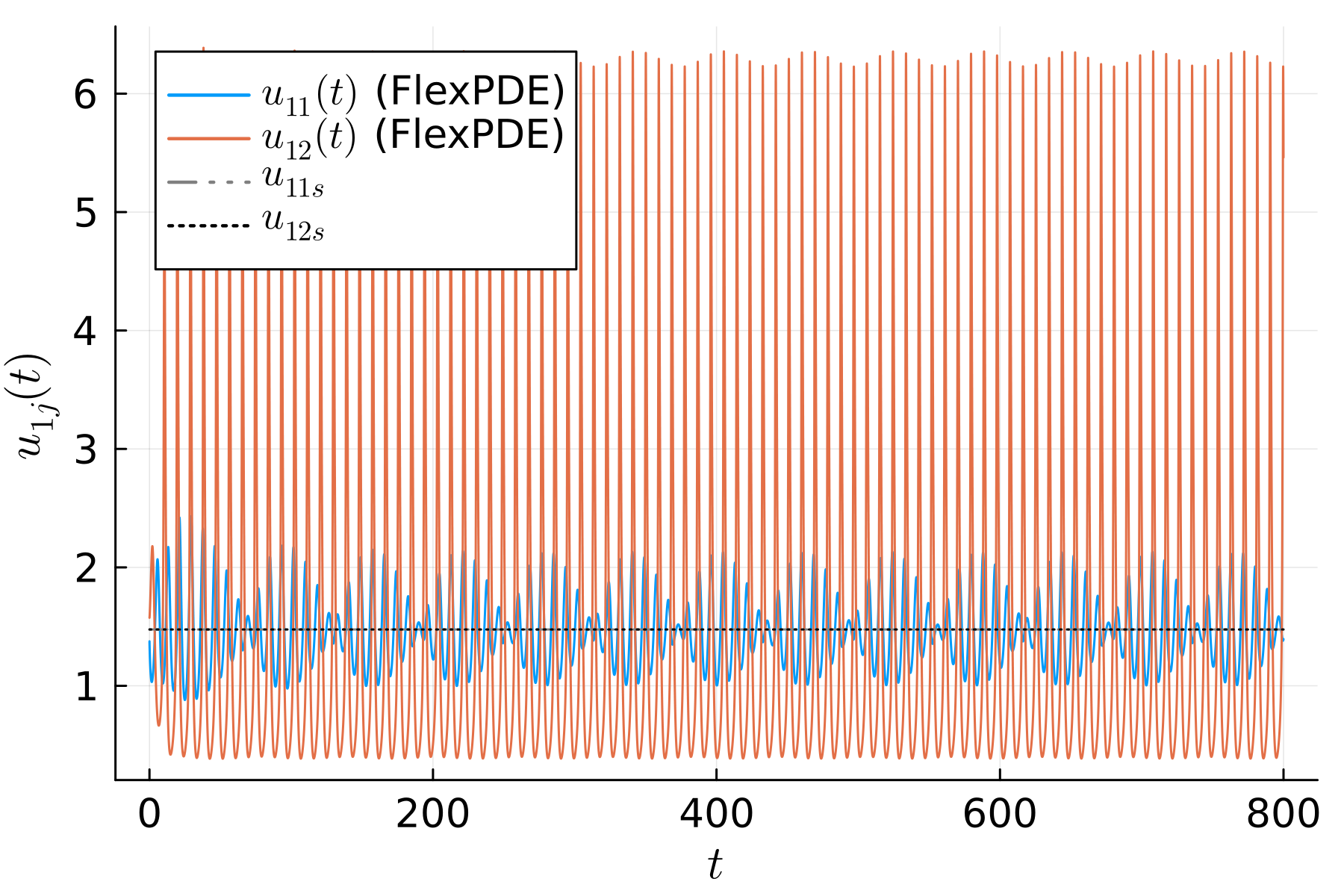}
   \caption{$u_{1j}$ FlexPDE}
  \end{subfigure}
  \begin{subfigure}[b]{0.49\textwidth}
 \includegraphics[width=\textwidth]{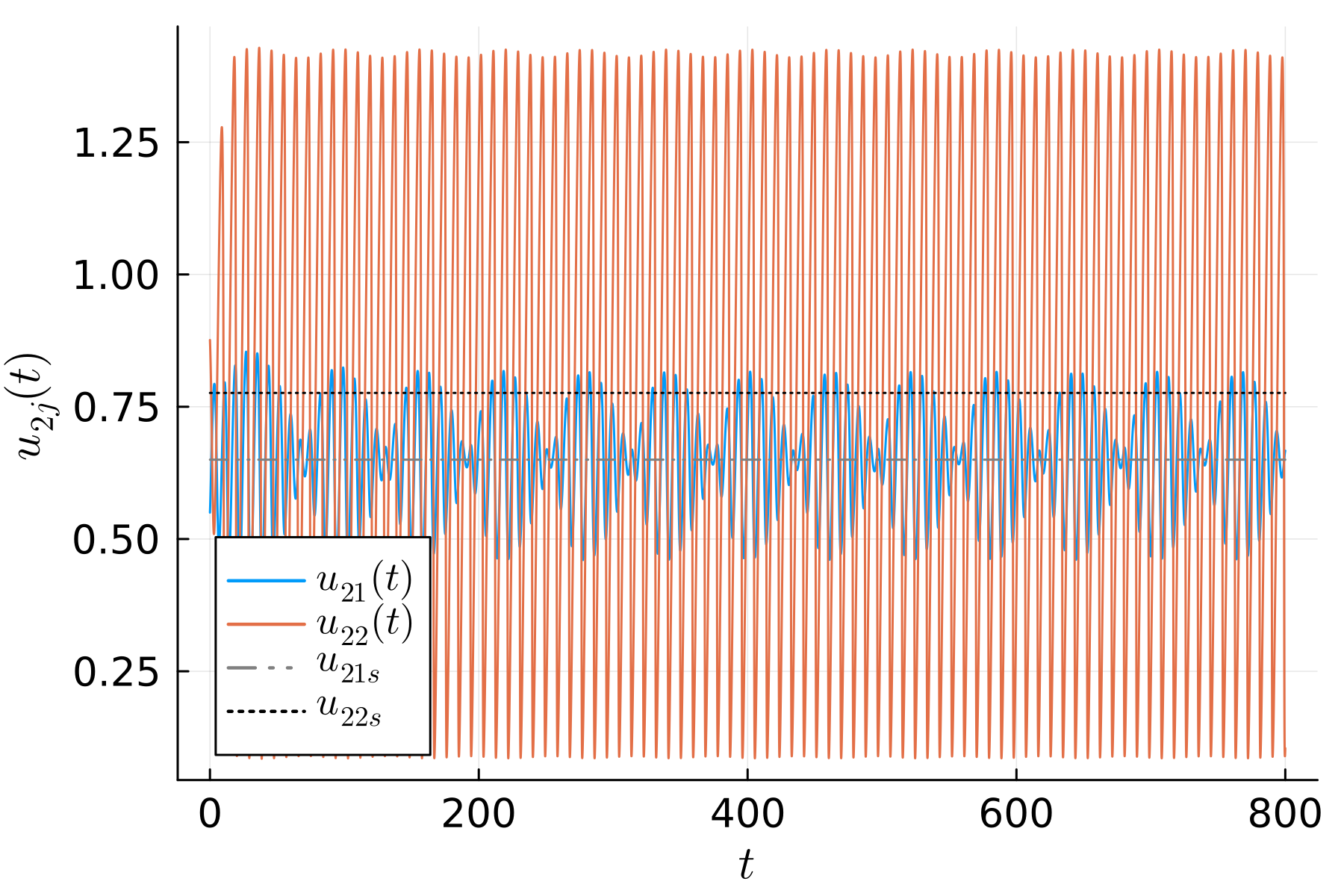}
        \caption{$u_{2j}$ hybrid scheme}
  \end{subfigure}
  \begin{subfigure}[b]{0.49\textwidth}
 \includegraphics[width=\textwidth]{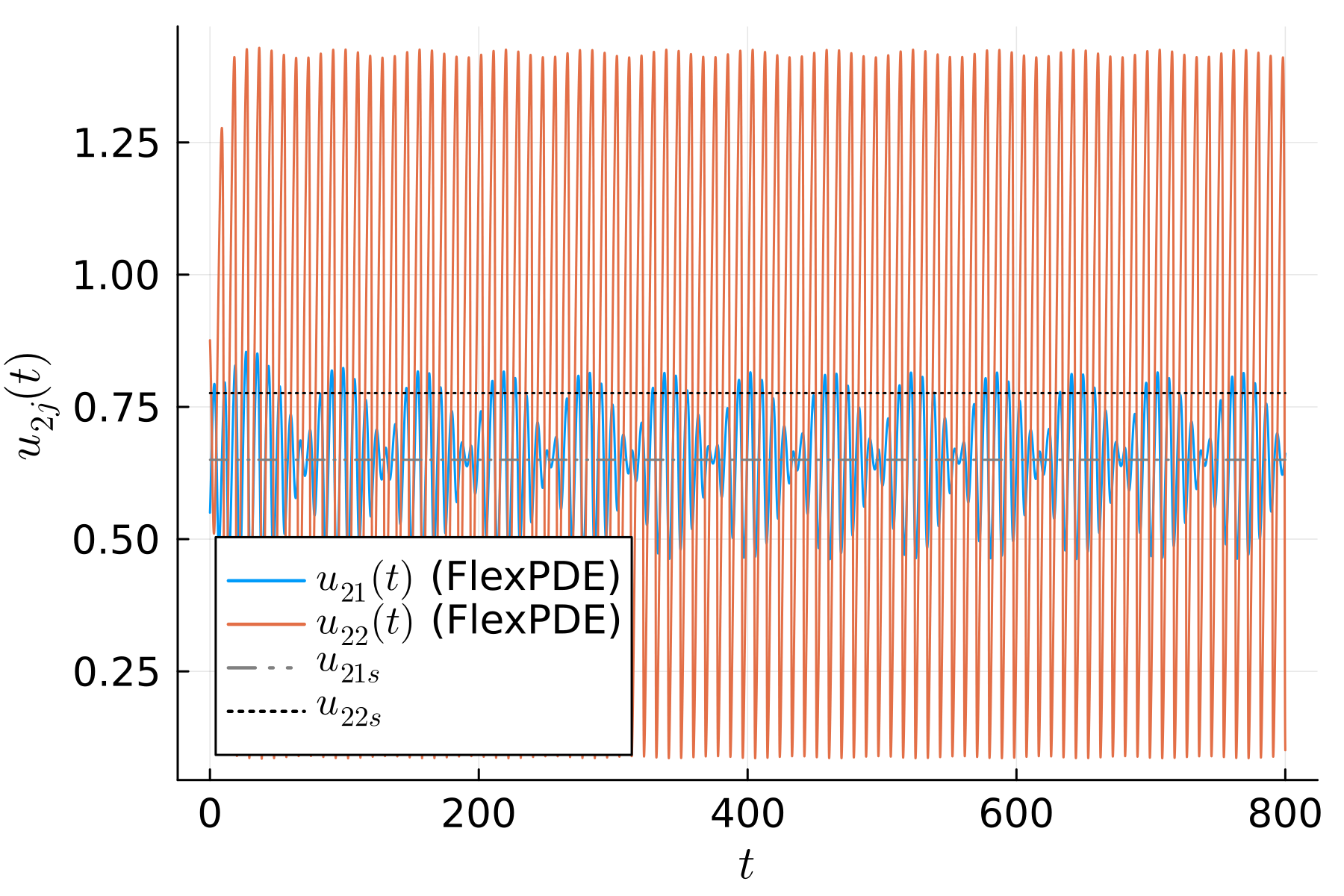}
   \caption{$u_{2j}$ FlexPDE}
  \end{subfigure}
  \caption{{Same caption as in Fig.~\ref{fig:D0p75Sig1over7} except
    that parameters are now $D=3$, $\sigma={1/2}$.  $\alpha_1=0.4$,
    $\alpha_2=0.9$ and $d_{1j}=1.5$ and $d_{2j}=0.2$ for
    $j\in\lbrace{1,2\rbrace}$ corresponding to the scatter plot in
    Fig.~\ref{fig:twocell_hopf_scatt_c}. The initial condition imposed
    was the steady-state with an anti-phase perturbation:
    $\v{u}_1^{(0)} = (u_{11s}, u_{21s})^T + 0.1 \cdot (1, -1)^T$, and
    similarly for $\v{u}_2^{(0)}$. The initial bulk solution for
    \eqref{DimLess_bulk} was $U(\v{x},0)=0$. The hybrid algorithm
    accurately reproduces the wave-packet oscillations of cell 2 over
    long time intervals. The unstable steady-states for cell 1 and cell 2 are
    the upper and lower black horizontal lines, respectively.}}
    \label{fig:D3sig0.5_comp}
\end{figure}

Next, we consider the effect on intracellular oscillations of
increasing the efflux permeability for the prior signaling cell
centered at $\x_1$ to $d_{21}=0.5$. From Fig.~\ref{fig:selkov_efflux}
we observe that limit cycle oscillations no longer occur for this cell
when it is uncoupled from the bulk, and so we refer to this cell as
being {\em deactivated}. Moreover, we decrease the influx parameter
for this cell to $d_{11}=0.4$, so that there is less influx from the
bulk medium as compared to the cell at $\x_2$. In
Fig.~\ref{fig:two_hopf_scatt_e} we show the resulting scatter plot in
the ${1/\sigma}$ versus $D$ plane. We now observe that there is a
parameter range (i.e. the white region) where the steady-state is
linearly stable. Along the dotted parameter path shown in
Fig.~\ref{fig:two_hopf_scatt_e}, in
Figs.~\ref{fig:two_hopf_f_mode1}-\ref{fig:two_hopf_f_mode2} we plot
$\mbox{Re}(\lambda)$ on the left vertical axis for each of the two
dominant modes. On the right vertical axes we plot our criterion in
(\ref{nstabform:c2}) that predicts which cell will have larger
amplitude oscillations near the steady-state.  In contrast to the
behavior in Fig.~\ref{fig:twocell_hopf_unstable} for the case where
the permeabilities were identical for the two cells, we observe from
the right vertical axes in
Figs.~\ref{fig:two_hopf_f_mode1}-\ref{fig:two_hopf_f_mode2} that, as
$D$ is increased with $\sigma=2$, the dominant mode of instability is
the one for which the deactivated cell with $\alpha_1=0.4$ now has
much smaller intracellular oscillations than the cell with
$\alpha_2=0.9$.

\begin{figure}[htbp]
  \centering
  \begin{subfigure}[b]{0.32\textwidth}  
 \includegraphics[width=\textwidth,height=4.4cm]{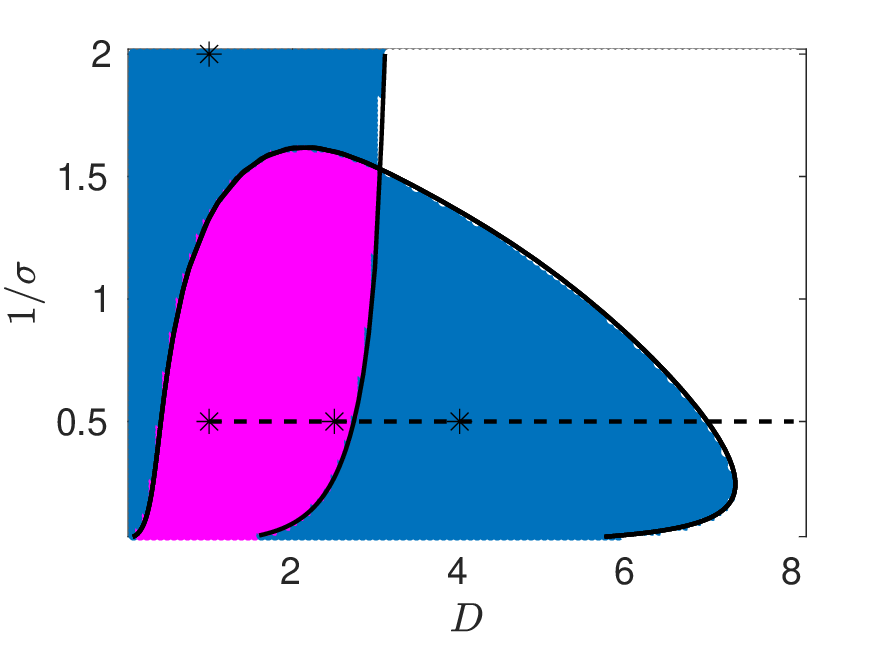}
        \caption{$d_{11}=0.4$ and $d_{21}=0.5$}
        \label{fig:two_hopf_scatt_e}
      \end{subfigure}
        \begin{subfigure}[b]{0.32\textwidth}  
         \includegraphics[width=\textwidth,height=4.4cm]{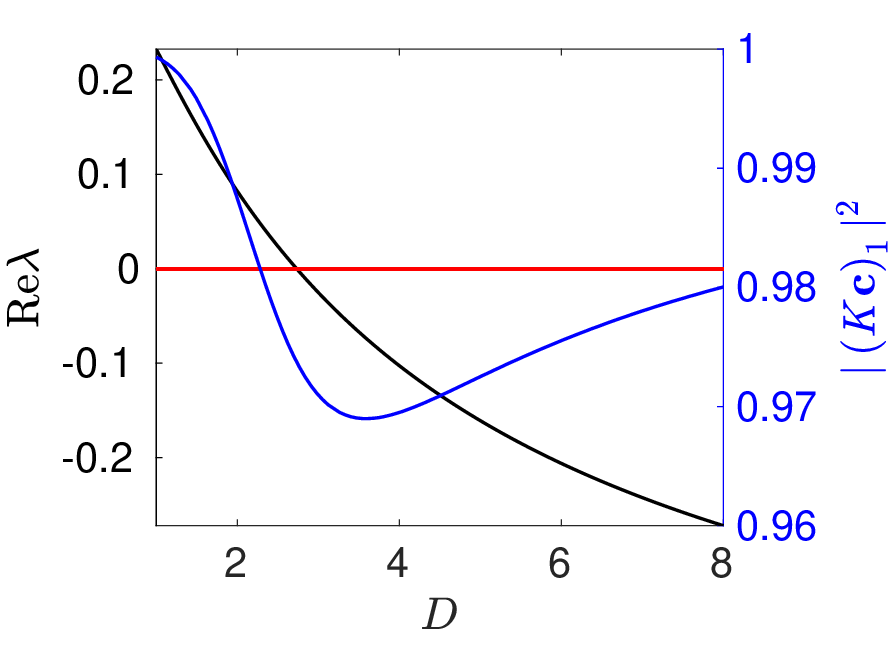}
        \caption{$\mbox{Re}(\lambda)$ for one mode}
        \label{fig:two_hopf_f_mode1}
      \end{subfigure}
      \begin{subfigure}[b]{0.32\textwidth}  
         \includegraphics[width=\textwidth,height=4.4cm]{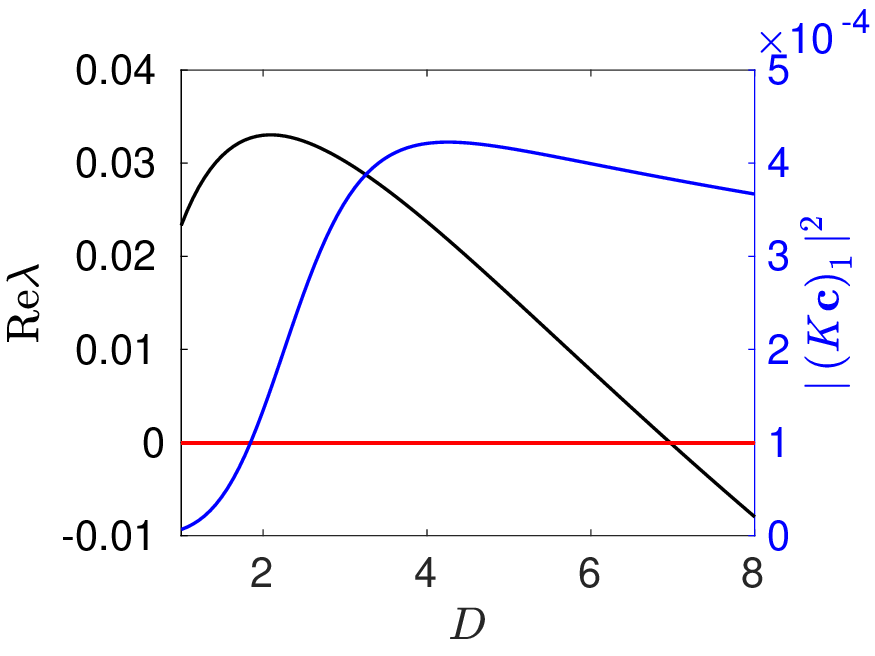}
        \caption{$\mbox{Re}(\lambda)$ for second mode}
        \label{fig:two_hopf_f_mode2}
      \end{subfigure}
      \caption{{Left: Scatter plot after decreasing the influx
          parameter to $d_{11}=0.4$ and increasing the efflux
          parameter to $d_{21}=0.5$ for the cell centered at $\x_1$
          with $\alpha_1=0.4$. Remaining parameter are as in
          Fig.~\ref{fig:twocell_hopf_unstable}. The number
          ${\mathcal Z}$ of destabilizing eigenvalues for the linearization
          of the steady-state is: ${\mathcal Z}=0$ (white),
          ${\mathcal Z}=2$ (blue) and ${\mathcal Z}=4$ (magenta).
          Middle: Left vertical axis is $\mbox{Re}(\lambda)$ along the
          dotted path in the left panel for the mode where
          oscillations have larger amplitude in cell 1. The right
          vertical axis is $|({\mathit K}\v{c})_1|^2$ (see
          (\ref{nstabform:c2})) that measures the relative amplitude
          of the oscillations near the steady-state for cell 1. Here
          we normalized $\sum_{j=1}^{2}|({\mathit
            K}\v{c})_j|^2=1$. Right: Same as middle plot but now for
          the other dominant mode where oscillations are predicted to
          be much larger in cell 2 since
          $|({\mathit K}\v{c})_1|^2={\mathcal O}(10^{-4})\ll 1$. The red
          horizontal lines denote the stability threshold
          $\mbox{Re}(\lambda)=0$.}}
\label{fig:twocell_hopf_f}
\end{figure}

{With the algorithm in \S \ref{sec:all_march} using $\Delta t=0.005$,
  and with uniformly random initial conditions near the steady-state
  values, in Fig.~\ref{fig:twocell_unstable_f_dyn} we show $u_{2j}$
  versus $t$ at the three star-labeled points with increasing $D$ along
  the parameter path in Fig.~\ref{fig:two_hopf_scatt_e}. As $D$
  increases, the amplitude of the intracellular oscillations in the
  deactivated cell 1 decreases as predicted by
  Figs.~\ref{fig:two_hopf_f_mode1}-\ref{fig:two_hopf_f_mode2}. However,
  when $D=1$ and $\sigma={1/2}$, the deactivated cell will have much
  larger amplitude oscillations than for the other cell centered at
  $\v{x}_2$.}

\begin{figure}[htbp]
  \centering
  \begin{subfigure}[b]{0.46\textwidth}  
     \includegraphics[width=\textwidth,height=4.4cm]{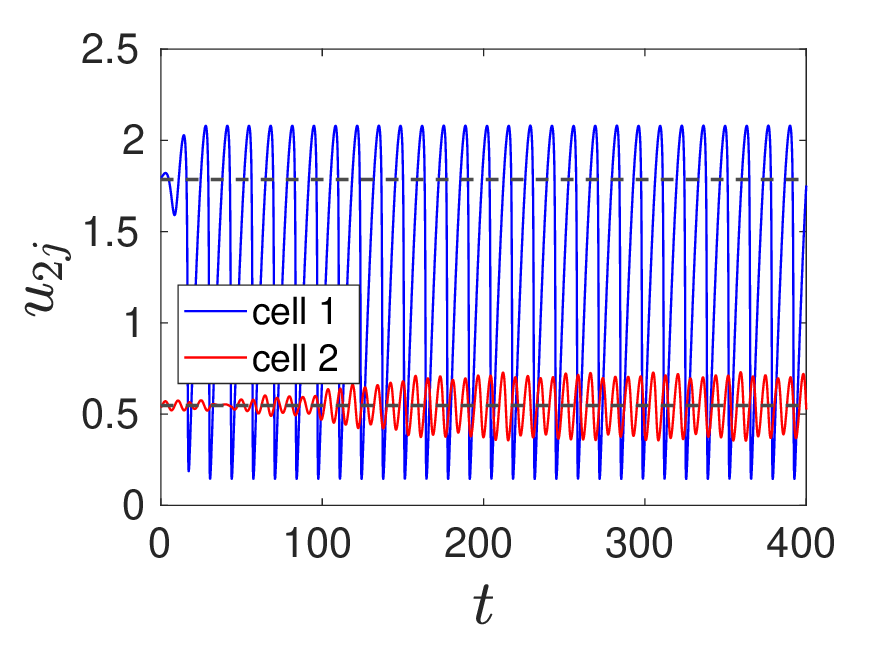}
  \caption{$D=1$, $\sigma=2$.} 
        \label{fig:unstable_center_f}
      \end{subfigure}
  \begin{subfigure}[b]{0.46\textwidth}  
     \includegraphics[width=\textwidth,height=4.4cm]{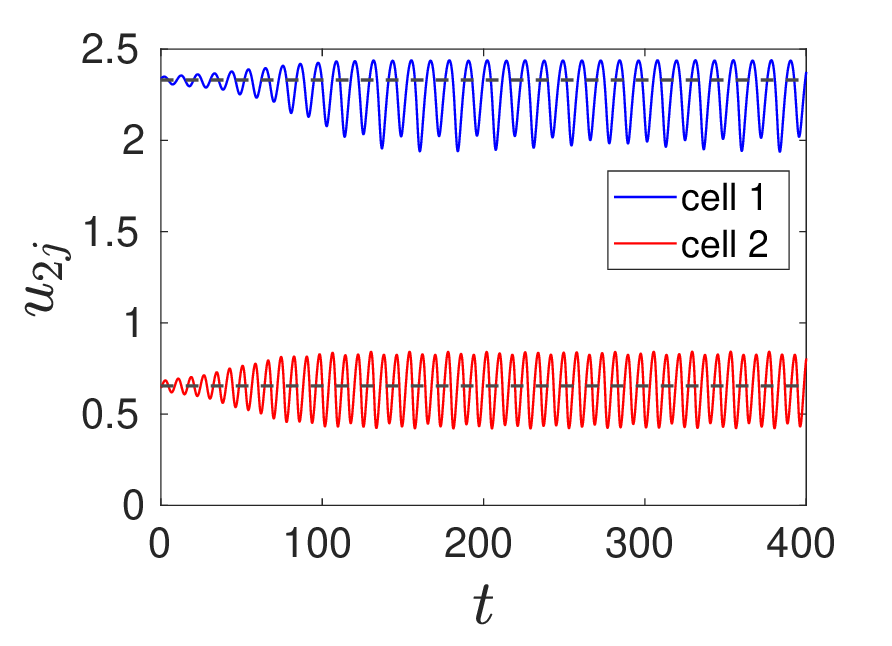}
  \caption{$D=2.5$, $\sigma=2$.} 
        \label{fig:unstable_middle_f}
      \end{subfigure}
  \begin{subfigure}[b]{0.46\textwidth}  
     \includegraphics[width=\textwidth,height=4.4cm]{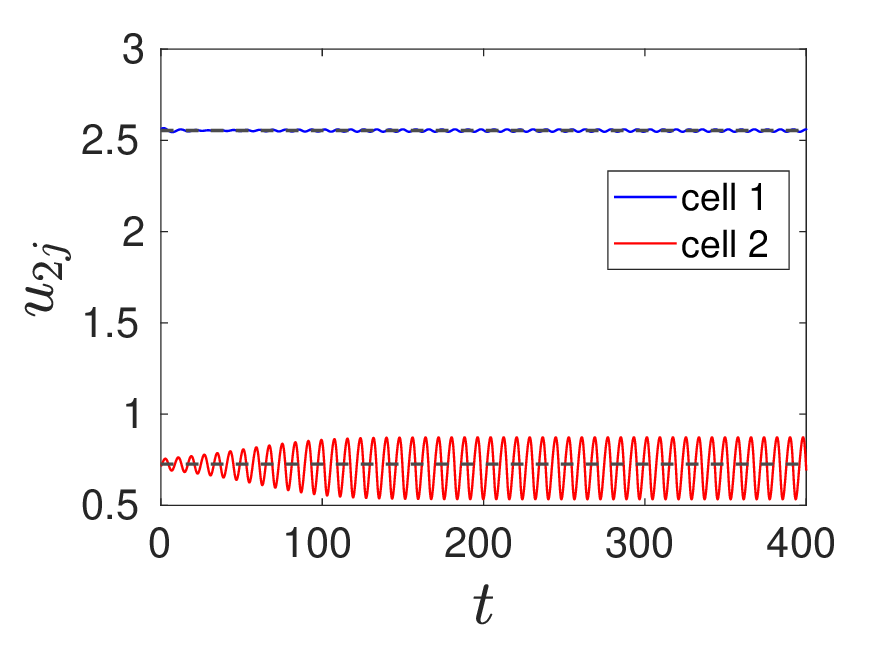}
  \caption{$D=4$, $\sigma=2$.}
  \label{fig:unstable_right_f}
\end{subfigure}
  \begin{subfigure}[b]{0.46\textwidth}  
     \includegraphics[width=\textwidth,height=4.4cm]{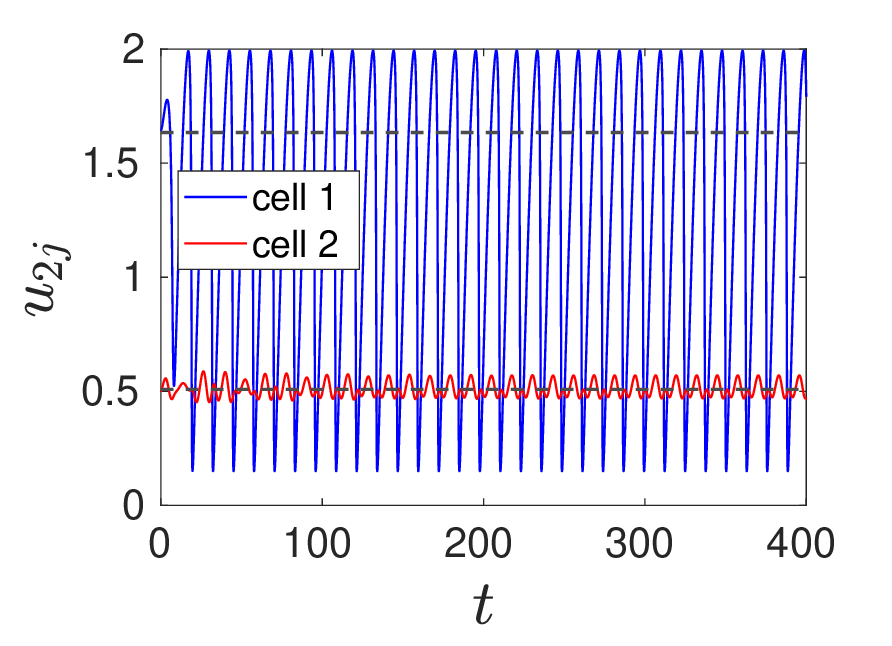}
  \caption{$D=1$, $\sigma={1/2}$.}
  \label{fig:unstable_top_f}
      \end{subfigure}
      \caption{{Top left, top right, bottom left: intracellular
          dynamics $u_{2j}(t)$ at the three star-labeled points along the path
          in Fig.~\ref{fig:two_hopf_scatt_e}. As $D$ increases, the
          amplitude of oscillations in the deactivated cell 1
          decreases. Bottom right: $u_{2j}(t)$ when $D=1$ and
          $\sigma={1/2}$ (see Fig.~\ref{fig:two_hopf_scatt_e}), where
          large oscillations in the deactivated cell occur.  Dashed
          horizontal lines are the steady-states.
          Parameters: $\alpha_1=0.4$, $\alpha_2=0.9$, $d_{11}=0.4$,
          $d_{21}=0.5$, $d_{12}=1.5$, $d_{22}=0.2$.}}
\label{fig:twocell_unstable_f_dyn}
\end{figure}

\subsection{A ring configuration of cells}\label{cell:fourcell}

{Next, we consider the four-cell configuration shown in
  Fig.~\ref{fig:fourcell}  where there are
  two pairs of cells with common kinetic parameters and
  permeabilities. For this case, and where $a$, $b$, $h$, and $d$
  can be identified from (\ref{5:gcep_2}), the GCEP matrix
  ${\mathcal M}$ in (\ref{5:gcep_2}) has the form
\begin{equation}\label{four:mat}
  {\mathcal M} = \begin{pmatrix}
                   a & h & d & h \\
                   h & b & h & d \\
                   d & h & a & h \\
                   h & d & h & b
                \end{pmatrix}\,.
\end{equation}
Omitting the details of the derivation, the
matrix spectrum of ${\mathcal M}$ is readily obtained as follows:

\begin{lemma} For the matrix in (\ref{four:mat}), we have
  \begin{equation}\label{four:mat_det}
   \det{\mathcal M} = (a-d)(b-d)\left[ (a+d)(b+d)-4h^2\right]\,.
\end{equation}
The matrix spectrum for ${\mathcal M}\v{c}=\chi\v{c}$ is
\begin{subequations}\label{four:mat_eig}
\begin{equation}\label{four:mat_eig_a}
  \begin{split}
    \v{c}_1&=(1,0,-1,0)^{T}\,, \quad \chi_1=a-d \,; \qquad
    \v{c}_2 =(0,1,0,-1)^{T} \,,\quad \chi_2=b-d \,, \\    
    \v{c}_{\pm}&=(1,f_{\pm},1,f_{\pm})^{T} \,, \quad \chi_{\pm}=
    (a+d)+2h f_{\pm}\,; \quad   f_{\pm}\equiv
    \frac{(b-a)}{4h} \pm \sqrt{ \frac{(b-a)^2}{(4h)^2}
      + 1 } \,.
  \end{split}
\end{equation}
\end{subequations}
\end{lemma}

To determine the eigenvalues $\lambda$ of the linearization of the
steady-state, which satisfy $\det{\mathcal M}(\lambda)=0$, we need
only find the union of the roots of the scalar root-finding problems
$a=d$, $b=d$ and $4h^2=(a+d)(b+d)$. The HB boundaries in the
${1/\sigma}$ versus $D$ plane are obtained by setting
$\lambda=i\lambda_I$. For identical cells where $a=b$, the modes
$\v{c}_1$ and $\v{c}_2$ are degenerate since $a=b$. For identical
cells, where $f_{\pm}=\pm 1$, the in-phase and anti-phase modes are
$\v{c}_{+}=(1,1,1,1)^T$ and $\v{c}_{-}=(1,-1,1,-1)^T$, respectively.

For identical cells with $d_{1j}=0.4$, $d_{2j}=0.2$ and $\alpha_j=0.9$
for $j\in\lbrace{1,\ldots,4\rbrace}$, and with a ring radius $r_c=1$, the
scatter plot is shown in Fig.~\ref{fig:fourcell_scatt_ident}. In
Fig.~\ref{fig:fourcell_ident_eig} we plot the real parts of the three
dominant eigenvalues of $\det{\mathcal M}(\lambda)=0$ along the
horizontal path shown in Fig.~\ref{fig:fourcell_scatt_ident}. In the
blue-shaded region, only the in-phase mode is destabilzing, while in the
green-shaded region all modes are destabilizing. Moreover, for $D=0.4$, all
the destabilizing modes have comparable growth rates. For this parameter
set, Fig.~\ref{fig:selkov_efflux} shows that no intracellular
oscillations would occur when there is no coupling to the bulk.}

\begin{figure}[htbp]
  \centering
  \begin{subfigure}[b]{0.48\textwidth}  
     \includegraphics[width=\textwidth,height=4.4cm]{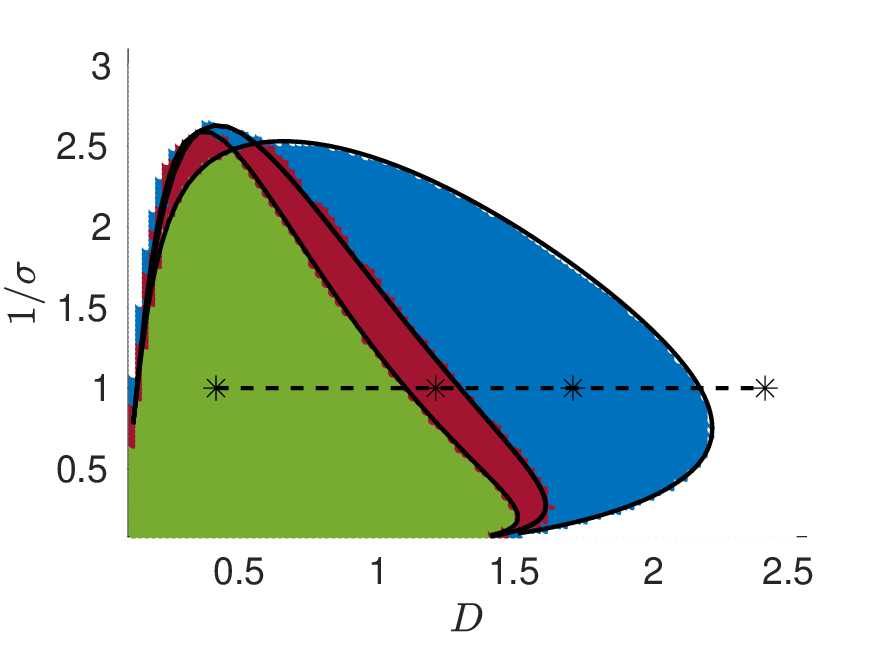}
        \caption{Four identical cells on a ring}
        \label{fig:fourcell_scatt_ident}
      \end{subfigure}
        \begin{subfigure}[b]{0.48\textwidth}  
           \includegraphics[width=\textwidth,height=4.4cm]{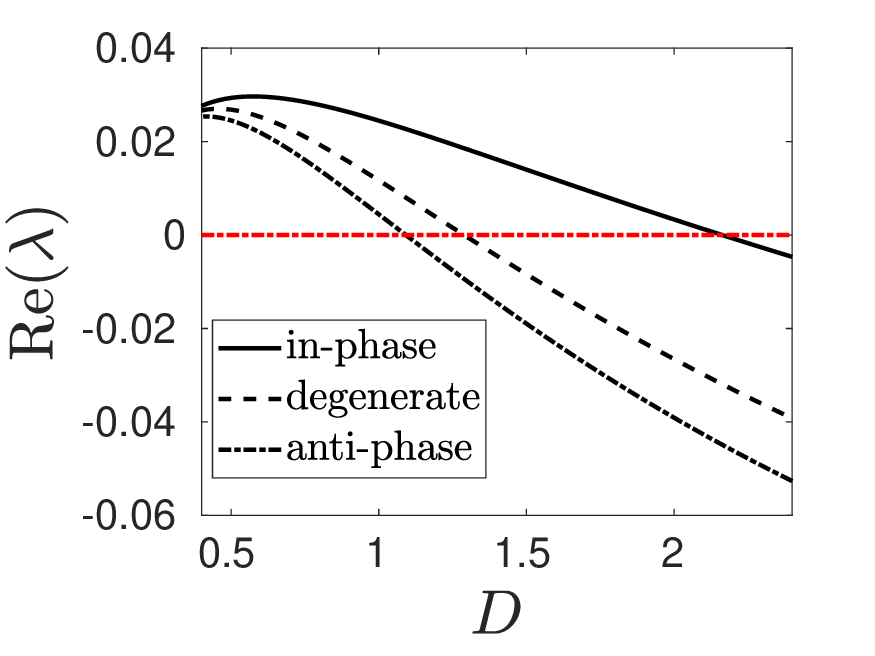}
        \caption{$\mbox{Re}(\lambda)$ on horizontal slice}
        \label{fig:fourcell_ident_eig}
      \end{subfigure}
      \caption{{Left: Scatter plot of the number ${\mathcal Z}$ of
          destabilizing eigenvalues, with ${\mathcal Z}=0$ (white),
          ${\mathcal Z}=2$ (blue), ${\mathcal Z}=6$ (rust) and
          ${\mathcal Z}=8$ (green) for the linearization of the
          steady-state for four identical cells with a ring radius
          $r_c=1$ (see Fig.~\ref{fig:fourcell}), with $d_{1j}=0.4$,
          $d_{2j}=0.2$ and $\alpha_j=0.9$ for
          $j\in\lbrace{1\,\ldots,4\rbrace}$. The HB boundaries (solid
          black curves) are superimposed. Right: Real part of the
          dominant eigenvalues of $\det{\mathcal M}(\lambda)=0$ along
          the horizontal path with $\sigma=1$ in the left
          panel. Owing to the circulant matrix structure of
          ${\mathcal M}(\lambda)$ there is mode degeneracy. In the
          blue region only the in-phase mode is destabilizing. For $D=0.4$,
          the three destabilizing modes have similar growth rates.}}
\label{fig:fourcell_hopf_identical}
\end{figure}

{In Fig.~\ref{fig:fourcell_iden_dyn} we plot $u_{2j}$ versus $t$ at
  the star-labeled points shown along the parameter path in
  Fig.~\ref{fig:fourcell_scatt_ident} as computed using the algorithm
  of \S \ref{sec:all_march} with $\Delta t=0.005$. The choice of
  initial conditions imposed near the steady-state values are
  indicated in the figure subcaptions. As predicted by the large
  growth rate of the in-phase mode in
  Fig.~\ref{fig:fourcell_ident_eig}, the intracellular oscillations
  become synchronous as time increases for $D=2.4$, $D=1.7$ and
  $D=1.2$. This is confirmed by the top row in
  Fig.~\ref{fig:fourcell_iden_dyn}.  Furthermore, for $D=0.4$, where
  all the destabilizing modes have comparable growth rates, the bottom
  row of Fig.~\ref{fig:fourcell_iden_dyn} shows, as predicted, that
  the long-term dynamics depends on the precise form of the initial
  conditions imposed.}

\begin{figure}[htbp]
  \centering
  \begin{subfigure}[b]{0.32\textwidth}  
     \includegraphics[width=\textwidth,height=4.4cm]{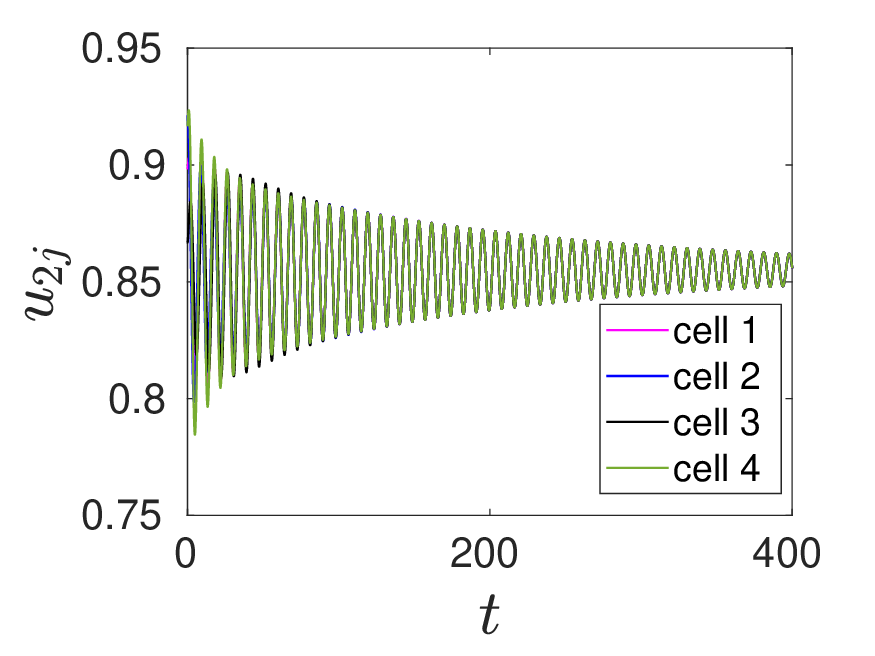}
  \caption{$D=2.4$ (random)} 
        \label{fig:iden:D2.4_arb}
  \end{subfigure}
  \begin{subfigure}[b]{0.32\textwidth}  
     \includegraphics[width=\textwidth,height=4.4cm]{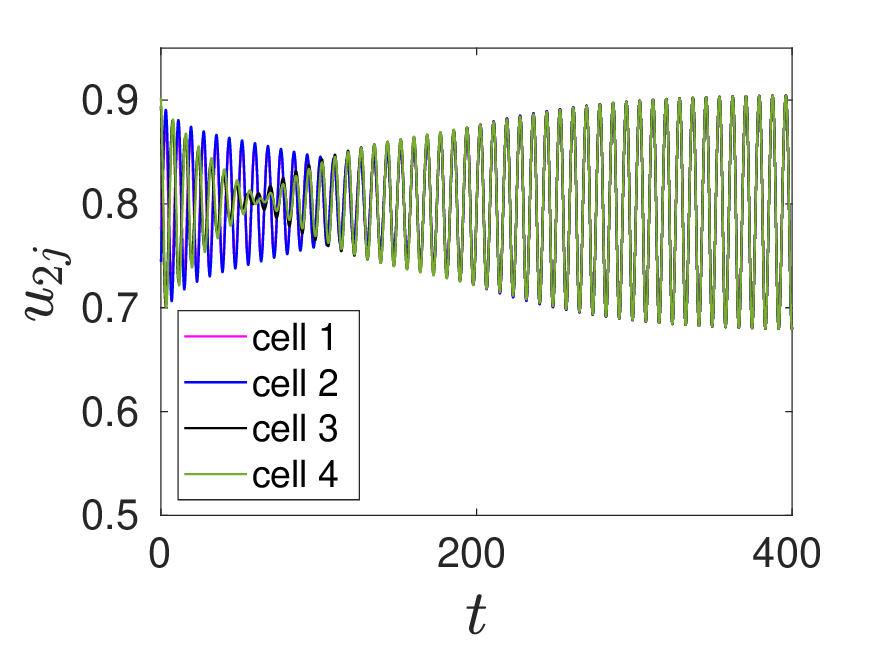}
  \caption{$D=1.7$ (random)} 
      \label{fig:iden:D1.7_arb}
  \end{subfigure}
  \begin{subfigure}[b]{0.32\textwidth}  
     \includegraphics[width=\textwidth,height=4.4cm]{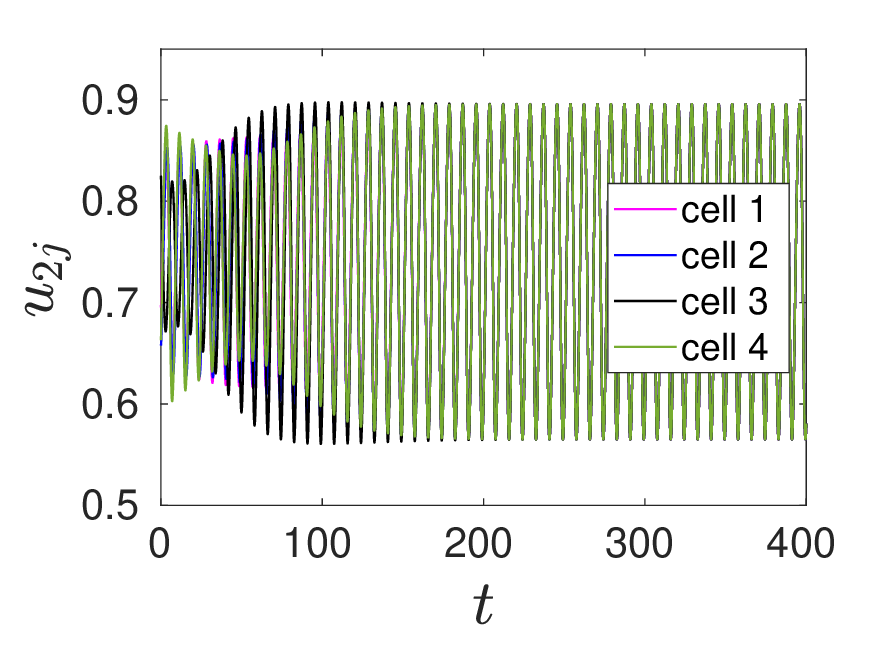}
  \caption{$D=1.2$ (random)}
  \label{fig:iden:D1.2_arb}
\end{subfigure}
\begin{subfigure}[b]{0.32\textwidth}  
   \includegraphics[width=\textwidth,height=4.4cm]{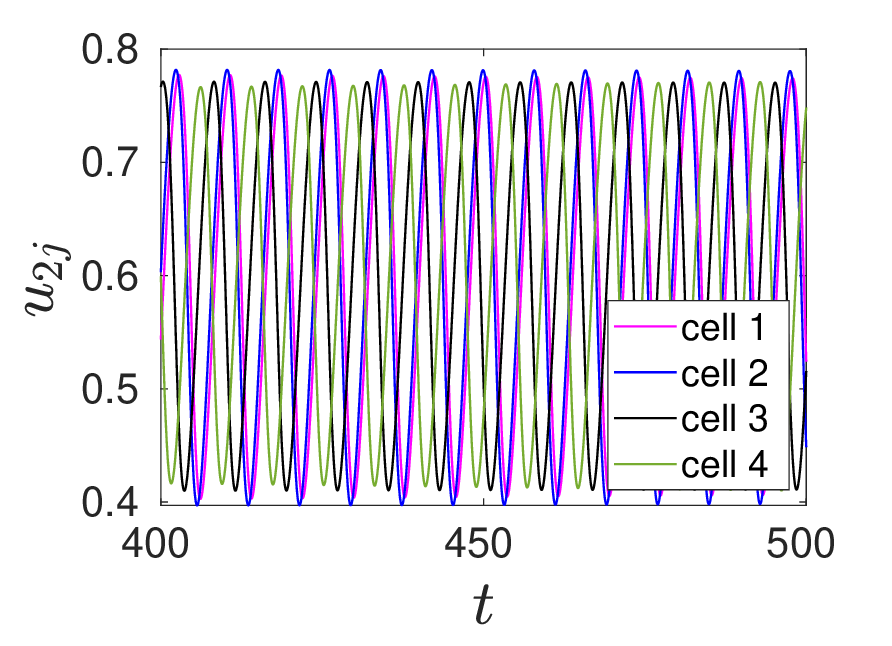}
  \caption{$D=0.4$ (random)} 
        \label{fig:iden:D0.4_arb}
      \end{subfigure}
  \begin{subfigure}[b]{0.32\textwidth}  
     \includegraphics[width=\textwidth,height=4.4cm]{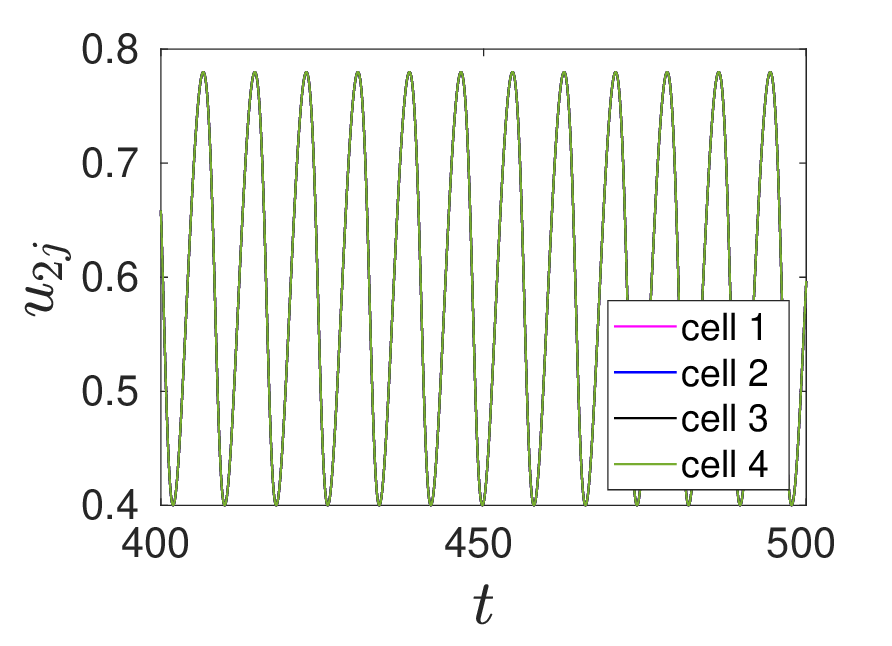}
  \caption{$D=0.4$ (in-phase)} 
        \label{fig:iden:D0.4_sync}
      \end{subfigure}
  \begin{subfigure}[b]{0.32\textwidth}  
     \includegraphics[width=\textwidth,height=4.4cm]{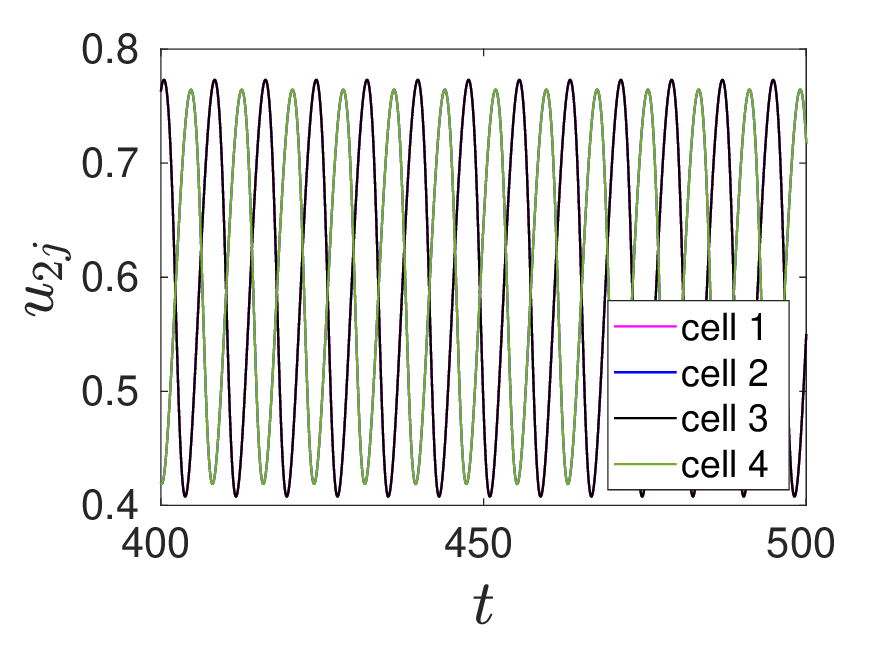}
  \caption{$D=0.4$ (anti-phase)}
  \label{fig:iden:D0.4_anti}
      \end{subfigure}
      \caption{{Intracellular dynamics $u_{2j}(t)$ at the star-labeled
          points along the horizontal slice with $\sigma=1$ in
          Fig.~\ref{fig:fourcell_scatt_ident} computed using the
          algorithm in \S \ref{sec:all_march} with $\Delta t=0.005$.
          The type of initial condition (IC) is indicated in the
          captions: Random: perturbing the steady-state values for
          $u_{1j}$ and $u_{2j}$ by $0.1 \cdot U$, where $U$ is uniformly
          distributed on $[-1,1]$. In-phase: perturbing steady-states
          by $0.1\cdot(1,1,1,1)^T$.  Anti-phase: perturbing steady-states
          by $0.1\cdot(1,-1,1,-1)^T$. Top left: For $D=2.4$, random IC lead
          to synchronous oscillations that slowly decay to the
          steady-state limit. Top middle and top right: For $D=1.7$
          and $D=1.2$, random IC lead to persistent in-phase
          oscillations.  Bottom row: For $D=0.4$, the choice of IC
          imposed leads to distinctly different long-time dynamics as
          suggested by Fig.~\ref{fig:fourcell_ident_eig}. Parameters:
          Identical cells with $\alpha_j=0.9$, $d_{1j}=0.4$ and
          $d_{2j}=0.2$ for $j\in\lbrace{1,\ldots,4\rbrace}$.}}
\label{fig:fourcell_iden_dyn}
\end{figure}

{Next, we consider two pairs of identical cells with
  $\alpha_2=\alpha_4=0.5$, $\alpha_1=\alpha_3=0.9$, and with the same
  permeabilities $d_{1j}=0.4$ and $d_{2j}=0.2$ for
  $j\in\lbrace{1,\ldots,4\rbrace}$ as in
  Fig.~\ref{fig:fourcell_scatt_ident}.  For this parameter set, from
  Fig.~\ref{fig:selkov_efflux} we observe that cells 2 and 4 would
  exhibit limit-cycle oscillations when uncoupled from the bulk. We
  refer to these two cells as signaling cells. In
  Fig.~\ref{fig:fourcell_scatt_a} we show the scatter plot, while in
  Fig.~\ref{fig:fourcell_a_eig} we plot the real parts of the four
  dominant eigenvalues of $\det{\mathcal M}(\lambda)=0$ along the
  parameter path shown in Fig.~\ref{fig:fourcell_scatt_a}. In the
  magenta-shaded region the dominant instability is where the
  signaling cells are activated. For $D=0.4$, there are four destabilizing
  modes.  In Fig.~\ref{fig:fourcell_a_dyn} we plot $u_{2j}$ versus $t$
  at the two star-labeled points along the parameter path in
  Fig.~\ref{fig:fourcell_scatt_a} as computed using the algorithm of
  \S \ref{sec:all_march} with $\Delta t=0.005$.  For $D=1.75$, the
  non-signaling cells 1 and 3 have in-phase small amplitude
  oscillations. However, for $D=0.5$, as a result of the multiple
  destabilizing modes shown in Fig.~\ref{fig:fourcell_a_eig}, the type of
  oscillation that occurs in the non-signaling cells depends on the
  choice of initial condition near the steady-state that is imposed
  (see Fig.~\ref{fig:a:D0.5_arb}--\ref{fig:a:D0.5_anti}).}

\begin{figure}[htbp]
  \centering
      \begin{subfigure}[b]{0.48\textwidth}  
         \includegraphics[width=\textwidth,height=4.4cm]{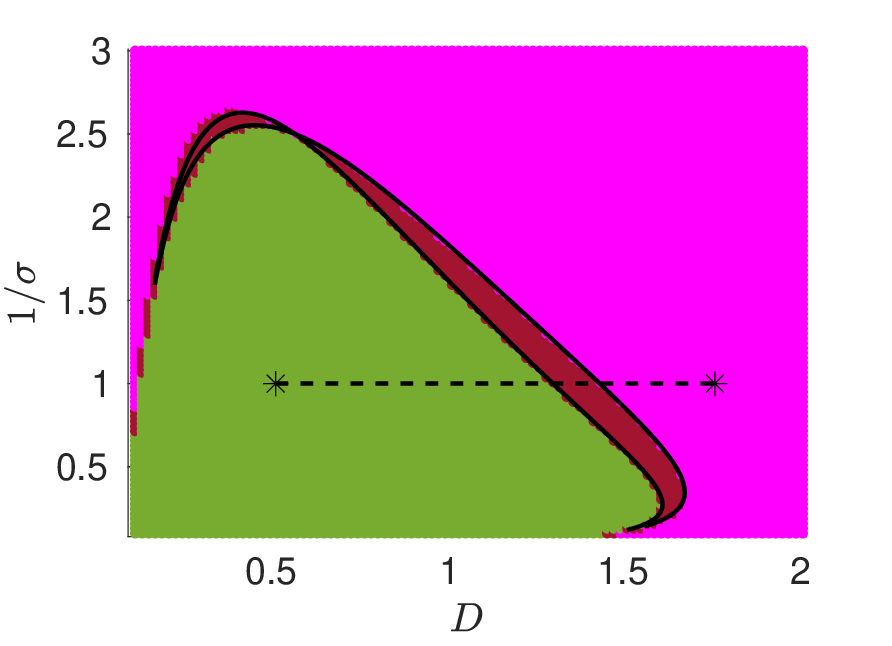}
        \caption{$\alpha_2=\alpha_4=0.5$}
        \label{fig:fourcell_scatt_a}
      \end{subfigure}
        \begin{subfigure}[b]{0.48\textwidth}  
          \includegraphics[width=\textwidth,height=4.4cm]{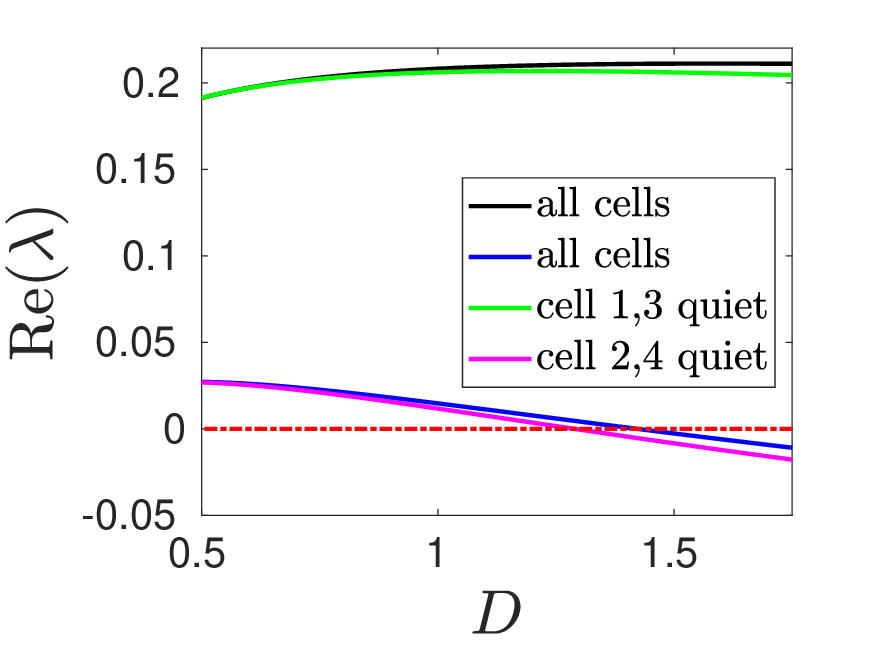}
        \caption{$\mbox{Re}(\lambda)$ on horizontal slice}
        \label{fig:fourcell_a_eig}
      \end{subfigure}
      \caption{{Left: Scatter plot of the number ${\mathcal Z}$ of
          destabilizing eigenvalues, with ${\mathcal Z}=4$ (magenta),
          ${\mathcal Z}=6$ (rust) and ${\mathcal Z}=8$ (green) for the
          linearization of the steady-state for two pairs of identical
          cells with $r_c=1$ where $\alpha_{2}=\alpha_{4}=0.5$ and
          $\alpha_1=\alpha_3=0.9$. The HB boundaries are superimposed.
          Same permeabilities $d_{1j}$ and $d_{2j}$ as in
          Fig.~\ref{fig:fourcell_scatt_ident}. Cells 2 and 4 admit
          limit cycle oscillations when uncoupled from the bulk.
          Right: Real part of the four dominant eigenvalues of
          $\det{\mathcal M}(\lambda)=0$ on the horizontal path
          with $\sigma=1$ in the left panel. The labels for the
          eigenmodes from (\ref{four:mat_eig_a}) are
          $\v{c}_1=(1,0,-1,0)^{T}$ (cells 2, 4 silent),
          $\v{c}_2 =(0,1,0,-1)^{T}$ (cells 1, 3 silent), and
          $\v{c}_{\pm}=(1,f_{\pm},1,f_{\pm})^{T}$ (all cells active). There
          are two destabilizing modes when $D=1.75$ and four when
          $D=0.5$.}}
\label{fig:fourcell_hopf_a}
\end{figure}

\begin{figure}[htbp]
  \centering
  \begin{subfigure}[b]{0.32\textwidth}  
     \includegraphics[width=\textwidth,height=4.4cm]{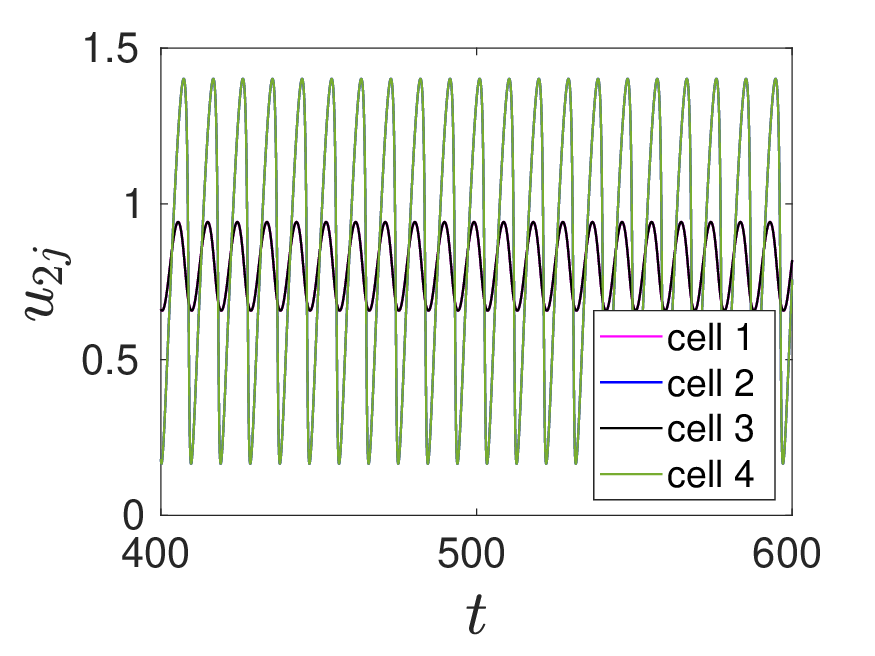}
  \caption{$D=1.75$ (random)} 
        \label{fig:a:D1.75_arb}
  \end{subfigure}
  \begin{subfigure}[b]{0.32\textwidth}  
     \includegraphics[width=\textwidth,height=4.4cm]{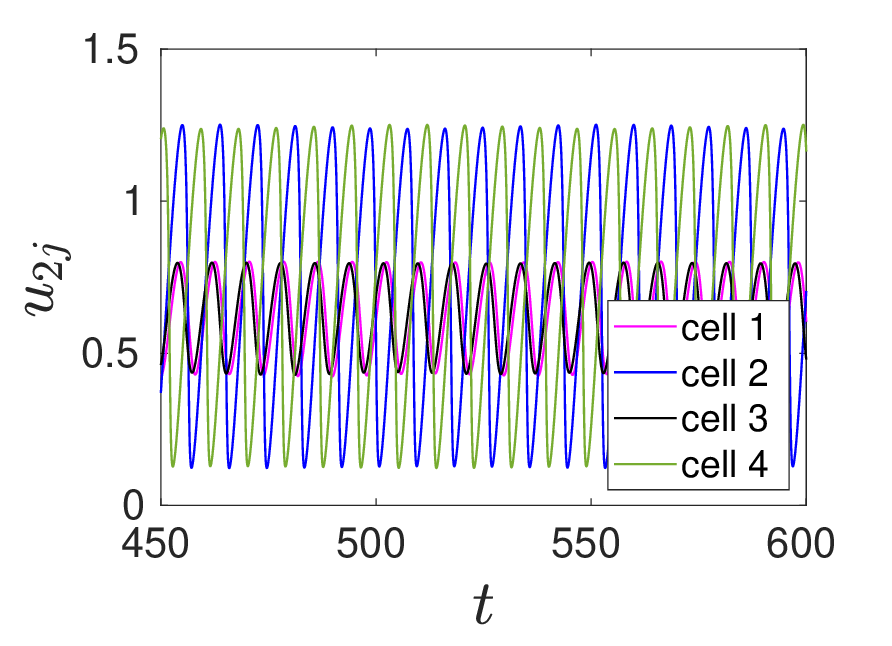}
  \caption{$D=0.5$ (random)} 
      \label{fig:a:D0.5_arb}
  \end{subfigure}
  \begin{subfigure}[b]{0.32\textwidth}  
     \includegraphics[width=\textwidth,height=4.4cm]{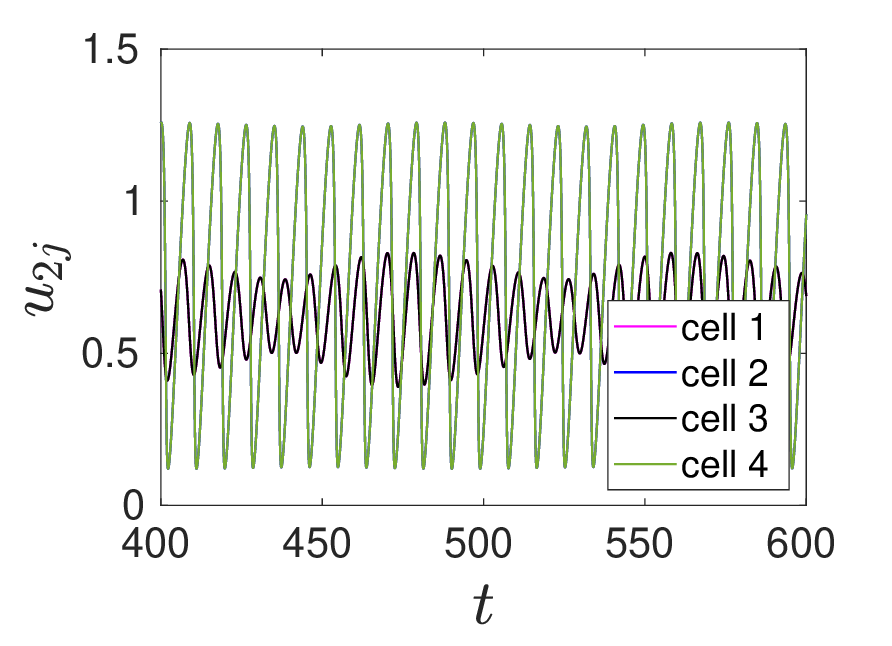}
  \caption{$D=0.5$ (anti-phase)}
  \label{fig:a:D0.5_anti}
\end{subfigure}
\caption{{Intracellular dynamics $u_{2j}(t)$ at the two star-labeled 
    points along the horizontal path with $\sigma=1$ in
    Fig.~\ref{fig:fourcell_scatt_a} computed using the algorithm in \S
    \ref{sec:all_march} with $\Delta t=0.005$.  The choice of initial
    condition (IC) is indicated (see caption of
    Fig.~\ref{fig:fourcell_a_dyn} for details). Left: For $D=1.75$
    cells 2 and 4 trigger small amplitude oscillations in cells 1 and
    3. Middle: for $D=0.5$ with a uniformly random IC, there is a
    phase shift in the oscillations of the signaling cells. Incoherent
    oscillations occur in the other cells. Right: For $D=0.5$ with
    anti-phase IC, the signaling cells oscillate in-phase. Cells 1
    and 3 exhibit mixed-mode type oscillations.}}
\label{fig:fourcell_a_dyn}
\end{figure}

{Finally, we explore the effect of increasing the efflux to
  $d_{22}=d_{24}=0.5$ for the signaling cells with
  $\alpha_2=\alpha_4=0.5$, while also increasing the influx to
  $d_{11}=d_{13}=2.0$ for the non-signaling cells where
  $\alpha_1=\alpha_3=0.9$. For this modified parameter set,
  Fig.~\ref{fig:selkov_efflux} shows that cells 2 and 4 no longer
  exhibit limit cycle oscillations when uncoupled from the bulk. As
  such, these two cells have now been deactivated.  In
  Fig.~\ref{fig:fourcell_scatt_b} we show how the scatter plot of
  Fig.~\ref{fig:fourcell_scatt_a} is modified, while in
  Fig.~\ref{fig:fourcell_b_eig} we plot the real parts of the four
  dominant eigenvalues of $\det{\mathcal M}(\lambda)=0$. In contrast
  to Fig.~\ref{fig:fourcell_a_eig}, we now observe that the dominant
  instability at the right-end of the horizontal path where
  $D=2$ and $\sigma=4$ in Fig.~\ref{fig:fourcell_scatt_b} is where the
  {deactivated cells 2 and 4 are now effectively quiet.}  In
  Fig.~\ref{fig:fourcell_b_dyn} we plot $u_{2j}$ versus $t$ at the
  four star-labeled points in Fig.~\ref{fig:fourcell_scatt_b}, as
  computed using the algorithm of \S \ref{sec:all_march} with
  $\Delta t=0.005$.  For $D=2$ and $\sigma=4$, in
  Fig.~\ref{fig:b:D2_sig4_arb} we observe that the non-signaling
  cells, cells 1 and 3 (which would not oscillate without cell-bulk
  coupling), have much larger amplitude oscillations than do the
  deactivated signaling cells, as is consistent with the prediction
  from Fig.~\ref{fig:fourcell_a_eig}. {Moreover, when both $D$ and
    $\sigma$ are smaller, i.e. $D=\sigma=1$,
    Fig.~\ref{fig:b:D1_sig1_arb} shows that the non-signaling cells
    have a much small amplitude oscillation in comparison with the
    deactivated signaling cells. To interpret this behavior, for
    smaller bulk diffusivity and with a lower degradation rate, the
    spatial gradient in (\ref{nstabform:c1}) that occurs in the
    vicinity of the deactivated cells can be large. Owing to the
    influx permeability, this large signaling gradient can effectively
    re-activate these cells when $D$ is small. We suggest that this
    behavior is related to the qualitative mechanism of
    \emph{diffusion-sensing} (cf.~\cite{DSQS}). By increasing $D$ to
    $D=1.65$, Fig.~\ref{fig:b:D1.65_sig1_arb} shows that both pairs of
    cells exhibit only small amplitude oscillations. Overall, the
    distinct time-dependent behaviors observed in
    Fig.~\ref{fig:fourcell_b_dyn} illustrate how the components of the
    eigenvector $\v{c}$ of the GCEP matrix encodes key predictive
    information on which cells will have larger oscillations near an
    unstable steady-state.}

\begin{figure}[htbp]
  \centering
        \begin{subfigure}[b]{0.48\textwidth}  
           \includegraphics[width=\textwidth,height=4.4cm]{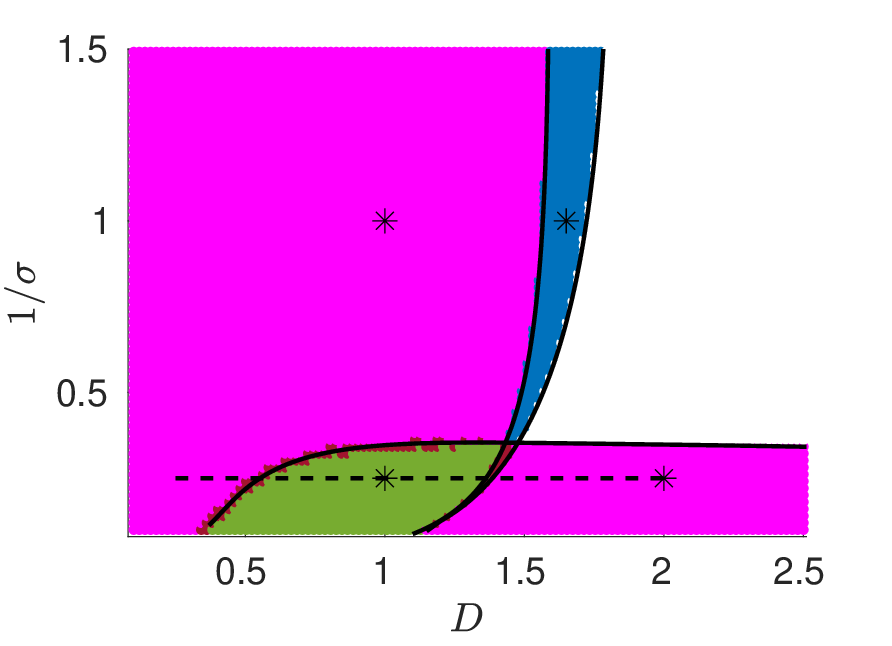}
        \caption{$d_{11}=d_{13}=2.0$, $d_{22}=d_{24}=0.5$}
        \label{fig:fourcell_scatt_b}
      \end{subfigure}
              \begin{subfigure}[b]{0.48\textwidth}  
            \includegraphics[width=\textwidth,height=4.4cm]{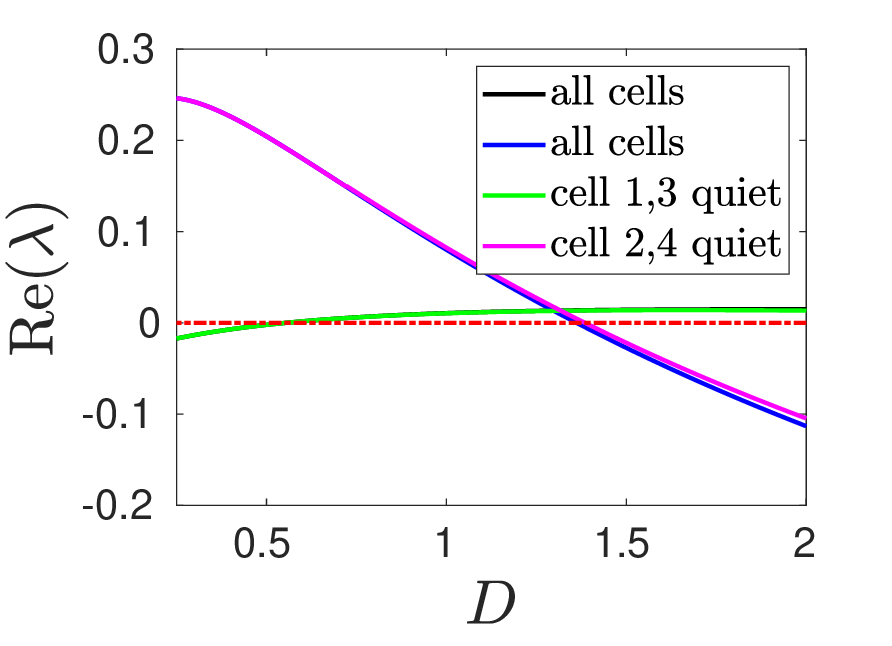}
        \caption{$\mbox{Re}(\lambda)$ on horizontal slice}
        \label{fig:fourcell_b_eig}
      \end{subfigure}
      \caption{{Left: Scatter plot with $r_c=1$,
          $\alpha_{2}=\alpha_{4}=0.5$ and $\alpha_1=\alpha_3=0.9$, but
          now $d_{11}=d_{13}=2.0$ and $d_{22}=d_{24}=0.5$. Remaining
          parameters are $d_{12}=d_{14}=0.4$ and
          $d_{21}=d_{23}=0.2$. The deactivated signaling cells 2 and
          4 now have an increased efflux and no longer admit limit
          cycle oscillations when uncoupled from the bulk. The influx
          permeability for the other two cells is much larger than in
          Fig.~\ref{fig:fourcell_scatt_a}. Right: Real part of the
          four dominant eigenvalues of $\det{\mathcal M}(\lambda)=0$
          on the horizontal path with $\sigma=0.25$ in the left
          panel.}}
\label{fig:fourcell_hopf_b}
\end{figure}

\begin{figure}[htbp]
  \centering
  \begin{subfigure}[b]{0.46\textwidth}  
     \includegraphics[width=\textwidth,height=4.4cm]{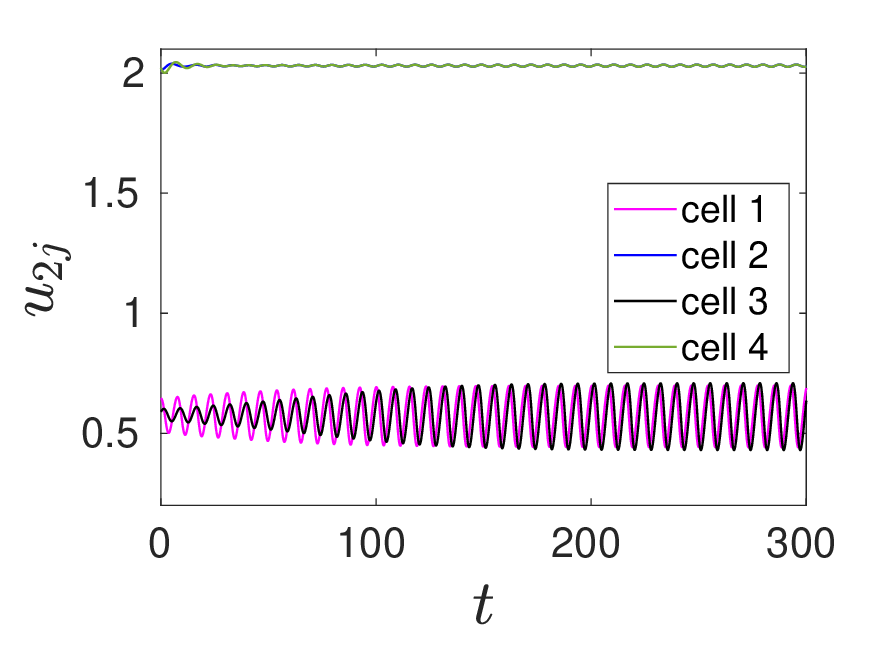}
  \caption{$D=2$, $\sigma=4$ (random)} 
        \label{fig:b:D2_sig4_arb}
      \end{subfigure}
  \begin{subfigure}[b]{0.46\textwidth}  
     \includegraphics[width=\textwidth,height=4.4cm]{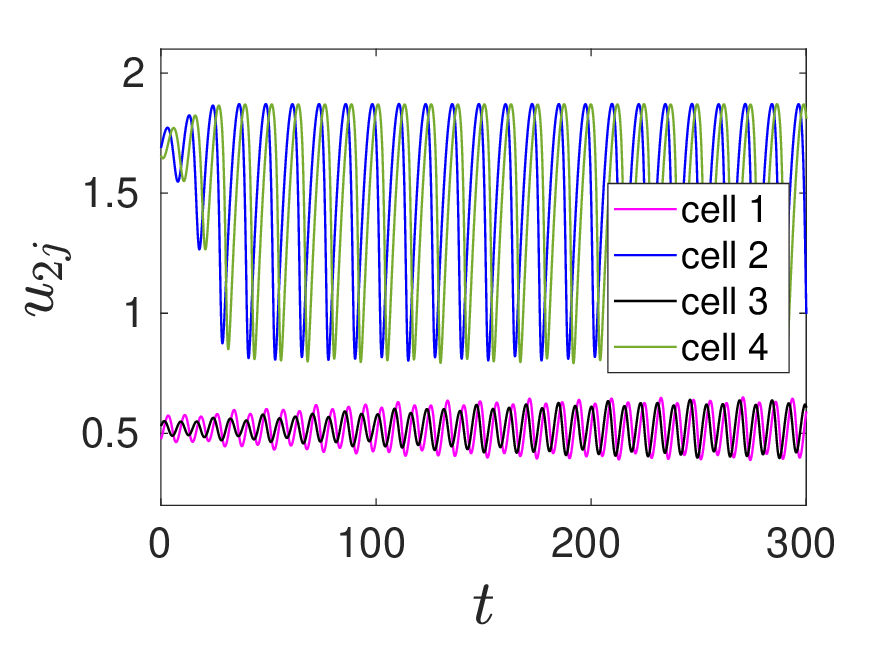}
  \caption{$D=1$, $\sigma=4$ (random)} 
        \label{fig:b:D1_sig4_arb}
      \end{subfigure}
  \begin{subfigure}[b]{0.46\textwidth}  
     \includegraphics[width=\textwidth,height=4.4cm]{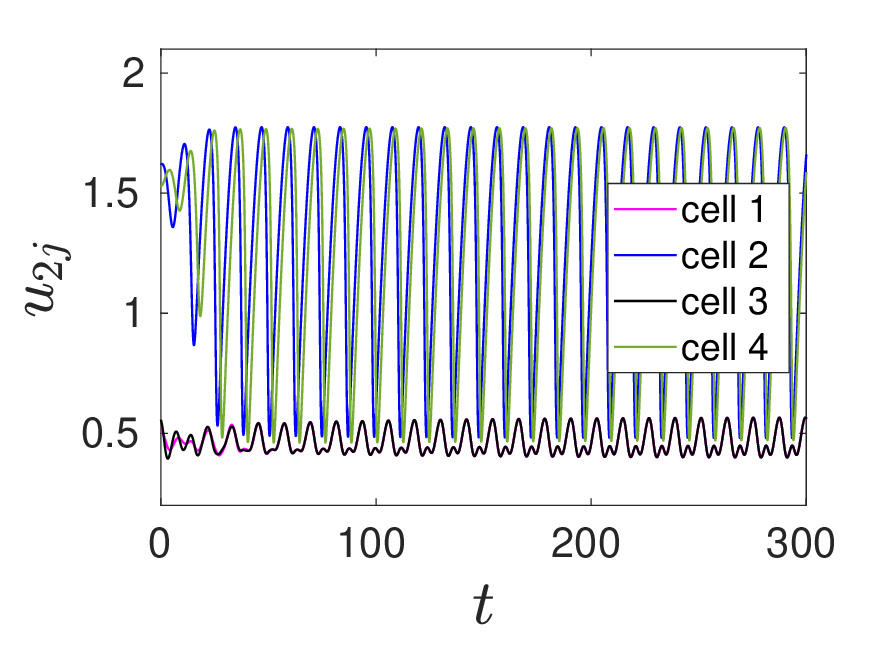}
  \caption{$D=1$, $\sigma=1$ (random)}
  \label{fig:b:D1_sig1_arb}
\end{subfigure}
  \begin{subfigure}[b]{0.46\textwidth}  
     \includegraphics[width=\textwidth,height=4.4cm]{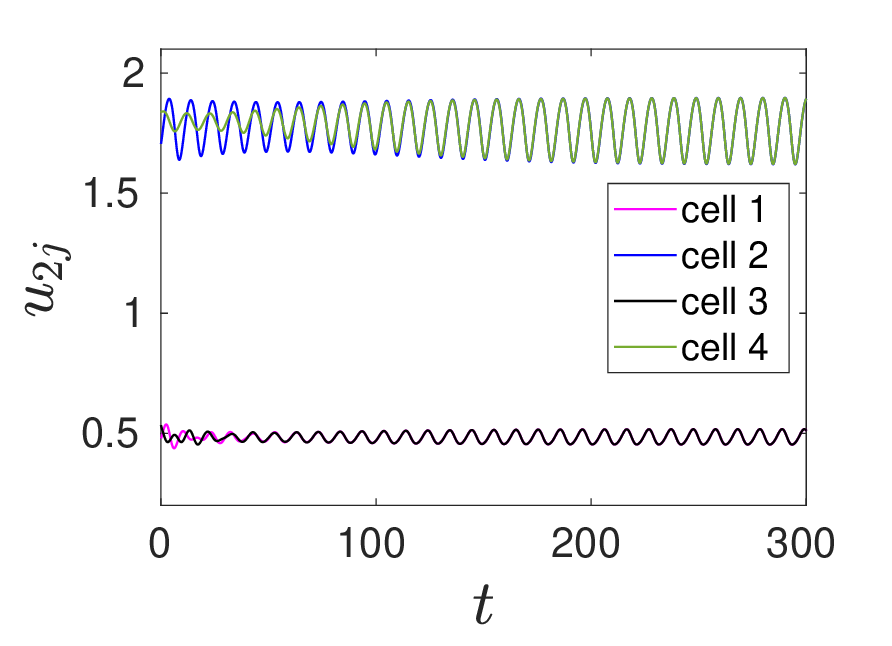}
  \caption{$D=1.65$, $\sigma=1$ (random)}
  \label{fig:b:D1.65_sig1_arb}
      \end{subfigure}
      \caption{{Intracellular dynamics $u_{2j}(t)$ at the four
          star-labeled points in Fig.~\ref{fig:fourcell_scatt_b}, as
          computed using the algorithm in \S \ref{sec:all_march} with
          $\Delta t=0.005$. Uniformly random initial perturbations of
          magnitude $0.1$ were used. Top left: Synchronized
          oscillations occur in cells 1 and 3, while the deactivated
          cells 2 and 4 are quiet. Top right: The
          deactivated cells oscillate out of phase, and mixed-mode
          oscillations occur in the other cells.  Bottom left: Large
          amplitude in-phase oscillations occur in the deactivated
          cells, with very small in-phase oscillations in the other
          cells.  Bottom right: both pairs of cells have only small
          amplitude in-phase oscillations. Parameters as in
          Fig.~\ref{fig:fourcell_hopf_b}.}}
      \label{fig:fourcell_b_dyn}
\end{figure}

\subsection{A ring configuration with a center cell}\label{cell:ring_center}

Next, we consider the centered-hexagonal cell configuration shown in
  Fig.~\ref{fig:ringcenter}. We label the ring cells in a
  counterclockwise orientation by
  $\x_j=r_c\left(\cos(\pi (j-1)/3),\sin(\pi (j-1)/3)\right)^T$ for
  $j\in\lbrace{1,\ldots,6\rbrace}$, while the center-cell at the origin
  is $\x_7=\v{0}$. The ring cells are assumed to be identical with
  $d_{1j}=d_{1r}$, $d_{2j}=d_{2r}$ and $\alpha_j=\alpha_r$ for
  $j\in\lbrace{1,\ldots,6\rbrace}$. The center cell can have different
  permeabilities $d_{17}$ and $d_{27}$ and a distinct kinetic
  parameter $\alpha_7$.

For this cell configuration, the GCEP matrix ${\mathcal M}$ in
(\ref{5:gcep_2}) can be partitioned as
\begin{subequations}\label{mat:ring_center}
\begin{equation}\label{mat:ring_center_a}
    {\mathcal M}(\lambda)=
    \left(
    \begin{array}{ccc|c}
    &  &   & a \\
    & {\mathcal M}_{6}  &  & \vdots \\
    &  &  &   a \\
    \hline
    a & \dots &  a & b
    \end{array}
    \right)\,,
\end{equation}
where ${\mathcal M}_6\in {\mathbb C}^{6,6}$ is
both circulant and symmetric. Here $a$ and $b$ are defined by
\begin{equation}\label{mat:ring_center_b}
  a \equiv \nu  K_0\left( \sqrt{\frac{\sigma+\lambda}{D}} |\v{x}_1-\v{x}_7|
  \right) \,, \quad
  b \equiv 1+\nu \left(
    \log\left(2 \sqrt{\frac{D}{\sigma+\lambda}}\right) - \gamma_e\right)
  +\frac{\nu D}{d_{17}}\left(1+ 2\pi d_{27} \mathit{K}_7\right)\,,
\end{equation}
\end{subequations}
with $\nu={-1/\log\varepsilon}$ and $\mathit{K}_7$ as given in
(\ref{selkov:kmat}). In (\ref{mat:ring_center_a}), ${\mathcal M}_6$ is
a cyclic permutation of its first row, labeled by $(m_1,\ldots,m_6)$,
in which $m_2=m_6$ and $m_3=m_5$ since ${\mathcal M}_6$ is
symmetric. The matrix spectrum of ${\mathcal M}_6$ consists of an
in-phase eigenvector $\v{e}_6\equiv (1,\ldots,1)^T\in\R^6$ and five
anti-phase eigenvectors, two of which are degenerate, which are all
orthogonal to $\v{e}_6$. For ${\mathcal M}$, these anti-phase modes
are lifted to $\R^7$ by requiring that the center cell is silent. The
in-phase mode for ${\mathcal M}_6$ will lead to two eigenmodes in
$\R^7$ for ${\mathcal M}$ where the center cell is active. The matrix
spectrum of ${\mathcal M}$ is characterized as follows:

\begin{lemma}\label{lem:ring_center}
 The matrix eigenvalues for ${\mathcal M}\v{c}=\chi\v{c}$, with ${\mathcal M}$
  as in (\ref{mat:ring_center_a}), are
\begin{subequations}\label{hex:mat_eig}
  \begin{equation}\label{spec:eig:center_hole}
    \begin{split}
    \chi_j&=\sum_{k=1}^{6} \cos\left(\frac{\pi j(k-1)}{3}\right) m_{k} \,,
    \quad j\in \lbrace{1,\ldots,5\rbrace}\,, \\
    \chi_{\pm}&=\omega_6 + a f_{\pm}\,, \qquad
    f_{\pm}\equiv \frac{\left(b-\omega_6\right)}{2a} \pm
    \sqrt{ \frac{\left(b-\omega_6\right)^2}{(2a)^2}
      + 6 } \,,
  \end{split}
\end{equation}
where $\det{{\mathcal M}} =
\left(\chi_{+}\chi_{-}\right)\left(\prod_{j=1}^{5} \chi_j \right)=
\left(b\omega_6 - 6 a^2\right)\left(\prod_{j=1}^{5} \chi_j \right)\,.$
Here $\omega_6$ is defined by
${\mathcal M}_{6}\v{e}_6=\omega_6\v{e}_6$ with
$\v{e}_6\equiv (1,\ldots,1)^T\in\R^6$, while $(m_1,\ldots,m_6)$ is
the first row of ${\mathcal M}_6$. The corresponding eigenvectors are
\begin{equation}\label{spec:vec:center_hole}
  \begin{split}
    \v{c}_{j}&= \left(1\,, \cos \left(\frac{ \pi j}{3}\right)\,, \ldots\,,
      \cos \left( \frac{5\pi j}{3} \right)\,, 0\right)^T \,,\,\,\,
    j\in \lbrace{1,2\rbrace}\,; \quad
 \v{c}_{3}=(1\,,-1\,,1\,,-1\,,1\,,-1\,,0)^T\,,\\
     \v{c}_{6-j}&=\left(0\,, \sin \left( \frac{\pi j}{3}\right)\,,
       \ldots\,, \sin \left( \frac{5\pi j}{3}\right)\,,0\right)^T \,,
     \,\,\, j\in \lbrace{1,2\rbrace}\,; \quad \v{c}_{\pm}
     =(1,\ldots,1,f_{\pm})^T \,,
   \end{split}
 \end{equation}
\end{subequations}
where $\v{c}_3$ is referred to as the sign-alternating anti-phase mode.
 Since $\chi_j=\chi_{6-j}$ for $j\in\lbrace{1,2\rbrace}$, two pairs of
 anti-phase modes for which the center cell is silent are
 degenerate. The roots $\lambda$ to $\det{{\mathcal M}(\lambda)}=0$
 are obtained from the union of the five scalar root-finding problems
 $\chi_1=0$, $\chi_2=0$, $\chi_3=0$ and $\chi_{\pm}=0$. Setting
 $\chi_{\pm}=0$, yields the parameter constraint $6a^2=b\omega_6$ and
 $f_{\pm}={\mp} \sqrt{6\omega_6/b}$.
\end{lemma}

\begin{proof}
  The proof of this result for $\chi_j$ and $\v{c}_j$ for
  $j\in\lbrace{1,\ldots,5\rbrace}$ is immediate since ${\mathcal M}_6$
  is a circulant matrix and $\v{e}^T\v{c}_j=0$ for
  $j\in\lbrace{1,\ldots,5\rbrace}$, where
  $\v{e}=(\v{e}_6,1)^T\in\R^7$. For the remaining two eigenmodes, we
  let $\v{c}=(\v{e}_6,f)^T$ and calculate using
  ${\mathcal M}_6\v{e}_6=\omega_6\v{e}_6$ that
  \begin{equation}\label{spec:two}
    {\mathcal M}\v{c}=
\begin{pmatrix}
   {\mathcal M}_6 \v{e}_6 + a f \v{e}_6 \\
    6a + bf
  \end{pmatrix}=
\begin{pmatrix}
   \left(\omega_6 + a f\right) \v{e}_6 \\
    6a + bf
  \end{pmatrix}= \chi
\begin{pmatrix}
  \v{e}_6\\
  f
\end{pmatrix}\,.
\end{equation}
This yields that $\omega_6 + a f = \chi$ and $6a + bf = \chi f$, which
gives a quadratic equation for $f$. The solution yields
(\ref{spec:eig:center_hole}) and (\ref{spec:vec:center_hole}) for
$\chi_{\pm}$ and $f_{\pm}$. Setting $\chi=0$ in (\ref{spec:two}) we
readily obtain the constraint $6a^2=b\omega_6$.
\end{proof}}

\begin{figure}[htbp]
  \centering
        \begin{subfigure}[b]{0.32\textwidth}  
    \includegraphics[width=\textwidth,height=4.4cm]{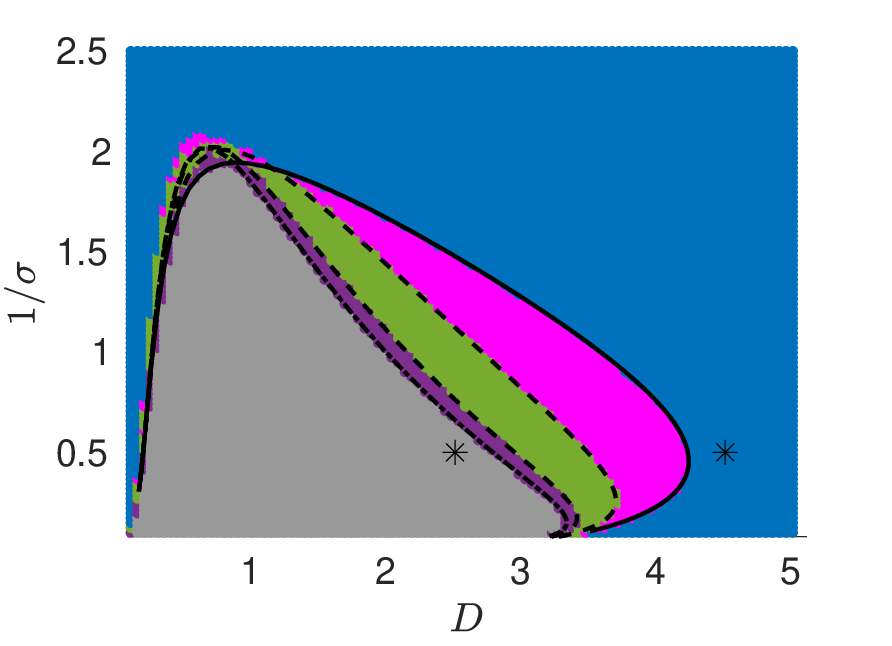}
        \caption{signaling cell on}
        \label{fig:7cell_scatt_oneshell}
      \end{subfigure}
              \begin{subfigure}[b]{0.32\textwidth}  
     \includegraphics[width=\textwidth,height=4.4cm]{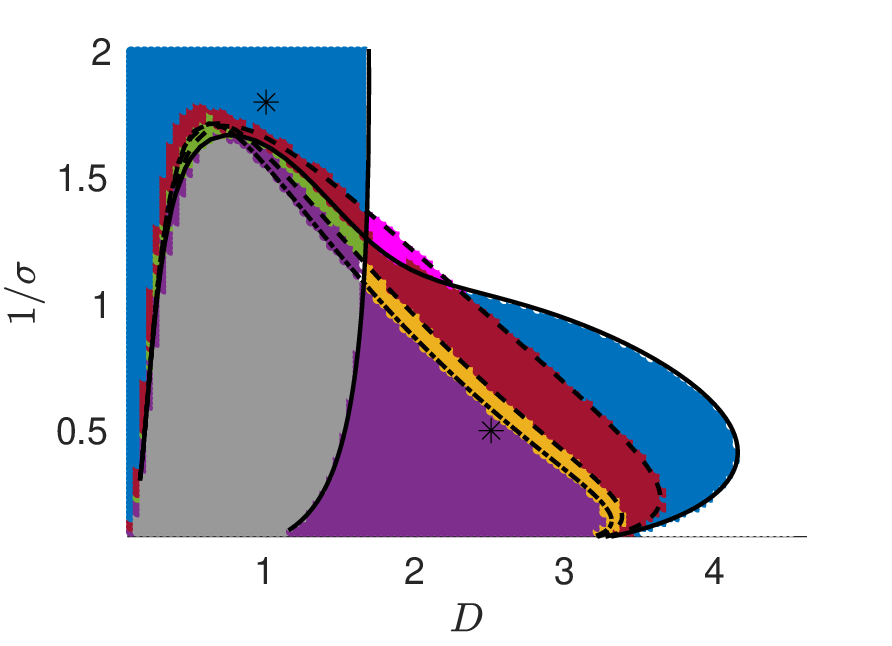}
        \caption{signaling cell deactivated}
        \label{fig:7cell_scatter_c}
      \end{subfigure}
        \begin{subfigure}[b]{0.32\textwidth}  
\includegraphics[width=\textwidth,height=4.4cm]{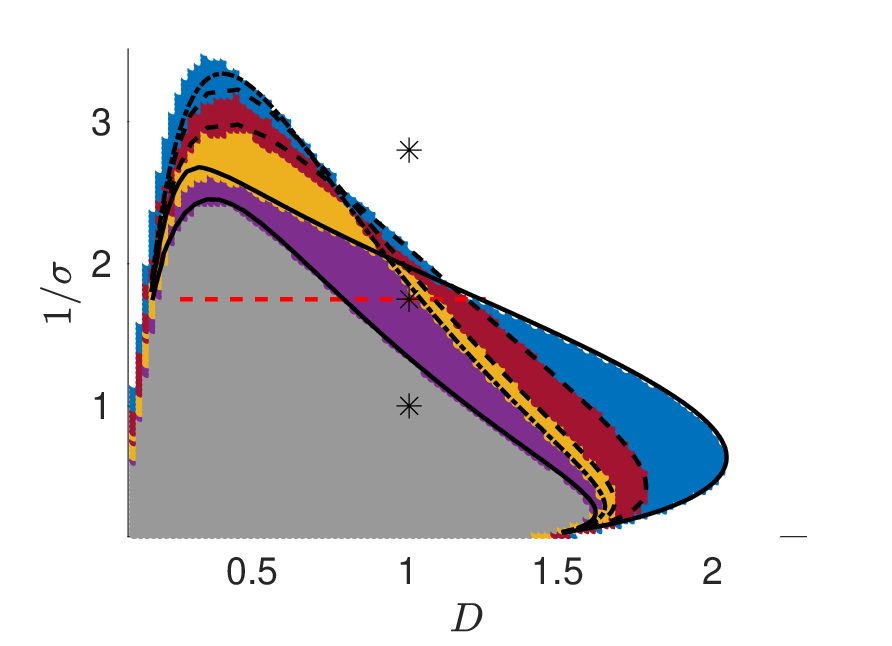}
        \caption{identical cells}
        \label{fig:7spot_scatter_ident}
      \end{subfigure}
      \caption{{Scatter plots for the ring and center cell
          configuration of Fig.~\ref{fig:ringcenter} with $r_c=2$. The
          number of destabilizing eigenvalues ${\mathcal Z}$ for the
          linearization of (\ref{DimLess_bulk}) around the
          steady-state solution is additionally coded as
          ${\mathcal Z}=10$ (orange), ${\mathcal Z}=12$ (purple) and
          ${\mathcal Z}=14$ (gray). The HB boundaries are shown:
          dashed black curves are degenerate anti-phase modes,
          dot-dashed black curve is the sign-alternating anti-phase
          mode, and solid black curves are the HB modes where the
          center cell is active. Left: Identical ring cells have
          $d_{1r}=0.8$, $d_{2r}=0.2$ and $\alpha_r=0.9$. The center
          cell is signaling with $d_{17}=0.4$, $d_{27}=0.2$ and
          $\alpha_7=0.5$. Middle: Same parameters as in the left panel
          except that the efflux from the center cell is increased to
          $d_{27}=0.5$.  The center cell is now deactivated. Right:
          All cells are identical with $d_{1r}=d_{17}=0.4$,
          $d_{2r}=d_{27}=0.2$ and $\alpha_r=\alpha_7=0.9$. Each cell
          would be in a quiescent state without any cell-cell
          coupling. The Kuramoto order parameter is computed on the
          red-dashed line.}}
\label{fig:7cell_scatter}
\end{figure}

Fixing $d_{1r}=0.8$, $d_{2r}=0.9$ and $\alpha_r=0.9$ for each of the
six ring cells, and with influx permeability $d_{17}=0.4$ and kinetic
parameter $\alpha_7=0.5$ for the center cell, in
Fig.~\ref{fig:7cell_scatt_oneshell} and Fig.~\ref{fig:7cell_scatter_c}
we show the scatter plots when the efflux permeability for the center
cell is either $d_{27}=0.2$ or $d_{27}=0.5$, respectively. When
$d_{27}=0.2$ the center cell is a signaling cell (see
Fig.~\ref{fig:selkov_efflux}) and would have limit cycle oscillations
when uncoupled from the bulk. However, when $d_{27}=0.5$ this
signaling cell is deactivated. For comparison, in
Fig.~\ref{fig:7spot_scatter_ident} we show the scatter plot for seven
identical cells with parameter values $d_{1j}=0.4$, $d_{2j}=0.2$ and
$\alpha_j=0.9$ for $j\in\lbrace{1,\ldots,7\rbrace}$. These cells would
all be in a quiescent state without any cell-bulk coupling. In
generating the scatter plots in Fig.~\ref{fig:7cell_scatter} we used
the analytical expression in Lemma \ref{lem:ring_center} for
$\det{\mathcal M}$ to calculate ${\mathcal Z}$ from (\ref{wind:form}).
In Fig.~\ref{fig:7cell_scatter} the HB boundaries for the various
modes are shown. Setting $\lambda=i\lambda_I$, the degenerate
anti-phase HB modes occur on the black dashed curves, for which
$\chi_1=0$ or $\chi_3=0$, the black dot-dashed curve corresponds to
the sign-alternating anti-phase HB mode, for which $\chi_3=0$, while
the solid black curves are the HB boundaries where the center cell is
active, as obtained from setting $\chi_{\pm}=0$.  The slight
raggedness in Fig.~\ref{fig:7cell_scatter} shows that our winding
number algorithm has some challenges in correctly calculating
${\mathcal Z}$ very close to the HB boundaries where the eigenvalues
associated with the degenerate anti-phase modes simultaneously cross
the imaginary axis of the spectral plane.

\begin{figure}[h!tbp]
  \centering
  \begin{subfigure}[b]{0.32\textwidth}  
  \includegraphics[width=\textwidth,height=4.2cm]{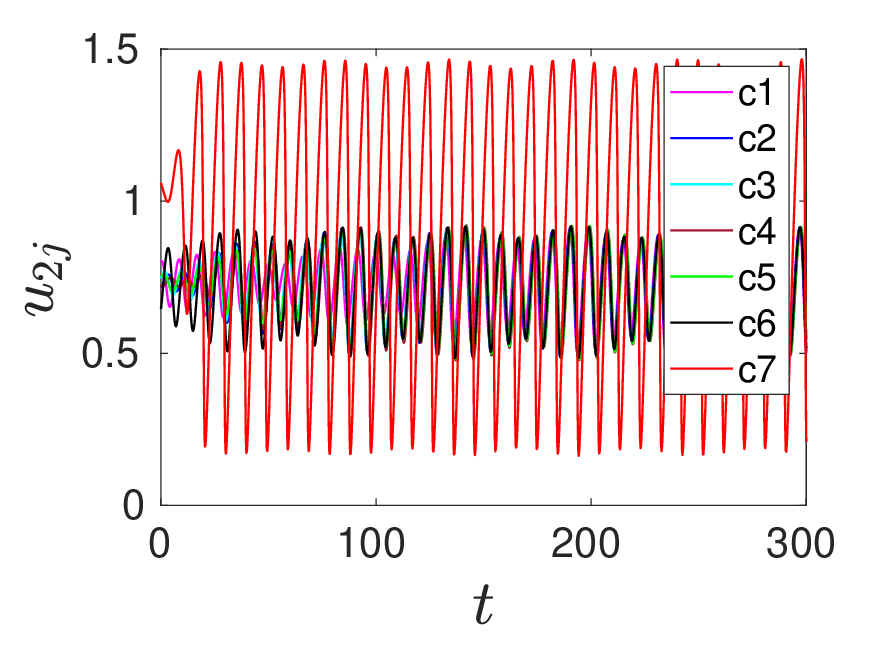}
  \caption{$D=2.5$, $\sigma=2$ (short)} 
        \label{fig:7_cell_on_shell_a}
      \end{subfigure}
  \begin{subfigure}[b]{0.32\textwidth}  
     \includegraphics[width=\textwidth,height=4.2cm]{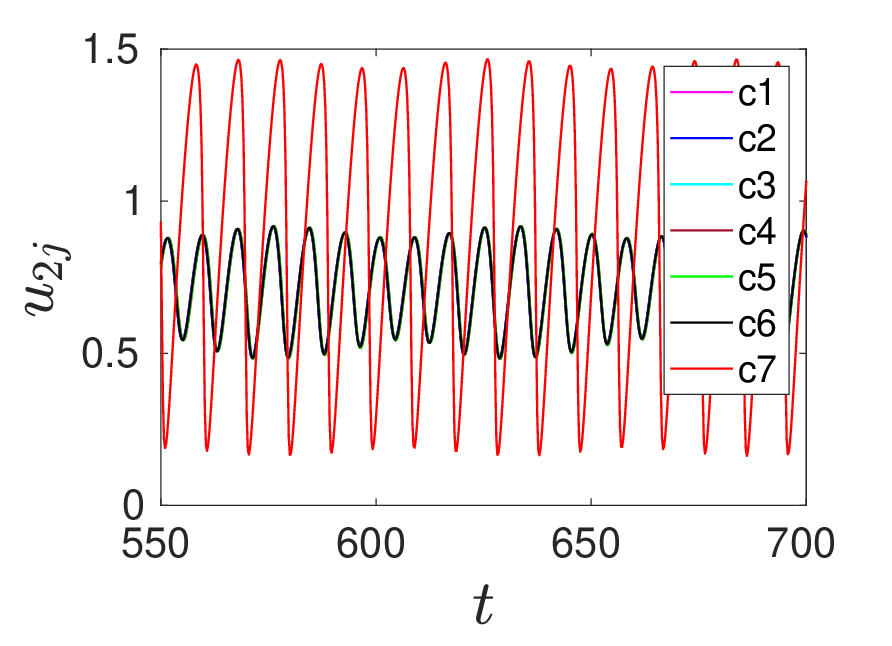}
  \caption{$D=2.5$, $\sigma=2$ (long)} 
        \label{fig:7_cell_on_shell_b}
      \end{subfigure}
  \begin{subfigure}[b]{0.32\textwidth}  
     \includegraphics[width=\textwidth,height=4.2cm]{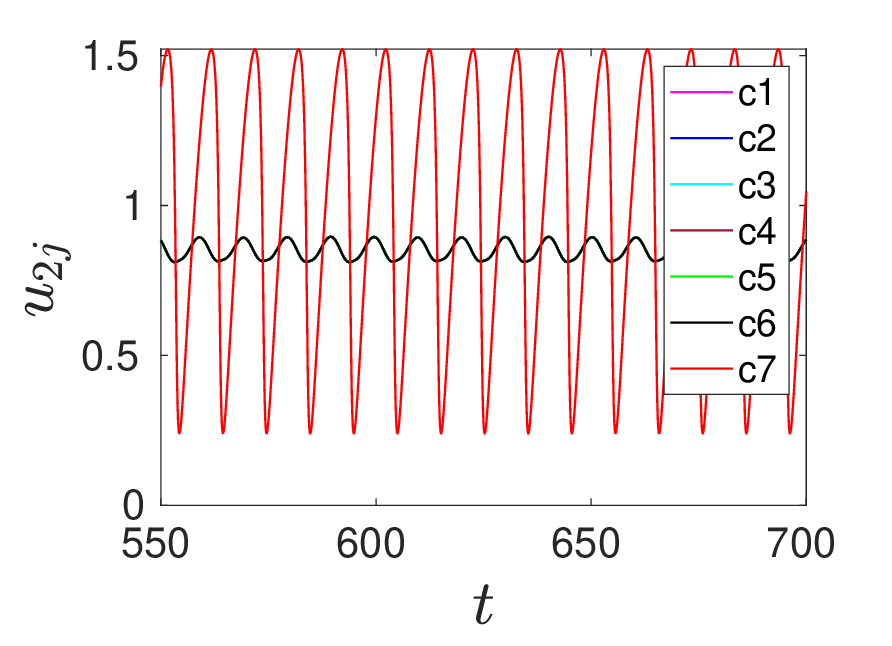}
  \caption{$D=4.5$, $\sigma=2$.}
  \label{fig:7_cell_on_shell_c}
\end{subfigure}
  \begin{subfigure}[b]{0.32\textwidth}  
    \includegraphics[width=\textwidth,height=4.2cm]{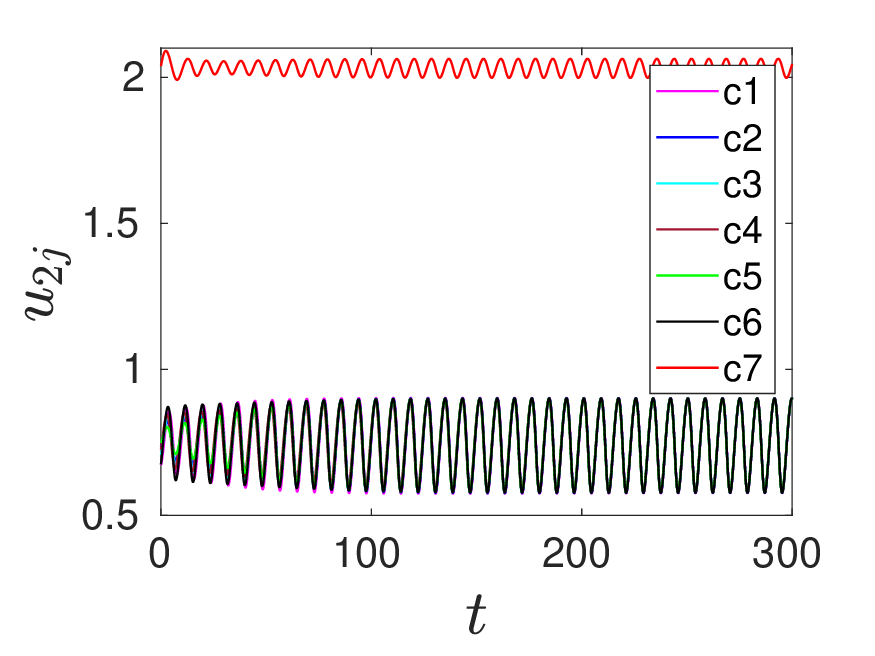}
  \caption{$D=2.5$, $\sigma=2$ (short)} 
        \label{fig:7_cell_off_a}
      \end{subfigure}
  \begin{subfigure}[b]{0.32\textwidth}  
     \includegraphics[width=\textwidth,height=4.2cm]{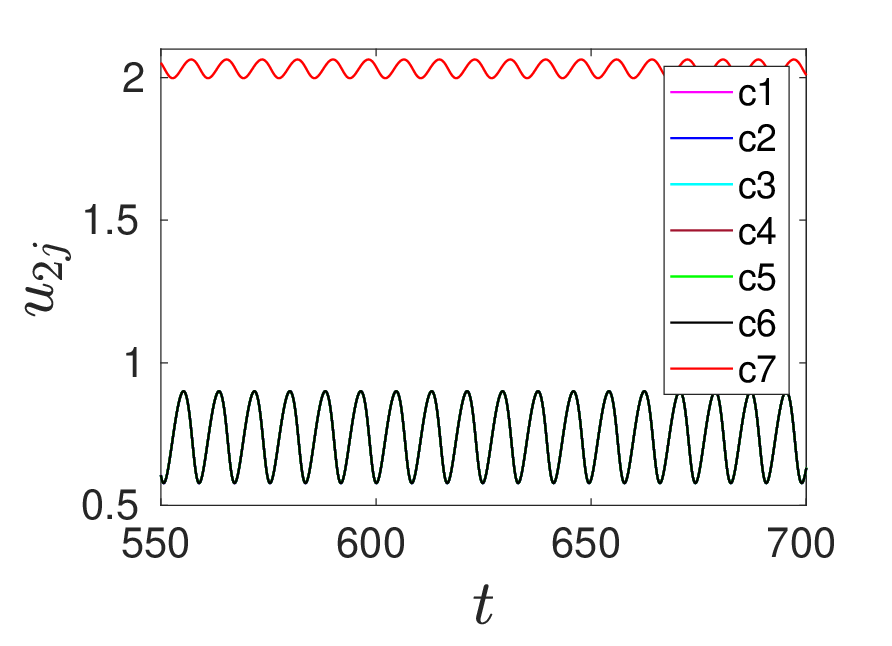}
  \caption{$D=2.5$, $\sigma=2$ (long)} 
        \label{fig:7_cell_off_b}
      \end{subfigure}
  \begin{subfigure}[b]{0.32\textwidth}  
    \includegraphics[width=\textwidth,height=4.2cm]{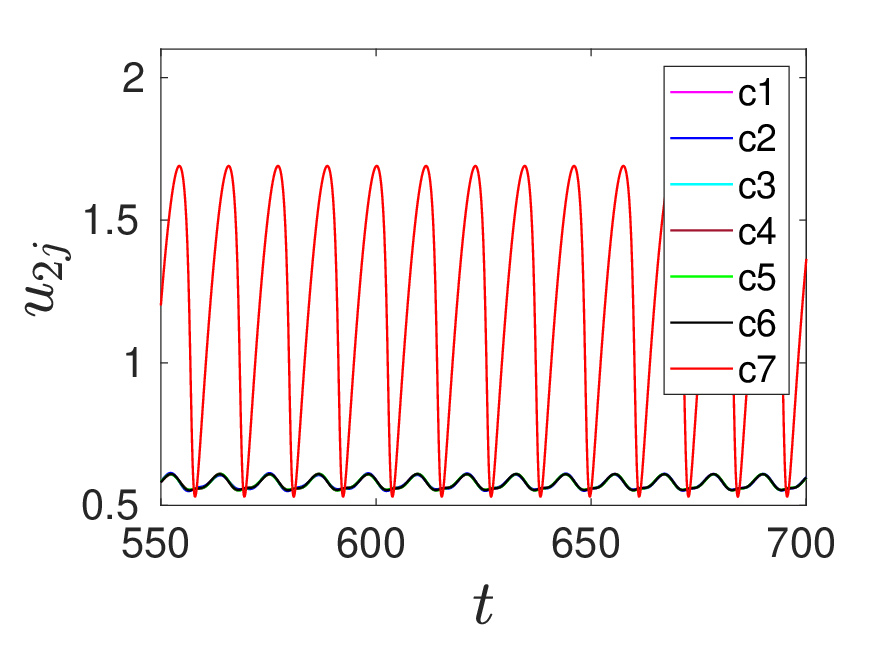}
  \caption{$D=1$, $\sigma={5/9}$}
  \label{fig:7_cell_off_c}
      \end{subfigure}
  \begin{subfigure}[b]{0.32\textwidth}  
     \includegraphics[width=\textwidth,height=4.2cm]{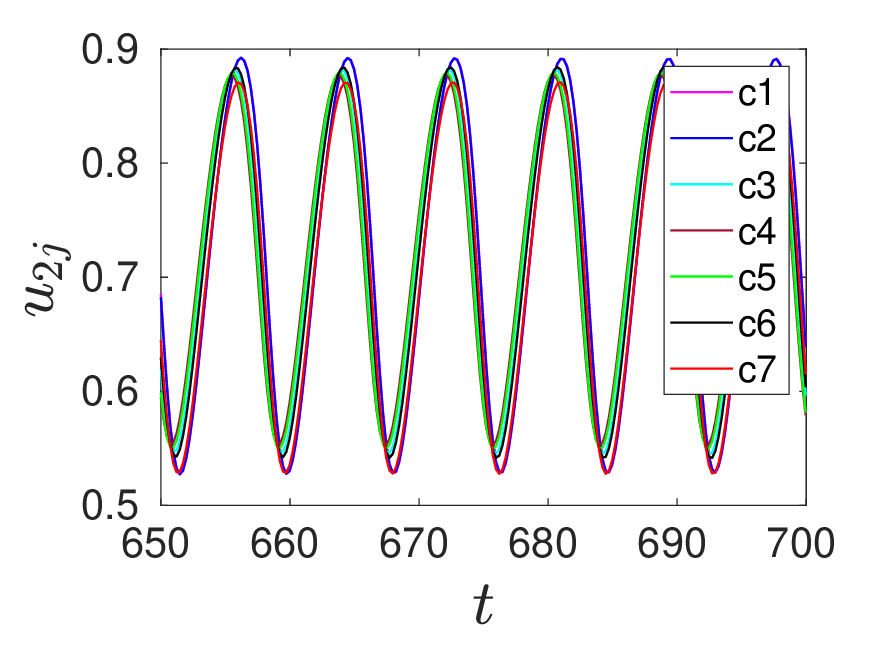}
  \caption{$D=1$, $\sigma=1$} 
        \label{fig:7_cell_ident_a}
      \end{subfigure}
  \begin{subfigure}[b]{0.32\textwidth}  
     \includegraphics[width=\textwidth,height=4.2cm]{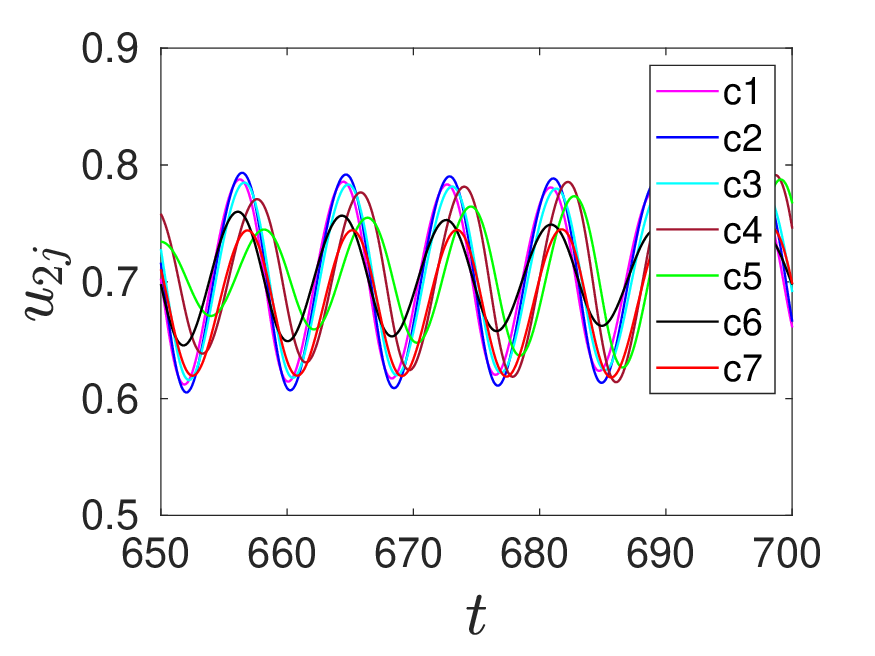}
  \caption{$D=1$, $\sigma={4/7}$} 
        \label{fig:7_cell_ident_b}
      \end{subfigure}
  \begin{subfigure}[b]{0.32\textwidth}  
     \includegraphics[width=\textwidth,height=4.2cm]{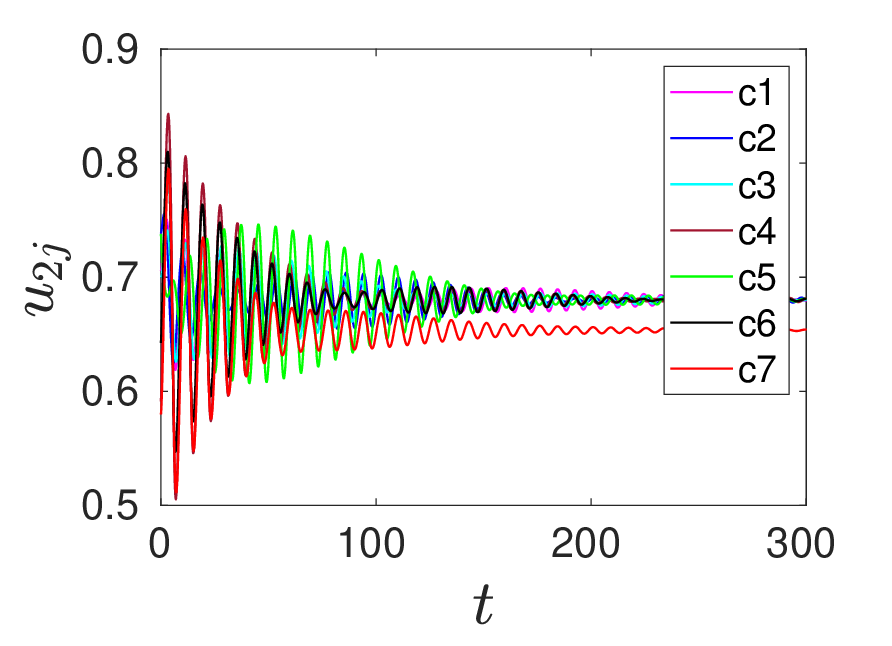}
  \caption{$D=1$, $\sigma={5/14}$}
  \label{fig:7_cell_ident_c}
      \end{subfigure}
      \caption{{Intracellular dynamics $u_{2j}(t)$ for the ring and
          center cell configuration of Fig.~\ref{fig:ringcenter} with
          $r_c=2$ computed using the algorithm of \S
          \ref{sec:all_march} with $\Delta t=0.005$ at the
          star-labeled points in Fig.~\ref{fig:7cell_scatter} using
          uniformly random initial conditions of magnitude $0.1$. The
          center cell is labeled as C7. Top row: Short- and long-time
          dynamics for $D=2.5$ and $\sigma=1$ in
          Fig.~\ref{fig:7cell_scatt_oneshell} is shown. A synchronous
          wave-packet solution emerges for the ring cells as
          $t$ increases. For $D=4.5$, the ring cells synchronize to a
          small amplitude regular oscillation. Middle row: For
          $D=2.5$ and $\sigma=2$ (left and middle) in
          Fig.~\ref{fig:7cell_scatter_c} the multiple destabilizing
          anti-phase modes of the ring cells eventually synchronize to
          a periodic oscillation that is larger than that of the
          center cell. For $D=1$ and $\sigma={5/9}$, the center cell
          now has much larger amplitude oscillations than do the ring
          cells. Bottom row: For the identical cells in
          Fig.~\ref{fig:7spot_scatter_ident} all cells exhibit roughly
          in-phase oscillations with a similar amplitude when
          $D=\sigma=1$. For $D=1$, but with a smaller bulk degradation
          $\sigma={4/7}$, the ring cells show large amplitude
          incoherent oscillations (middle). The center cell has a
          small amplitude regular oscillation. For the smaller
          value $\sigma={5/14}$, the mixed-mode oscillations in the
          ring cells and the regular oscillation of the center cell
          tend to their linearly stable steady-state values.}}
\label{fig:7cell_dynamics}
\end{figure}

In Fig.~\ref{fig:7cell_dynamics} we plot $u_{2j}$ versus $t$ at the
star-labeled points in the scatter plots of
Fig.~\ref{fig:7cell_scatter}, as computed using the algorithm of \S
\ref{sec:all_march} with $\Delta t=0.005$ and with uniformly random
initial conditions of magnitude $0.1$ about the steady-state. For
$D=2.5$ and $\sigma=2$ in the scatter plot of
Fig.~\ref{fig:7cell_scatt_oneshell} where all the modes are destabilizing,
in Figs.~\ref{fig:7_cell_on_shell_a}--\ref{fig:7_cell_on_shell_b} we
show that the ring cells synchronize on a wave packet solution and
that large oscillations occur in the center signaling cell. By
increasing $D$ to $D=4.5$, Fig.~\ref{fig:7_cell_on_shell_c} shows that
the ring cells now exhibit regular in-phase oscillations but with
rather small amplitude.  In contrast, when the center signaling cell
has been deactivated, for the values $D=2.5$ and $\sigma=2$ in the
scatter plot of Fig.~\ref{fig:7cell_scatter_c} the results in
Figs.~\ref{fig:7_cell_off_a}--\ref{fig:7_cell_off_b} show that the
ring cells eventually synchronize to a rather large in-phase
oscillation in comparison to that of the deactivated center
cell. Observe from the HB boundaries in Fig.~\ref{fig:7cell_scatter_c}
that when $D=2.5$ and $\sigma=2$ only one of the two modes where the
center cell is active is destabilizing. This is different than in
Fig.~\ref{fig:7cell_scatt_oneshell} where, for $D=2.5$ and $\sigma=2$,
both modes where the center cell are active are destabilizing. For
the pair $D=1$ and $\sigma={5/9}$ in Fig.~\ref{fig:7cell_scatter_c}
all anti-phase modes where the center cell is silent are linearly
stable, while exactly one mode for which the center cell is not
quiescent and the ring cells are in-phase is destabilizing.  As expected,
Fig.~\ref{fig:7_cell_off_c} shows that the center cell has much larger
amplitude oscillations than do the ring cells, which oscillate
in-phase.  Finally, for identical cells that are all quiescent without
cell-bulk coupling, the last row of Fig.~\ref{fig:7cell_dynamics}
shows the intracellular dynamics at the three star-labeled points in
Fig.~\ref{fig:7spot_scatter_ident}. For $D=1$ and $\sigma=1$,
Fig.~\ref{fig:7_cell_ident_a} shows that all cells eventually
oscillate in-phase and with a comparable amplitude. For $\sigma={4/7}$
(purple-shaded region in the scatter plot), where there are twelve
destabilizing eigenvalues for the linearization of the steady-state,
Fig.~\ref{fig:7_cell_ident_b} shows that the ring cells have
incoherent oscillations while the center cell has a small amplitude
regular oscillation. When the degradation is decreased further to
$\sigma={5/14}$, Fig.~\ref{fig:7_cell_ident_c} shows that the
amplitudes of the center cell and of the mixed-mode oscillations for
the ring cells decay to their steady-state limits, as predicted by
the scatter plot in Fig.~\ref{fig:7spot_scatter_ident}.

\subsection{Measuring phase synchronization}\label{cell:hexagonal}

For several cell configurations we now quantify the phase
synchronization or coherence of intracellular oscillations based on
the Kuramoto order parameter
\begin{equation}\label{more:kuram}
  Q(t) \equiv \frac{1}{N} \sum_{j=1}^{N} e^{i\theta_j} \,, \quad \mbox{where}
  \quad
  \theta_j \equiv \mbox{arctan}\left(\frac{ u_{2j}(t)-u_{2js}}{u_{1j}(t)-u_{1js}}
  \right) \in (0,2\pi)\,.
\end{equation}
Here $u_{2js}$ and $u_{1js}$ are the steady-state values obtained from
(\ref{selkov:equil}). To safely disregard the effect of transients, we
compute a time-averaged order parameter $Q_{ave}$, which we define by
\begin{equation}\label{more:kuram_ave}
  Q_{ave}= \frac{1}{t_{up}-t_{low}} \int_{t_{low}}^{t_{up}} Q(t) \, dt \,,
\end{equation}
and where we chose $t_{low}=1300$ and $t_{up}=1500$ for the
simulations described below.  To use our fast marching algorithm in \S
\ref{sec:all_march} with $\Delta t=0.005$ for the longer time interval
$t\leq 1500$, we set $n=114$, which gives $229$ quadrature points for
discretizing the Laplace space contour (see Tables \ref{tab:G2eps} and
\ref{tab:E1eps}).

For the ring with center cell configuration of identical cells, in
Fig.~\ref{fig:7cell_kuram} we plot $Q_{ave}$ versus $D$ along the
parameter path $0.25\leq D\leq 1.25$ with $\sigma={4/7}$ in the
scatter plot in Fig.~\ref{fig:7spot_scatter_ident} (red-dashed path)
for two choices of initial condition. For either choice of initial
condition, we observe an apparent transition to complete phase
coherence when $D$ is near unity. This transition to phase coherence
is further illustrated for $D=0.75$, $D=1$ and $D=1.25$ in
Fig.~\ref{fig:7cell_kuram_dynamics}, where we plot the trajectories
$u_{j2}(t)$ versus $t$ on both short- and long-time intervals, as
computed from our fast algorithm of \S \ref{sec:all_march}. Overall,
this example shows that that a ring and center cell configuration,
where each cell would be in a quiescent state without any cell-bulk
coupling, can exhibit fully phase coherent intracellular oscillations
due to the intercellular communication, as mediated by the bulk diffusion
field, when the bulk diffusivity exceeds a threshold.

\begin{figure}[!htbp]
  \centering
    \centering
  \begin{subfigure}[b]{0.48\textwidth}  
  \includegraphics[width=\textwidth,height=4.3cm]{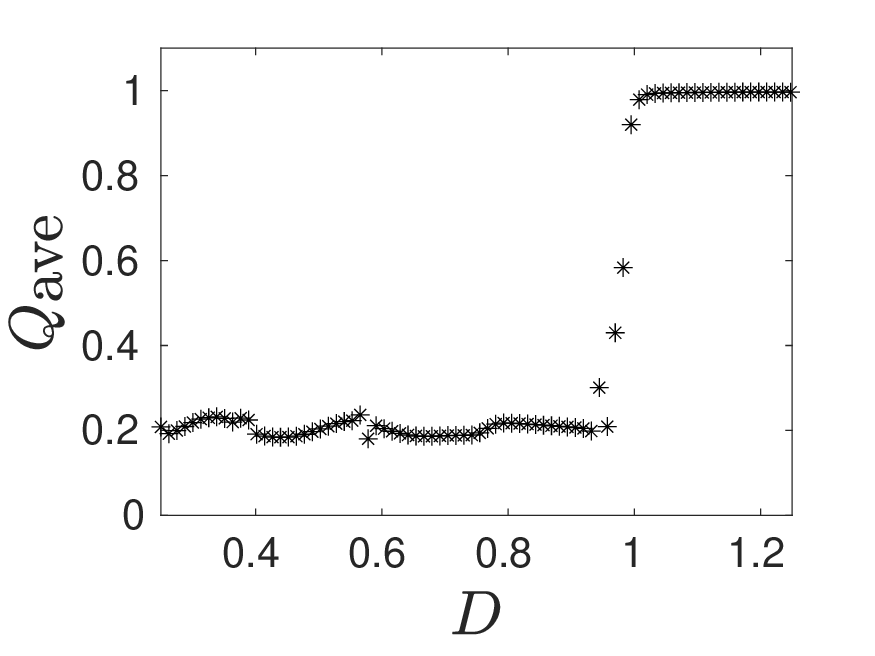}
  \caption{random IC} 
        \label{fig:7cell_kuram_rand}
      \end{subfigure}
  \begin{subfigure}[b]{0.48\textwidth}  
     \includegraphics[width=\textwidth,height=4.3cm]{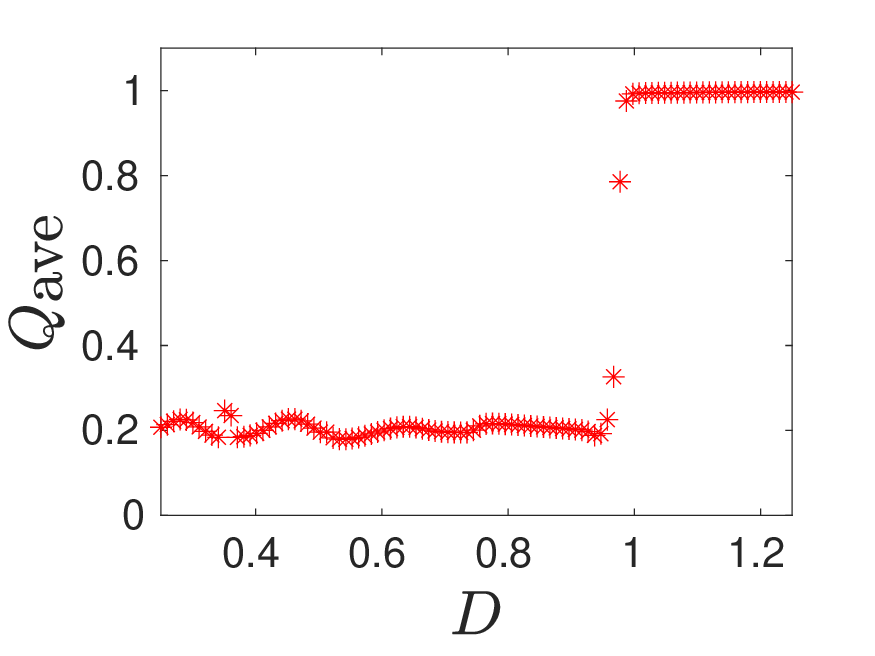}
   \caption{sign-alternating IC} 
        \label{fig:7cell_kuram_alt}
      \end{subfigure}
      \caption{{Averaged order parameter (\ref{more:kuram_ave}) for
          the ring and center cell configuration of identical cells in
          Fig.~\ref{fig:ringcenter} along the red-dashed path in the
          scatter plot of Fig.~\ref{fig:7spot_scatter_ident} with
          $\sigma={4/7}$ and $0.25<D<1.25$. Black (left panel) and red
          (right panel) marked points are for either a uniformly
          random initial perturbation or a sign-alternating initial
          perturbation of the steady-state of magnitude $0.1$. A
          transition to phase coherence occurs for $D$ near unity for
          either choice of initial condition.  Parameters: $r_c=2$,
          $d_{1j}=0.4$, $d_{2j}=0.2$ and $\alpha_j=0.9$ for
          $j\in \lbrace{1,\ldots,7\rbrace}$.}}
\label{fig:7cell_kuram}
\end{figure}

Next, we consider the two-ring configuration of cells shown in
Fig.~\ref{fig:eightcell}. For the rightmost ring, the cells are
identical with $d_{1j}=1.5$, $d_{2j}=0.2$ and $\alpha_j=0.9$ for
$j\in \lbrace{1,\ldots,4\rbrace}$, and the cell centers are at
$\x_1=(4.5,0)$, $\x_2=(2.5,2)$, $\x_3=(0.5,0)$ and
$\x_4=(2.5,-2)$. These cells are in a quiescent state when uncoupled
from the bulk. For the leftmost ring, with cell centers at
$\x_5=(-0.5,0.0)$, $\x_6=(-1.5,1)$, $\x_7=(-2.5,0)$ and
$\x_8=(-1.5,-1)$, the cells have identical permeabilities $d_{1j}=0.4$
and $d_{2j}=0.2$ for $j\in\lbrace{5,\ldots,8\rbrace}$, but are chosen
to have distinct kinetic parameters $\alpha_5=0.6$, $\alpha_6=0.45$,
$\alpha_7=0.4$ and $\alpha_8=0.5$.  From Fig.~\ref{fig:selkov_efflux},
we observe that, when uncoupled from the bulk, these leftmost ring
cells are activated and would have different frequencies of
intracellular oscillations.

\begin{figure}[h!tbp]
  \centering
  \begin{subfigure}[b]{0.48\textwidth}  
  \includegraphics[width=\textwidth,height=4.2cm]{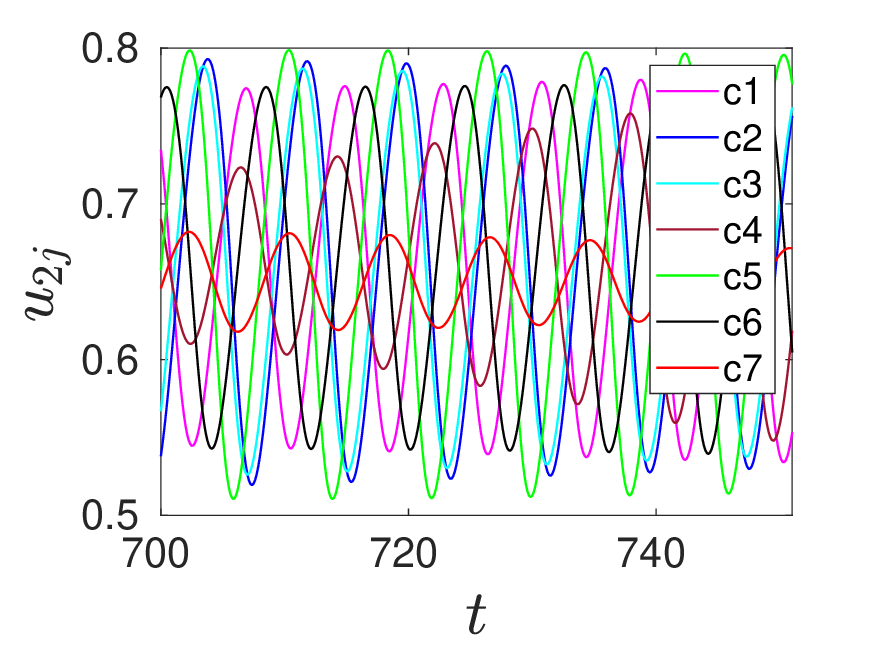}
  \caption{$D=0.75$ (short)} 
        \label{fig:7_kuram_D075_short}
      \end{subfigure}
  \begin{subfigure}[b]{0.48\textwidth}  
  \includegraphics[width=\textwidth,height=4.2cm]{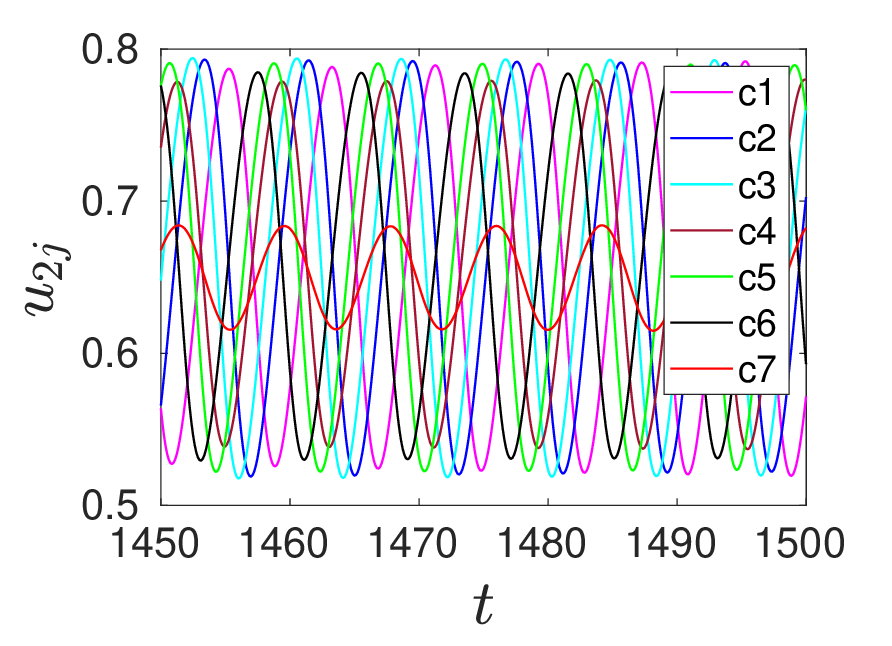}
  \caption{$D=0.75$ (long)} 
        \label{fig:7_kuram_D075_long}
      \end{subfigure}
  \begin{subfigure}[b]{0.48\textwidth}  
  \includegraphics[width=\textwidth,height=4.2cm]{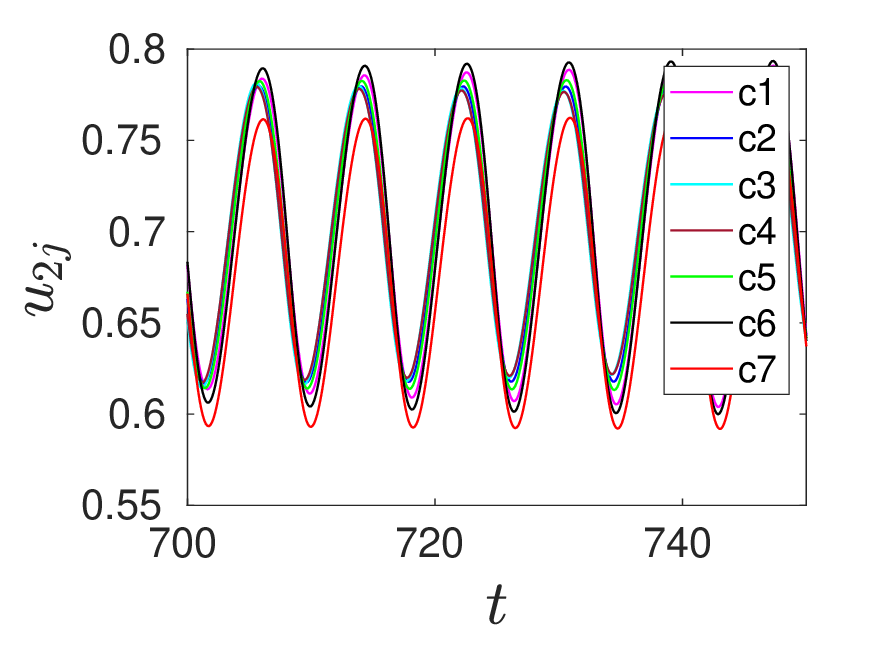}
  \caption{$D=1$ (short)}
          \label{fig:7_kuram_D1_short}
\end{subfigure}
  \begin{subfigure}[b]{0.48\textwidth}  
    \includegraphics[width=\textwidth,height=4.2cm]{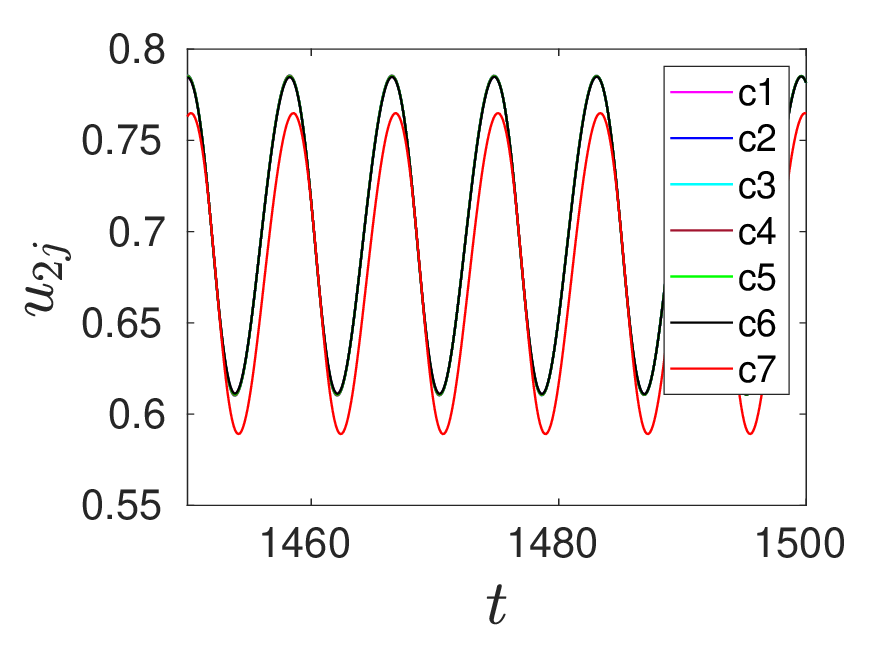}
  \caption{$D=1$ (long)} 
        \label{fig:7_kuram_D1_long}
      \end{subfigure}
  \begin{subfigure}[b]{0.48\textwidth}  
  \includegraphics[width=\textwidth,height=4.2cm]{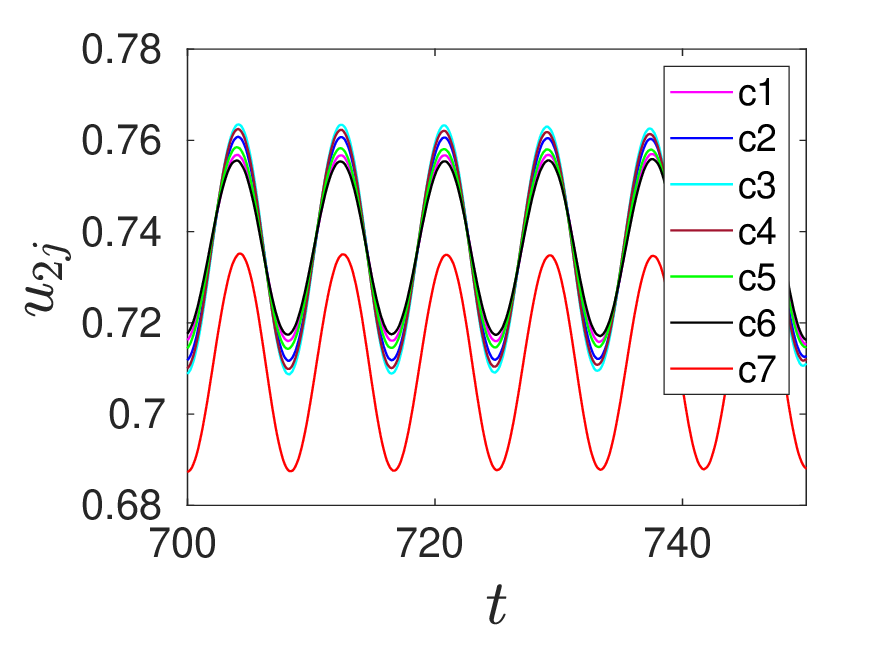}
  \caption{$D=1.25$ (short)} 
        \label{fig:7_kuram_D125_short}
      \end{subfigure}
  \begin{subfigure}[b]{0.48\textwidth}  
    \includegraphics[width=\textwidth,height=4.2cm]{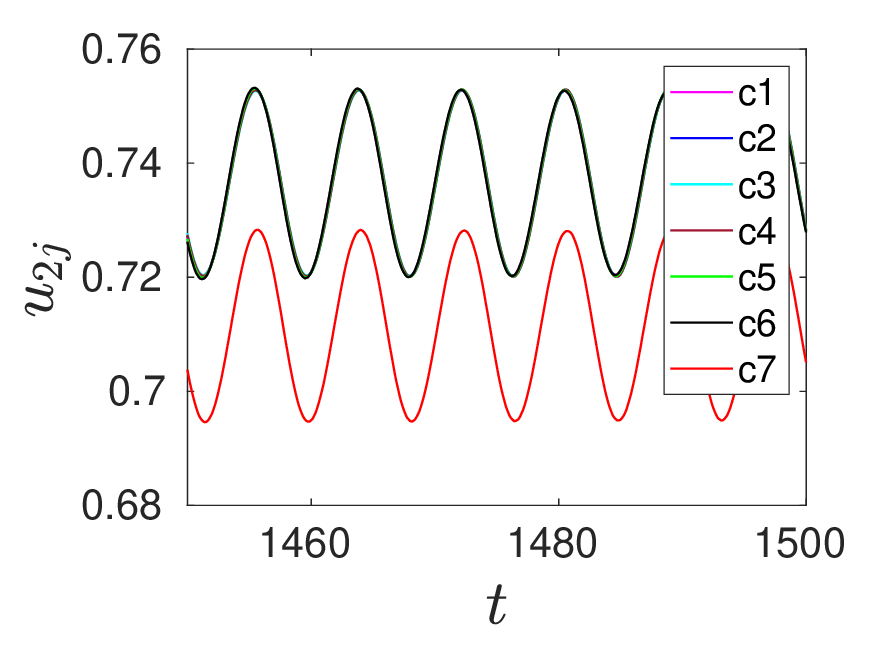}
  \caption{$D=1.25$ (long)}
  \label{fig:7_kuram_D125_long}
      \end{subfigure}
      \caption{Intracellular dynamics $u_{2j}(t)$ for the ring and
          center cell configuration of Fig.~\ref{fig:ringcenter} 
          computed using the algorithm of \S \ref{sec:all_march} with
          $\Delta t=0.005$ at three values of $D$ along the red-dashed
          path in Fig.~\ref{fig:7spot_scatter_ident}. Uniformly
          random initial conditions of magnitude $0.1$ near the
          steady-state were used. The center cell is C7. Short- and
          long-time dynamics are shown. There is a transition to phase
          coherence near $D=1$. Parameters as in Fig.~\ref{fig:7cell_kuram}.}
\label{fig:7cell_kuram_dynamics}
\end{figure}

\begin{figure}[h!tbp]
  \centering
    \begin{subfigure}[b]{0.48\textwidth}
      \includegraphics[width=\textwidth,height=4.4cm]{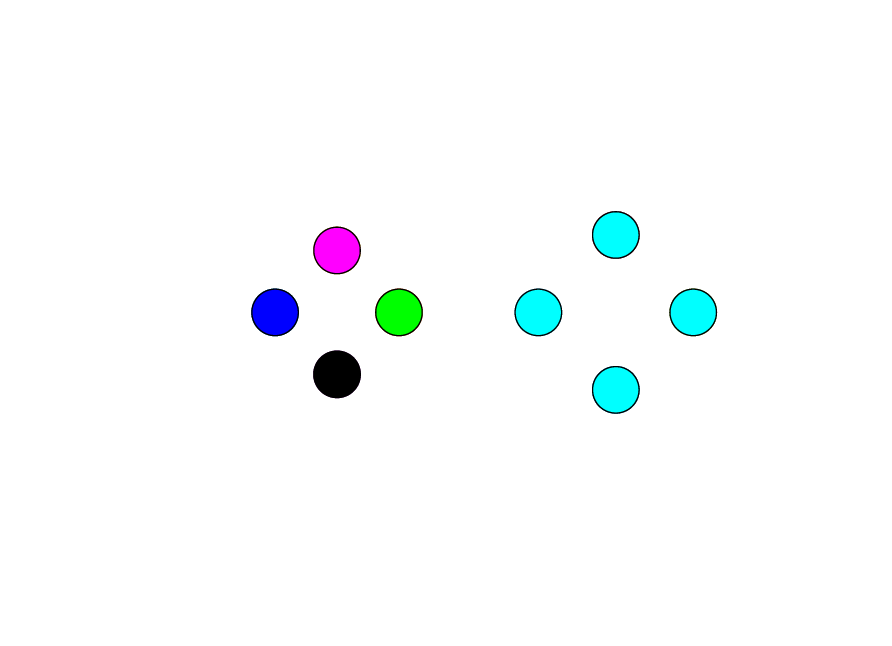}
        \caption{Two-ring configuration}
        \label{fig:eightcell}
    \end{subfigure}
    \begin{subfigure}[b]{0.48\textwidth}
      \includegraphics[width=\textwidth,height=4.4cm]{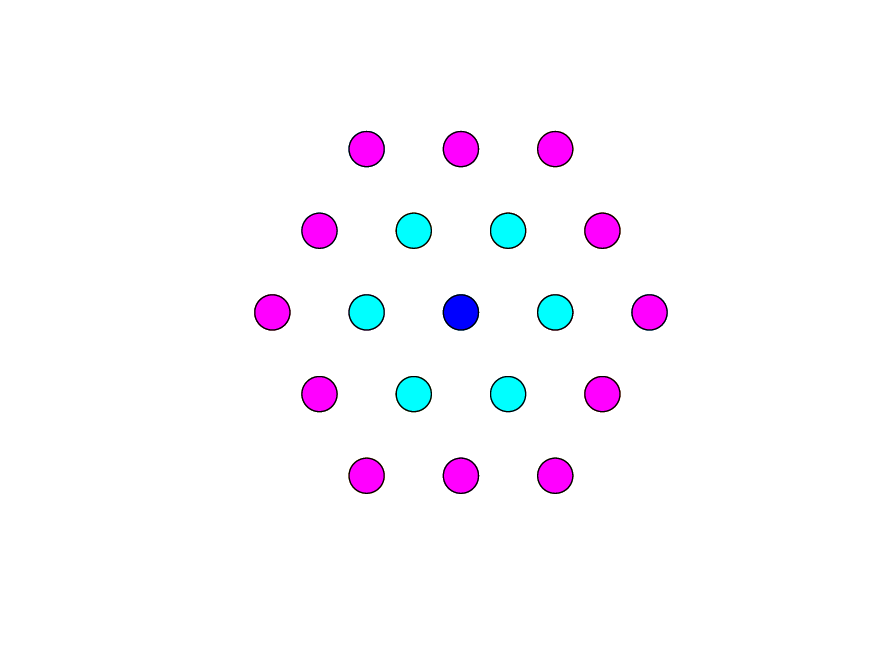}
        \caption{Centered hexagon with two shells}
        \label{fig:centerhex}
      \end{subfigure}
      \caption{Left: Two non-concentric rings of cells, where the
        two rings have different radii. The rightmost ring has identical
        cells, while the other ring has cell-dependent
        parameters. Right: A centered hexagonal configuration of cells
        with two shells. Each shell has identical cells, while the
        center cell is the signaling or pacemaker cell.}
\label{fig:schem:hex}
\end{figure}

By using our fast algorithm of \S \ref{sec:all_march}, numerical
results for $Q_{ave}$, defined in (\ref{more:kuram_ave}), are given in
Table \ref{eightcell:kuramoto} for various pairs of $(D,\sigma)$.  For
$\sigma=1$, the intracellular dynamics are shown for three values of $D$ in
Fig.~\ref{fig:8cell_kuram_dynamics}. We observe that an increase in
the bulk diffusivity $D$ has the dual effect of enhancing the phase
coherence of the activated cells in the leftmost ring as well as for
triggering large intracellular oscillations in the otherwise
quiescent cells in the rightmost ring. Moreover, as $\sigma$
decreases, there is a stronger intercellular communication via the
bulk medium owing to the weaker spatial decay for the bulk species. As
a result, a more pronounced phase coherence is expected when $\sigma$
decreases. This trend is confirmed in Table \ref{eightcell:kuramoto}.

\begin{table}[httbp]
\centering
  \begin{tabular}{ccc}
    \hline
  \multirow{2}{*}{$D$} & \multicolumn{2}{c}{$Q_{ave}$} \\
    & $\sigma=1$ & $\sigma=0.5$ \\
    \hline
   $0.5$ & $0.329$ & $0.427$  \\
   $1.0$ & $0.419$ & $0.540$  \\
  $2.0$ & $0.496$ & $0.624$  \\
  $5.0$ & $0.596$ & $0.826$  \\
  $10.0$ & $0.763$ & $0.875$  \\
  \hline
  \end{tabular}
  \caption{{Averaged Kuramoto order parameter $Q_{ave}$ from
      (\ref{more:kuram_ave}) for the two-ring configuration of
      Fig.~\ref{fig:eightcell} for various bulk diffusivities $D$ when
      either $\sigma=1$ or $\sigma=0.5$. The IC's were the steady-state
      values.  Phase coherence increases as $D$ increases for fixed
      $\sigma$. Moreover, when the bulk degradation rate $\sigma$ is
      smaller, the bulk signal has a weaker spatial decay and so is
      more able to enhance phase synchronization.}}
  \label{eightcell:kuramoto}
  \end{table}

\begin{figure}[h!tbp]
  \centering
  \begin{subfigure}[b]{0.48\textwidth}  
  \includegraphics[width=\textwidth,height=4.3cm]{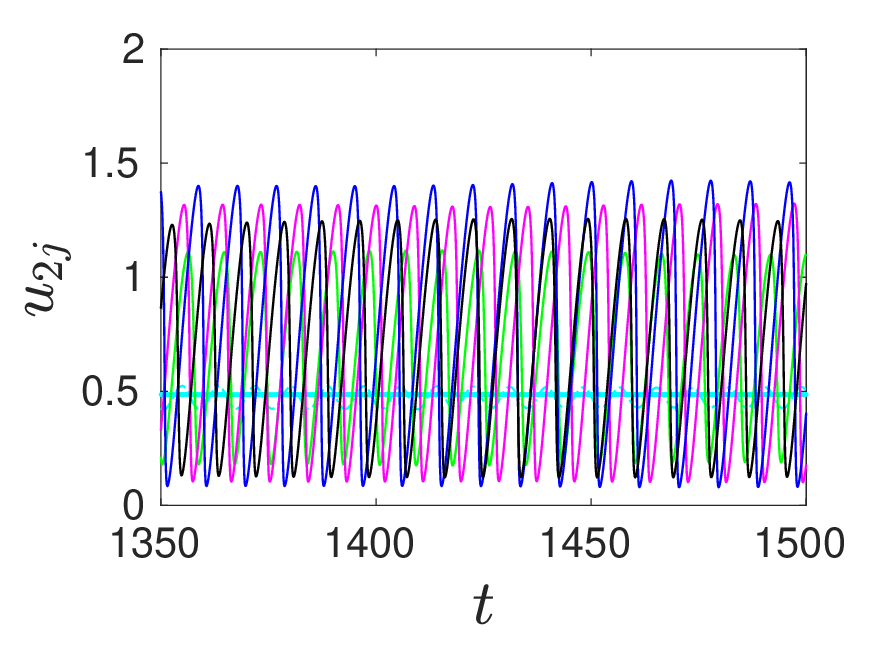}
  \caption{$D=0.5$ (all cells)} 
        \label{fig:8_kuram_D05sig1}
      \end{subfigure}
  \begin{subfigure}[b]{0.48\textwidth}  
  \includegraphics[width=\textwidth,height=4.3cm]{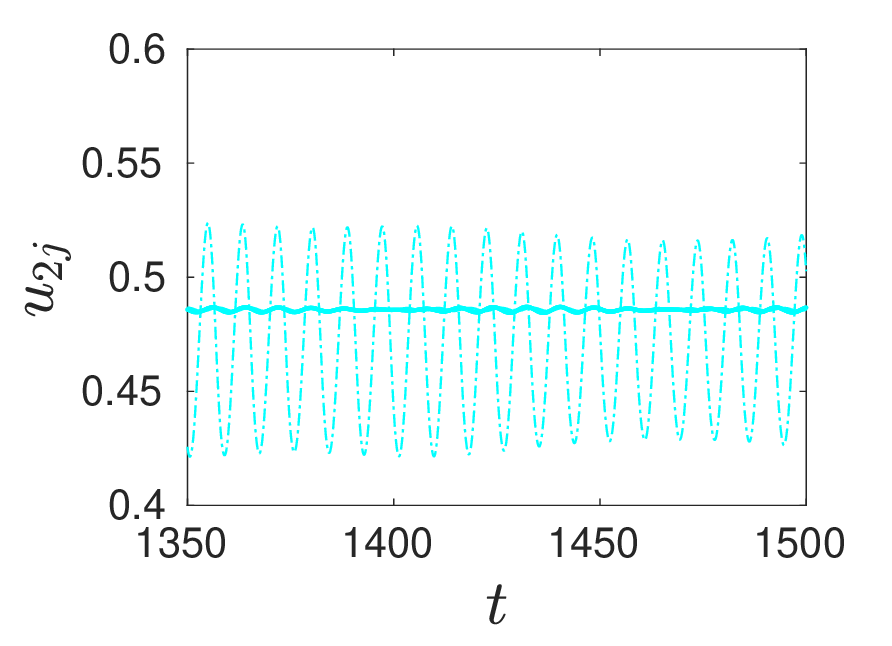}
  \caption{$D=0.5$ (rightmost ring cells)} 
        \label{fig:8_kuram_D05sig1_right}
      \end{subfigure}
  \begin{subfigure}[b]{0.48\textwidth}  
  \includegraphics[width=\textwidth,height=4.3cm]{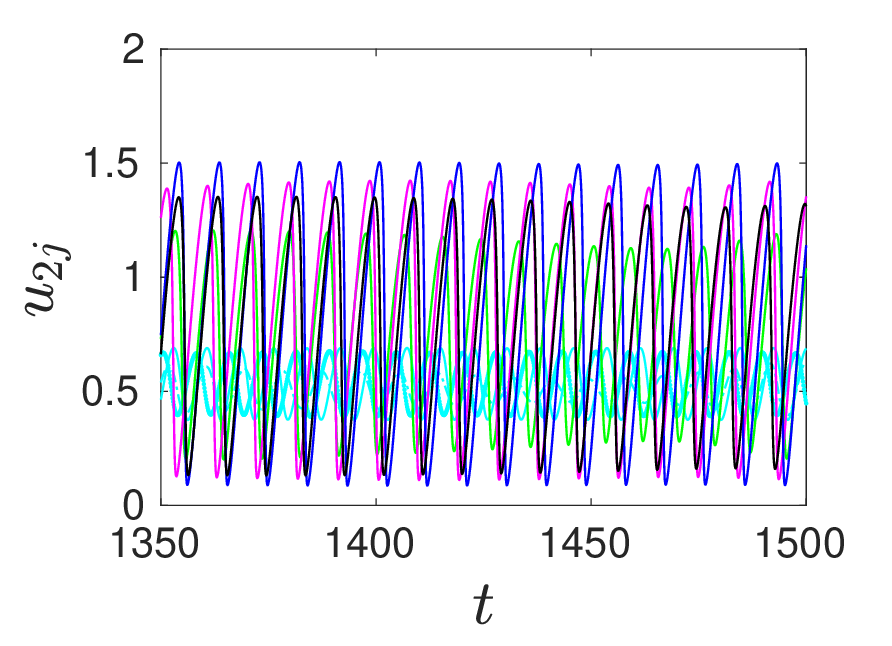}
  \caption{$D=1$ (all cells)}
          \label{fig:8_kuram_D1sig1}
\end{subfigure}
  \begin{subfigure}[b]{0.48\textwidth}  
    \includegraphics[width=\textwidth,height=4.3cm]{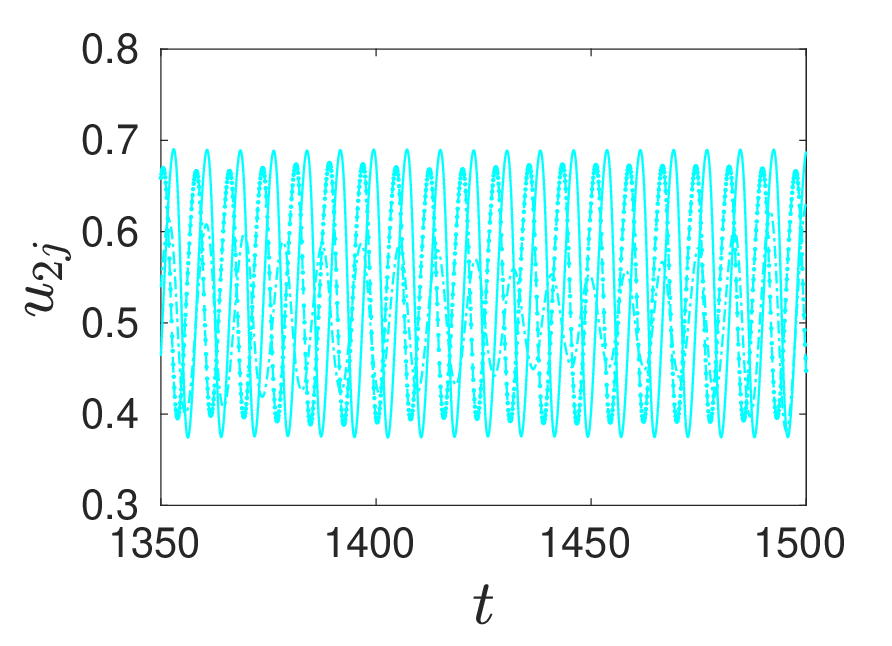}
  \caption{$D=1$ (rightmost ring cells)} 
        \label{fig:8_kuram_D1sig1_right}
      \end{subfigure}
  \begin{subfigure}[b]{0.48\textwidth}  
 \includegraphics[width=\textwidth,height=4.3cm]{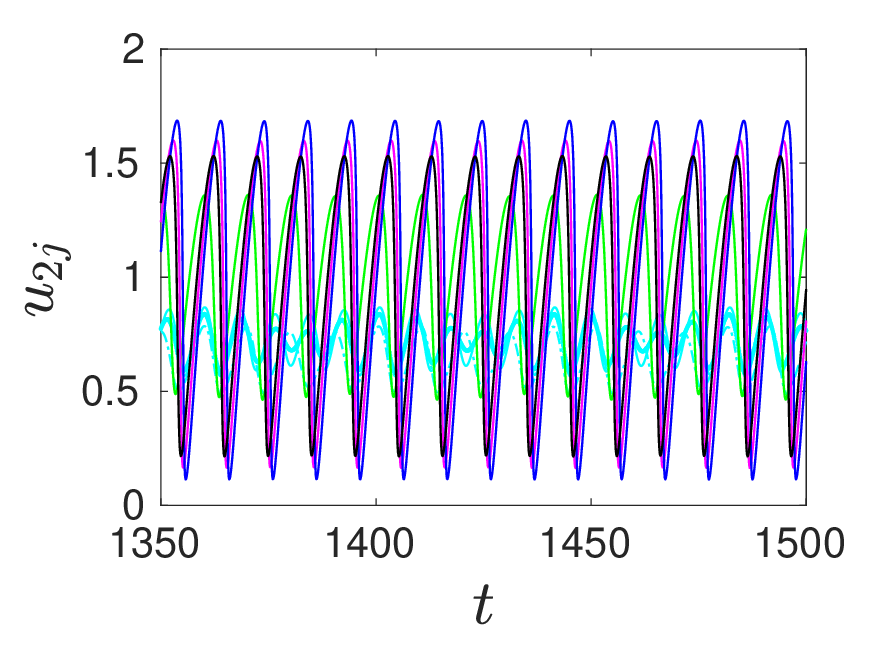}
  \caption{$D=5$ (all cells)} 
        \label{fig:8_kuram_D5sig1}
      \end{subfigure}
  \begin{subfigure}[b]{0.48\textwidth}  
    \includegraphics[width=\textwidth,height=4.3cm]{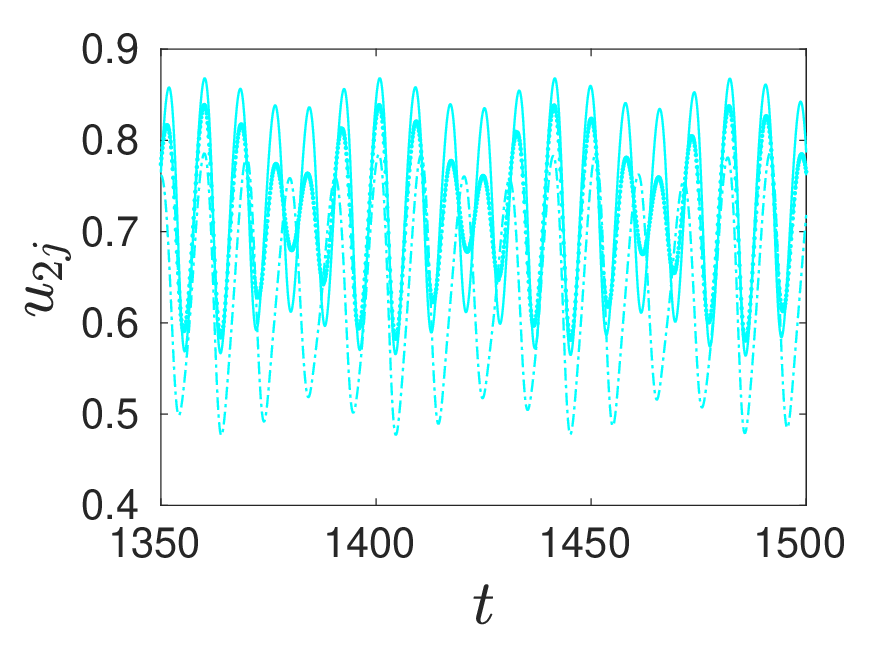}
  \caption{$D=5$ (rightmost ring cells)}
  \label{fig:8_kuram_D5sig1_right}
      \end{subfigure}
      \caption{Intracellular dynamics $u_{2j}(t)$ for the two-ring
          configuration of cells in Fig.~\ref{fig:eightcell} computed
          using the algorithm of \S \ref{sec:all_march} with
          $\Delta t=0.005$ for three values of $D$ when
          $\sigma=1$. The left panels show the dynamics in all the
          cells, identified according to the colors in
          Fig.~\ref{fig:eightcell}.  The right panels show the
          intracellular dynamics for the cells in the rightmost ring,
          where the solid cyan curve is for the rightmost cell
          $\x_1$. With a counterclockwise orientation, solid-dashed,
          dashed-dotted and dotted curves identifies the dynamics in
          the cells $\x_2$, $\x_3$ and $\x_4$, respectively.  Top row:
          For $D=0.05$, the leftmost ring cells show little phase
          coherence (left panel), and the only cell in the rightmost
          ring that exhibits oscillations is $\x_3$ (which is the cell
          closest to the leftmost ring). Middle row: For $D=1$ there
          is enhanced phase coherence in the leftmost ring cells, and all
          cells in the rightmost ring now have significant
          oscillations. Bottom row: For $D=5$, the cells on the leftmost
          ring show significant phase coherence, while the rightmost
          ring cells exhibit in-phase mixed-mode
          oscillations. Parameter values are given in the text.}
\label{fig:8cell_kuram_dynamics}
\end{figure}

For our last numerical experiment we consider the centered hexagonal
arrangement of cells shown in Fig.~\ref{fig:centerhex} that has two
shells of cells and a center cell. The generators for this finite
lattice are taken to be
$\lb_1 = \left(\left(\frac 43\right)^{1/4},0\right)^T$ and
$\lb_2= \left(\frac 43\right)^{1/4}\left(\frac 12, \frac{\sqrt 3}{2}
\right)^T$, so that the area $|\lb_1\times\lb_2|$ of the primitive
cell is unity.

We fix $d_{1j}=0.8$ and $d_{2j}=0.2$ for all cells, i.e.~for
$j\in \lbrace{1,\ldots,19 \rbrace}$, and we set the bulk
degradation rate as $\sigma=1$.  For the center cell we take
$\alpha_1=0.5$, so that from Fig.~\ref{fig:selkov_efflux} this cell is
activated when uncoupled from the the bulk. For the cells on the inner
and outer shells we take $\alpha_j=0.7$ for
$j\in\lbrace{2,\ldots,7\rbrace}$ and $\alpha_j=0.8$ for
$j\in\lbrace{8,\ldots,19\rbrace}$, respectively. These cells are 
quiescent when uncoupled from the bulk.

  \begin{table}[htbp]
\centering
  \begin{tabular}{c|c||c|c}
        \hline
 $D$   & $Q_{ave}$ & $D$  & $Q_{ave}$ \\
        \hline
   $0.2$ & $0.171$ & $0.6$ & $0.798$\\
   $0.3$ & $0.237$ & $0.7$ & $0.935$\\
   $0.4$ & $0.498$ & $0.8$ & $0.958$ \\
   $0.5$ & $0.664$ & $0.9$ & $0.973$\\
  \hline
  \end{tabular}
  \caption{{Averaged Kuramoto order parameter $Q_{ave}$ from
      (\ref{more:kuram_ave}) for the finite hexagonal lattice
      configuration of Fig.~\ref{fig:centerhex} for various bulk
      diffusivities $D$. Uniformly random IC's of magnitude $0.1$ were
      used.  The phase coherence increases significantly with
      $D$. Parameter values as in the caption of
      Fig.~\ref{fig:hex_kuram_dynamics}.}}
  \label{hexcell:kuramoto}
  \end{table}

  The numerical results for $Q_{ave}$ in Table \ref{hexcell:kuramoto}
  show that the intracellular oscillations become increasingly
  phase-coherent as $D$ increases. In
  Fig.~\ref{fig:8cell_kuram_dynamics} we plot $u_{j2}(t)$ versus $t$
  for three values of $D$ as computed using our algorithm in \S
  \ref{sec:all_march} with $\Delta t=0.005$, and with uniformly random
  initial conditions of magnitude $0.1$. For $D=0.3$, in
  Figs.~\ref{fig:hex_kuram_D03sig1}--\ref{fig:hex_kuram_D03sig1_long}
  we observe that the cells on each ring now oscillate with different
  amplitudes and phases, and that there is little phase
  coherence. However, upon increasing $D$ to $D=0.7$ for which
  $Q_{ave}=0.935$, in
  Figs.~\ref{fig:hex_kuram_D07sig1}--\ref{fig:hex_kuram_D07sig1_long}
  we observe that cells on each of the two shells oscillate with a
  common amplitude and phase, and that there is only relatively minor
  phase differences between the intracellular oscillations on the two
  shells and the center cell. Overall, this example clearly show that
  a single activated cell at the center can trigger intracellular
  oscillations in the otherwise quiescent cells that form the inner
  and outer shells of a truncated lattice, and that the coherence of
  these oscillations increases strongly as the bulk diffusivity
  increases.

      \begin{figure}[h!tbp]
  \centering
  \begin{subfigure}[b]{0.48\textwidth}  
  \includegraphics[width=\textwidth,height=4.3cm]{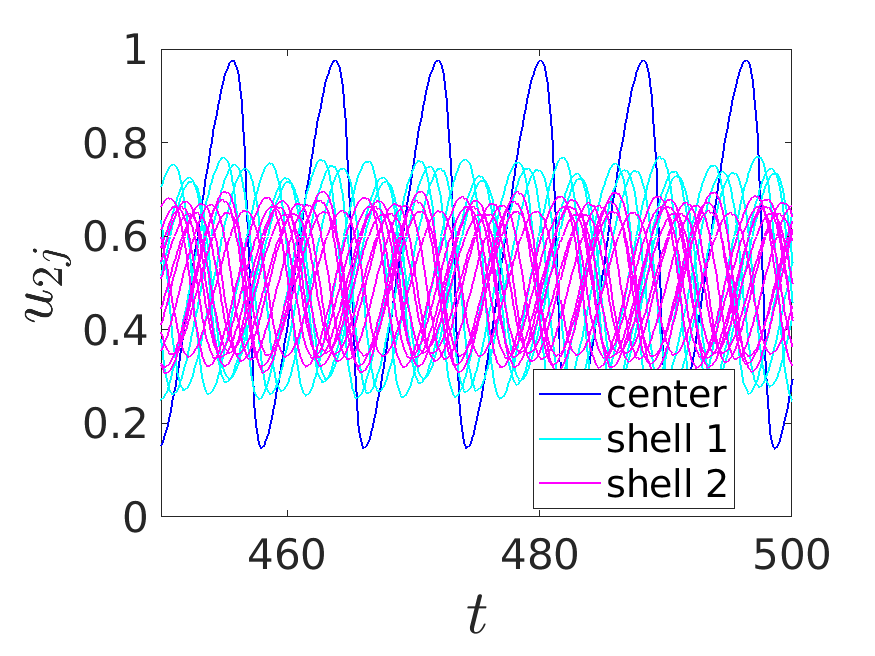}
  \caption{$D=0.3$ (short time)} 
        \label{fig:hex_kuram_D03sig1}
      \end{subfigure}
  \begin{subfigure}[b]{0.48\textwidth}  
  \includegraphics[width=\textwidth,height=4.3cm]{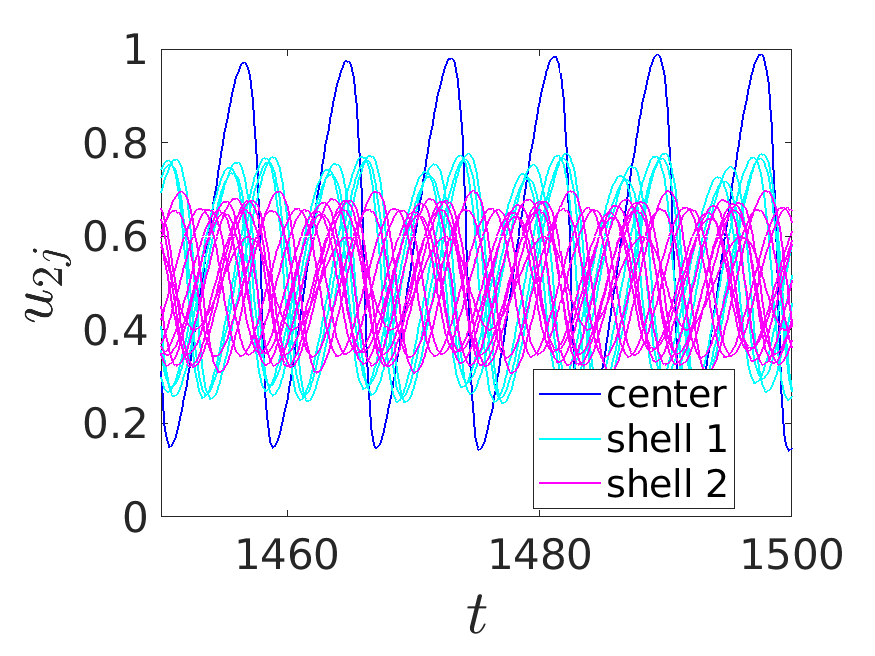}
  \caption{$D=0.3$ (long time)} 
        \label{fig:hex_kuram_D03sig1_long}
      \end{subfigure}
  \begin{subfigure}[b]{0.48\textwidth}  
  \includegraphics[width=\textwidth,height=4.3cm]{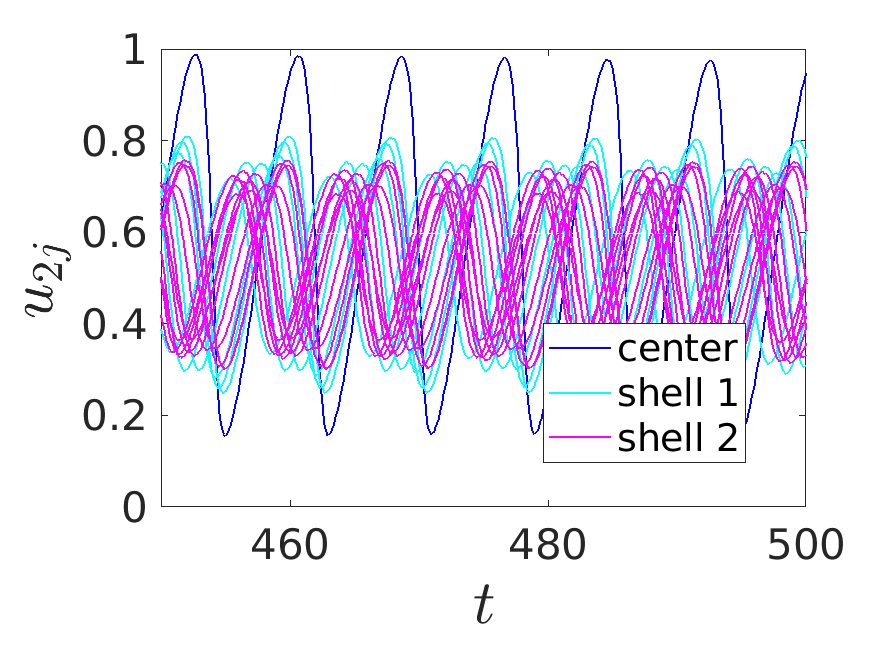}
  \caption{$D=0.5$ (short time)}
          \label{fig:hex_kuram_D05sig1}
\end{subfigure}
  \begin{subfigure}[b]{0.48\textwidth}  
    \includegraphics[width=\textwidth,height=4.3cm]{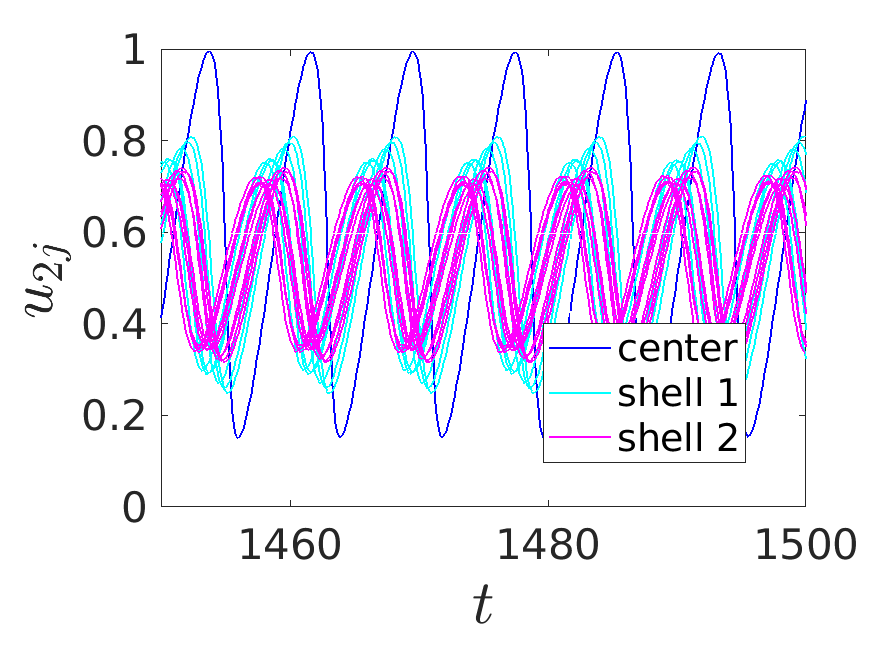}
  \caption{$D=0.5$ (long time)} 
        \label{fig:hex_kuram_D05sig1_long}
      \end{subfigure}
  \begin{subfigure}[b]{0.48\textwidth}  
 \includegraphics[width=\textwidth,height=4.3cm]{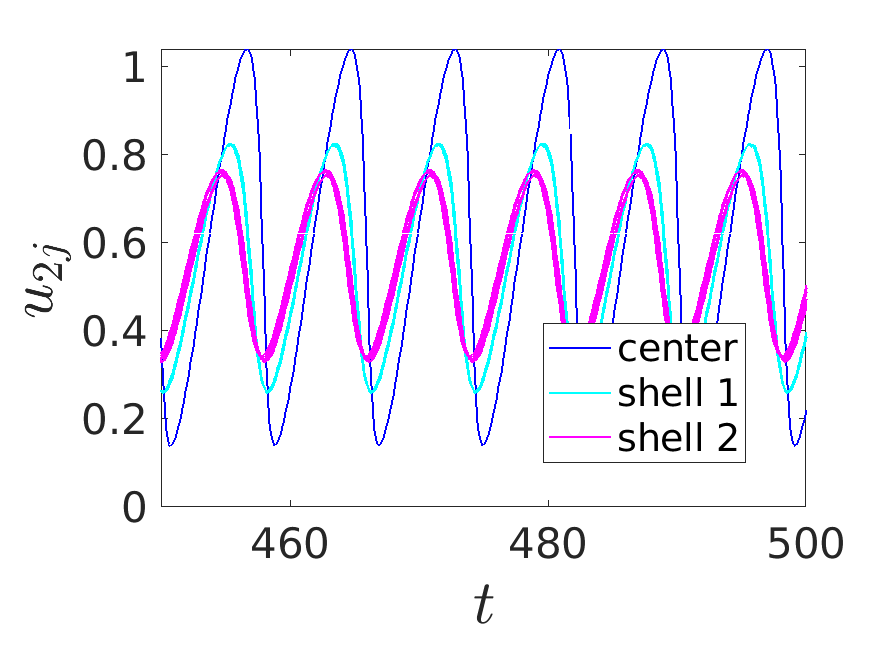}
  \caption{$D=0.7$ (short time)} 
        \label{fig:hex_kuram_D07sig1}
      \end{subfigure}
  \begin{subfigure}[b]{0.48\textwidth}  
    \includegraphics[width=\textwidth,height=4.3cm]{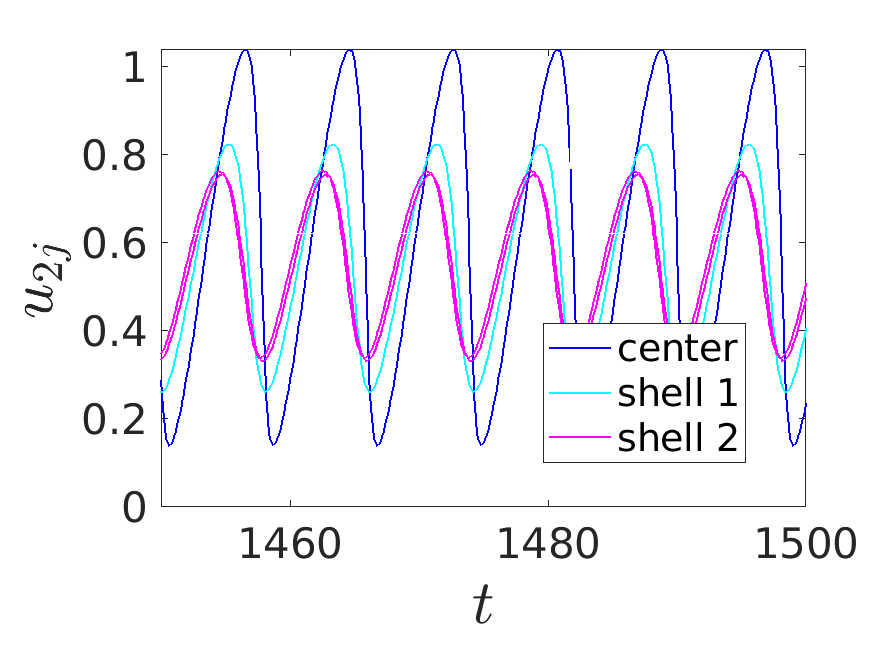}
  \caption{$D=0.7$ (long time)}
  \label{fig:hex_kuram_D07sig1_long}
      \end{subfigure}
      \caption{{Short- (left panels) and long-time (right panels)
          intracellular dynamics $u_{2j}(t)$ for the finite lattice
          configuration in Fig.~\ref{fig:centerhex} as computed using
          the algorithm of \S \ref{sec:all_march} with
          $\Delta t=0.005$ and for three values of $D$. The cell at
          the center and on the two shells are identified by the
          colours in Fig.~\ref{fig:centerhex}. Only the center cell is
          activated when uncoupled from the bulk. Uniformly random IC
          of magnitude $0.1$ around the steady-state intracellular
          values were given. Phase coherence increases substantially
          near $D=0.5$. Parameters: $\sigma=1$, $d_{1j}=0.8$,
          $d_{2j}=0.2$ for $j\in \lbrace{1,\ldots,19\rbrace}$,
          $\alpha_1=0.5$, $\alpha_j=0.7$ for
          $j\in\lbrace{2,\ldots,7\rbrace}$ (inner shell) and
          $\alpha_j=0.8$ for $j\in\lbrace{8,\ldots,19\rbrace}$ (outer
          shell).}}
\label{fig:hex_kuram_dynamics}
\end{figure}

\section{Discussion}\label{sec:discussion}

The hybrid asymptotic-numerical theory developed herein has provided
  new theoretical and computationally efficient approaches for
  studying how oscillatory dynamics associated with spatially
  segregated dynamically active ``units'' or ``cells'' are regulated
  by a PDE bulk diffusion field that is produced and absorbed by the
  entire cell population. For the cell-bulk model (\ref{DimLess_bulk})
  we extended the strong localized perturbation theory, surveyed in
  \cite{ward2018spots}, to a time-dependent setting in order to derive
  a new integro-differential ODE system that characterizes
  intracellular dynamics in a memory-dependent bulk-diffusion
  field. For this nonlocal system, a fast time-marching scheme,
  relying on the {\em sum-of-exponentials} method of
  \cite{greengard_2015} and \cite{beylkin_2005}, was developed to
  reliably numerically compute solutions over long time intervals.

  For the special case of Sel'kov reaction kinetics, we have used our
  hybrid approach to investigate the effect that the influx and efflux
  permeability parameters, the bulk degradation rate and bulk
  diffusivity, and the reaction-kinetic parameters have in both
  triggering intracellular oscillations and in mediating oscillator
  synchronization for the entire collection of cells. Quorum- and
  diffusion-sensing behavior were illustrated for various cell
  configurations. Comparisons of numerical results from our fast
  time-marching algorithm for the integro-differential ODE system with
  the numerical results for the full cell-bulk system
  (\ref{DimLess_bulk}) computed using the commercial solver FlexPDE
  \cite{flexpde2015} have shown that our hybrid asymptotic-numerical
  theory accurately reproduces, within roughly a minute of CPU time on
  a laptop, the intricate mixed-mode oscillatory dynamics that can
  occur for the full system (\ref{DimLess_bulk}) over long-time
  intervals. From a computation of the Kuramoto order parameter, we
  have exhibited an apparent phase transition to complete phase
  coherence for a hexagonal arrangement of identical cells (see
  Fig.~\ref{fig:7cell_kuram}).  These cells would not exhibit
  intracellular oscillations without any inter-cell
  coupling. Moreover, in Fig.\ref{fig:hex_kuram_dynamics} we have
  shown that a single pacemaker or signaling cell can trigger
  intracellular oscillations for all cells on a truncated lattice, and
  that these oscillations become increasingly phase-coherent as the
  bulk diffusivity increases.

Although we have only illustrated and benchmarked our theory for a
moderate number of cells, with our fast time-marching scheme, possibly
applied to Fitzhugh-Nagumo kinetics, it is viable numerically to
compute target or spiral wave patterns of oscillatory dynamics, or to
seek to identify chimera states, associated with a large $N\ge 100$
collection of cells. We emphasize that our integro-differential ODE
system can also be used for specific multi-component reaction kinetics
such as those modeling the initiation and synchronization of
glycolysis oscillations (see \cite{hauser} for a recent survey) or for
those modeling quorum-sensing transitions between bistable states that
are believed to be responsible for bioluminscence phenomena for the
marine bacterium {\em Aliivibrio fischeri}
(cf.~\cite{taga2003chemical}, \cite{qs_jump}, \cite{ridgway} see also
the references therein).

Moreover, our new theoretical and computational approaches for first
deriving and then computing solutions to integro-differential systems
of ODE's for cell-bulk models with one diffusing species can readily
be extended to $\R^3$ and to the case of finite domains
$\Omega\in \R^d$ with $d\in\lbrace{1,2,3\rbrace}$. In a 1-D setting,
steady-states and their linear stability properties for cell-bulk
models consisting of discrete oscillators located at either the domain
boundaries or at interior points in the domain, and which are coupled
by a bulk diffusion field, have been studied in \cite{gou2015},
\cite{gou2017} and \cite{paquin_memb} (see also the references
therein). Our approach is also readily extended to cell-bulk models
with two bulk diffusing species, such as in \cite{turing2d} and
\cite{turing1d}, where the focus was mainly on describing bifurcation
and pattern-formation properties of the model. Cell-bulk models,
allowing for a spatially uniform bulk drift velocity, such as in
\cite{Mueller2014uecke}, should also be tractable to analyze with our
approach.  The main two features that we require to derive the fast
time-marching algorithm are an exact analytical representation of the
Laplace-transformed bulk diffusion field and that this representation
has singularities only along the negative real axis in the Laplace
transform space. These properties are inherent for Green's functions
for heat-type equations with bulk degradation in either bounded or
unbounded domains.

Finally, we discuss three numerical challenges that warrant further
investigation. Firstly, the primary numerical errors in our
time-marching scheme of \S \ref{sec:all_march} are due to the ETD2
discretization of the Duhamel integrals as well as the Forward-Euler
discretization in (\ref{alg:step_i}) for the dynamics. By increasing
the truncation order of this aspect of our algorithm it should be
possible to take larger time steps than we have done while still
preserving a comparable level of accuracy.  Secondly, the computation
of roots of $\det{{\mathcal M}}(\lambda)=0$ for the GCEP in
(\ref{5:gcep_2}), which determines the HB boundaries and oscillation
frequencies near the steady-state, becomes highly challenging when the
number of cells becomes large. As a result, efficient numerical
methods to compute path-dependent solutions to nonlinear matrix
eigenvalue problems that are complex-symmetric, but non-Hermitian, are
required. Finally, the computation of the winding number in
(\ref{wind:form}) for an arbitrary cell configuration with a large
number of cells is highly challenging as there is typically no
analytical expression for $\det{{\mathcal M}}(\lambda)$. To overcome
this difficulty, a Takagi factorization (cf.~\cite{horn}) of the
complex symmetric matrix ${\mathcal M}(\lambda)$ that smoothly depends
on a bifurcation parameter may provide a numerically well-conditioned
approach to compute the winding number (cf.~\cite{dieci}).

\begin{appendices}
\renewcommand{\theequation}{\Alph{section}.\arabic{equation}}
\setcounter{equation}{0}

\section{Non-dimensionalization of the cell-bulk
  model}\label{app:nondim}

We summarize the non-dimensionalization of
(\ref{Dim_bulk}) into the form (\ref{DimLess_bulk}).
Labeling $\left[\gamma\right]$ as the dimensions of the
variable $\gamma$, we have
\begin{equation}\label{app:dimensions}
  \begin{aligned}
   \left[{\mathcal U}\right] &= \frac{\mbox{moles}}{\mbox{(length)}^2}
   \,, \quad \left[\v{\mu}_j\right] = \mbox{moles} \,, \quad
   \left[\mu_c \right] = \mbox{moles} \,, \quad \left[D_B \right]
   =\frac{\mbox{(length)}^2}{\mbox{time}} \,, \quad\\ \left[k_B \right]
   &=\left[k_R\right] = \frac{1}{\mbox{time}} \,, \quad \left[\beta_1
     \right] =\frac{\mbox{length}}{\mbox{time}} \,,
   \quad\ \left[\beta_2 \right]
   =\frac{1}{\mbox{length}\times\mbox{time}} \,,
\end{aligned}
\end{equation}
for the dimensions of the terms in (\ref{Dim_bulk}).  We introduce the
dimensionless variables $t$, $\v{x}$, $U$ and $\v{u}$ as
\begin{equation}\label{form:non_dim}
  t = k_R T \,, \quad  \v{x} = {\v{X}/L} \,, \quad
  U= L^2 {{\mathcal U}/\mu_c} \,, \quad \v{u}_j = 
  {\v{\mu}_j/\mu_c} \,, \quad  \mbox{where} \quad L\equiv
  \mbox{min}_{i\neq j}|\v{x}_i-\v{x}_j| \,.
\end{equation}
We define the small parameter $\eps={R_0/L}$ where $R_0$ is the common
cell radius. By noting that ${\rm d}S_{\v{X}}=L{\rm d}S_{\v{x}}$, we
readily obtain (\ref{DimLess_bulk}) where the dimensionless parameters
are defined in (\ref{dim:param}).

\section{Long-Time Dynamics of the Transient
  solution}\label{app:trans}

We follow \cite{LS} and summarize some results for inverting the
Laplace transform in (\ref{3:small_s}). We first define
$\hat{H}(s)\equiv {1/\left[s\log{s}\right]}$, with a branch cut
along the negative real axis. To determine
$H(\tau) = {\mathcal L}^{-1}\left[\hat{H}\right]$ we integrate over a
Bromwich contour
\begin{equation*}
 H(\tau)= \frac{1}{2\pi i} \int_{c-i\infty}^{c+i\infty} \frac{e^{s\tau}}
  {s\log{s}} \, ds\,, \quad \mbox{where} \quad c>1\,.
\end{equation*}
By calculating the residue at the simple pole $s=1$ and from
deforming the integration path we obtain
\begin{equation*}
  H(\tau)=e^{\tau} + \frac{1}{2\pi i} \int_{0}^{\infty} \left(
    \frac{e^{-\xi\tau}}{\xi\left(\log\xi+i\pi\right)} -
    \frac{e^{-\xi\tau}}{\xi\left(\log\xi-i\pi\right)}\right) \, d\xi \,,
\end{equation*}
where the two integrand terms represent contributions from the top and
the bottom of the branch cut along the negative real axis. In this
way, we obtain the identity
\begin{equation}\label{appb:int}
  N(\tau) \equiv \int_{0}^{\infty} \frac{e^{-\tau\xi}}{\xi\left[\pi^2+
      (\log \xi)^2\right]} \, d\xi = e^{\tau} - H(\tau) \,,
\end{equation}
where $N(\tau)$ is Ramanujan's integral. By taking the Laplace transform
of (\ref{appb:int}) we obtain the identity
\begin{equation*}
  \hat{N}(s) = \frac{1}{s-1} - \frac{1}{s\log{s}}\,.
\end{equation*}
Upon using $\log{s}\sim (s-1)- {(s-1)^2/2}$ as $s\to 1$, we get
$\hat{N}(s)\sim {1/2}$ as $s\to 1$, so that $s=1$ is a removable
singularity for $\hat{N}$. Therefore, $\hat{N}$ is analytic in
$\mbox{Re}(s)>0$ and satisfies
$\hat{N}(s)\sim -\left[s\log{s}\right]^{-1} + {\mathcal O}(1)$ as
$s\to 0$. By the scaling relation
${\mathcal L}\left[N({\tau/c})\right]=c \hat{N}(cs)$ for the Laplace
transform with $c=\kappa_{0j}>0$, we obtain (\ref{3:N1}).

\section{Improved approximation for
  $B_j(\Delta t)$}\label{app:bdeltat}
{For $u_{1j}(0)\neq 0$, we derive the approximation $B_j(\Delta t)$ in
(\ref{alg:b1}) that is an improvement of that given in (\ref{mat:bj}).

By substituting $D_{j}(\Delta t)$ and $C_{jk}(\Delta t)$ from
(\ref{dhj:march_1}) and (\ref{chj:march_1}) into
(\ref{duh:r2}), we obtain 
\begin{equation}\label{app:imp_1}
  \left[E_1(\sigma \Delta t) - \eta_j\right] B_j(\Delta t) =
  \gamma_j u_{1j}(\Delta t) - I  + \sum_{k\ne j}^{N} B_{k}(\Delta t)
  E_{1}\left({a_{jk}^2/\Delta t}\right)\,, \quad
  \mbox{for} \quad j\in\lbrace{1,\ldots,N\rbrace}\,,
\end{equation}
where
\begin{equation}\label{app:Ieval}
    I \equiv \int_{0}^{\Delta t} \left( B_{j}(\Delta t) - B_{j}(\Delta t-z)\right)
  \frac{ e^{-\sigma z}}{z} \, dz \,.
\end{equation}
We let $\Delta t\ll 1$, and use $E_1(\sigma \Delta t)=-
\log\left(\sigma \Delta t\right)-\gamma_e+{\mathcal O}(\Delta t)$, and
neglect the last term on the right-hand side of (\ref{app:imp_1})
since $E_1(z)$ is exponentially small for $z\gg 1$. This yields that
\begin{equation}\label{app:imp_2}
  B_j(\Delta t) = -\frac{u_{1j}(\Delta t) \gamma_j}{\log\left({\Delta t/(
        \kappa_j e^{-\gamma_e})}\right)} +
     \frac{I}{\log\left({\Delta t/(
        \kappa_j e^{-\gamma_e})}\right)} \,,
\end{equation}
where we used $\sigma e^{\eta_j}={1/\kappa_j}$. Upon substituting
the two-term expansion for $B_j(\Delta t)$, given by
\begin{equation}\label{app:imp_3}
  B_j(\Delta t) = \frac{c_j}{\log\left({\Delta t/(
        \kappa_j e^{-\gamma_e})}\right)} +
     \frac{d_j}{\left[\log\left({\Delta t/(
        \kappa_j e^{-\gamma_e})}\right)\right]^2} \,,
\end{equation}
with $c_j=-u_{1j}(\Delta t)\gamma_j$ into (\ref{app:imp_2}) and
(\ref{app:Ieval}), we determine $d_j$ as
\begin{equation}
  d_j = c_j \int_{0}^{\Delta t} \frac{ \log\left( 1-{z/\Delta t}\right)
    e^{-\sigma z}}{z \log\left( {(\Delta t-z)/(\kappa_je^{-\gamma_e})}
    \right)} \, dz = c_j \int_{0}^{1} \frac{\log(1-\xi) e^{-\sigma \xi \Delta
      t}}{\xi \log\left(\frac{\Delta t}{\kappa_j e^{-\gamma_e}}(1-\xi)\right)}
  \, d\xi \,.
\end{equation}
By using $e^{-\sigma \xi\Delta t}\approx 1$ and
$|\log\left({\Delta t/(\kappa_j e^{-\gamma_e})}\right)|\gg 1$, we obtain
for $\Delta t\ll 1$ that
\begin{equation}\label{app:dj}
  d_j \sim \frac{c_j}{\log\left(\frac{\Delta t}{\kappa_j e^{-\gamma_e}}
    \right)} \int_{0}^{1} \frac{\log(1-\xi)}{\xi} \, d\xi\,.
\end{equation}
Using $\int_{0}^{1} \xi^{-1}\log(1-\xi)\, d\xi=-{\pi^2/6}$, we
substitute (\ref{app:dj}) and $c_j=-u_{1j}(\Delta t)\gamma_j$
into (\ref{app:imp_3}) to get (\ref{alg:b1}).
}

\section{An artificial boundary condition for full
  PDE simulations}\label{app:artificial}

{To numerically solve the cell-bulk system (\ref{DimLess_bulk})
using FlexPDE 6.50 \cite{flexpde2015}, we must truncate $\R^2$ by both
choosing and formulating a boundary condition on a confining boundary
$\partial\Omega$ that encloses all of the cells.

For the full PDE computations shown in Fig.~\ref{fig:D0p75Sig1over7}
and Fig.~\ref{fig:D3sig0.5_comp} we chose the artificial boundary
$\partial \Omega$ to have approximately the same shape as the
asymptotic levels sets of the steady-state solution away from the two
cells. Assuming that the solution is approximately radially symmetric
due to the cell arrangement with center of mass at $\v{x}_c=\v{0}$, we
took $\Omega = \{\v{x}\in\R^2 \, | \, |\v{x}| \leq R_a\}$ with
$R_{a}\gg 1$. The steady-state far away from the asymptotic support of
the two cells satisfies
$\Delta U_s - \omega^2 U_s \approx C_1\delta(\v{x})$, for some $C_1$
where $\omega \equiv \sqrt{\sigma/D}$. This yields the asymptotic
behavior $U_s = C_2 K_0(\omega |\v{x}|)$ where
$C_2=-{C_1/(2\pi)}$. Upon using
$K_{0}(z)\sim \sqrt{\pi/2}z^{-1/2} e^{-z} \left(1+{\mathcal
    O}(z^{-1})\right)$ for $z\gg 1$, we let $r=|\v{x}|$, and readily
calculate
\begin{equation*}
  \frac{d}{dr} U_s  \sim -\left(\omega + \frac{1}{2r}\right) U_{s}
  + \mathcal{O}\left(\frac{1}{r^2}\right) U_s \,, \qquad \mbox{for}
  \quad r\gg 1\,.
\end{equation*}
Neglecting the ${\mathcal O}(r^{-2})$ term, the artificial boundary
condition that we imposed in our FlexPDE simulations, which
asymptotically approximates the steady-state on $\R^2$, was
\begin{equation}\label{app:art_eq}
  \partial_r U(r, t) = -\left(\omega + \frac{1}{2r}\right)\,
  U(r,t) \quad \text{on} \quad r = R_a\,,
\end{equation}
where $\omega=\sqrt{\sigma/D}$.  We chose $R_a=6$ for the results
shown in Fig.~\ref{fig:D0p75Sig1over7} and Fig.~\ref{fig:D3sig0.5_comp}.

More sophisticated, memory-dependent, artificial boundary conditions
for parabolic PDE problems have been formulated in
\cite{artificialBC}. However, owing to the very close agreement
between our asymptotic and FlexPDE results using (\ref{app:art_eq}),
we did not find it essential to implement the approach in
\cite{artificialBC}.}

\end{appendices}

\section*{Acknowledgments}
Merlin Pelz was partially supported by a UBC Four-Year Doctoral
Fellowship. Michael Ward gratefully acknowledges the
support of the NSERC Discovery Grant Program.

\bibliographystyle{plain}
\bibliography{diffkuramoscillref}

\end{document}